\documentclass[10pt,journal]{IEEEtran}
\usepackage{amsmath,amsfonts,amssymb}
\usepackage{algorithmic}
\usepackage{algorithm}
\usepackage{array,soul}
\usepackage[caption=false,font=normalsize,labelfont=sf,textfont=sf]{subfig}
\usepackage{textcomp}
\usepackage{stfloats,colortbl,xcolor,graphicx}
\usepackage{url,hyperref}
\usepackage{verbatim}
\usepackage{graphicx}
\usepackage{cite}
\usepackage{rsfso}
\usepackage{scalerel}
\usepackage{enumerate}
\newtheorem{theorem}{Theorem}

\newtheorem{definition}{Definition}
\definecolor{Gray}{gray}{0.93}

\definecolor{Gray}{gray}{0.85}

\usepackage{etoolbox}
\makeatletter
\patchcmd{\@makecaption}
  {\scshape}
  {}
  {}
  {}
\makeatother

\begin{document}
\title{Correction of Pooling Matrix Mis-specifications in Compressed Sensing Based Group Testing}

\author{Shuvayan Banerjee, Radhendushka Srivastava, James Saunderson and Ajit Rajwade ~\IEEEmembership{Senior Member,~IEEE}
\thanks{Shuvayan Banerjee is with the IITB-Monash Research Academy. Radhendushka Srivastava is with the Department of Mathematics, IITB. James Saunderson is with the Department of Electrical and Computer Systems Engineering, Monash University. Ajit Rajwade is with the Department of Computer Science and Engineering, IITB.}}


\maketitle
\begin{abstract}
Compressed sensing, which involves the reconstruction of sparse signals from an under-determined linear system, has been recently used to solve problems in group testing. In a public health context, group testing aims to determine the health status values of $p$ subjects from $n\ll p$ pooled tests, where a pool is defined as a mixture of small, equal-volume portions of the samples of a subset of subjects. This approach saves on the number of tests administered in pandemics or other resource-constrained scenarios. In practical group testing in time-constrained situations, a technician can inadvertently make a small number of errors during pool preparation, which leads to errors in the pooling matrix, which we term `model mismatch errors' (MMEs). This poses difficulties while determining health status values of the participating subjects from the results on $n \ll p$ pooled tests. In this paper, we present an 
algorithm to \emph{correct} the MMEs in the pooled tests directly from the pooled results and the available (inaccurate) pooling matrix. Our approach then reconstructs the signal vector from the corrected pooling matrix, in order to determine the health status of the subjects. We further provide theoretical guarantees for the correction of the MMEs and the reconstruction error from the corrected pooling matrix. We also provide several supporting numerical results. 
\end{abstract}

\section{Introduction}

\subsection{Group Testing }
Group testing (GT) originated with the seminal work of Dorfman~\cite{Dorfman1943} as an efficient strategy to identify a small number of defective (infected or diseased) individuals in a large population using far fewer tests than individual testing. In the classical formulation (e.g.,~\cite{Atia2012,Chan2011,Dorfman1943}), a population of size $p$ is represented by a binary vector $\boldsymbol{\beta}^*$, where each entry indicates whether the corresponding individual is defective or healthy. Testing is performed on pooled subsets of samples, whose compositions are encoded by a binary pooling matrix $\boldsymbol{B} \in \{0,1\}^{n \times p}$. Each row $\boldsymbol{b}_{i,.}$ of the pooling matrix specifies the set of samples (to be precise, small equal-volume portions of each sample from the set) included in pool $i$, and the outcome vector $\boldsymbol{z} \in \{0,1\}^n$ records whether each pool tests positive or negative.

The forward model in this classical setting is governed by Boolean operations: a pool tests positive if and only if at least one defective sample is present. This is commonly expressed as $\boldsymbol{z} = \mathfrak{N}(\boldsymbol{B\beta}^*),$
where $\mathfrak{N}$ denotes a noise process that flips bits in the Boolean outcomes. This Boolean formulation highlights the inherently nonlinear and combinatorial structure of traditional group testing.

\noindent\textbf{Adaptive and Non-Adaptive GT:}
GT algorithms can broadly be classified into adaptive and non-adaptive schemes. Adaptive algorithms~\cite{Dorfman1943,Heiderzadeh2021,Hwang1972} proceed in multiple stages, using outcomes of earlier stages to design subsequent pools. Non-adaptive algorithms~\cite{Shental2020,Ghosh2021,gilbert2008group,Bharadwaja2022}, on the other hand, fix all pools in advance and conduct all tests simultaneously. While adaptive designs can offer stronger theoretical guarantees, non-adaptive strategies are highly attractive in time-critical or resource-limited scenarios, such as large-scale population screening during pandemics~\cite{Larremore2021,Ghosh2021}. In this work, we focus exclusively on non-adaptive designs.

\noindent\textbf{Compressed Sensing:}
Compressed sensing (CS) is a foundational framework in signal processing that enables the recovery of sparse, high-dimensional signals from a relatively small number of linear measurements. Formally, the aim is to reconstruct a signal $\boldsymbol{\beta}^* \in \mathbb{R}^p$ from an underdetermined linear system
\begin{equation}\label{eq:forward_model_gt}
    \boldsymbol{z} = \boldsymbol{B\beta}^* + \boldsymbol{\tilde{\eta}},
\end{equation}
where $\boldsymbol{B} \in \mathbb{R}^{n \times p}$ is the sensing matrix, $\boldsymbol{z} \in \mathbb{R}^n$ is the measurement vector, and $n \ll p$. Accurate recovery is possible if $\boldsymbol{\beta}^*$ is sparse and $\boldsymbol{B}$ satisfies structural conditions such as the restricted isometry property or the restricted null space property~\cite{Candes2008}.

\noindent\textbf{Compressed Sensing Based Group Testing:}
Recent developments in diagnostic technologies have motivated the transition from binary pooled results and binary health status values to real-valued pooled results and real-valued health status values which quantify information about infection/defect \emph{levels}. In such quantitative GT systems~\cite{saeedi2022group,Heiderzadeh2021,Ghosh2021,gilbert2008group}, each individual is associated with a nonnegative real value (e.g., viral load), and each pooled measurement equals the aggregate viral load of the individuals included in the pool. This replaces Boolean arithmetic with a linear model over the reals, bringing GT into close alignment with the framework of compressed sensing (CS).

In real-valued quantitative GT, the unknown vector $\boldsymbol{\beta}^* \in \mathbb{R}^p$ contains the viral loads or infection levels of the $p$ individuals, and the pooling matrix $\boldsymbol{B} \in \{0,1\}^{n \times p}$ acts as a binary sensing matrix. The observation model becomes $\boldsymbol{z} = \boldsymbol{B\beta}^* + \boldsymbol{\tilde{\eta}},$
where $\boldsymbol{z} \in \mathbb{R}^n$ is the vector of infection levels in each pool and $\boldsymbol{\tilde{\eta}}$ denotes additive noise. The typically low prevalence of infections naturally implies that $\boldsymbol{\beta}^*$ is sparse, a structure that can be exploited by compressed sensing (CS) techniques for reliable recovery~\cite{Ghosh2021,Bharadwaja2022}. This formulation has led to a growing body of work leveraging CS for reconstructing real-valued infection levels from pooled tests~\cite{gilbert2008group,Ghosh2021,saeedi2022group,Heiderzadeh2021,bah2024compressed}. The sparsity of $\boldsymbol{\beta^*}$ and the choice of a binary sensing matrix $\boldsymbol{B}$ both align well with CS assumptions. The CS formulation enables principled reconstruction algorithms for estimating not only which individuals are infected but also their infection \emph{levels}. In the context of biological testing modalities such as real-time polymerase chain reaction (RT-PCR) or digital drop PCR (DDPCR), the noise in $\boldsymbol{z}$ respectively follows a multiplicative lognormal \cite{Ghosh2021} or Poisson model \cite{majumdar2017poisson}. In both cases, these noise models can be expressed in an additive signal-dependent form. The algorithms presented in this work apply to both these models, as will be discussed later in the paper.

\subsection{Errors in Group Membership Specifications}\label{sec:GMSE}
Most GT or CS algorithms assume that the pooling/sensing matrix $\boldsymbol{B}$ is known accurately. However a technician may make errors while implementing the pooling procedure~\cite{Fenichel2021, Zabeti2021, Grobe2021}.

\noindent \textbf{Permutation Errors:} If the results of two different pools are mistakenly exchanged (swapped) by the technician, it is termed a permutation error. Here, the technician mistakenly swaps the results of two groups. For example, the technician unknowingly records $z_{i_1} = \boldsymbol{b}_{i_2,.} \boldsymbol{\beta}^* + \tilde{\eta}_{i_2}$ (where $\boldsymbol{b}_{i_2,.}$ is the $i_2$th row of $\boldsymbol{B}$) and $z_{i_2} = \boldsymbol{b}_{i_1,.} \boldsymbol{\beta}^* + \tilde{\eta}_{i_1}$, although the original intention was to record $z_{i_1} = \boldsymbol{b}_{i_1,.} \boldsymbol{\beta}^* + \tilde{\eta}_{i_1}$ and $z_{i_2} = \boldsymbol{b}_{i_2,.} \boldsymbol{\beta}^* + \tilde{\eta}_{i_2}$.

\noindent \textbf{Bitflips:} Consider the case where due to errors in mixing of the samples, the pools are generated using an \emph{unknown} matrix $\widetilde{\boldsymbol{B}}$ instead of the pre-specified and \emph{known} matrix $\boldsymbol{B}$. The elements of matrices $\widetilde{\boldsymbol{B}}$ and $\boldsymbol{B}$ are equal everywhere except for the mis-specified samples in each pool. 
We refer to these errors in group membership specifications as \textbf{bit-flips}. For example, suppose that the $i^{\text{th}}$ pool is specified to consist of samples $j_1,j_2,j_3 \in [p]$. But due to errors during pool creation, the $i^{\text{th}}$ pool is generated using samples $j_1,j_2,j_3'$ where $j_3 \ne j_3'$. In this specific instance, $b_{i,j_3} \ne \tilde{b}_{i,j_3} $ and $b_{i,j_3'} \ne \tilde{b}_{i,j_3'}$. We present two different types of bit-flips, though our techniques are applicable to other types as well:
\\
\noindent\textit{Single Switch Model} (\textsf{SSM}):
In this model, a bit-flipped pool contains exactly one bit-flip at a randomly chosen index. Given the $i$th pool ($i \in [n]$),  exactly one of the following two scenarios can occur given the bit-flip model \textsf{SSM}: (1) some $j$th sample that was intended to be included in the pool (as defined in $\boldsymbol{B}$) is excluded, or (2) some $j$th sample that was not intended to be part of the pool (as defined in $\boldsymbol{B}$) is included. These two cases lead to the following structure of the $i$th row of $\boldsymbol{\widetilde{B}}$, and in both cases the choice of $j \in [p]$ is random uniform: Case 1: $\tilde{b}_{ij} = 0$ but $b_{ij} = 1$, Case 2: $\tilde{b}_{ij} = 1$ but $b_{ij}=0$.
\\
\noindent\textit{Adjacent Switch Model} (\textsf{ASM}):
In this model, a bit-flipped pool contains bit-flips at two adjacent indices. Suppose the $i$th pool contains bit-flips. Then under the \textsf{ASM} scheme, either (1) the $j$th sample that was not intended to be in the pool is included and the $j'$th sample where $j' \triangleq \textrm{mod}(j+1,p)$ that was intended to be in the pool is excluded, or (2) the $j$th sample that is intended to be in the pool is excluded and the $j'$th sample where $j' \triangleq \textrm{mod}(j+1,p)$ that is not intended to be in the pool, is included at random. This leads to the following structure of the $i$th row of $\boldsymbol{\tilde{B}}$, and in both cases the choice of $j$ is uniform random:
Case 1: $\tilde{b}_{ij'} =0, \tilde{b}_{ij}=1 \ \mbox{and} \ b_{ij'} =1,b_{ij}=0$, Case 2: $\tilde{b}_{ij'} =1, \tilde{b}_{ij}=0 \ \mbox{and} \ b_{ij'} =0,b_{ij}=1$.

The forward model in the presence of MMEs (both bit-flips and permutations) is expressed as follows:
\begin{equation}\label{eq:model_non_centered}
\boldsymbol{z}=\widetilde{\boldsymbol{B}}\boldsymbol{\beta}^*+\boldsymbol{\tilde{\eta}} = \boldsymbol{B\beta^*} +\boldsymbol{\tilde{\delta}}^* + \boldsymbol{\tilde{\eta}}.
\end{equation}
Here, $\boldsymbol{\tilde{\delta}^*} \triangleq (\boldsymbol{B}-\boldsymbol{\widetilde{B}})\boldsymbol{\beta}^*$ is a large and possibly-signal dependent error which may arise due to bit-flips or permutations. We collectively term such errors in the sensing matrix \textbf{model mismatch errors} (MMEs). If the technician is fairly competent, then $\boldsymbol{\tilde{\delta}}^*$ is a sparse vector, i.e. $\|\boldsymbol{{\tilde{\delta}}}^*\|_0=r<n$. 

\subsection{Existing Methods for Detecting/ Correcting MMEs}
In general, there is very little literature on GT/CS given bit-flips in the pooling/sensing matrix. In our previous work \cite{Banerjee24}, we focused on \emph{identifying} pooled measurements with bit-flips simultaneously with robust signal estimation via a hypothesis testing mechanism, without attempting to \emph{correct} the bit-flipped measurements. 
In \cite{BanerjeeICASSP}, the hypothesis test procedure was extended to the correction of permutation errors. The method proposed in this paper is superior to the one in \cite{BanerjeeICASSP}. This is particularly true for \textsf{ASM}
or \textsf{SSM} errors, where the magnitude of the error in the measurement due to the presence of an MME may not be very large (cf. \eqref{eq:fm_delta_indiv}), leading to larger 
Type-II error (false negative) in the hypothesis testing procedure.
 
The work in \cite{Cheragchi2011} focuses on binary group testing where the pooling matrix undergoes perturbations expressed by a probabilistic model of one-to-zero bit-flips. Here again, there is no emphasis on correction of the measurements, zero-to-one bitflips are not considered, and no quantitative information is considered. There are some prior works in the CS literature which deal with perturbations in the sensing matrix \cite{Parker2011,Zhu2011,Ince2014,Aldroubi2012,Herman2010,Fosson2020,herman2010mixed}. However, the perturbations considered in these papers are expressed as a dense matrix with bounded values for every entry \cite{Zhu2011,Ince2014,Parker2011,Aldroubi2012,Herman2010,Fosson2020,herman2010mixed}. In contrast, our work specifically considers a sparse set of errors in the sensing/pooling matrix. Furthermore, the aforementioned techniques focus on robust signal estimation, without any attempt to correct the sensing matrix errors.

Handling permutation errors in GT falls under the area of `unlabeled sensing' (ULS) or `shuffled regression' \cite{unnikrishnan2018unlabeled}, which is the class of regression problems where the correspondence between the entries of the measurement vector ($\{z_i\}_{i=1}^n$) and the rows of the pooling matrix ($\{\boldsymbol{b}_{i,.}\}_{i=1}^n$) is completely or partially lost. With few exceptions such as \cite{peng2021homomorphic}, ULS has been mostly studied in the over-sampled regime where $n > p$ \cite{unnikrishnan2018unlabeled,slawski2019linear}. 
GT with some pooled results permuted is a special case of ULS, with the specialty that (\textit{i}) the measurements are compressive ($n < p$), (\textit{ii}) the signal of interest $\boldsymbol{\beta}^*$ is sparse, and (\textit{iii}) the permutations are also sparse. Sparse permutations in ULS have been considered before in works such as \cite{slawski2019linear,slawski2019sparse,akrout2025unlabeled}, via a robust subspace estimation framework applicable when the linear measurements of a \emph{collection} of signals are available simultaneously, all acquired by a \emph{common} sensing matrix and subject to a \emph{common} permutation matrix. Thus the problem solved in these works is applicable to a different setting than ours, since we here consider the pooled measurements of a \emph{single} sparse signal. Moreover, the model of a \emph{common} permutation matrix for different signals is not applicable to GT, as technicians will typically make different errors while preparing pools for different sets of $p$ subjects. Compressive recovery of sparse signals with sparse permutations is considered in \cite{sresthunlabelled2025}, but with a small number of perfectly known correspondences available as side-information. Such side-information is not available in the GT problem at hand. For the problem of binary GT with sparse permutation errors, an integer linear programming approach is presented in \cite{Zabeti2021}. In contrast, we present a more efficient convex optimization approach that is applicable in the real-valued setting as opposed to the binary setting.

\subsection{Contributions of this paper}
In this work, we present a technique to \emph{correct} different types of MMEs in the pooling matrix. In other words, we design a technique to estimate $\widetilde{\boldsymbol{B}}$ directly from $\boldsymbol{z}, \boldsymbol{B}$. This technique requires a step to identify pooled measurements containing MMEs. This step is accomplished via a hypothesis testing procedure based on the concept of debiasing a robust form of the \textsc{Lasso} estimator, based on our previous work in \cite{Banerjee24}, although other techniques can also be employed for MME identification. The correction technique is based on the residuals of the Robust \textsc{Lasso} estimator. 

We also provide theoretical guarantees to show that our method reduces reconstruction errors following MME correction. Additionally, we present empirical results showing improved signal reconstruction outcomes after MME correction across different models for MMEs: bit-flips of different types, as well as permutation errors. We further propose a variant that performs the correction via multiple estimation stages to ensure thorough correction of MMEs, for which we provide intuitive potential functions as stopping criteria. We empirically show that the value of these potential functions reduces with each stage of correction.

The component of our approach that performs MME detection requires a carefully constructed `debiasing matrix' corresponding to the available pooling matrix $\boldsymbol{B}$. The debiasing matrix can be interpreted as an approximate form of the inverse of the covariance matrix of the rows of the pooling matrix. The debiasing matrix is defined as the solution of an optimization problem that is, in general, expensive to solve. We show, however, that in our setting the exact solution of this optimization problem can be obtained in closed form with inexpensive computations. This is particularly important to speed up the multi-stage variant of our algorithm. 

\noindent\textbf{Notations:} Throughout this paper, $\boldsymbol{I_n}$ denotes the identity matrix of size $n \times n$, and $\boldsymbol{e_i}$ denotes its $i$th column. We denote $[n] := \{1,2,\cdots,n\}$ for $n \in \mathbb{Z}_{+}$. Given a matrix $\boldsymbol{A}$, its $i^\text{th}$ row is $\boldsymbol{a_{i.}}$, its $j^\text{th}$ column is  $\boldsymbol{a_{.j}}$ and the $(i,j)^{\text{th}}$ element is  $a_{ij}$. For any vector $\boldsymbol{z} \in \mathbb{R}^n$ and index set $S \subseteq [n]$, we define $\boldsymbol{z}_S \in \mathbb{R}^n$ such that $\forall i \in S, (z_S)_i = z_i$ and $\forall i \notin S, (z_S)_i = 0$. $S^c$ denotes the complement of set $S$ and $|S|$ is the cardinality of $S$. We define the entrywise $l_\infty$ norm of a matrix $\boldsymbol{A}$ as $|\boldsymbol{A}|_\infty \triangleq \underset{i,j}{\max} |a_{ij}|$. 

\section{Techniques for Correction of MMEs in Compressed Sensing based Group Testing}
\subsection{Problem Setup}
We describe the model setup employed in this work. 
Let the signal vector (or defect status vector) be $\boldsymbol{\beta}^* \in \mathbb{R}^p$, and consider a pooling matrix $\boldsymbol{B} \in \mathbb{R}^{n \times p}$ with i.i.d. Bernoulli entries, $\boldsymbol{B}_{ij} \sim \text{Bernoulli}(\theta)$ for $\theta \in (0,1)$. Such random Bernoulli designs have been extensively studied in GT \cite{bharadwaja2025probably} as well as in CS \cite{duarte2008single}; in particular, $\theta = 0.5$ has been shown to be optimal in some CS applications~\cite{hunt2018data}. 

The signal $\boldsymbol{\beta}^*$ is assumed sparse, with at most $s \ll p$ non-zero distinct entries. The additive noise $\boldsymbol{\tilde{\eta}} \in \mathbb{R}^n$ in~\eqref{eq:forward_model_gt} is modeled as i.i.d. $\mathcal{N}(0,\tilde{\sigma}^2)$ with known variance $\tilde{\sigma}^2$. If $\tilde{\sigma}$ is unknown, methods to estimate $\tilde{\sigma}$ can be incorporated in our framework; in fact, we empirically show that the correction algorithms in this paper work well even for the multiplicative lognormal noise model of RT-PCR. The \textsc{Lasso} estimator for $\boldsymbol{\beta}^*$ is
\begin{equation}
\boldsymbol{\hat{\beta}} = \arg \min_{\boldsymbol{\beta}} \frac{1}{2n}\|\boldsymbol{z}-\boldsymbol{B\beta}\|^2_2 + \lambda \|\boldsymbol{\beta}\|_1,
\label{eq:lasso_z}
\end{equation}
which achieves consistency under appropriate $\lambda>0$ when $\boldsymbol{B}$ satisfies the Restricted Eigenvalue Condition (REC)~\cite[Chapter 11]{THW2015}.\footnote{REC: For $\psi \geq 1$ and $S \subseteq [p]$, define $C_\psi(S) \triangleq \{\boldsymbol{\Delta} \in \mathbb{R}^n : \|\boldsymbol{\Delta}_{S^c}\|_1 \leq \psi \|\boldsymbol{\Delta}_S\|_1\}$. $\boldsymbol{B}$ satisfies REC if $\exists \gamma>0$ such that $\tfrac{1}{n}\|\boldsymbol{B\Delta}\|_2^2 \geq \gamma \|\boldsymbol{\Delta}\|_2^2$ for all $\boldsymbol{\Delta}\in C_\psi(S)$. The parameter $\gamma$ is the restricted eigenvalue constant. $\boldsymbol{\Delta}$ is the error vector, i.e. the difference between the true signal vector and its estimate via (say) the LASSO.} 
We will now describe the forward model in the presence of MMEs.

\noindent\textbf{MME formulation on the Pooling Matrix $\boldsymbol{B}$:}
We consider mis-specified pooling arising from model mismatch errors (MMEs). 
Let $\boldsymbol{\widetilde{B}}$ denote the erroneous and unknown pooling matrix 
actually used during pooling, where technician mistakes introduce deviations 
from the intended design. We represent this as $\boldsymbol{\widetilde{B}} = \boldsymbol{B} + \boldsymbol{\Delta B},$
with $\boldsymbol{\Delta B}$ capturing the MMEs.
These errors induce signal-dependent perturbations in the measurements through the interaction of the true signal $\boldsymbol{\beta}^*$ with the mismatched rows of $\boldsymbol{\widetilde{B}}$. For each $i \in [n]$, the mismatch introduces a signal dependent MME of the form $\widetilde{\delta}_i^* \triangleq (\boldsymbol{\Delta b}_{i,.}) \boldsymbol{\beta}^*=(\boldsymbol{ \tilde{b}}_{i,.}-\boldsymbol{ {b}}_{i,.})\boldsymbol{\beta}^*.$
Thus, the measurement model becomes
\begin{equation} \label{eq:fm_delta_indiv}
    z_i = (\boldsymbol{b}_{i,.}) \boldsymbol{\beta}^* +\widetilde{\delta}_i^* + \eta_i,
\end{equation}
or in vector form as in \eqref{eq:model_non_centered}.
We assume that the MME vector 
$\boldsymbol{\widetilde{\delta}}^* \in \mathbb{R}^n$ is sparse, with $\tilde{r} \triangleq \|\boldsymbol{\widetilde{\delta}}^*\|_0 \ll n,$
reflecting practical scenarios in which only a small number of pooling 
operations incur MMEs given a competent technician.

In order to estimate $\boldsymbol{\beta^*}$ and $\boldsymbol{\widetilde{\delta}^*}$ from $\boldsymbol{z}$ and $\boldsymbol{B}$, one can use a Robust version of the \textsc{Lasso} estimator provided in \cite{Nguyen2013}. This estimator called the Robust \textsc{Lasso} (\textsc{Rl}) given as follows:
\begin{eqnarray}
\label{eq:robustLasso}
\begin{pmatrix}
\boldsymbol{\hat{\beta}_{\lambda_1}} \\
\boldsymbol{\hat{\delta}_{\lambda_2}}
\end{pmatrix} \!=\! \arg\min_{\boldsymbol{\beta},\boldsymbol{\widetilde{\delta}}} \frac{1}{2n}\!\left\|\boldsymbol{z}\!-\!\boldsymbol{B\beta}\!-\!\boldsymbol{\widetilde{\delta}}\right\|^2_2 \!+\! \lambda_1 \|\boldsymbol{\beta}\|_1 \!+\! \lambda_2 \left\|\boldsymbol{\widetilde{\delta}}\right\|_1 \label{eq:lasso_delta}
\end{eqnarray}
where $\lambda_1, \lambda_2$ are appropriately chosen regularization parameters.

\subsection{Algorithm: Correction of MMEs}\label{sec:Corr_Perm}
In this subsection, we present a procedure for correcting the MMEs that occur in the pooling matrix $\boldsymbol{B}$ under a specified MME model (see Sec.~\ref{sec:GMSE}). Since any correction method must begin by identifying the rows of $\boldsymbol{B}$ that contain MMEs, we first employ the detection mechanism introduced in Sec.~\ref{subsec:detect}--though other methods for MME detection (eg: based on thresholding the estimates of $\boldsymbol{\tilde{\delta}^*}$) can also be used. Throughout this paper, we denote by $\mathcal{J}$ the set of all rows of $\boldsymbol{B}$ that are flagged as MME-affected by this identification method. For consistency, we employ the approach outlined in Sec.~\ref{subsec:detect}.

Once the set $\mathcal{J}$ has been determined, and given the data $(\boldsymbol{z},\boldsymbol{B})$, we apply an algorithm (see Alg.~\ref{alg:gen_corr}) that corrects MMEs in $\boldsymbol{B}$ according to the assumed MME model. Following the correction step, we then re-estimate $\boldsymbol{\beta}^*$ using Robust \textsc{Lasso}, applied to the measurement vector $\boldsymbol{z}$ together with the corrected pooling matrix $\boldsymbol{\hat{B}}$. In order to develop this approach in more detail, we begin by defining the notions of \emph{effective} and \emph{correctable} MMEs.

\begin{definition}[Effective MMEs]\label{def:effective_mmes}
Given the measurement vector $\boldsymbol{z}$ and pooling matrix $\boldsymbol{B}$, the $i^{\text{th}}$ measurement $z_i$ is said to contain an \emph{effective} MME if $\tilde{\delta}^*_i = (\widetilde{\boldsymbol{b}}_{i\cdot} - \boldsymbol{b}_{i\cdot})\boldsymbol{\beta}^* \neq 0,$
for $i \in [n]$ as defined in \eqref{eq:fm_delta_indiv}. \hfill $\blacksquare$
\end{definition}

An MME may fail to be effective, for example, when bit-flips occur only at indices corresponding to zero-valued components of $\boldsymbol{\beta}^*$, or when permutation-induced pooling results satisfy $\widetilde{\boldsymbol{b}}_{i\cdot}\boldsymbol{\beta}^* = \boldsymbol{b}_{i\cdot}\boldsymbol{\beta}^*$.

\begin{definition}[Model-based Correctable Effective MMEs]\label{def:correctable_MME}
For $i \in [n]$, let $\mathcal{C}_i$ denote the set of all model-based perturbed versions of the $i^{\text{th}}$ row of $\boldsymbol{B}$. The $i^{\text{th}}$ measurement is said to have a \emph{correctable effective} MME if the values
$(\bar{\boldsymbol{b}}_{i\cdot} - \boldsymbol{b}_{i\cdot})\boldsymbol{\beta}^*
$ are distinct for all $\bar{\boldsymbol{b}}_{i\cdot} \in \mathcal{C}_i$. \hfill $\blacksquare$
\end{definition}
Different MME models are formally defined in Sec.~\ref{sec:GMSE}. The choice of model determines the construction of $\mathcal{C}_i$ in Definition~\ref{def:correctable_MME}. Importantly, these definitions are independent of any particular correction algorithm. If the $i^{\text{th}}$ effective MME is not correctable, this occurs due to aliasing when for two distinct elements $\bar{\boldsymbol{b}}_{i\cdot}, \bar{\boldsymbol{d}}_{i\cdot} \in \mathcal{C}_i$,
we have $(\bar{\boldsymbol{b}}_{i\cdot} - \boldsymbol{b}_{i\cdot})\boldsymbol{\beta}^*
=
(\bar{\boldsymbol{d}}_{i\cdot} - \boldsymbol{b}_{i\cdot})\boldsymbol{\beta}^*$. 
As a consequence, no algorithm can resolve such cases. Note that the unknown true row $\widetilde{\boldsymbol{b}}_{i\cdot}$ always belongs to $\mathcal{C}_i$ by construction. In Sec.~S.I of the supplemental, we specify $\mathcal{C}_i$ for each MME model in Sec.~\ref{sec:GMSE}.

Ideally, a detection mechanism should identify all MME-affected rows almost surely. In practice, such complete accuracy is rarely achievable; even a strong detector typically identifies these rows only with high probability. We therefore introduce two closely related correction strategies corresponding to these two realistic regimes.

\subsubsection{Case 1: All MME-affected rows are known}
In this scenario, we consider the setup in which either the detection mechanism is capable of identifying all MME-affected rows in $\mathcal{J}$ almost surely, or an oracle is assumed to provide the locations of all MME-affected rows in advance. Therefore, we know in advance that the set $\mathcal{J}^c$ is a clean set of measurements without any MMEs. 
Hence, we can utilize the \textsc{Lasso} estimate on the clean set $(\boldsymbol{z}_{\mathcal{J}^c},\boldsymbol{B}_{\mathcal{J}^c})$ to correct for the given MMEs. However the REC property has been proved for matrices with entries drawn i.i.d. from zero-mean distributions such as zero-mean Gaussian or the Rademacher distribution \cite{raskutti2010restricted}. Since the entries of $\boldsymbol{B}$ have a non-zero mean, some centering transformations have to be applied in order to obtain a mean-zero sensing matrix for which REC-based analysis could be used. 

\noindent \textbf{Centering operation:} Recall the measurement model $\boldsymbol{z} = \boldsymbol{B}\boldsymbol{\beta}^* + \widetilde{\boldsymbol{\delta}}^* + \tilde{\boldsymbol{\eta}}.$
To obtain a mean-zero sensing design, we apply the centering transformation for all $i \in \{1,2,\ldots,\lfloor n/2 \rfloor\}$
\begin{align}\label{eq:center_y}
y_i &= \frac{z_i - z_{\lfloor n/2 \rfloor+i}}{2\theta(1-\theta)}, \quad 
\boldsymbol{a}_{i\cdot} = \frac{\boldsymbol{b}_{i\cdot} - \boldsymbol{b}_{\lfloor n/2 \rfloor+i\cdot}}{2\theta(1-\theta)}, \\ \delta^*_i &=\frac{\tilde{\delta}^*_i-\tilde{\delta}^*_{\lfloor n/2 \rfloor+i}}{2\theta(1-\theta)},  \quad \eta_i =\frac{\tilde{\eta}^*_i-\tilde{\eta}^*_{\lfloor n/2 \rfloor+i}}{2\theta(1-\theta)},
\end{align}
where, $\lfloor n/2 \rfloor$ indicates the smallest integer less than or equals to $n/2$.
This yields the equivalent forward model
\begin{equation}\label{eq:cent_model}
\boldsymbol{y} = \boldsymbol{A}\boldsymbol{\beta}^* + \boldsymbol{\delta}^* + \boldsymbol{\eta}.
\end{equation}
 In this representation, 
 the error vector $\boldsymbol{\delta}^*$ will be sparse if $\boldsymbol{\tilde{\delta}^*}$ is sparse, with $r = \|\boldsymbol{\delta}^*\|_0 \ll \lfloor n/2 \rfloor$, reflecting the presence of only a small number of pooling errors. Note that, for any $i \in \mathcal{J}^c$, $\delta^*_i = 0$.
The centering operation reduces the effective number of measurements for the centered forward model given in \eqref{eq:cent_model} to $n'=\lfloor n/2 \rfloor$.
We require it mainly for theoretical performance bounds, since the existing theoretical results for estimators such as the \textsc{Lasso} require matrices whose entries have mean zero. The centering operation is employed in practical numerical experiments for detection of measurements with bit-flips, but is not required for correction of bit-flips.

Detection and the consequent correction of MMEs requires that the MMEs be \emph{effective} in the centered measurement model $(\boldsymbol{y},\boldsymbol{A})$.
However, an MME that is both \emph{effective} and \emph{correctable} in the non-centered measurements $(\boldsymbol{z},\boldsymbol{B})$ need not remain effective after the centering operation in \eqref{eq:center_y}.
To illustrate this phenomenon, consider an index $i \in [n']$ such that both measurements $z_i$ and $z_{n'+i}$ contain effective and correctable MMEs, i.e., $\tilde{\delta}^*_i \neq 0$ and $\tilde{\delta}^*_{n'+i} \neq 0$. It is possible that
the corresponding centered measurement $y_i = \frac{z_i - z_{n'+i}}{2\theta(1-\theta)}$
does not contain an effective MME, namely $\delta^*_i = 0$.
By the centering transformation in \eqref{eq:center_y}, such cancellation occurs when $\tilde{\delta}^*_i = \tilde{\delta}^*_{n'+i}$, which is equivalent to $ (\widetilde{\boldsymbol{b}}_{i\cdot} - \boldsymbol{b}_{i\cdot}) \boldsymbol{\beta}^*
    =
    (\widetilde{\boldsymbol{b}}_{n'+i\cdot} - \boldsymbol{b}_{n'+i\cdot}) \boldsymbol{\beta}^* .$
In this case, the mismatch contributions cancel exactly under centering, rendering the MME ineffective in the centered model.

Therefore, we require the assumption that the MMEs in the non-centered measurements $(\boldsymbol{z},\boldsymbol{B})$ are distinct, i.e., 
\begin{itemize}
    \item \textbf{A0 (Distinctness of MMEs):} For any two distinct indices $i_1, i_2 \in [n]$ with $i_1 \neq i_2$, the corresponding MME values are different; that is, $\tilde{\delta}^*_{i_1} \neq \tilde{\delta}^*_{i_2}$.
\end{itemize}
This assumption ensures that each MME-affected row can be uniquely identified and corrected. However, in scenarios where this distinctness assumption is violated, we develop an algorithmic procedure to handle non-distinct MMEs, as detailed in Sec.~\ref{sec:mult_Corr}.

To correct for the MMEs in this setting, we employ the Robust \textsc{Lasso} estimator given in \eqref{eq:lasso_delta}, denoted by $\boldsymbol{\hat{\beta}}_{\mathcal{J}^c}$, computed using the centered, MME-free subset of the measurements $(\boldsymbol{y}_{\mathcal{J}^c}, \boldsymbol{A}_{\mathcal{J}^c})$. This estimator then serves as the basis for correcting the MMEs present in the remaining rows.
In this setting, for any $i \in \mathcal{J}$, to correct an \emph{effective} and \emph{correctable} MME in row $\boldsymbol{b}_{i\cdot}$, we compute a distance measure $d(z_i,\bar{\boldsymbol{b}}_{i\cdot}\boldsymbol{\hat{\beta}}_{\mathcal{J}^c})$
for every $\bar{\boldsymbol{b}}_{i\cdot} \in \mathcal{C}_i$. We then set
\begin{equation}\label{eq:prox_meas}
    \hat{\boldsymbol{b}}_{i\cdot}
    \triangleq
    \arg\min_{\bar{\boldsymbol{b}}_{i\cdot} \in \mathcal{C}_i}
    d(z_i,\bar{\boldsymbol{b}}_{i\cdot}\boldsymbol{\hat{\beta}}_{\mathcal{J}^c}).
\end{equation}
The row $\boldsymbol{b}_{i\cdot}$ is then replaced by the corrected version $\hat{\boldsymbol{b}}_{i\cdot}$ for each $i \in \mathcal{J}$.

\noindent\textbf{Absolute Prediction Error (APE):}  
For $i \in \mathcal{J}$, we define the Absolute Prediction Error as
\begin{equation}\label{eq:APE}
    d(z_i, \bar{\boldsymbol{b}}_{i\cdot}\boldsymbol{\hat{\beta}}_{\mathcal{J}^c})
    = |z_i - \bar{\boldsymbol{b}}_{i\cdot}\boldsymbol{\hat{\beta}}_{\mathcal{J}^c}|.
\end{equation}

\subsubsection{Case 2: MME-affected rows are unknown}
We now consider the challenging setting in which the detection mechanism is capable of identifying the MME-affected rows only with high probability. Consequently, we cannot guarantee almost surely that the set $\mathcal{J}$ contains all MME-affected rows. As a result, the \textsc{Lasso} estimate obtained using $(\boldsymbol{y}_{\mathcal{J}^c}, \boldsymbol{A}_{\mathcal{J}^c})$ may be unreliable (where $\boldsymbol{y}, \boldsymbol{A}$ are as defined in \eqref{eq:center_y}), since the index set $\mathcal{J}^c$ may still contain MME-affected rows. In such cases, we instead employ the Robust \textsc{Lasso} estimator \cite{Nguyen2013}, introduced in Sec.~\ref{sec:odrlt}, and compute $\boldsymbol{\hat{\beta}}_{\mathcal{J}^c}$ using the pair $(\boldsymbol{y}_{\mathcal{J}^c}, \boldsymbol{A}_{\mathcal{J}^c})$. Given $\boldsymbol{\hat{\beta}}_{\mathcal{J}^c}$, the rest of the correction algorithm for this scenario follows the same steps as in Case 1 as described in \eqref{eq:prox_meas} and \eqref{eq:APE}.

In Theorem~\ref{th:minval_error}, we formally establish that the solution $\hat{\boldsymbol{b}}_{i\cdot}$ obtained via \eqref{eq:prox_meas} with the APE metric equals the true row $\widetilde{\boldsymbol{b}}_{i\cdot}$ with high probability in the case $\boldsymbol{\eta}=\boldsymbol{0}$, provided that the detection mechanism identifies all \emph{effective} MMEs. We also show that this conclusion is stable under zero mean i.i.d. Gaussian noise with variance $\sigma^2$.

\noindent \textbf{General correction algorithm.}
We now present the pseudo-code for the general correction algorithm in Alg.~\ref{alg:gen_corr}.

\begin{algorithm}[h!]
\caption{General Correction Algorithm for Model Mismatch Error}
\begin{algorithmic}[1]
\ENSURE Measurement vector $\boldsymbol{z}$, sensing matrix $\boldsymbol{B}$, MME index set $\mathcal{J}$, \textsc{Lasso} estimate $\boldsymbol{\hat{\beta}}_{\mathcal{J}^c}$.
\STATE $\boldsymbol{\hat{B}} = \boldsymbol{B}$.
\FOR{$i \in \mathcal{J}$}
    \STATE Construct the perturbation set $\mathcal{C}_i$ based on the MME model\footnotemark for row $\boldsymbol{b}_{i\cdot}$.
    \STATE Set $\hat{\boldsymbol{b}}_{i\cdot}
    =
    \arg\min_{\bar{\boldsymbol{b}}_{i\cdot} \in \mathcal{C}_i}
    d(z_i,\bar{\boldsymbol{b}}_{i\cdot} \boldsymbol{\hat{\beta}}_{\mathcal{J}^c})$, where
    $\boldsymbol{\hat{\beta}}_{\mathcal{J}^c}$ is the Robust \textsc{Lasso} estimator on $(\boldsymbol{y}_{\mathcal{J}^c},\boldsymbol{A}_{\mathcal{J}^c})$ given in \eqref{eq:lasso_delta}.
\ENDFOR
\STATE \textbf{return} $\boldsymbol{\hat{B}}$.
\end{algorithmic}
\label{alg:gen_corr}
\end{algorithm}
\footnotetext{see Sec.~S.I of the supplemental}

We emphasize that MME correction refers specifically to estimating the underlying \emph{pooling matrix} $\boldsymbol{\widetilde{B}}$ from $(\boldsymbol{B},\boldsymbol{z})$ under a given MME model. The corrected matrix is denoted by $\boldsymbol{\hat{B}}$. Our correction procedure does \emph{not} modify the measurement vector $\boldsymbol{z}$.

\subsection{Identification of MMEs using \textsc{Odrlt} }\label{sec:odrlt}

We briefly summarize the MME identification procedure introduced in \cite{Banerjee24}, termed as the \emph{Optimal Debiased Robust Lasso Test} (\textsc{Odrlt})--note that one can also explore other techniques from the robust regression literature for MME identification. The \textsc{Odrlt} method operates on the centered pooling model constructed from paired Bernoulli($\theta$) rows of the original pooling matrix. 
 Recall that the forward centered model is given by $\boldsymbol{y} = \boldsymbol{A\beta}^* + \boldsymbol{\delta}^* + \boldsymbol{\eta}.$

To obtain initial estimates of $(\boldsymbol{\beta}^*, \boldsymbol{\delta}^*)$, we use the Robust Lasso estimator of \cite{Nguyen2013} given in \eqref{eq:lasso_delta}. Since such estimators are biased, \cite{Banerjee24} constructs a debiased estimator for the MME vector. Let $\hat{\boldsymbol{\beta}}_{\lambda_1}$ and $\hat{\boldsymbol{\delta}}_{\lambda_2}$ denote the preliminary Robust \textsc{Lasso} estimates obtained from the centered measurements $(\boldsymbol{y},\boldsymbol{A})$. 
The debiased MME estimator is defined as
\begin{equation}\label{eq:deb_delta}
\hat{\boldsymbol{\delta}}_{W}
=
\hat{\boldsymbol{\delta}}_{\lambda_2}
+
\left( \boldsymbol{I}_{n'} - \frac{1}{n'}\boldsymbol{W A}^{\top} \right)
\left( \boldsymbol{y} - \boldsymbol{A\hat{\beta}}_{\lambda_1} - \hat{\boldsymbol{\delta}}_{\lambda_2} \right),
\end{equation}
where $\boldsymbol{W}$ is obtained from the convex optimization program given in Algorithm~2 of \cite{Banerjee24} (also described in \eqref{eq:opt_W}) and $n'=\lfloor n/2 \rfloor$. Theorem~4.1 therein provides the asymptotic variance expression associated with each component of $\hat{\boldsymbol{\delta}}_{W}$, enabling statistical inference on the entries of $\boldsymbol{\delta}^*$.

The \textsc{Odrlt} determines whether measurement $i$ contains an MME by testing $H_0: \delta_i^* = 0 
\,\text{versus}\,
H_1: \delta_i^* \neq 0.$
We reject $H_0$ against $H_1$ at $\alpha \%$ level of significance if,
\begin{equation}\label{eq:test_stat}
\frac{\,|\hat{\delta}_{W_i}|\,}{\sqrt{\Sigma_{\delta_{ii}}}} > \tau_{\alpha/2},
\end{equation}
where $\boldsymbol{\Sigma}_{\delta}=\sigma^2\left(\boldsymbol{I_n}-\frac{1}{n'}\boldsymbol{AW}^{\top}\right)\left(\boldsymbol{I_n}-\frac{1}{n'}\boldsymbol{AW}^{\top}\right)^{\top}$, and where $\tau_{\alpha/2}$ is the upper $(\alpha/2)^{\text{th}}$ quantile of a standard normal random variable. Under $H_0$, the test statistic is asymptotically standard normal, yielding the p-value $ 1 - \Phi\left(\frac{\,|\hat{\delta}_{W_i}|\,}{\sqrt{\Sigma_{\delta_{ii}}}}\right)$ where, $\Phi(.)$ is the cdf of the standard normal distribution.
All indices $i$ with p-value $ < \alpha$ form the detected MME set $\mathcal{J}$, with $\hat{r} = |\mathcal{J}| \le r_U$ as guaranteed by the detection procedure described in Sec.~\ref{subsec:detect}. Since $r_U \ll n'$, the complement $\mathcal{J}^{c}$ retains a sufficiently large number of uncorrupted measurements for accurate recovery of $\boldsymbol{\beta}^*$.

\subsection{Algorithm: Discarding MMEs}\label{subsec:detect}
We now discuss a technique to estimate $\boldsymbol{\beta}^*$ after discarding the set of measurements detected to have MMEs by using a specific detection technique. We refer to this estimate of $\boldsymbol{\beta}^*$ as `MME Rejection'(\textsc{Mmer}). 
The method of MME detection considered is \textsc{Odrlt}, presented in Sec.~\ref{sec:odrlt}.
For developing our theory, we assume knowledge of an upper bound $r_U$ on the number of MMEs $r$. Such an upper bound is reasonable to assume given a competent technician. 

Recall that \textsc{Odrlt}, operates on the centered measurements $(\boldsymbol{y}, \boldsymbol{A})$. 
Let $\mathcal{J} \subseteq [n']$ denote the set of indices of the centered measurements identified as corrupted under the chosen detection rule, and let $\hat{r} := |\mathcal{J}|$. Since each centered measurement $y_i$ is constructed from a known pair of rows of the original pooling matrix $\boldsymbol{B}$ (cf. \eqref{eq:center_y}), the identification of $\mathcal{J}$ immediately determines the corresponding $2\hat{r}$ rows of $\boldsymbol{B}$ that have been identified as containing MMEs by \textsc{Odrlt}. This mapping is exact because the pairing structure used in forming $\boldsymbol{A}$ is deterministic and known \textit{a priori}.

Note that the chosen identification technique in Sec.~\ref{sec:odrlt} does not guarantee the detection of all the measurements with effective MMEs (see Definition~\ref{def:effective_mmes}), due to measurement noise $\boldsymbol{\eta}$ which causes error in MME detection. Due to this, we perform the detection method in an iterative manner to obtain $\mathcal{J}$ as summarized in Alg.~\ref{alg:mme_detection}. Due to these errors, we obtain estimates of $\boldsymbol{\beta}^*$ and $\boldsymbol{\delta}^*_{\mathcal{J}^c}$ based on a measurement subset, i.e. $\boldsymbol{y}_{\mathcal{J}^c}$ and $\boldsymbol{A}_{\mathcal{J}^c}$, using the robust \textsc{Lasso} estimator from \eqref{eq:robustLasso}. 

\begin{algorithm}[h]
\caption{Detection of Measurements with Model Mismatch Errors (MMEs)}
\label{alg:mme_detection}
\begin{algorithmic}[1]
\STATE \textbf{Input:}  $\boldsymbol{y}$ (centered measurement vector), $\boldsymbol{A}$ (centered sensing matrix), $r_U$ (upper bound on the number of MMEs)
\STATE Set $i \gets 1$
\STATE $\mathcal{J}_1\triangleq$ Set of indices of $\boldsymbol{y}$ detected with MMEs with input $(\boldsymbol{y},\boldsymbol{A})$ (see Sec.~\ref{sec:odrlt})
\STATE Assign $\mathcal{J}=\mathcal{J}_1$
\WHILE{True}
    \IF{$|\mathcal{J}| = r_U$}
        \STATE \textbf{break} \COMMENT{Maximum possible MMEs reached}
    \ELSIF{$|\mathcal{J}| > r_U$}
        \STATE Retain indices of the $r_U$ strongest MMEs in $\mathcal{J}$ as per a criterion (See Sec.~\ref{sec:odrlt})
        \STATE \textbf{break}
    \ENDIF
    \STATE $\mathcal{J}_{i+1}\triangleq$ Set of indices of $\boldsymbol{y}$ detected to have MMEs via debiased robust \textsc{Lasso} (see Sec.~\ref{sec:odrlt}) with only measurements from   
$(\boldsymbol{y}_{\mathcal{J}^c},\boldsymbol{A}_{\mathcal{J}^c})$ as input.
        \IF{$\mathcal{J}_{i+1} = \emptyset$}
        \STATE \textbf{break} \COMMENT{No additional MMEs detected} 
    \ENDIF
    \STATE Set $\mathcal{J} \gets \bigcup_{k=1}^{i+1} \mathcal{J}_k$
     \STATE $i \gets i + 1$
\ENDWHILE
\STATE \textbf{Output:} Set $\mathcal{J}$
\end{algorithmic}
\end{algorithm}

\subsection{Multiple stages of correction}\label{sec:mult_Corr}
We have observed that after the first stage of correction using \textbf{APE}, 
a small number of MMEs may remain uncorrected and a small number of new MMEs may be falsely reported (respectively due to Type-I and Type-II error of the \textsc{Odrlt}). 
Hence we resort to multiple stages of correction as described in Alg.~\ref{alg:mult_cape}. 

We obtain the final Robust \textsc{Lasso} estimate after the last stage of correction, 
i.e., after the stopping criterion is met. We refer to this method as 
``Correction using \textbf{APE}'' (\textsc{Cape}). 
We define the stopping function $f_{ape}$ of \textsc{Cape} for the correction algorithm 
Alg.~\ref{alg:gen_corr} based on the sum of the squares of the 
\textbf{APE} distance measures:
\begin{equation}\label{eq:halt_est}
    f_{ape} \triangleq 
    \frac{1}{n}\|\boldsymbol{z} - \boldsymbol{\hat{B}\hat{\beta}_{\lambda}}\|_2^2,
\end{equation}
where $\boldsymbol{\hat{\beta}_{\lambda}}$ is the \textsc{Lasso} estimate on $\boldsymbol{z}$ and $\boldsymbol{\hat{B}}$ is the corrected pooling matrix in the last stage of Alg.~\ref{alg:mult_cape}.
The criterion $f_{ape}$ is motivated by the fact that in Alg.~\ref{alg:gen_corr}, using \textbf{APE} as the distance measure, we expect the sum of squares of the prediction errors to reduce. We will later show empirically that $f_{ape}$ across successive iterations in Alg.~\ref{alg:mult_cape} converges to a minimum value.

We stop the iterations if $|f_{ape}(k)-f_{ape}(k-1)| \leq \epsilon$, where $f_{ape}(k)$ is the value 
in the $k^{\text{th}}$ iteration of $f_{ape}$, 
and $\epsilon>0$ is a pre-defined stopping threshold.

\begin{algorithm}
\caption{Multi-stage Correction}
\label{alg:mult_cape}
\begin{algorithmic}[1]
\STATE \textbf{Input:} Stage index $k \gets 1$, stopping function $f_{ape}$ 
(see \eqref{eq:halt_est}), stopping threshold $\epsilon > 0$, 
and the corrected pooling matrix $\boldsymbol{\hat{B}}$ after the first stage.
\REPEAT
    \STATE Randomly shuffle the $n$ rows of $\boldsymbol{B}$ \emph{along with} the corresponding entries in $\boldsymbol{z}$ and perform centering on the shuffled rows using \eqref{eq:center_y}.
    \STATE Use the Robust \textsc{Lasso} to obtain estimates 
    $\boldsymbol{\hat{\beta}_{\lambda_1}}$ and 
    $\boldsymbol{\hat{{\delta}}_{\lambda_2}}$ 
    of $\boldsymbol{\beta}^*$ and $\boldsymbol{{\delta}}^*$ respectively, 
    using the centered version of the measurement vector $\boldsymbol{z}$ and corrected pooling matrix 
    $\boldsymbol{\hat{B}}$ (see~\eqref{eq:prox_meas}).
    \STATE Debias the Robust \textsc{Lasso} estimate to obtain 
    $\boldsymbol{\hat{{\delta}}_W}$.
    \STATE Update the set of measurement indices with MMEs, $\mathcal{J}$, 
    using the \textsc{Odrlt} approach (Sec.~\ref{sec:odrlt}) on 
    $\boldsymbol{\hat{{\delta}}_W}$.
    \STATE Perform correction based on the updated $\mathcal{J}$ 
    using Alg.~\ref{alg:gen_corr}.
    \STATE Compute the stopping function value $f_{ape}(k)$ (\textit{e.g.} using \eqref{eq:halt_est}).
    \STATE $k \gets k + 1$
\UNTIL{$|f_{ape}(k) - f_{ape}(k-1)| > \epsilon$}
\STATE \textbf{Output:} Final Robust \textsc{Lasso} estimate $\boldsymbol{\beta}^*$ 
using \eqref{eq:robustLasso} after the last stage of correction.
\end{algorithmic}
\end{algorithm}
Recall that detection and correction of MMEs in the centered model $(\boldsymbol{y}, \boldsymbol{A})$ requires that their mismatch contributions remain \emph{effective} after centering. However, effective MMEs in the non-centered measurements $(\boldsymbol{z}, \boldsymbol{B})$ may cancel out each other under the centering transformation, particularly when paired measurements satisfy $\tilde{\delta}^*_i = \tilde{\delta}^*_{n'+i}$, where $n'=\lfloor n/2 \rfloor$. To avoid such cancellations, we previously imposed Assumption \textbf{A0}, requiring all non-centered MMEs to be distinct. 
However, Step~2 of the multi-stage correction algorithm in Alg.~\ref{alg:mult_cape} across multiple stages of correction allows us to relax this assumption and correct MMEs even when \textbf{A0} does not hold. Specifically, before each centering operation, the $n$ measurements are randomly shuffled, so that the pairing $(z_i, z_{n'+i})$ used to construct the centered measurement $y_i$ changes independently across iterations. As a result, even if two effective and correctable MMEs satisfy $\tilde{\delta}^*_i = \tilde{\delta}^*_{n'+i}$ in a given stage leading to cancellation under centering and rendering $\delta^*_i = 0$, the probability that the same two corrupted measurements are paired again in the next stage (or subsequent stages) is negligible. In this way, the multi-stage correction algorithm reduces the inherent ambiguity introduced by centering and guarantees that \emph{effective} MMEs are eventually identified and corrected, provided a sufficient number of stages are performed.

\section{Theoretical Results}
\subsection{Stability of MME correction}
In this subsection, we will determine the probability that the correction algorithm Alg.~\ref{alg:gen_corr} using the distance measure \textbf{APE}, actually succeeds in correcting an MME. We further determine the probability, that the correction algorithm does not wrongly introduce an MME in a measurement that do not contain any MMEs.
This is equivalent to showing that (\textit{i}) for any $i \in \mathcal{J}$ such that $\tilde{\delta}^*_i \ne 0$, we would have $\boldsymbol{\hat{b}_{i.}}=\boldsymbol{\tilde{b}_{i.}}$ and that, (\textit{ii}) for any $i \in \mathcal{J}$ such that $\tilde{\delta}^*_i =0$, we have $\boldsymbol{\hat{b}_{i.}}=\boldsymbol{{b}_{i.}}$, where $\boldsymbol{\hat{b}_{i.}}$ is the corrected row as defined in \eqref{eq:prox_meas}. 

We presented the concept of \emph{effective} and \emph{correctable} MMEs in Defns.~\ref{def:effective_mmes} and \ref{def:correctable_MME}. However, note that in our model, the measurement vector $\boldsymbol{z}$ is corrupted by the additive noise vector $\boldsymbol{\eta}$ whose elements are drawn i.i.d.\ from $\mathcal{N}(0,\sigma^2)$. Hence, we need to establish the notion of \emph{effective} and \emph{correctable} MMEs in the presence of such noise. Since $\eta_i$ is Gaussian, we know that $|\eta_i| \leq k\sigma$ with probability at least $1 - e^{-k^2/2}$, which is large for any $k \geq 3$ (see Lemma S.2 of the supplemental). Therefore, we present the following assumptions:

\begin{itemize}\label{assumptions}
    \item \textbf{A1:} For any $i \in [n]$, if the $i^{\text{th}}$ measurement has an \emph{effective} MME as per Defn.~\ref{def:effective_mmes}, then for some $k \geq 3$,  
    \[
    |(\boldsymbol{\tilde{b}}_{i.}-\boldsymbol{b}_{i.})\boldsymbol{\beta}^*| \geq 2(k+1)\sigma.
    \]

    \item \textbf{A2:} For any $i \in [n]$, if the $i^{\text{th}}$ measurement has a correctable MME, then for any $\bar{\boldsymbol{b}}_{i.} \in \mathcal{C}_i$ and for some $k \geq 3$ 
    \[
    |(\boldsymbol{\bar{b}}_{i.}-\boldsymbol{b}_{i.})\boldsymbol{\beta}^*
      -(\boldsymbol{\tilde{b}}_{i.}-\boldsymbol{b}_{i.})\boldsymbol{\beta}^*|
      = |(\boldsymbol{\bar{b}}_{i.}-\boldsymbol{\tilde{b}}_{i.})\boldsymbol{\beta}^*|
      \geq 2(k+1)\sigma.
    \]
\end{itemize}

These two assumptions present notions of effective and correctable MMEs in the presence of noise $\boldsymbol{\eta}$ in $\boldsymbol{z}$ and therefore impose stronger conditions on $\boldsymbol{z}$ than those in Defns.~\ref{def:effective_mmes} and \ref{def:correctable_MME}.

Since our correction algorithm for Case 1 uses the \textsc{Lasso} as its foundation, we require the number of MMEs to be typically small, as presented in assumption \textbf{A3} below:

\begin{itemize}
    \item \textbf{A3:} The number of MMEs given by $r$ is bounded w.r.t.\ $n$, i.e., $r \leq r_U$ where $r_U = \zeta n' \leq \zeta \frac{n}{2}$ for some $\zeta \in (0,1)$.
\end{itemize}

The assumptions \textbf{A1}, \textbf{A2}, and \textbf{A3} are essential for proving the following theorem.

\begin{theorem}\label{th:minval_error}
   Let $\boldsymbol{B} \in \mathbb{R}^{n \times p}$ be the known pooling matrix with i.i.d.\ entries drawn from a Bernoulli($\theta$) distribution, and let $\boldsymbol{\tilde{B}}$ denote an unknown corrupted version of $\boldsymbol{B}$ due to model mismatch errors (MMEs). For each $i \in [n]$, let $\boldsymbol{b}_{i.}$ and $\boldsymbol{\tilde{b}}_{i.}$ respectively denote the $i^{\text{th}}$ rows of $\boldsymbol{B}$ and $\boldsymbol{\tilde{B}}$, and let $\boldsymbol{\bar{b}}_{i.}$ denote a model-based perturbation of $\boldsymbol{b}_{i.}$. Define $\mathcal{C}_i$ to be the collection of allowable model-based perturbations of $\boldsymbol{b}_{i.}$. 
Let $\hat{\boldsymbol{b}}_{i.}$ be the corrected row obtained via Alg.~\ref{alg:gen_corr} using the \textbf{APE} criterion, and let $\boldsymbol{\hat{\beta}}_{\mathcal{J}^c}$ be the \textsc{Lasso} estimate of $\boldsymbol{\beta}^*$ based on $(\boldsymbol{y}_{\mathcal{J}^c},\boldsymbol{A}_{\mathcal{J}^c})$. Under assumptions \textbf{A1}, \textbf{A2}, \textbf{A3}, if
$n \geq \frac{C^2 s^2}{1-\zeta}\log p$,
then the following statements are true for any $k > 0$:

\begin{enumerate}[\label=(1)]
    \item For any $i \in \mathcal{J}$ such that $\delta^*_i \neq 0$,
    \begin{equation}\label{eq:correct_prob}
    P\left(\hat{\boldsymbol{b}}_{i.} = \boldsymbol{\tilde{b}}_{i.}\right)
    \geq 1 - 3(e^{-k^2/2} + 1/p^2).
    \end{equation}

    \item For any $i \in \mathcal{J}$ such that $\delta^*_i = 0$,
    \begin{equation}\label{eq:misdetect_prob}
    P\left(\hat{\boldsymbol{b}}_{i.} = \boldsymbol{b}_{i.}\right)
    \geq 1 - 2(e^{-k^2/2} + 1/p^2).
    \end{equation}
\end{enumerate}

\hfill{$\blacksquare$}
\end{theorem}

\textbf{Remarks:}
\begin{enumerate}
    \item In the noiseless scenario, i.e., when $\boldsymbol{\eta} = \boldsymbol{0}$, result (1) holds with probability at least  
$    1 - 3/p^2$,
    and result (2) holds with probability at least  
$1 - 2/p^2$.     Thus, an MME can be corrected with high probability asymptotically.

    \item The constant $k>0$ defines a probabilistic range for the additive noise in $\boldsymbol{\eta}$. For instance, $k = 3$ corresponds approximately to a $99.7\%$ confidence interval for Gaussian noise, meaning $\eta_i$ lies in $[-3\sigma,3\sigma]$ with probability $0.997$.

    \item Apart from those defined in Sec.~\ref{sec:GMSE}, Theorem~\ref{th:minval_error} holds for any MME model admitting a finite set $\mathcal{C}_i$ of allowable model-based perturbations.

    \item This result extends naturally to matrices drawn from any bounded distribution, up to modified constant factors in the sample-size requirement.  

    \item This result extends naturally to the scenario in which the index set $\mathcal{J}$ does not contain all rows affected by MMEs. 
In this regime, the correction algorithm is applied using the Robust \textsc{Lasso} estimator, rather than the standard \textsc{Lasso} estimator computed on the reduced measurements $(\boldsymbol{y}_{\mathcal{J}^c}, \boldsymbol{A}_{\mathcal{J}^c})$. 
This modification allows us to characterize the performance of the correction procedure even when $\mathcal{J}$ fails to capture all MME-corrupted rows.
The analysis follows along similar lines, with the only difference being that the error bounds associated with the Robust \textsc{Lasso} estimator replace those of the standard \textsc{Lasso}. 
As a consequence, two changes arise: 
(\textit{i}) the required sample complexity increases to $n \geq \frac{48^2 \kappa^{-4}}{1-\zeta}(s+r)^2 \log p,$
and 
(\textit{ii}) the lower bounds on the success probabilities in \eqref{eq:correct_prob} and \eqref{eq:misdetect_prob} are modified to $1 - 4\left(e^{-k^2 / 2} + \frac{1}{p} + \frac{1}{n}\right)$ and $1 - 2\left(e^{-k^2 / 2} + \frac{1}{p} + \frac{1}{n}\right)$
respectively.
\end{enumerate}

\subsection{Reconstruction error in MME Rejection}
Recall that $\mathcal{J}$ is the set of all measurements detected to contains MMEs, as per Alg.~\ref{alg:mme_detection}.
Let $\boldsymbol{A}_{\mathcal{J}^c}$ be the sub-matrix of $\boldsymbol{A}$ containing only the rows from $\mathcal{J}^c$. Since the elements of $\boldsymbol{A}$ are drawn i.i.d. from the centered version of the Bernoulli($\theta$) (henceforth denoted by $CB(\theta)$), i.e., the elements of $\boldsymbol{A}$ obey the following probability mass function:
\begin{eqnarray*}
    P(a_{ij}=-1) &=&\theta(1-\theta)= P(a_{ij}=1), \\  P(a_{ij}=0)&=&\theta^2 + (1-\theta)^2.   
    \end{eqnarray*} 
 Clearly, the rows of the matrix  $\boldsymbol{A}_{\mathcal{J}^c}$ are also i.i.d. $CB(\theta)$. As Alg.~\ref{alg:mme_detection} constrains $\hat{r} < r_U$, we can assume that the cardinality of $\hat{r} \triangleq |\mathcal{J}| < n$ if we choose $r_U < n$.

Note that from Theorem~1 of \cite{Banerjee24}, we have the error bounds on the Robust \textsc{Lasso} estimate for any sensing matrix $\boldsymbol{A}$ that satisfies an extended version of the Restricted Eigenvalue Condition (EREC). Furthermore, it is shown in Lemma~6 of \cite{Banerjee24}, that any sensing matrix $\boldsymbol{A}$ with i.i.d. entries obtained from a distribution with mean $0$, variance $\leq 1$ and defined on a bounded domain $[-h,h]$ where $h > 0$, satisfies the EREC with high probability. Clearly $CB(\theta)$ satisfies these conditions. Therefore from Theorem 1 of \cite{Banerjee24}, we have the following error bound on the Robust \textsc{Lasso} estimate given $\lambda_1 \triangleq \frac{4\sigma\sqrt{\log p}}{\sqrt{n'-\hat{r}}}, \lambda_2 \triangleq \frac{4\sigma\sqrt{\log n'}}{{n'-\hat{r}}}$:
\begin{equation}\label{eq:RL_estimate}
P\left(\left\|\boldsymbol{\hat{\beta}_{\lambda_1}}\!-\!\boldsymbol{\beta}^*\right\|_1 \!\leq \!48\kappa^{-2} (s+r)\sigma\sqrt{\frac{{\log(p)}}{{n'-\hat{r}}}}\right)\! \geq\! 1\!-\!\left(\frac{1}{p}\!+\!\frac{1}{n'}\right),
\end{equation}
where $\kappa>0$ is the Extended Restricted Eigenvalue constant of $\boldsymbol{A}$ and $n'=\lfloor n/2 \rfloor$.

\section{Fast Debiasing of the Robust \textsc{Lasso} estimator}
The multi-stage correction algorithm given in Alg.~\ref{alg:mult_cape} iteratively makes updates to (say) row $\boldsymbol{a}_{i.}$ of the pooling matrix $\boldsymbol{A}$, based on the \textbf{APE} criterion in \eqref{eq:APE}. The algorithm requires computation of the debiased estimate  $\boldsymbol{\hat{\delta}_{W}}$, which requires computation of the debiasing matrix $\boldsymbol{W}$ that corresponds to $\boldsymbol{\hat{A}}$, the \emph{corrected} version of $\boldsymbol{A}$. The algorithm involves the following convex optimization procedure (see \cite{Banerjee24} for the derivation of this procedure):
\begin{eqnarray}\label{eq:opt_W}
 \underset{\boldsymbol{W}}{\textrm{min}}& &\quad 
\sum_{j=1}^p\boldsymbol{w_{.j}}^{\top}\boldsymbol{w_{.j}} \\ \nonumber
\textrm{s.t.} && 
\mathsf{C0} : \boldsymbol{w_{.j}}^{\top}\boldsymbol{w_{.j}}/n \leq 1+\sqrt{\frac{\log p}{n}} \ \forall \ j \in [p] 
\\ && \nonumber \mathsf{C1} : \left| \left(\boldsymbol{I_p}-\frac1n\boldsymbol{W^{\top}A}\right)\right|_\infty \leq \mu_1=  2 \sqrt{\frac{2 \log p}{n}} , \\ \nonumber && \mathsf{C2} : \left|\frac{1}{p}\left(\boldsymbol{I_n}-\frac{1}{n}\boldsymbol{A}\boldsymbol{W}^{\top}\right) \boldsymbol{A}\right|_\infty \leq  4\sqrt{\frac{\log 2np}{np}}+\frac{1}{n}, \\ \nonumber && \mathsf{C3} : \left|\left(\frac{\boldsymbol{AW^{\top}}}{p}-\boldsymbol{I_n}\right)\right|_{\infty} \leq \mu_3= 2 \sqrt{\frac{2 \log n}{p}}.
\end{eqnarray}
At first glance, it would appear that an iterative algorithm is necessary for obtaining $\boldsymbol{W}$, which would be very computationally expensive, especially because $\boldsymbol{W}$ would need to be computed afresh in every stage of the multi-stage procedure in Alg.~\ref{alg:mult_cape}. We now show that for $CB(\theta)$ matrices with $n<p$, this convex optimization problem admits an exact, closed-form optimal solution. 
\begin{theorem}\label{th:opt_W}
Let $\boldsymbol{A}$ be an $n \times p$ $CB(\theta)$ matrix.
Define $\tau(\boldsymbol{A}) := \max_{i\neq j} \frac{|\boldsymbol{a}_{i.}^\top \boldsymbol{a}_{j.}|}{\|\boldsymbol{a}_{i.}\|_2^2}$.
Then the optimal solution of the optimization problem in \eqref{eq:opt_W} is with high probability,
\begin{equation}\label{eq:exact_RL}
    \boldsymbol{w}_{i.}=\frac{p(1-\mu_3)}{\|\boldsymbol{a}_{i.}\|_2^2}\boldsymbol{a}_{i.}, \quad \forall \, i \in [n],
\end{equation}
 where,
 $\mu_3=2h^2\sqrt{\tfrac{2\log(n)}{p}}$ and $h=1/2\theta(1-\theta)$. 
\hfill{$\blacksquare$}
\end{theorem}
The proof of Theorem~\ref{th:opt_W} is provided in the supplementary material.

\noindent \textbf{Remarks:}
\begin{enumerate}
    \item The exact closed-form solution of the optimisation problem \eqref{eq:opt_W} given in Theorem~\ref{th:opt_W} accelerates multi-stage correction schemes by avoiding repeated numerical optimization. This is because given $\boldsymbol{A}$, one can directly implement the optimal solution of Alg. 2 of \cite{Banerjee24} in the form~\eqref{eq:exact_RL} for all $i \in [n]$.
    \item The condition $\tau/(1+\tau) \leq \mu_3 \leq 1$ is necessary and sufficient for the closed-form expression in~\eqref{eq:exact_RL} to be optimal for~\eqref{eq:opt_W}. This is similar to the condition $\rho/(1+\rho) \leq \mu \leq 1$ for the exact solution of the optimization problem in Theorem~1 of \cite{Banerjee25FD}. 
\end{enumerate}
In \cite{Banerjee25FD}, we have previously provided an exact closed form solution for the debiasing matrix used for debiasing the estimates of the \textsc{Lasso} (and not the robust \textsc{Lasso}). Theorem~\ref{th:opt_W} obtains the exact closed-form solution for the debiasing matrix $\boldsymbol{W}$ from \eqref{eq:opt_W} for \emph{robust} \textsc{Lasso} estimates, given $CB(\theta)$ pooling matrices. Note that the debiasing procedures for the robust \textsc{Lasso} (\eqref{eq:opt_W}) and the \textsc{Lasso} (Alg. 1 of \cite{Javanmard2014}) are very different from each other.

\section{Experimental Results}
We now present simulation results demonstrating the effectiveness of the proposed MME-correction framework.

\noindent\textbf{Data Generation:}  
Sparse signals $\boldsymbol{\beta}^* \in \mathbb{R}^{200}$ were generated with non-zero entries drawn i.i.d.\ from $U(100,1000)$ and placed at random locations, satisfying Assumptions~\textbf{A1}--\textbf{A2}. The pooling matrix $\boldsymbol{B}$ was sampled from a Bernoulli($\theta$) model for $\theta \in \{0.1,0.3,0.5\}$. Noisy measurements were produced as $\boldsymbol{z}=\boldsymbol{B\beta}^* + \boldsymbol{\eta},\qquad 
\sigma \triangleq f_\sigma \frac{1}{n}\sum_{i=1}^n | \boldsymbol{b}_{i.}\boldsymbol{\beta}^* |,$
with $\boldsymbol{\eta} \sim \mathcal{N}(0,\sigma^2 \boldsymbol{I_n})$.  
Effective MMEs were introduced adversarially by flipping bits in randomly chosen rows among the first $n$ rows, restricted to the non-zero indices of $\boldsymbol{\beta}^*$. This yielded the perturbed matrix $\boldsymbol{\tilde{B}}$ for all three MME models.  
For \textsc{Odrlt}, $(\boldsymbol{z},\boldsymbol{B})$ were centered as in~\eqref{eq:center_y} to form $(\boldsymbol{y},\boldsymbol{A})$, used only for detection.

\noindent\textbf{Experimental Setup:} 
Let $f_{sp}=s/p$ and $f_{adv}=r/n$ denote sparsity and corruption fractions. Four different simulations were conducted as per the following settings:
\begin{itemize}
    \item \textsf{EA}: $f_{\text{adv}} \in [0.02{:}0.02{:}0.1]$; $n=80$, $f_{sp}=0.05$, $f_\sigma=0.01$.
    \item \textsf{EB}: $n \in [50{:}10{:}90]$; $f_{\text{adv}}=0.02$, $f_{sp}=0.05$, $f_\sigma=0.01$.
    \item \textsf{EC}: $f_\sigma \in [0.01{:}0.01{:}0.05]$; $n=80$, $f_{\text{adv}}=0.01$, $f_{sp}=0.05$.
    \item \textsf{ED}: $f_{sp} \in [0.01{:}0.03{:}0.09]$; $n=80$, $f_{\text{adv}}=0.01$, $f_\sigma=0.01$.
\end{itemize}
Each configuration was repeated over 25 trials. The sensing matrix $\boldsymbol{A}$ was fixed across trials in \textsf{EA}, \textsf{EC}, and \textsf{ED}, but regenerated in \textsf{EB} since $n$ varied. The signal $\boldsymbol{\beta}^*$ was fixed except in \textsf{ED}, where varying sparsity required re-sampling. 
Throughout these experiments, we choose $r_U=0.2n$ to ensure that $r \leq r_U$ for all levels of $f_{adv}$.

\noindent\textbf{Selection of Regularization Parameters:}  
A two-stage procedure was used. A coarse grid search over $\log(\lambda_1)$, $\log(\lambda_2) \in [1,7]$ (step size $0.25$) was first screened using the Lilliefors test at the $1\%$ level, retaining pairs for which at least $70\%$ of standardized debiased coordinates were approximately Gaussian. Among these, the final pair minimized the average 10-fold cross-validation error computed using robust LASSO fits.

\noindent\textbf{Implementation Platform:}  
All estimators were implemented in \texttt{MATLAB} using \texttt{CVX} with \texttt{SDPT3}.

\noindent\textbf{Evaluation Criteria:}  
Detection performance was quantified using sensitivity and specificity at a $5\%$ significance level. Let $\boldsymbol{\hat{b}}_\beta \in \{0,1\}^p$ denote hypothesis-test decisions on coordinates of $\boldsymbol{\hat{\beta}}_W$. A coordinate was declared defective if its $p$-value fell below the threshold. Sensitivity and specificity were defined as $\text{sensitivity}=\frac{\#\text{true defectives}}{\#\text{true defectives}+\#\text{false non-defectives}},$ and $\text{specificity}=\frac{\#\text{true non-defectives}}{\#\text{true non-defectives}+\#\text{false defectives}}.$
Reconstruction accuracy was evaluated using the relative RMSE given by $\text{RRMSE}=\frac{\|\boldsymbol{\beta}^*-\hat{\boldsymbol{\beta}}\|_2}{\|\boldsymbol{\beta}^*\|_2}.$

\subsection{Performance of the Correction Algorithm under MMEs and Additive Noise}
We evaluate the performance of \textsc{Cape} for correcting \textsf{SSM} errors where $\boldsymbol{\eta}$ follows $\mathcal{N}(0,\sigma^2)$ and compare it with three baselines: (\textit{i}) Robust \textsc{Lasso} (\textsc{Rl}) from~\eqref{eq:RL_estimate}~\cite{Nguyen2013}, (\textit{ii}) \textsc{Odrlt}~\cite{Banerjee24} without correction, and (\textit{iii}) the MME rejection scheme \textsc{Mmer} (Sec.~\ref{subsec:detect}). For \textsc{Cape}, the estimate of $\boldsymbol{\beta}^*$ is obtained using the corrected sensing matrix $\boldsymbol{\hat{A}}$. All comparisons are performed for design matrices whose rows are sampled i.i.d.\ from $CB(\theta)$.

Fig.~\ref{fig:Sens_spec_ssm_0.1} and~\ref{fig:Sens_spec_ssm_0.5} present the sensitivity and specificity values (averaged over 25 noise realizations) for all four experiment settings \textsf{EE}, \textsf{EB}, \textsf{EC}, and \textsf{ED}, for $\theta \in \{0.1,0.5\}$. For the same configurations, the corresponding RRMSE comparisons (for \textsc{Cape}, \textsc{Rl} and \textsc{Mmer}) are shown in Fig.~\ref{fig:Rmse_ssm_0.1} and~\ref{fig:Rmse_ssm_0.5}. We omit \textsc{Odrlt} from the RRMSE comparisons as it produces only binary estimates. Across all experiment classes, \textsc{Cape} consistently outperforms the baselines. The improvement over \textsc{Mmer} is most pronounced for small $n$, larger corruption levels, and higher noise. A similar set of results are shown in the supplemental material for \textsf{ASM} and permutation errors, where again \textsc{Cape} outperforms all baselines.


\begin{figure*}[t]
   \centering
    \includegraphics[scale=0.19]{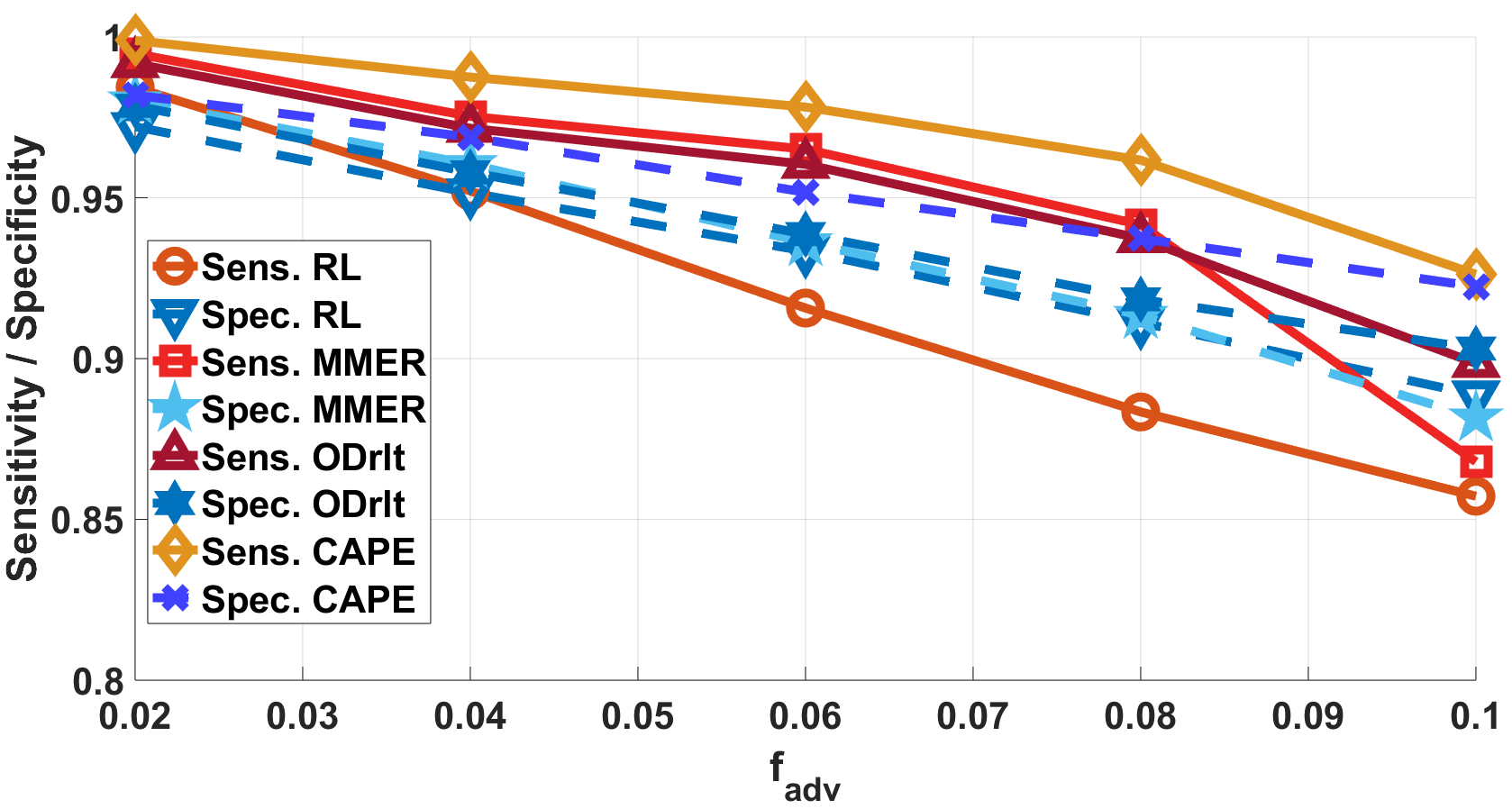}
    \includegraphics[scale=0.19]{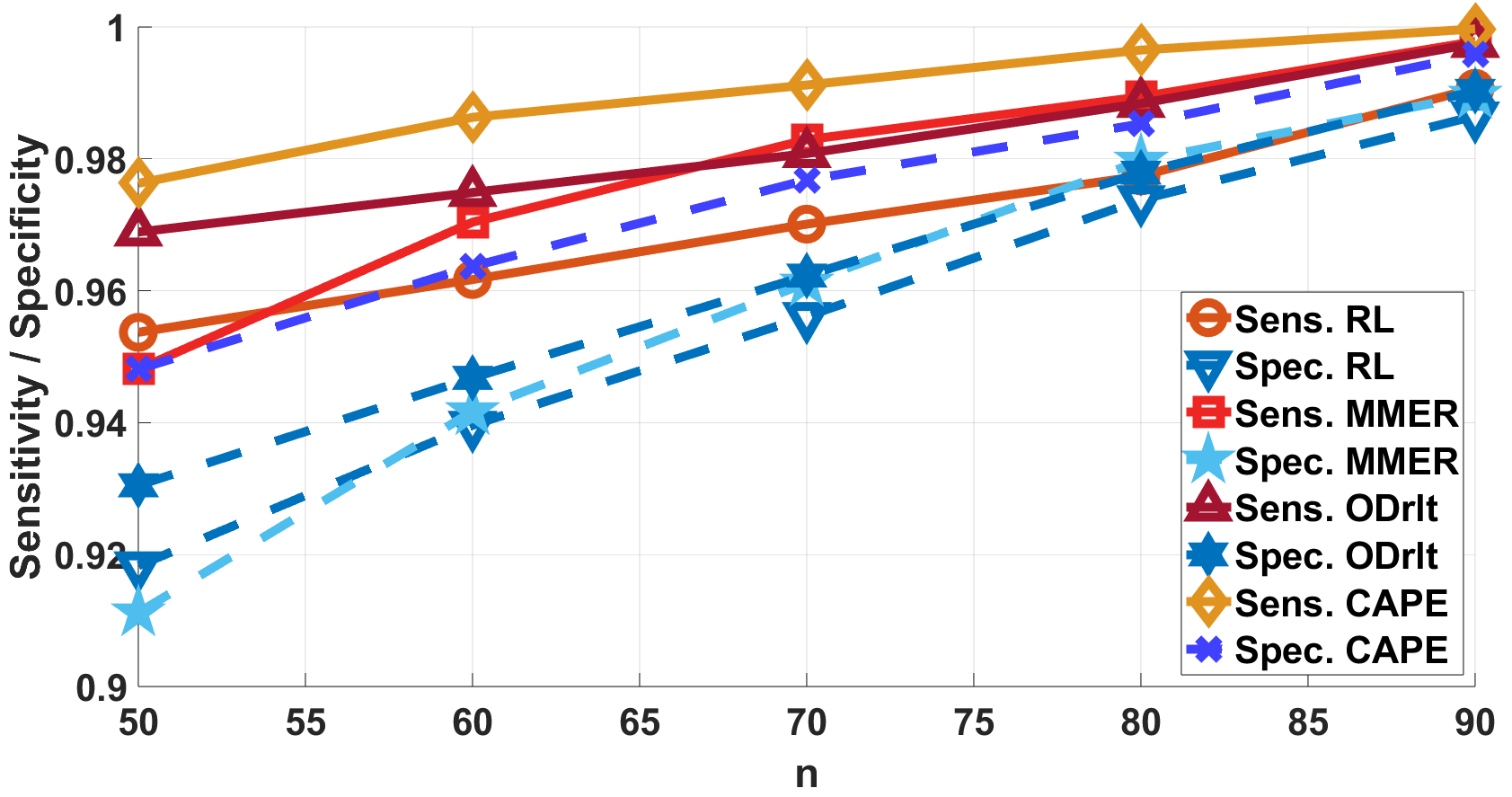}\\
    \includegraphics[scale=0.19]{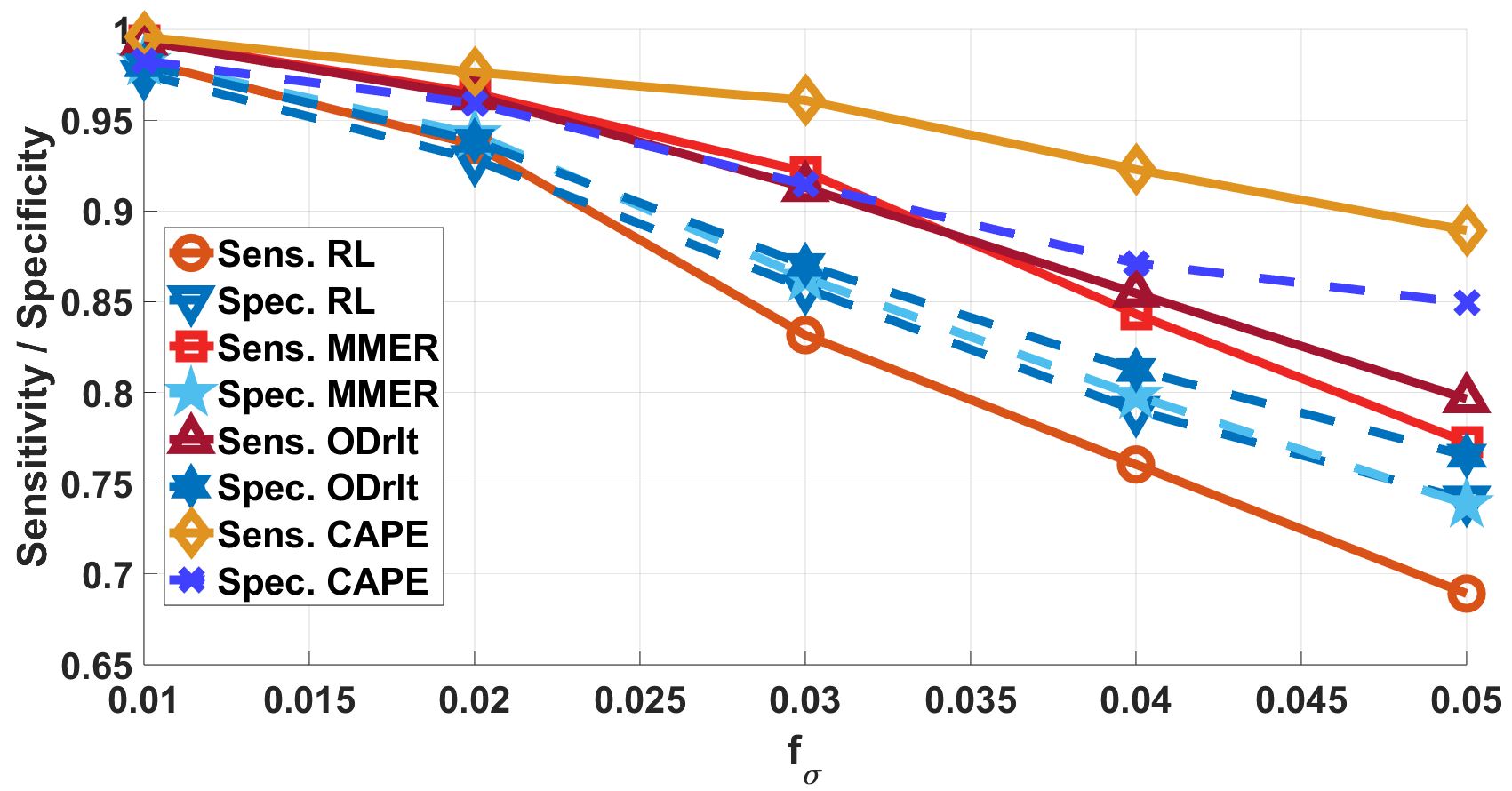}
    \includegraphics[scale=0.19]{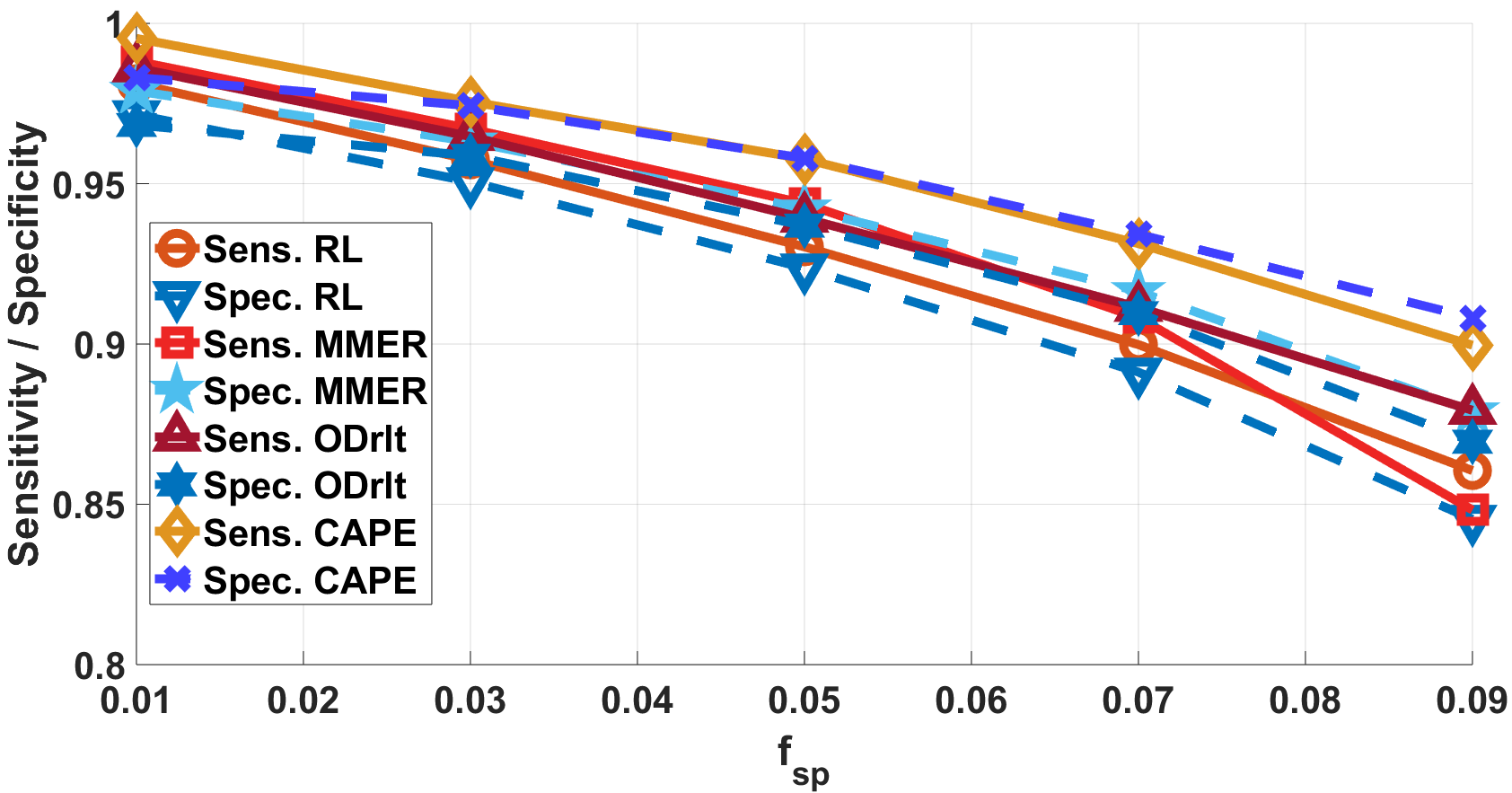}
    \caption{Average sensitivity and specificity (100 noise realizations; fixed $\boldsymbol{\beta}^*$, $\boldsymbol{A}$, $\boldsymbol{\delta}^*$) for detecting non-zero entries of $\boldsymbol{\beta}^*$ using \textsc{Cape}, RL, \textsc{Mmer}, and \textsc{Odrlt} in the presence of \textsf{SSM} MMEs. $\boldsymbol{A}$ is $CB(0.1)$. Panels correspond to (\textsf{EA}), (\textsf{EB}), (\textsf{EC}), (\textsf{ED}).}
    \label{fig:Sens_spec_ssm_0.1}
\end{figure*}

\begin{figure*}[t]
   \centering
    \includegraphics[scale=0.19]{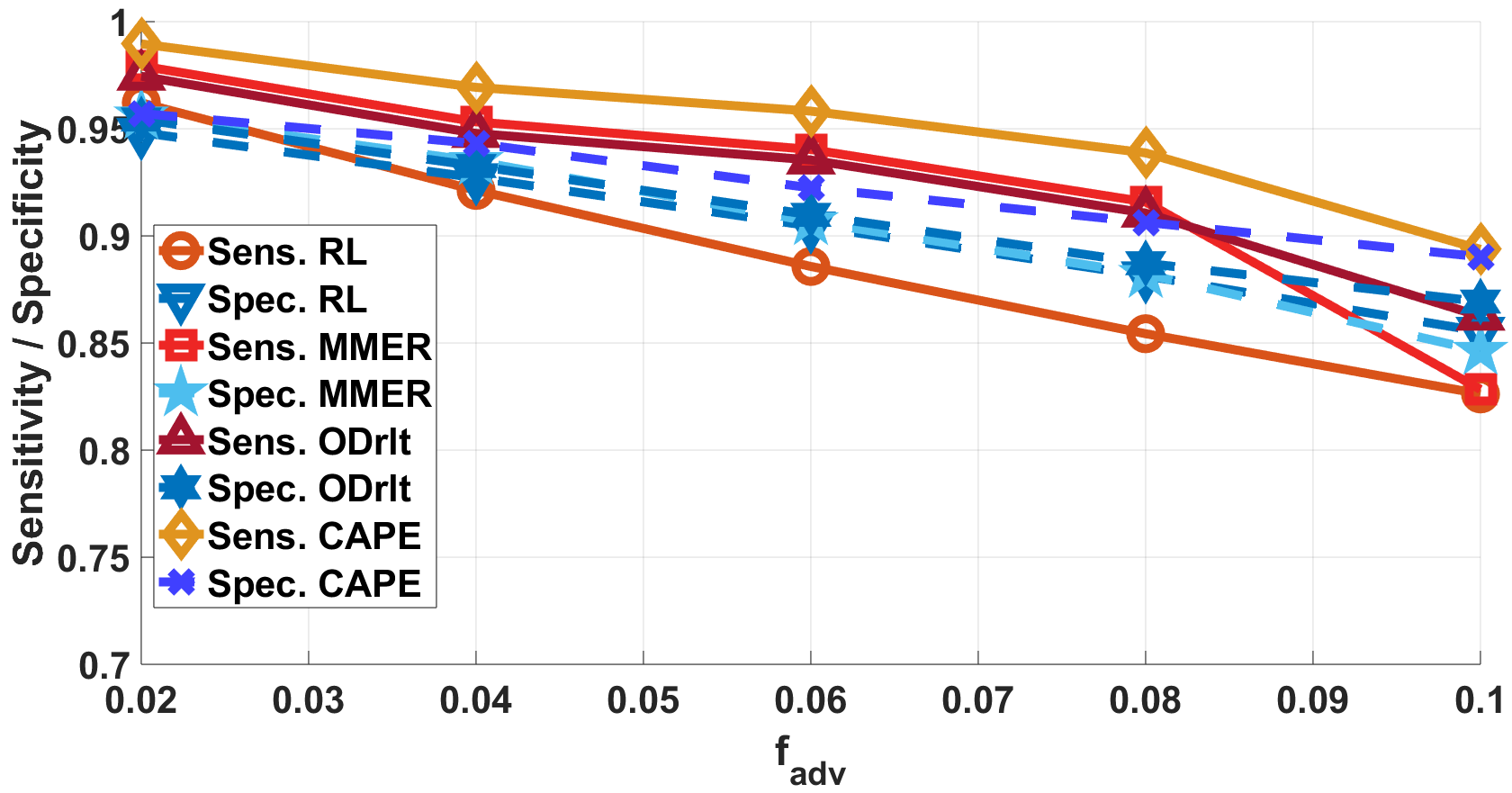}
    \includegraphics[scale=0.19]{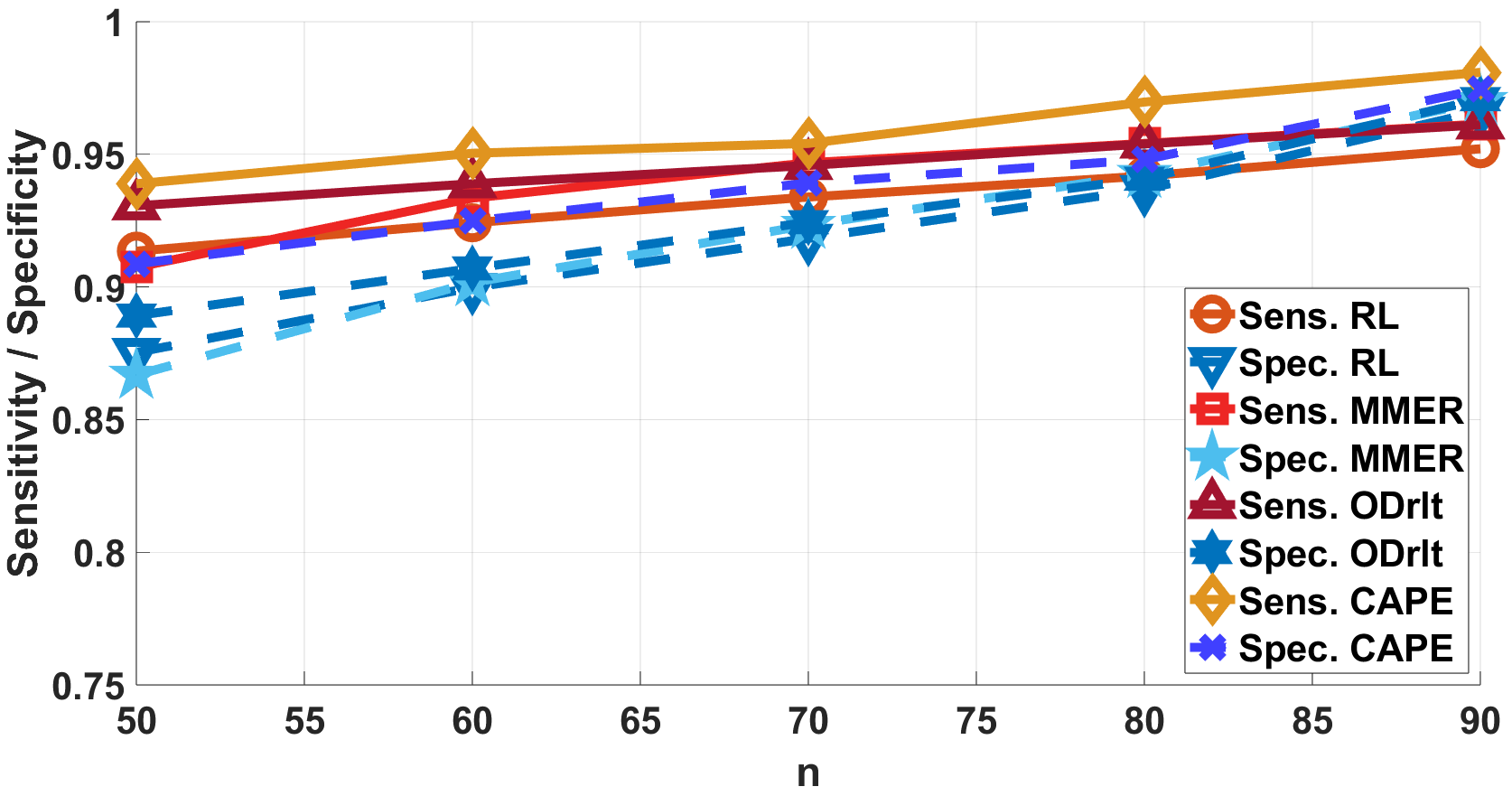}\\
    \includegraphics[scale=0.19]{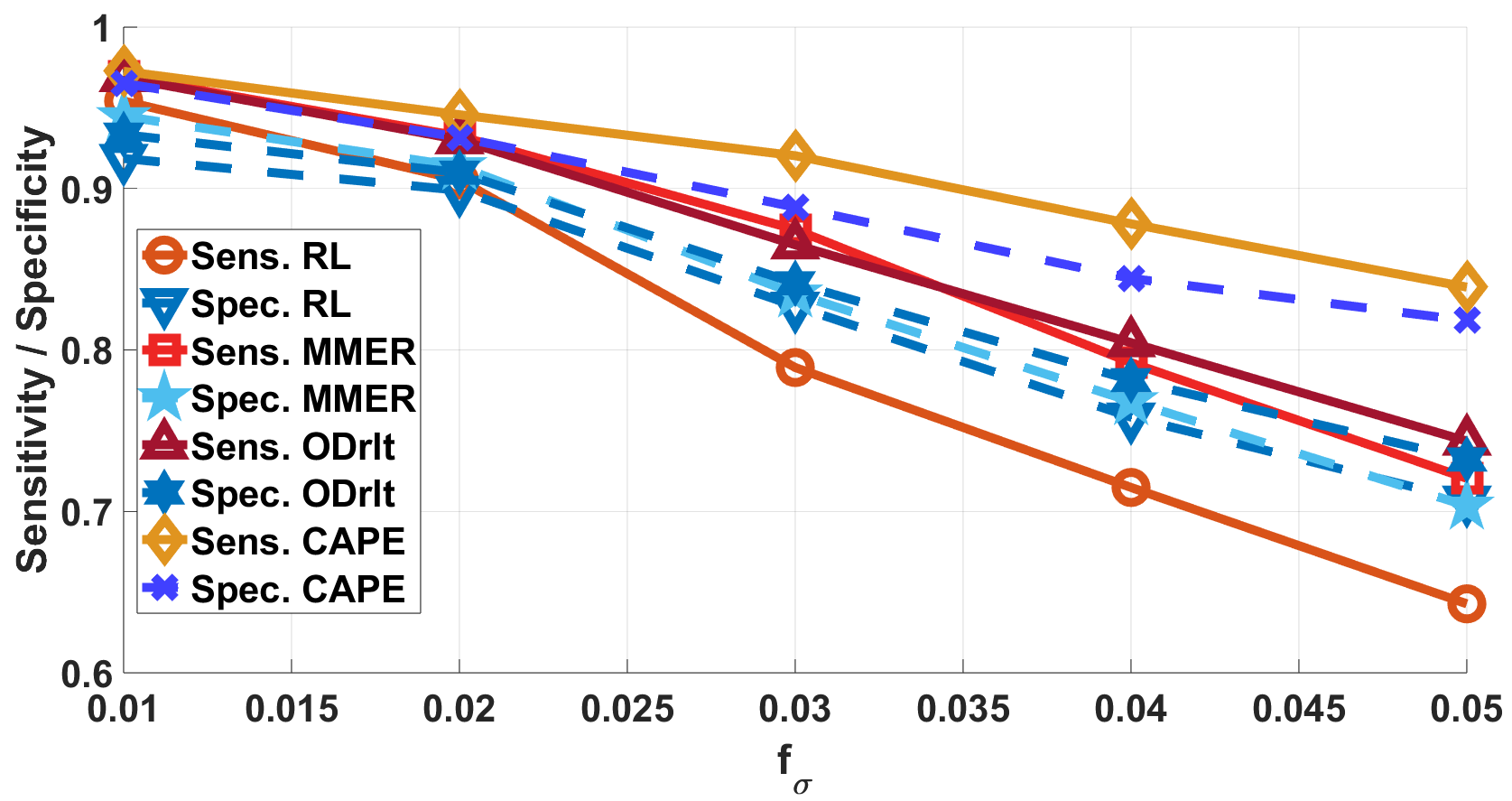}
    \includegraphics[scale=0.19]{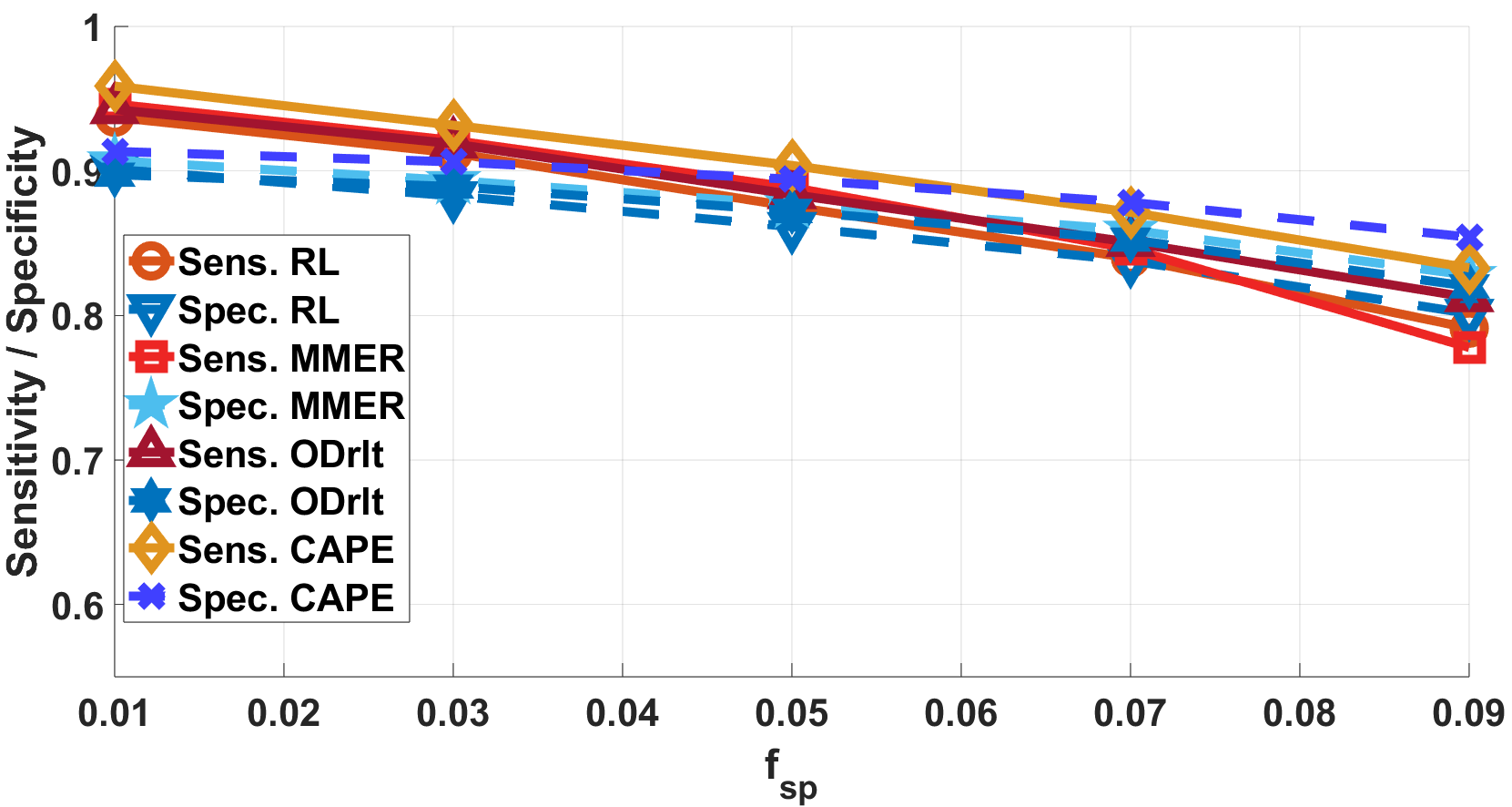}
    \caption{Average sensitivity and specificity (100 noise realizations; fixed $\boldsymbol{\beta}^*$, $\boldsymbol{A}$, $\boldsymbol{\delta}^*$) for detecting non-zero entries of $\boldsymbol{\beta}^*$ using \textsc{Cape}, RL, \textsc{Mmer}, and \textsc{Odrlt} under \textsf{SSM} MMEs. $\boldsymbol{A}$ is $CB(0.5)$. Panels correspond to (\textsf{EA}), (\textsf{EB}), (\textsf{EC}), (\textsf{ED}).}
    \label{fig:Sens_spec_ssm_0.5}
\end{figure*}

\begin{figure*}[t]
   \centering
    \includegraphics[scale=0.19]{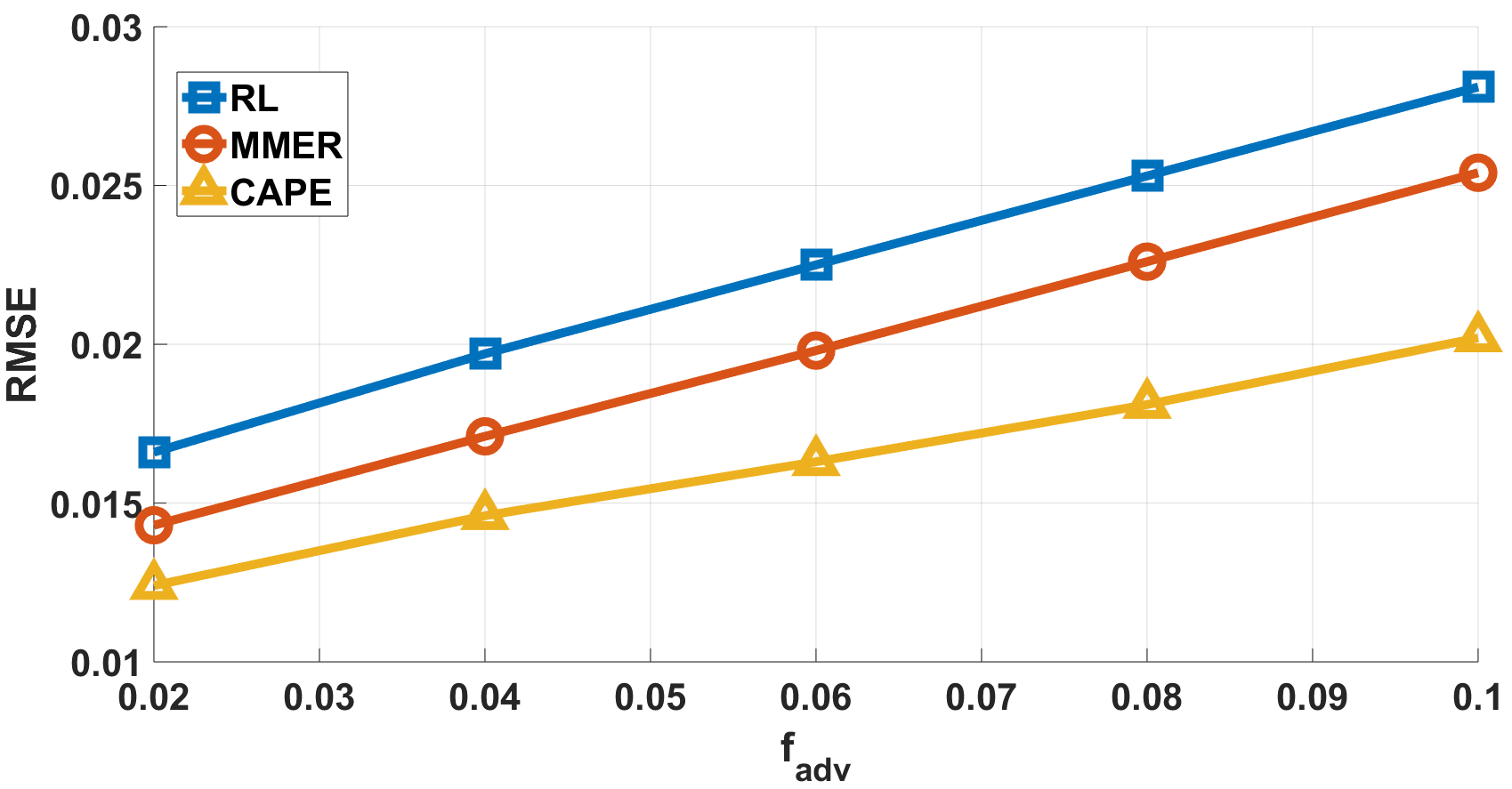}
    \includegraphics[scale=0.19]{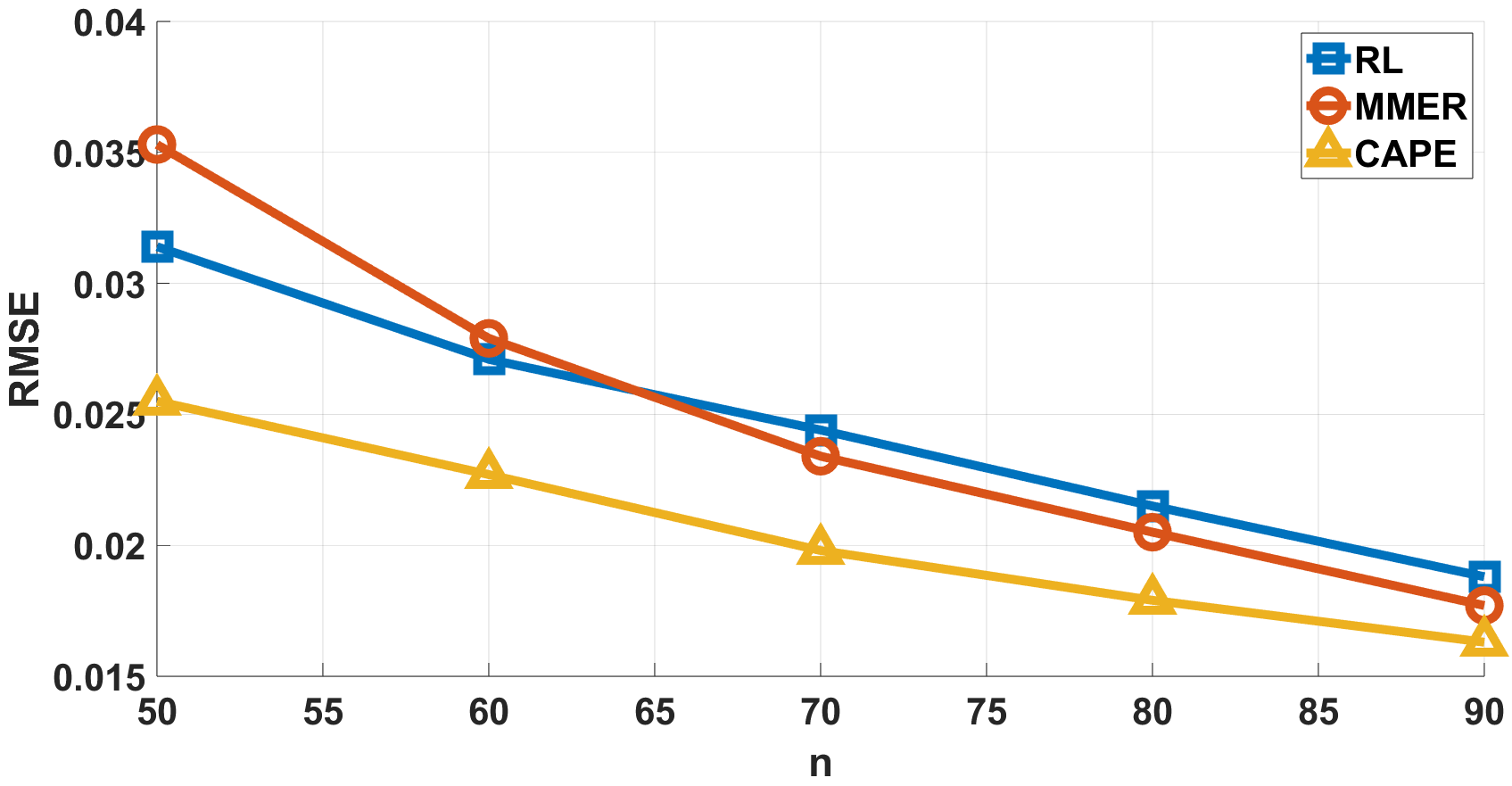}\\
    \includegraphics[scale=0.19]{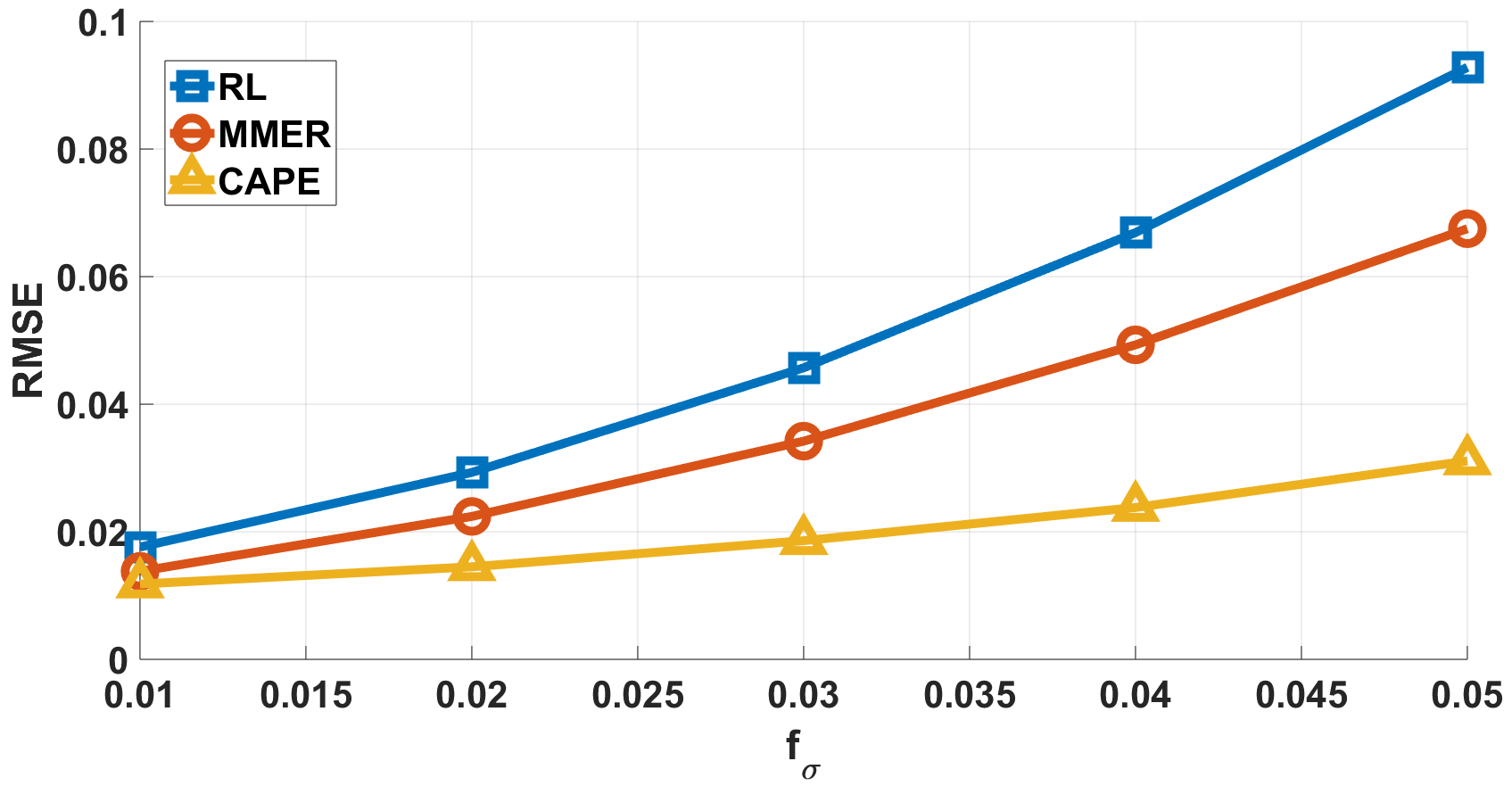}
    \includegraphics[scale=0.19]{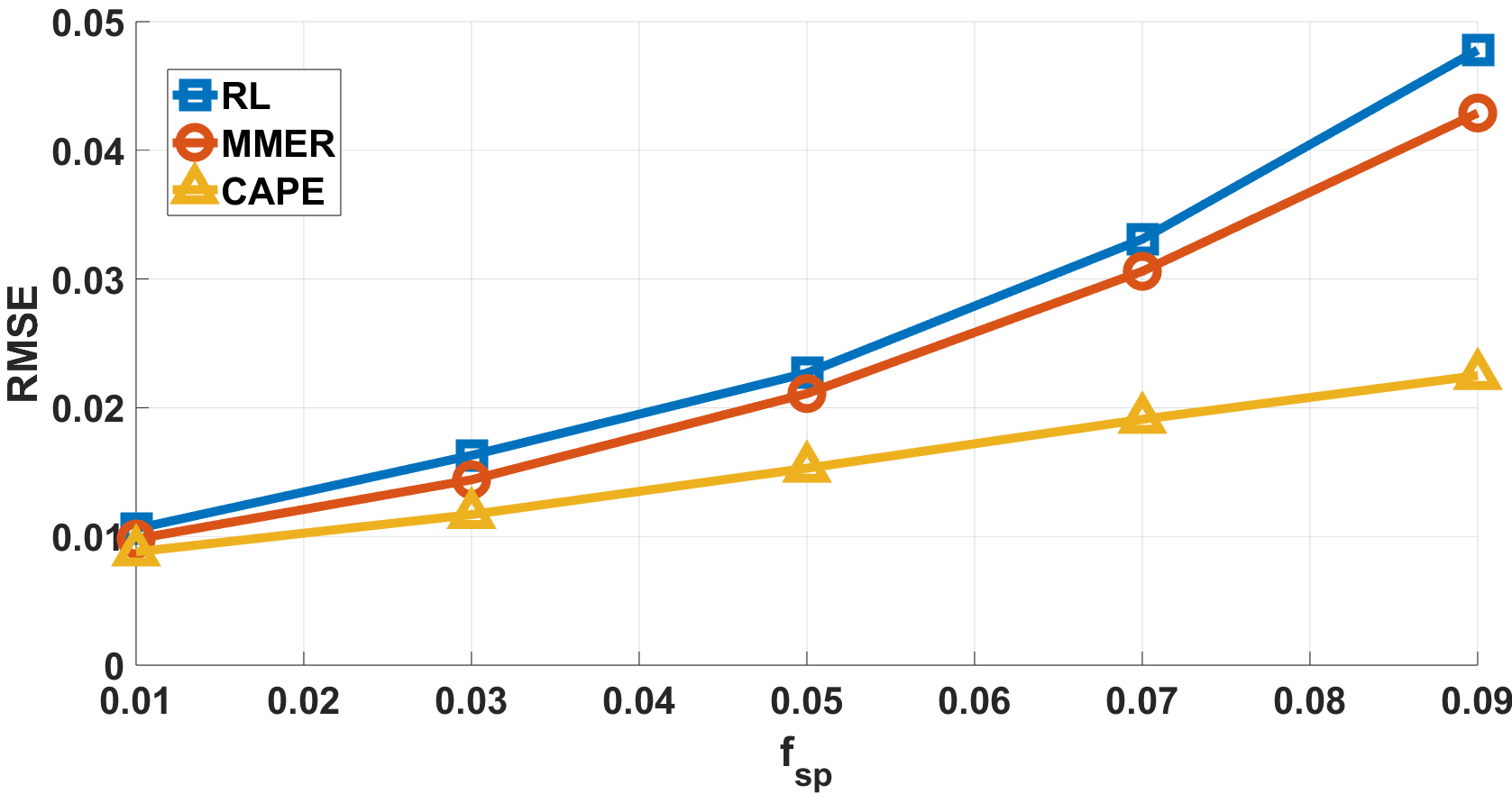}
    \caption{Average RRMSE (100 noise realizations) under \textsf{SSM} mismatches for $\theta=0.1$. Methods: \textsc{Cape}, RL, \textsc{Mmer}. Panels correspond to (\textsf{EA}), (\textsf{EB}), (\textsf{EC}), (\textsf{ED}).}
    \label{fig:Rmse_ssm_0.1}
\end{figure*}

\begin{figure*}[t]
   \centering
    \includegraphics[scale=0.19]{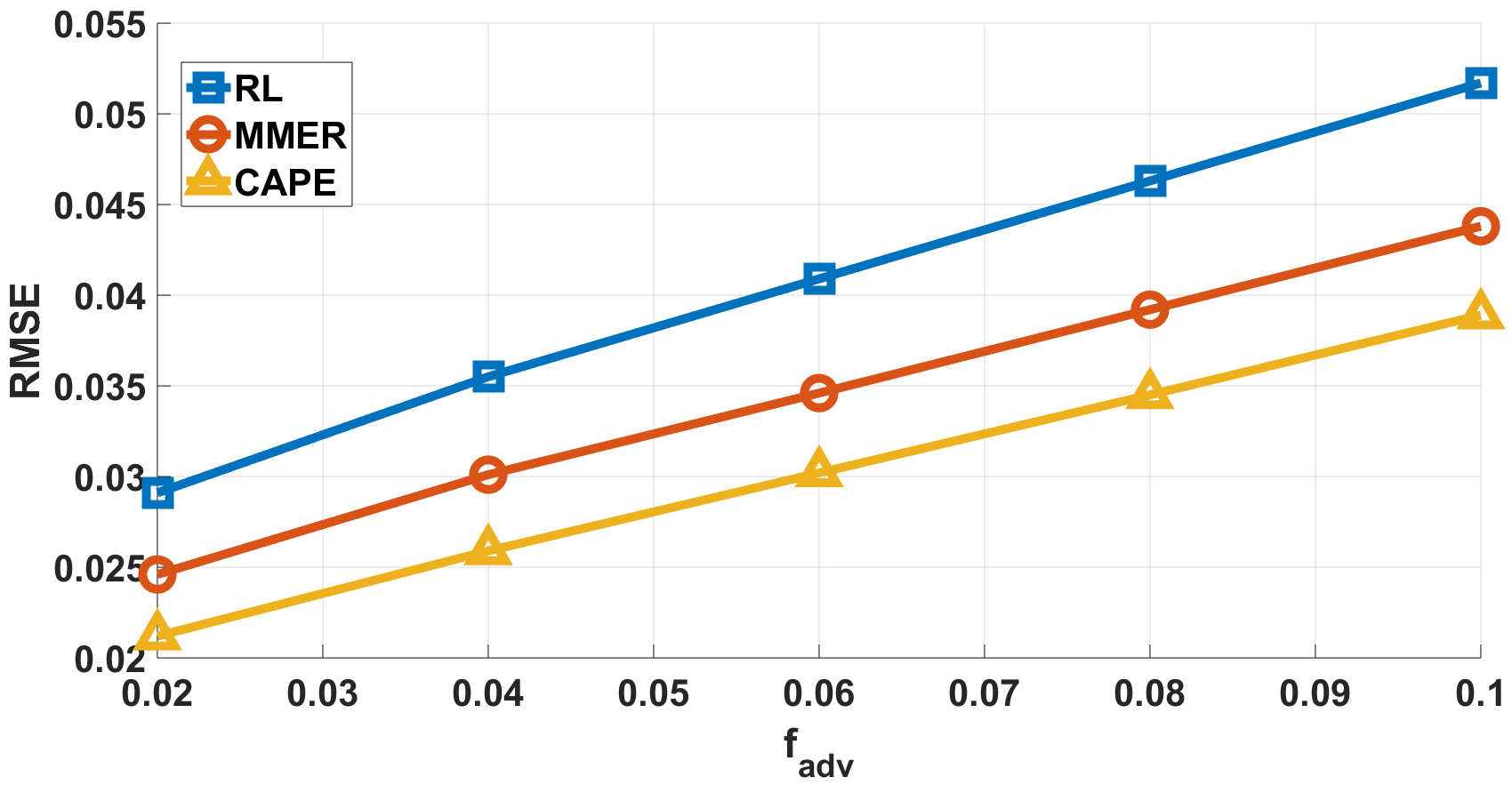}
    \includegraphics[scale=0.19]{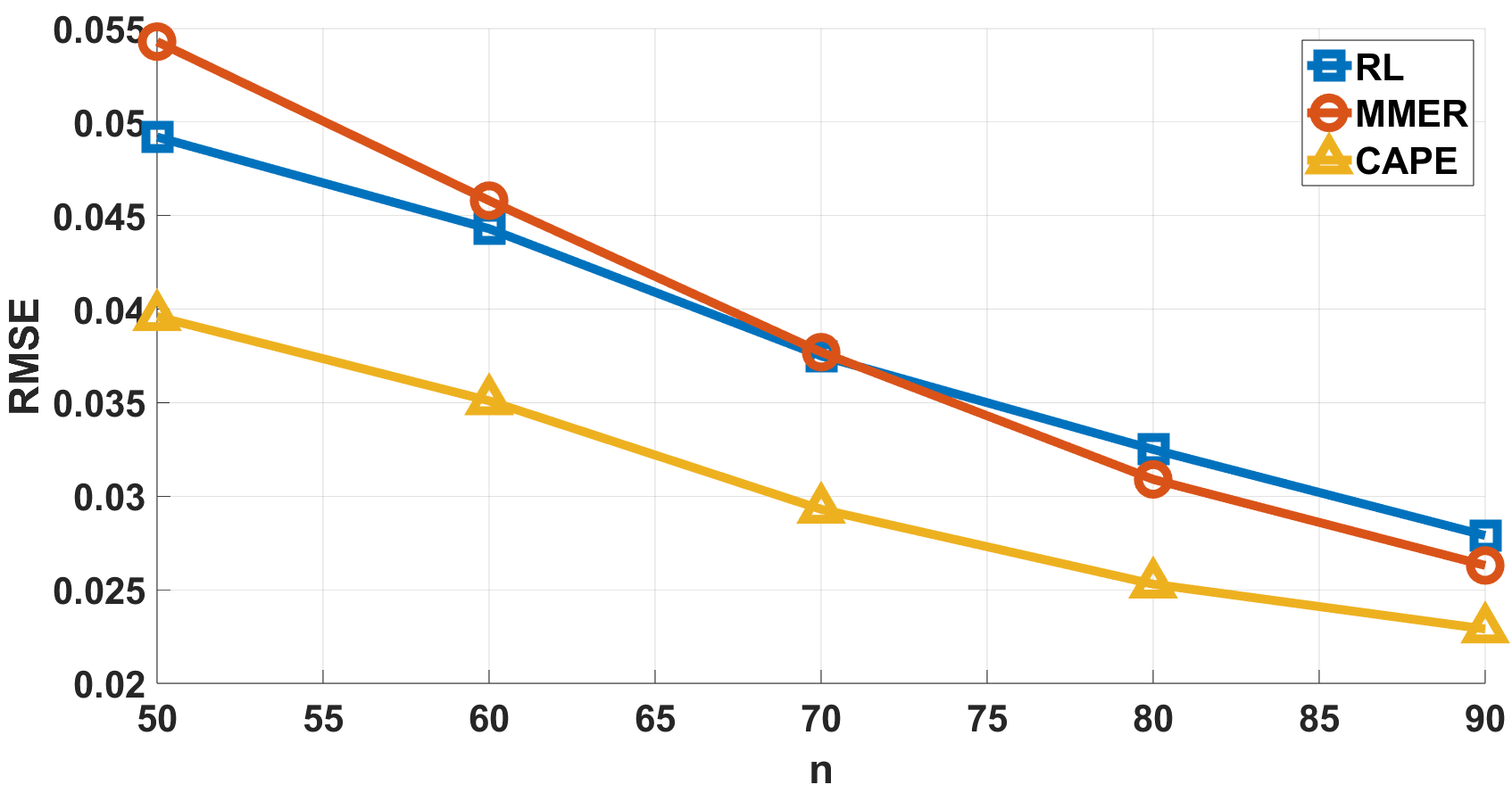}\\
    \includegraphics[scale=0.19]{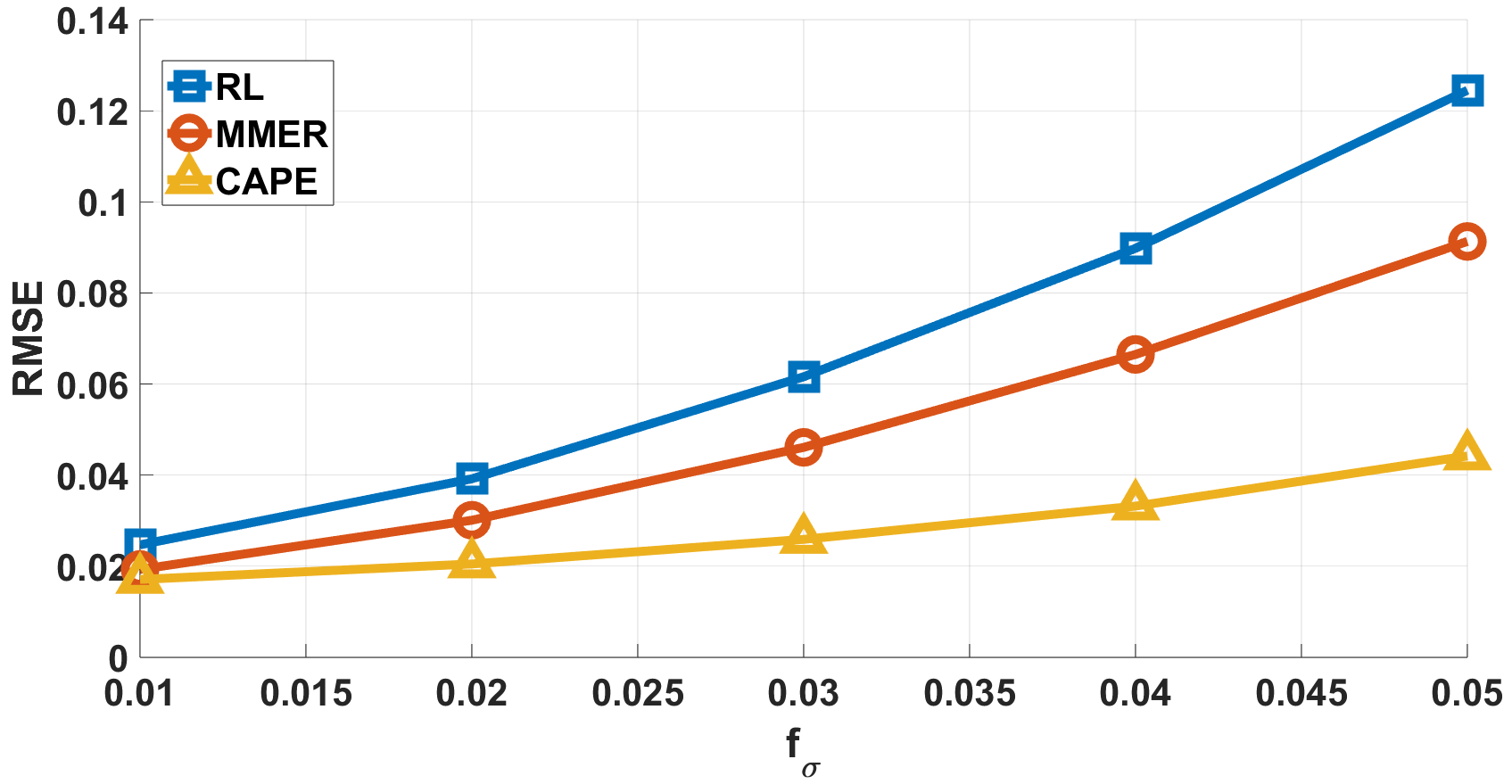}
    \includegraphics[scale=0.19]{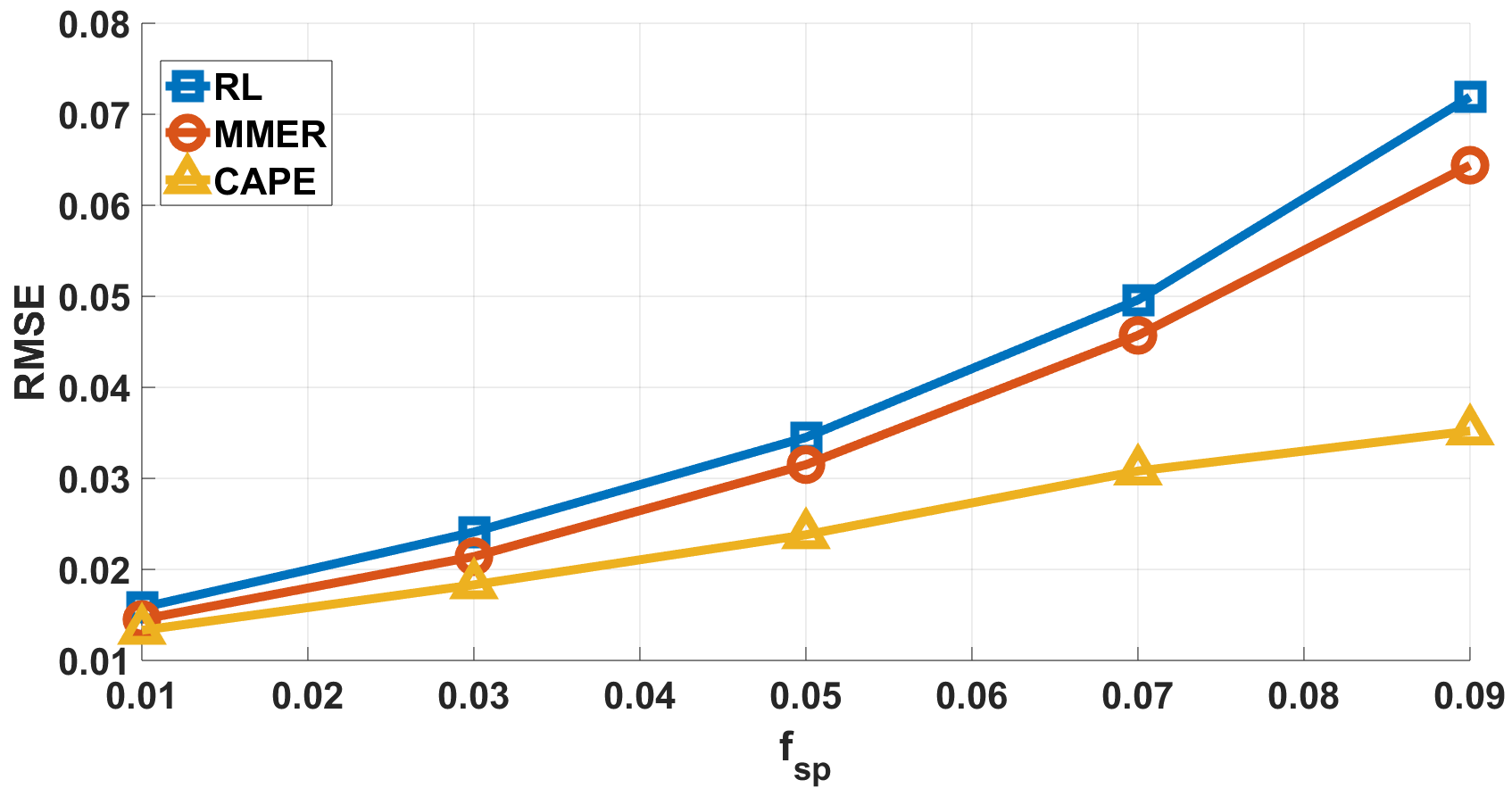}
    \caption{Average RRMSE (100 noise realizations) under \textsf{SSM} mismatches for $\theta=0.5$. Methods: \textsc{Cape}, RL, \textsc{Mmer}. Panels correspond to (\textsf{EA}), (\textsf{EB}), (\textsf{EC}), (\textsf{ED}).}
    \label{fig:Rmse_ssm_0.5}
\end{figure*}

\subsection{Performance Under Log-normal (Gaussian-Approximated) PCR Noise Model}
\begin{table*}[t]
\centering
\caption{Comparison of \textsc{Rl}, \textsc{Mmer} and \textsc{Cape} in the presence of log-normal noise and permutation errors under varying noise standard deviation. The fixed parameters are $n=400, s=10, r=8$.}
\renewcommand{\arraystretch}{1.15}
\setlength{\tabcolsep}{6pt}
\begin{tabular}{|c|ccc|ccc|ccc|}
\hline
$f_{\sigma}$ &
RMSE$_\mathrm{RL}$ & RMSE$_\mathrm{MMER}$ & RMSE$_\mathrm{CAPE}$ &
Sens$_\mathrm{RL}$ & Sens$_\mathrm{MMER}$ & Sens$_\mathrm{CAPE}$ &
Spec$_\mathrm{RL}$ & Spec$_\mathrm{MMER}$ & Spec$_\mathrm{CAPE}$ \\
\hline
0.002 & 0.0269 & 0.0223 & 0.0178 & 0.9814 & 0.9996 & 1.0000 & 0.9918 & 1.0000 & 1.0000 \\
0.004 & 0.0312 & 0.0279 & 0.0207 & 0.9722 & 0.9829 & 0.9945 & 0.9779 & 0.9902 & 0.9956 \\
0.006 & 0.0378 & 0.0308 & 0.0248 & 0.9498 & 0.9695 & 0.9872 & 0.9502 & 0.9725 & 0.9838 \\
0.008 & 0.0509 & 0.0411 & 0.0315 & 0.9017 & 0.9331 & 0.9583 & 0.9268 & 0.9498 & 0.9621 \\
0.010 & 0.0782 & 0.0575 & 0.0398 & 0.8437 & 0.8948 & 0.9194 & 0.8764 & 0.9073 & 0.9235 \\
\hline
\end{tabular}
\label{tab:log_noise}
\end{table*}

In this section, we evaluate the performance of three estimation--correction strategies for sparse permutation errors in pooled testing under multiplicative lognormal measurement noise: 
(i) the Robust Lasso (\textsc{Rl}), 
(ii) the Model-Mismatch Error Rejection (\textsc{Mmer}) algorithm, and 
(iii) the proposed \textsc{Cape} correction algorithm.
The measurement model and noise characteristics follow the \textit{Tapestry} framework for RT-PCR based pooled testing~\cite{Ghosh2021}.

For each positive pool $j$, the clean viral load is given by $z_j = A_j x,$ and the noisy measurement produced by the RT-PCR machine is $y_j = z_j (1+q)^{e_j},$ with $e_j \sim \mathcal{N}(0,\sigma^2),$
as derived in Eqn.~(7) of \cite{Ghosh2021}. 
Here $q \in (0,1)$ represents the per-cycle replication factor of cDNA, and $e_j$ denotes the discrepancy between the true cycle threshold $\tau_j$ and the observed threshold $t_j$, where $\tau_j = t_j + e_j$ and $e_j$ is modeled as $\mathcal{N}(0,\tilde{\sigma}^2)$. In addition to this physically derived noise model, our setting includes sparse permutation errors of sparsity $r$ in the pooling matrix, representing mislabeled or misplaced sample contributions.

We first study reconstruction performance as a function of the noise standard deviation $f_{\mathrm{sig}}$. The pooling parameters are fixed as $n = 400, s = 10, r = 8$ and $r_U=0.2 n$. In Table ~\ref{tab:log_noise}, we see that \textsc{Cape} consistently obtains the lowest RMSE and the highest sensitivity and specificity across all noise levels. 
\textsc{Mmer} comes second with \textsc{Rl} being the worst. 
All methods deteriorate as noise grows, but the degradation is steepest for \textsc{Rl}, indicating greater sensitivity to log-normal noise and permutation errors. \textsc{Cape} maintains its advantage throughout the tested regime. For $f_{\sigma}=0.01$, due to log-normal noise, the first-stage detection procedure rejects a substantial number of measurements, resulting in $|\mathcal{J}| > r_U$ due to Type-I errors. Consequently, some MMEs remain undetected and uncorrected at this stage. Nonetheless, the vast majority of these missed MMEs are successfully identified and corrected in subsequent stages, demonstrating the robustness of \textsc{Cape}.

\begin{table*}[t]
\centering
\caption{Comparison of \textsc{Rl}, \textsc{Mmer} and \textsc{Cape} in the presence of log-normal noise and permutation errors under varying number of measurements $n$. The fixed parameters are $f_{\sigma}=0.01, s=10, r=8$.}
\renewcommand{\arraystretch}{1.15}
\setlength{\tabcolsep}{6pt}
\begin{tabular}{|c|ccc|ccc|ccc|}
\hline
$n$ &
RMSE$_\mathrm{RL}$ & RMSE$_\mathrm{MMER}$ & RMSE$_\mathrm{CAPE}$ &
Sens$_\mathrm{RL}$ & Sens$_\mathrm{MMER}$ & Sens$_\mathrm{CAPE}$ &
Spec$_\mathrm{RL}$ & Spec$_\mathrm{MMER}$ & Spec$_\mathrm{CAPE}$ \\
\hline
250 & 0.1143 & 0.0989 & 0.0892 & 0.8255 & 0.8463 & 0.8812 & 0.8579 & 0.8942 & 0.9195 \\
300 & 0.0941 & 0.0778 & 0.0512 & 0.8922 & 0.9084 & 0.9236 & 0.9119 & 0.9388 & 0.9544 \\
350 & 0.0592 & 0.0403 & 0.0294 & 0.9317 & 0.9503 & 0.9618 & 0.9515 & 0.9649 & 0.9795 \\
400 & 0.0332 & 0.0234 & 0.0178 & 0.9685 & 0.9788 & 0.9872 & 0.9821 & 0.9885 & 0.9924 \\
450 & 0.0252 & 0.0156 & 0.0114 & 0.9906 & 0.9974 & 1.0000 & 0.9986 & 1.0000 & 1.0000\\
\hline
\end{tabular}
\label{tab:log_n}
\end{table*}

Next, we investigate performance under varying number of pools $n$, while keeping $s=10, r=8,$ and $f_{\sigma}=0.01$ fixed.
From Table.\ref{tab:log_n}, increasing the number of measurements significantly improves performance for all estimators. 
\textsc{Cape} benefits the most from additional redundancy, achieving perfect sensitivity and specificity for $n \ge 400$. 
\textsc{Mmer} consistently improves upon RL by removing permutation affected rows before reconstruction. 
The performance gap between \textsc{Cape} and the other algorithms is largest in the moderate measurement regime ($n = 250$--$350$), where explicit MME correction is especially beneficial even in the presence of log-normal noise.

\subsection{Convergence of Stopping functions}
Here, we show that the values of the stopping function given in \eqref{eq:halt_est} based on the \textbf{APE} distance measure converge after multiple stages of correction--see of Alg.~\ref{alg:gen_corr}. We computed the natural logarithm of the function $f_{ape}$ as defined in \eqref{eq:APE} for varying number of permutation errors $r \in \{0,2,4,6,8,10\}$, keeping $\boldsymbol{z}, \boldsymbol{B}, \boldsymbol{\eta}$ fixed. In each case, we computed the stopping function value for $10$ iterations, at which convergence was observed. The other experimental parameters were set to $p=200,n=80,f_{sp}=0.05,f_{\sigma}=0.01$. 
\begin{figure}
    \centering
    \includegraphics[scale=0.2]{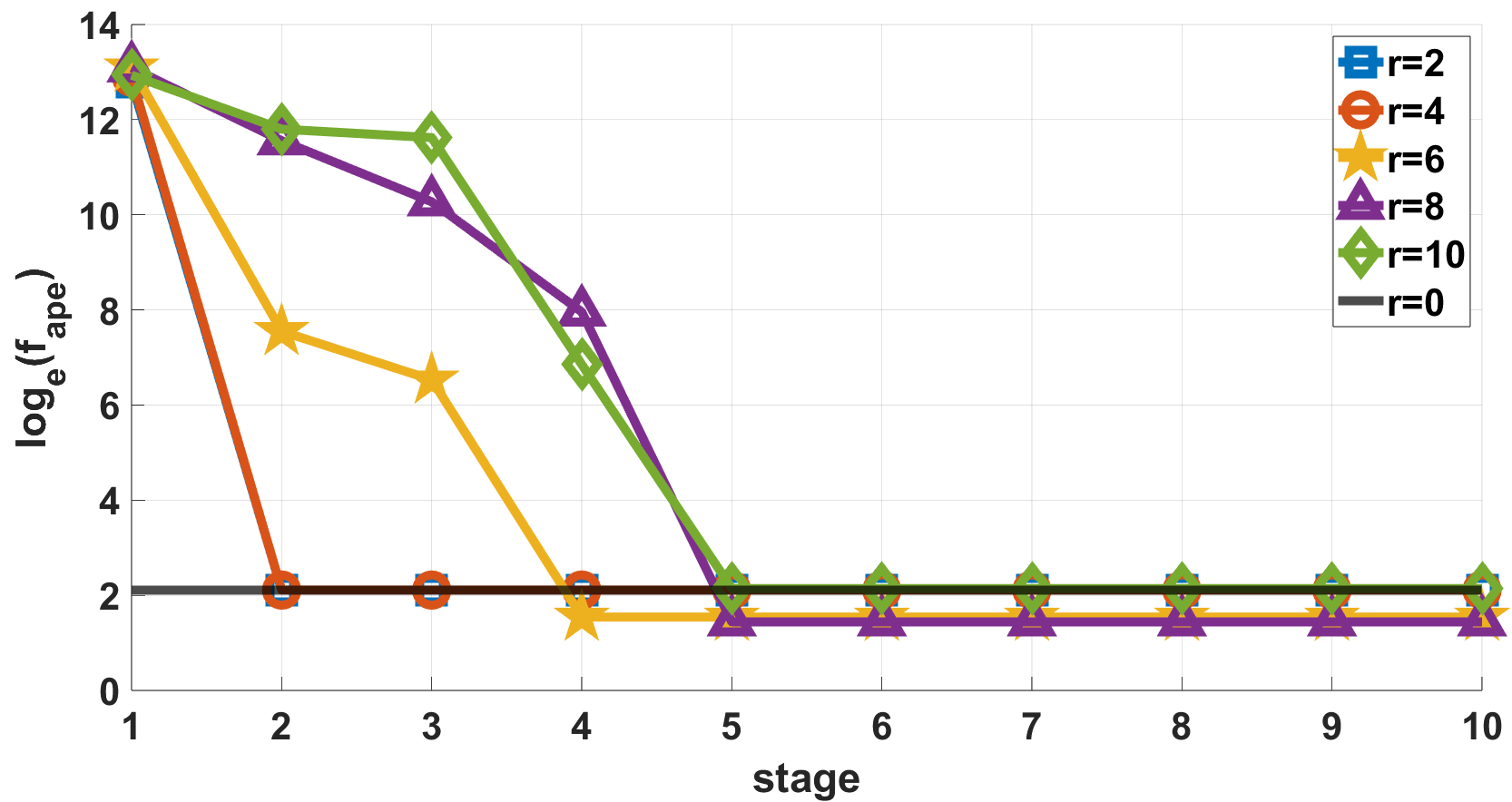}
    \caption{Convergence of Stopping functions: Convergence of stopping function $f_{ape}$ for correction algorithm \textsc{Cape} for permutations errors. Each line in the plot represents the respective function values at all $10$ stages of correction for different numbers of permutation errors given by  $r \in \{0,2,4,6,8,10\}$. The fixed parameters are $p=100,n=80,f_{sp}=0.05,f_{\sigma}=0.01$.}
    \label{fig:conv_func}
\end{figure}
Fig.~\ref{fig:conv_func} shows that for \textsc{Cape}, the stopping function $f_{ape}$ converges for every $r$ to a value close to $f_{ape}$ for $r=0$. This implies that the correction algorithm fixes all or most of the permutation errors. Furthermore, the permutations that miscorrected do not affect the model adversely by the end of stage $10$.

\subsection{Correction of MMEs in multiple stages}
We now empirically study the behavior of \textsc{Cape} across multiple stages of correction as described in Sec.~\ref{sec:mult_Corr}. In this experiment for both the algorithms, the setup was as follows: $p=200, n=100, s=5, r=4, f_{\sigma}=0.01$. After each stage of correction, we reported the average number of measurements correctly detected to have MMEs ($N_{ET}$), the average number of measurements incorrectly detected to have MMEs ($N_{EF}$) and the average actual number of MMEs remaining ($N_E$). In each case, the average was computed over 20 independent noise runs, for a fixed signal $\boldsymbol{\beta}^*$, a fixed pooling matrix $\boldsymbol{B}$ and a fixed set of induced MMEs. We performed this experiment separately for \textbf{SSM}, \textbf{ASM} and Permutation errors. In Table \ref{tab:multiplestages}, we see that by the fourth stage, \textsc{Cape} corrects all $4$ errors in all three cases. We also observe that $N_E, N_{EF}, N_{ET}$ decrease rapidly after the first stage of correction and steadily in subsequent iterations.

\begin{table}
\centering
\begin{tabular}{|c|c|c|c|c|}
\hline
MME type & Stage & $N_E$ & $N_{EF}$ & $N_{ET}$ \\
\hline
Perm. & 1 & 4   & 3.6 & 1.6 \\
Perm. & 2 & 0.4 & 0.2 & 0.4 \\
Perm. & 3 & 0.2 & 0.2 & 0.2 \\
Perm. & 4 & 0   & 0   & 0.2 \\
\hline
\textsf{SSM} & 1 & 4   & 3.6 & 1.6 \\
\textsf{SSM} & 2 & 0.4 & 0.4 & 0.1 \\
\textsf{SSM} & 3 & 0 & 0 & 0.1 \\
\hline
\textsf{ASM} & 1 & 4   & 3.6 & 1.6 \\
\textsf{ASM} & 2 & 0.5 & 0.5 & 0.3 \\
\textsf{ASM} & 3 & 0 & 0 & 0.2 \\
\hline
\end{tabular}
\caption{Avg. \#MMEs ($N_P$), avg. \#MMEs detected correctly ($N_{PT}$) and incorrectly ($N_{PF}$), after each stage of correction for \textsc{Cape} for $r=4$ MMEs of different types. The average is computed over 20 independent noise runs.}
\label{tab:multiplestages}
\end{table}

\subsection{Performance of \textsc{Cape} under Variation in MME Numbers $r$ with a Matched Cutoff Parameter $r_U$}
We now compare the performance of \textsc{Cape} with that of \textsc{Mmer} and the Robust \textsc{Lasso} for different numbers of MMEs given by $r = f_{adv}n$ for all three different MME models keeping $n = 80, p = 200, f_{sp} = 0.05, f_{\sigma} = 0.01$ fixed and with elements of $B$ drawn from Bernoulli($0.1$). This experiment is similar to the setting \textsf{EA} defined in the beginning of this section. But in setting \textsf{EA}, for the detection stage in Alg.~\ref{alg:mme_detection}, we set $r_{U} := \zeta n$ with fixed $\zeta = 0.2$ for all levels $f_{adv}$. However, in this experiment, we set $\zeta := f_{adv} =\{0.02,0.04,0.06,\ldots,0.3\}$ due to which $r_U = r$. Thus, here our aim is to investigate the behavior of both the detection algorithm (\textsc{Mmer}) and the correction method (\textsc{Cape}) when the cut-off parameter $r_U$ matches the true MME count $r$. Note that the detection step may fail to identify all $r$ corrupted rows due to the Type-I error of the \textsc{Odrlt} tests and the because the cutoff on $|\mathcal{J}|$ now equals $r$. Thus this experiment assesses the robustness of \textsc{Cape} in the presence of such initial detection for all three MME models. 

For the correction algorithm, we compute the RRMSE of the Robust \textsc{Lasso} estimate of $\boldsymbol{\beta}^*$ using $\boldsymbol{z}$ and the corrected pooling matrix $\boldsymbol{\hat{B}}$.
In Fig.~\ref{fig:Rmse_zeta}, we plot $\ln (\text{RRMSE})$ (averaged over 25 independent noise runs) for all four estimators.
For all three different MME models, we see that until about $r \approx 0.2 n$, \textsc{Cape} performs the best followed by \textsc{Mmer} and \textsc{Rl}. This implies the correction algorithm works well till $r=0.2n$. The regime beyond that is not practical as one does not expect the technician to make so many mistakes.
\begin{figure}
   \centering
    \includegraphics[scale=0.19]{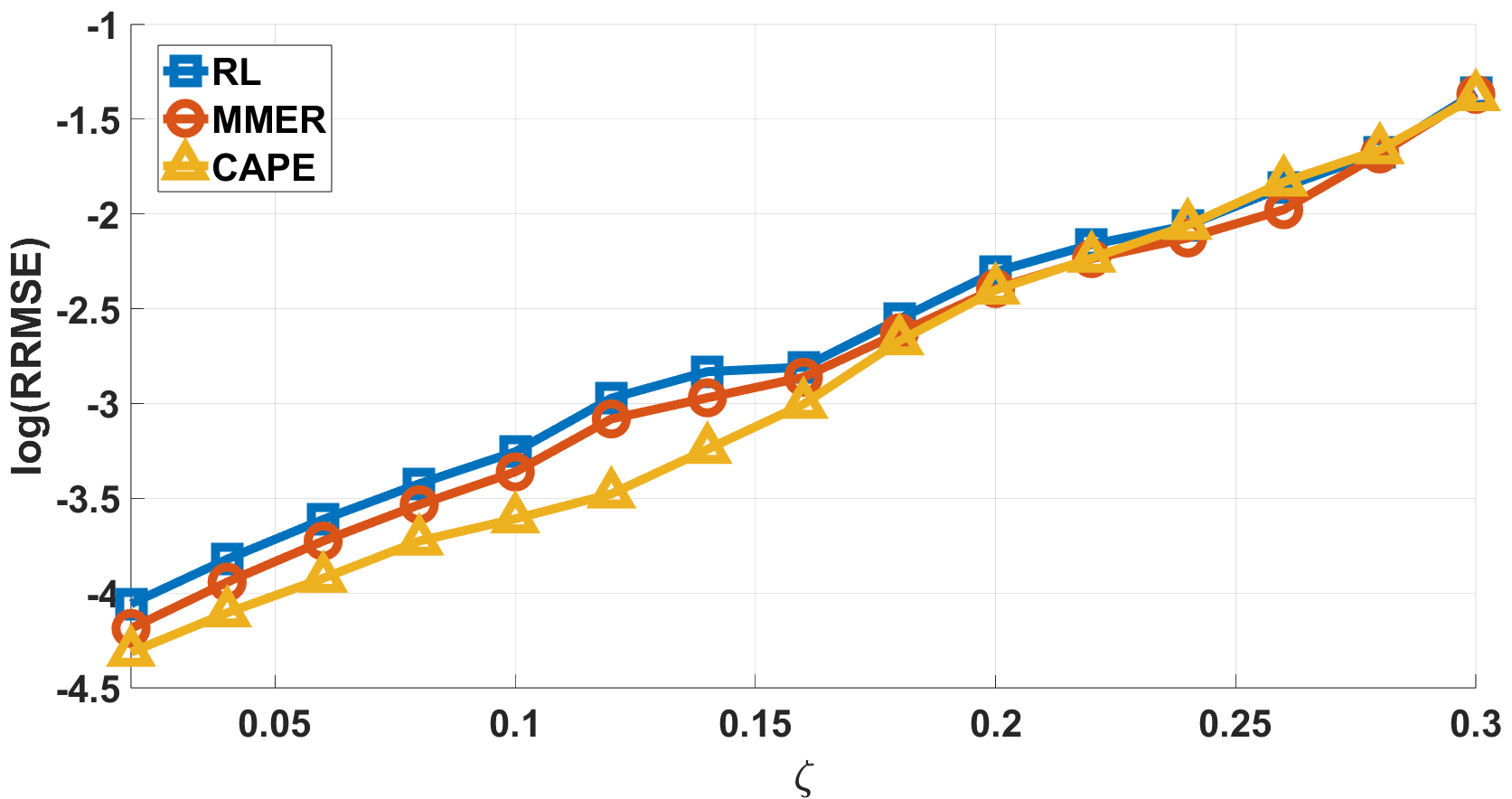}
    \includegraphics[scale=0.19]{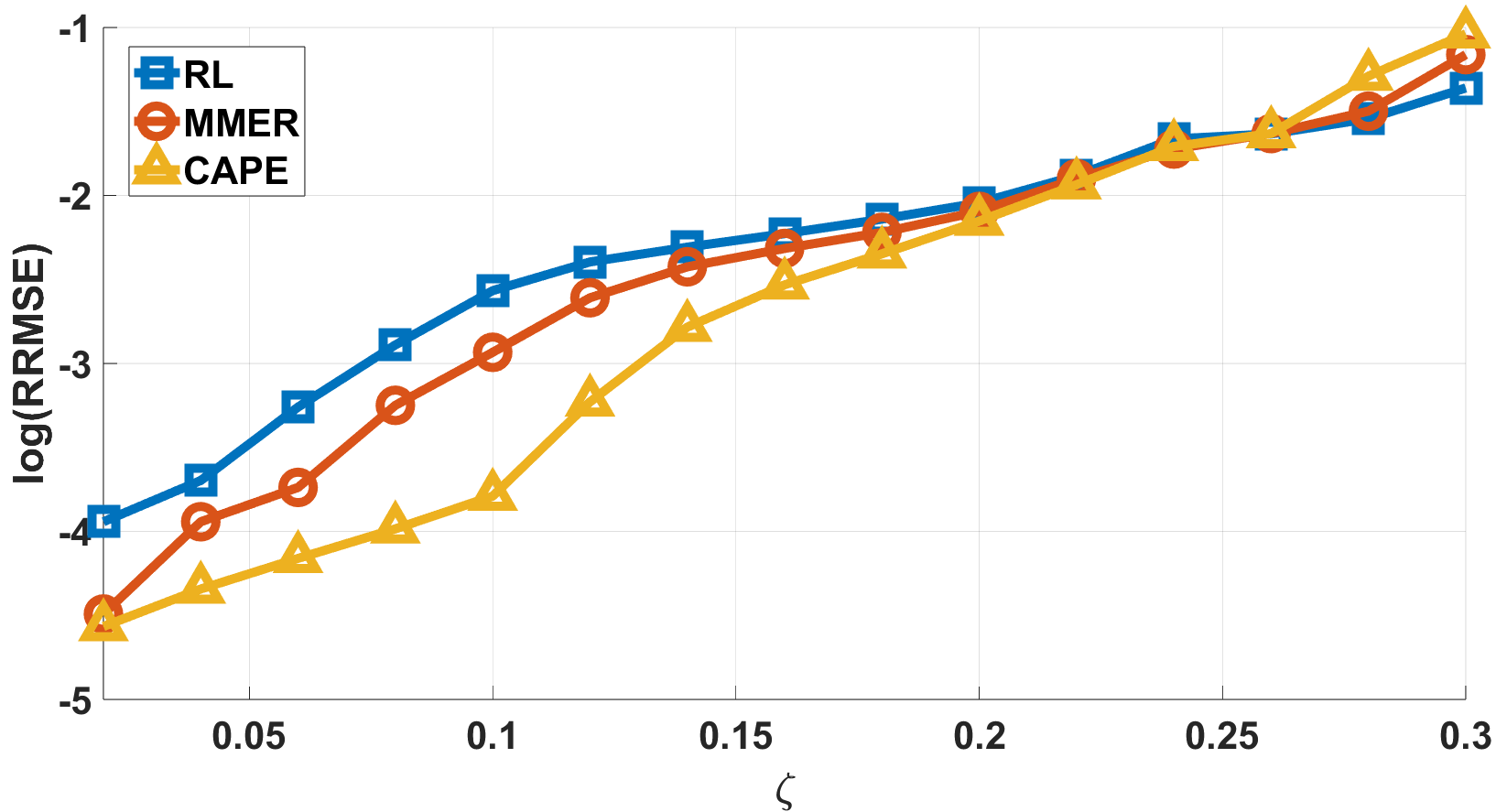}\\
    \includegraphics[scale=0.19]{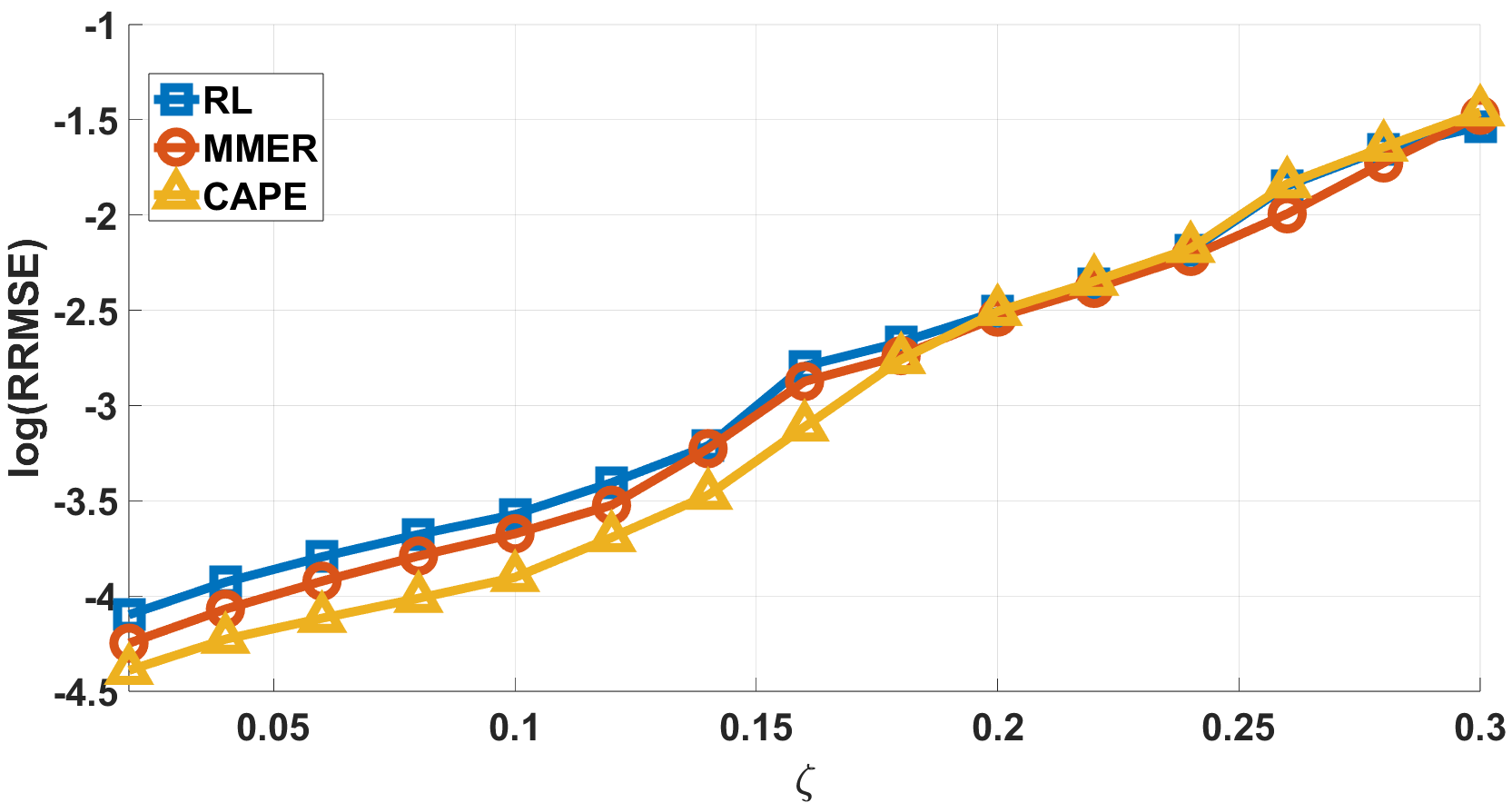}
    \caption{Average $\ln (\text{RRMSE})$ plots (over 100 independent noise runs keeping $\boldsymbol{\beta}^*$, $\boldsymbol{A}$ and $\boldsymbol{\delta}^*$ fixed) over varying $\zeta=r/n=\{0.02,0.04,\ldots,0.3\}$ for detecting defective samples (i.e., non-zero values of $\boldsymbol{\beta}^*$) using
    \textsc{Cape}, Robust \textsc{Lasso} (\textsc{Rl}) and \textsc{Mmer} in the presence of \textsf{SSM} MMEs(top), \textsf{ASM} MMEs(middle) and Permutation error (bottom).
    The design matrix $\boldsymbol{A}$ is has rows i.i.d. from Centered Bernoulli($0.1$). }  
    \label{fig:Rmse_zeta}    
\end{figure}

\section{Conclusions}
\label{sec:conclusions}
We have presented a novel technique for correction of model mismatch errors in pooled measurements in real-valued quantitative group testing, with theoretical guarantees for the quality of correction. It is interesting to note that when the number of MMEs is not too large, the measurements (even under noise and MMEs) contain enough information to allow detection as well as correction of pooled measurements with errors. In particular, neither the detection nor the correction of MMEs requires creation of new pools or any form of re-testing. We have worked with three different models for the MMEs: \textsf{SSM}, \textsf{ASM} and permutations. The \textsf{ASM} model can be easily extended to the case where bit-flips occur in not just adjacent pairs of locations, but any possible pair of locations in a given row of the pooling matrix. In principle, our technique is general enough to be adapted to a large ensemble of other MME models as well as to the case when MMEs of different models co-occur. If the precise model for MMEs is known only partially, the proposed method will still work, albeit at significantly higher computational cost due to the potentially large size of the model perturbation sets $\{\mathcal{C}_i\}_{i=1}^n$. A complete study of this scenario is left as future work. Our work here focuses on random matrices. Deterministic matrices have proved useful in group testing (e.g., Kirkman triples in \cite{Ghosh2021}, matrices such as those in \cite{devore2007deterministic}) and are also much sparser than random generated matrices. MMEs in deterministic matrices will inherently give rise to model perturbation sets $\{\mathcal{C}_i\}_{i=1}^n$ that are relatively smaller. However there is no guarantee that the corrected matrix $\boldsymbol{\hat{A}}$ will satisfy the theoretical conditions such as RIP, REC, etc. for successful signal reconstruction. Extending this work to handle such deterministic matrices is an interesting direction for future work.

\bibliographystyle{plain}
\bibliography{refs}
\end{document}


\title{Supplementary Material: ``Correction of Pooling Matrix Mis-specifications in Compressed Sensing Based Group Testing''}

\author{Shuvayan Banerjee, Radhendushka Srivastava, James Saunderson and Ajit Rajwade ~\IEEEmembership{Senior Member,~IEEE}}

\maketitle

\section{MME models and Construction of Perturbation set $\mathcal{C}_i$}\label{sec:MME_model_cent}

Let $\boldsymbol{B} \in \{0,1\}^{n \times p}$ be the binary pooling matrix, with elements drawn i.i.d. from $\text{Bernoulli}(\theta)$, where $0 < \theta < 1$. Let $n'=\lfloor n/2 \rfloor$.  
Define the centered matrix $\boldsymbol{A} \in \{-1,0,1\}^{n' \times p}$ as for all $i \in \{1,2,\ldots,\lfloor n/2 \rfloor\}$,
\begin{equation}
\boldsymbol{a}_{i\cdot} = \frac{\boldsymbol{b}_{i\cdot} - \boldsymbol{b}_{\lfloor n/2 \rfloor+i\cdot}}{2\theta(1-\theta)}.
\end{equation}  
For each error model, let $\boldsymbol{B}$ denote the binary pooling matrix after the occurrence of measurement matrix errors (MMEs), and let $\boldsymbol{A}$ denote its centered version obtained using the same transformation.
We define $\mathcal{C}_i$ as the set of all \emph{allowable} perturbations of the row $\boldsymbol{b}_{i\cdot}$ under the given error model. 
\subsubsection{\textbf{SSM (Single-Switch Model)}}
Under the Single-Switch Model (SSM), exactly one entry in a row of the binary pooling matrix is flipped. Specifically, for a given row index $i \in \{1,\dots,n\}$, one coordinate $j \in \{1,\dots,p\}$ is chosen uniformly at random, and the corresponding entry undergoes either an exclusion error $b_{ij}=1 \to 0$ or an inclusion error $b_{ij}=0 \to 1$. Let $\bar{\boldsymbol{b}}_{i\cdot}$ denote the perturbed version of $\boldsymbol{b}_{i\cdot}$. The set of allowable perturbations of $\boldsymbol{b}_{i\cdot}$ under the SSM is therefore given by
\begin{equation}
\mathcal{C}_i^{\mathrm{SSM}}
=
\bigl\{
\bar{\boldsymbol{b}}_{i\cdot} \in \{0,1\}^p :
\|\bar{\boldsymbol{b}}_{i\cdot} - \boldsymbol{b}_{i\cdot}\|_0 = 1
\bigr\}.
\end{equation}

\subsubsection{\textbf{ASM (Adjacent-Switch Model)}}
In the Adjacent-Switch Model (ASM), an error consists of swapping the values of two adjacent coordinates within a row of $\boldsymbol{B}$. More precisely, for a given row $i$, an index $j \in \{1,\dots,p\}$ is selected uniformly at random, and the pair of entries at positions $(j,j')$, where $j' = (j \bmod p) + 1$, are exchanged. This results in one of the following two cases, each occurring with equal probability: $(b_{ij}, b_{ij'}) = (0,1) \to (1,0), (b_{ij}, b_{ij'}) = (1,0) \to (0,1)$. Accordingly, the set of allowable perturbations of $\boldsymbol{b}_{i\cdot}$ under the ASM is defined as follows (with $j' = (j \bmod p) + 1$): 
\begin{eqnarray}
\mathcal{C}_i^{\mathrm{ASM}}
=
\bigl\{
\bar{\boldsymbol{b}}_{i\cdot} \in \{0,1\}^p :
\exists\, j \in [p] \text{ s.t. }
\|\bar{\boldsymbol{b}}_{i\cdot} - \boldsymbol{b}_{i\cdot}\|_0 = 2, \nonumber \\
(\bar{b}_{ij}, \bar{b}_{ij'}) \in \{(1,0),(0,1)\}
\bigr\}.
\end{eqnarray}

\subsubsection{\textbf{Permutation Errors}}
In the permutation error model, two distinct measurements are interchanged due to a labeling or handling error. Concretely, this corresponds to swapping two rows $i_1$ and $i_2$ of the binary pooling matrix $\boldsymbol{B}$. As a result, the row appearing in position $i$ after the error may be any one of the original rows of $\boldsymbol{B}$. Thus, the set of allowable perturbations for the $i^{\text{th}}$ row under permutation errors is
\begin{equation}
\mathcal{C}_i^{\mathrm{Perm}}
=
\bigl\{
\boldsymbol{b}_{k\cdot} : k \in \{1,\dots,n\}
\bigr\}.
\end{equation}

We now provide an algorithm to create the set $\mathcal{C}_i$ given a row $\boldsymbol{b}_{i.}$ and a specific MME model. The algorithm returns a matrix $\boldsymbol{P}$ which consists

\begin{algorithm}[h]
\caption{Generation of the set of all valid Perturbations Under SSM, ASM, and PERM}
\label{alg:perturbations}
\begin{algorithmic}[1]
\REQUIRE Binary row $\boldsymbol{b} = (b_1,\dots,b_p) \in \{0,1\}^p$ (row $i$ of $\boldsymbol{B}$), model $\in \{\textsf{SSM}, \textsf{ASM}, \textsf{PERM}\}$, index set $\mathcal{J} \subseteq [n']$ (used only for PERM)
\ENSURE Matrix $\mathcal{P}$ of all valid perturbations of $\boldsymbol{b}$
\STATE $\mathcal{P} \gets \emptyset$
\IF{model $=$ \textsf{SSM}\COMMENT{Single-Switch Model}} 
    \FOR{$j = 1$ to $p$}
        \STATE $\tilde{\boldsymbol{b}} \gets \boldsymbol{b}$
        \STATE $\tilde{b}_j \gets 1 - b_j$  \COMMENT{Flip bit $j$}
        \STATE Append $\tilde{\boldsymbol{b}}$ to $\mathcal{P}$
    \ENDFOR
\ELSIF{model $=$ \textsf{ASM} \COMMENT{Adjacent-Switch Model}} 
    \FOR{$j = 1$ to $p$}
        \STATE $j' \gets (j \bmod p) + 1$ \COMMENT{Circular adjacency}
        \IF{$b_j = 0$ \textbf{and} $b_{j'} = 1$ \COMMENT{Case 1: $(0,1)\rightarrow(1,0)$}} 
            \STATE $\tilde{\boldsymbol{b}} \gets \boldsymbol{b}$
            \STATE $\tilde{b}_j \gets 1,\ \tilde{b}_{j'} \gets 0$
            \STATE Append $\tilde{\boldsymbol{b}}$ to $\mathcal{P}$
        \ENDIF
        \IF{$b_j = 1$ \textbf{and} $b_{j'} = 0$  \COMMENT{Case 2: $(1,0)\rightarrow(0,1)$}}
            \STATE $\tilde{\boldsymbol{b}} \gets \boldsymbol{b}$
            \STATE $\tilde{b}_j \gets 0,\ \tilde{b}_{j'} \gets 1$
            \STATE Append $\tilde{\boldsymbol{b}}$ to $\mathcal{P}$
        \ENDIF
    \ENDFOR
\ELSIF{model $=$ \textsf{PERM}\COMMENT{Permutation Error Model}} 
    \FOR{$k \in \mathcal{J}$ such that $k \neq i$}
        \STATE Append row $\boldsymbol{b}_{k\cdot}$ of $\boldsymbol{B}$ to $\mathcal{P}$
    \ENDFOR
\ENDIF
\STATE \RETURN$\mathcal{P}$
\end{algorithmic}
\end{algorithm}

\section{Proof of Theoretical Results}

\subsection{Proof of Theorem~1:}
\noindent\textbf{Proof of Result (1), For any $i \in \mathcal{J}$ such that $\tilde{\delta}^*_i \ne 0$:}  
Recall that $\boldsymbol{\tilde{b}_{i.}}$ is the $i^{\text{th}}$ row of $\boldsymbol{\tilde{B}}$ corresponding to the $i^{\text{th}}$ pool without any MMEs. Alg. 1 from the main paper using APE corrects the MME in the $i^{\text{th}}$ pool if it selects $\boldsymbol{\hat{b}_{i.}}$ to be equal to $\boldsymbol{\tilde{b}_{i.}}$ from $\mathcal{C}_i$.

Recall that $\hat{\boldsymbol{b}}_{i.} \triangleq \arg\min_{\boldsymbol{\bar{b}_{i.}} \in \mathcal{C}_i} d(z_i,\boldsymbol{\bar{b}_{i.}}, \boldsymbol{\hat{\beta}_{\mathcal{J}^c}})$ with $d(z_i,\boldsymbol{\bar{b}_{i.}}, \boldsymbol{\hat{\beta}_{\mathcal{J}^c}})=|z_i-\boldsymbol{\bar{b}_{i.}}\boldsymbol{\hat{\beta}_{\mathcal{J}^c}}|$. To show $\boldsymbol{\hat{b}_{i.}}=\boldsymbol{\tilde{b}_{i.}}$, it suffices to show that $d(z_i,\boldsymbol{\tilde{b}_{i.}}, \boldsymbol{\hat{\beta}_{\mathcal{J}^c}}) \leq d(z_i,\boldsymbol{\bar{b}_{i.}}, \boldsymbol{\hat{\beta}_{\mathcal{J}^c}})$ for all $\boldsymbol{\bar{b}_{i.}} \in \mathcal{C}_i$.  

Using $z_i = \boldsymbol{\tilde{b}}_{i,.}\boldsymbol{\beta}^* + \eta_i$ and the definition of $\tilde{\delta}^*_i$, we expand:

\begin{eqnarray}
   \nonumber |z_i-\boldsymbol{\tilde{b}_{i.}}\boldsymbol{\hat{\beta}_{\mathcal{J}^c}}| &=& |\eta_i+\boldsymbol{\tilde{b}}_{i.}\boldsymbol{\beta}^*-\boldsymbol{\tilde{b}_{i.}}\boldsymbol{\hat{\beta}_{\mathcal{J}^c}}| \\ &=& |\eta_i + \boldsymbol{\tilde{b}}_{i.}(\boldsymbol{\beta}^*-\boldsymbol{\hat{\beta}_{\mathcal{J}^c}})| \label{eq:val_nobf} \\
   \nonumber |z_i - \boldsymbol{{b}}_{i.}\boldsymbol{\hat{\beta}_{\mathcal{J}^c}}| &=& |\eta_i +\boldsymbol{{b}}_{i.}\boldsymbol{\beta}^*+(\boldsymbol{\tilde{b}}_{i.}-\boldsymbol{b}_{i.})\boldsymbol{\beta}^* - \boldsymbol{{b}}_{i.}\boldsymbol{\hat{\beta}_{\mathcal{J}^c}}|  \nonumber \\ &=& |\eta_i\! +\!(\boldsymbol{\tilde{b}}_{i.}-\boldsymbol{b}_{i.})\boldsymbol{\beta}^*\! +\! \boldsymbol{{b}}_{i.}(\boldsymbol{\beta}^*\!-\!\boldsymbol{\hat{\beta}_{\mathcal{J}^c}})| \label{eq:val_truebf} \\
   \nonumber |z_i - \boldsymbol{\bar{b}}_{i.}\boldsymbol{\hat{\beta}_{\mathcal{J}^c}}| &=& |\eta_i+\boldsymbol{\tilde{b}}_{i.}\boldsymbol{\beta}^*- \boldsymbol{\bar{b}}_{i.}\boldsymbol{\hat{\beta}_{\mathcal{J}^c}}| \nonumber \\ &=& \! |\eta_i \!+\!(\boldsymbol{\tilde{b}}_{i.}\!-\!\boldsymbol{\bar{b}}_{i.})\boldsymbol{\beta}^*\! +\! \boldsymbol{\bar{b}}_{i.}(\boldsymbol{\beta}^*\!-\!\boldsymbol{\hat{\beta}_{\mathcal{J}^c}})| \label{eq:val_indubf}
\end{eqnarray}

Note that, in this case, $\mathcal{J}$ has all MME-affected rows almost surely. Let $\boldsymbol{\hat{\beta}_{\mathcal{J}^c}}$ be the \textsc{Lasso} estimate obtained from the set of $n'-\hat{r}$ clean (i.e. without MME) centered measurements $(\boldsymbol{y}_{\mathcal{J}^c}, \boldsymbol{A}_{\mathcal{J}^c})$ with $n'=\lfloor n/2 \rfloor$. Standard results for the \textsc{Lasso}, as given in Theorem 11.1(b) of \cite{THW2015}, yield:
\begin{equation}
P\left(\|\boldsymbol{\hat{\beta}_{\mathcal{J}^c}}-\boldsymbol{\beta}^*\|_1 \leq C s\sigma\sqrt{\frac{\log p}{n'-\hat{r}}}\right) \geq 1-\frac{1}{p^2}.
\end{equation}
for some constant $C>0$. Under assumption \textbf{A3} and given the technique for MME detection from the main paper, we have $\hat{r}\leq r_U := \zeta n' \leq \zeta \frac{n}{2}$, $0<\zeta<1$. With $n \geq 2\frac{C^2 s^2}{1-\zeta}\log p$, we have $n'-\hat{r} \geq n(1-\zeta)/2 \geq C^2 s^2\log p$, ensuring:
\begin{equation}
Cs\sigma\sqrt{\frac{\log p}{n'-\hat{r}}} \leq \sigma.
\end{equation}
Thus, under $n' \geq \frac{2 C^2 s^2}{1-\zeta}\log p$:
\begin{equation}\label{eq:bound_beta_sigma}
P\left(\|\boldsymbol{\hat{\beta}_{\mathcal{J}^c}}-\boldsymbol{\beta}^*\|_1 \leq \sigma\right) \geq 1-\left(\frac{1}{p^2}\right).
\end{equation}
This implies $|\hat{\beta}_j- \beta^*_j| \leq \sigma$ for all $j\in[p]$ with probability at least $1-\left(\frac{1}{p^2}\right)$.

Let $\boldsymbol{u} \in \mathbb{R}^p$ have i.i.d. entries from $\text{Bernoulli}(\theta)-\theta$ (centered). Since $\|\boldsymbol{u}\|_\infty \leq 1$, we have with probability $\geq 1-\left(\frac{1}{p^2}\right)$:
\begin{equation}\label{eq:cent_bnd_pred}
|\boldsymbol{u}^{\top}(\boldsymbol{\beta}^*-\boldsymbol{\hat{\beta}_{\mathcal{J}^c}})| \leq \|\boldsymbol{\beta}^*-\boldsymbol{\hat{\beta}_{\mathcal{J}^c}}\|_1 \leq \sigma.
\end{equation}
Since $\eta_i \sim N(0,\sigma^2)$, Lemma~\ref{eq:exponential_normal} gives $P(|\eta_i| \leq k\sigma)\geq 1-e^{-k^2/2}$. By triangle inequality and union bound:
\begin{equation}\label{eq:base_bound}
P\left(|\eta_i+ \boldsymbol{\tilde{b}}_{i.}(\boldsymbol{\beta}^*-\boldsymbol{\hat{\beta}_{\mathcal{J}^c}})| \leq (k+1)\sigma\right)\geq 1-\left(e^{-k^2/2}+\frac{1}{p^2}\right).
\end{equation}
From \eqref{eq:val_nobf}:
\begin{equation}\label{eq:val_no_bound}
P\left(|z_i-\boldsymbol{\tilde{b}_{i.}}\boldsymbol{\hat{\beta}_{\mathcal{J}^c}}| \leq (k+1)\sigma\right)\geq 1-\left(e^{-k^2 / 2}+\frac{1}{p^2}\right).
\end{equation}

Under \textbf{A1}, $|(\boldsymbol{\tilde{b}}_{i.}-\boldsymbol{b}_{i.})\boldsymbol{\beta}^*| \geq 2(k+1)\sigma$, $k>0$. Using reverse triangle inequality with \eqref{eq:base_bound} and \eqref{eq:cent_bnd_pred}, with probability $1-\left(e^{-k^2 / 2}+\frac{1}{p^2}\right)$:
\begin{eqnarray}\label{eq:bf_bound}
\nonumber & &|\eta_i  + \boldsymbol{{b}}_{i.}(\boldsymbol{\beta}^*-\boldsymbol{\hat{\beta}_{\mathcal{J}^c}})+(\boldsymbol{\tilde{b}}_{i.}-\boldsymbol{b}_{i.})\boldsymbol{\beta}^*| \\
\nonumber &\geq& \left||(\boldsymbol{\tilde{b}}_{i.}-\boldsymbol{b}_{i.})\boldsymbol{\beta}^*|-|\eta_i  + \boldsymbol{{b}}_{i.}(\boldsymbol{\beta}^*-\boldsymbol{\hat{\beta}_{\mathcal{J}^c}})|\right| \\
&\geq& 2(k+1)\sigma - (k+1)\sigma = (k+1)\sigma .
\end{eqnarray}
Thus:
\begin{eqnarray}\label{eq:bf1_bound}
P\left(|\eta_i  + \boldsymbol{{b}}_{i.}(\boldsymbol{\beta}^*-\boldsymbol{\hat{\beta}_{\mathcal{J}^c}})+(\boldsymbol{\tilde{b}}_{i.}-\boldsymbol{b}_{i.})\boldsymbol{\beta}^*| \geq  (k+1)\sigma\right) \nonumber \\ \geq 1-\left(e^{-k^2 / 2}+\frac{1}{p^2}\right).
\end{eqnarray}
From \eqref{eq:val_truebf}:
\begin{equation}\label{eq:val_true_bound}
P\left( |z_i - \boldsymbol{{b}}_{i.}\boldsymbol{\hat{\beta}_{\mathcal{J}^c}} | \geq (k+1)\sigma\right)   \geq 1-\left(e^{-k^2 / 2}+\frac{1}{p^2}\right).
\end{equation}
Combining \eqref{eq:val_no_bound} and \eqref{eq:val_true_bound}: $|z_i - \boldsymbol{{b}}_{i.}\boldsymbol{\hat{\beta}_{\mathcal{J}^c}} | \geq (k+1)\sigma \geq |z_i-\boldsymbol{\tilde{b}_{i.}}\boldsymbol{\hat{\beta}_{\mathcal{J}^c}}|$ with high probability. By union bound:
\begin{eqnarray}\label{eq:truebf_nobf}
\nonumber  P\Big(|z_i - \boldsymbol{{b}}_{i.}\boldsymbol{\hat{\beta}_{\mathcal{J}^c}} | \geq  |z_i-\boldsymbol{\tilde{b}_{i.}}\boldsymbol{\hat{\beta}_{\mathcal{J}^c}}| \Big) \geq 1-2\left(e^{-k^2 / 2}+\frac{1}{p^2}\right).
\end{eqnarray}
Thus $|z_i-\boldsymbol{\tilde{b}_{i.}}\boldsymbol{\hat{\beta}_{\mathcal{J}^c}}| \leq |z_i - \boldsymbol{{b}}_{i.}\boldsymbol{\hat{\beta}_{\mathcal{J}^c}} |$ with high probability.

For all $\boldsymbol{\bar{b}_{i.}} \in \mathcal{C}_i- \{\boldsymbol{b}_{i,.},\boldsymbol{\tilde{b}}_{i,.}\}$, using \textbf{A1} and \eqref{eq:cent_bnd_pred}, with probability $\geq 1-1/p-1/n$:
\begin{equation}\label{eq:max_bound_b_hat}
\underset{\boldsymbol{\bar{b}_{i.}}\in \mathcal{C}_i}{\max}|\boldsymbol{\bar{b}}_{i.}(\boldsymbol{\beta}^*-\boldsymbol{\hat{\beta}_{\mathcal{J}^c})}|\leq \|\boldsymbol{\beta}^*-\boldsymbol{\hat{\beta}_{\mathcal{J}^c}}\|_1 \leq \sigma.
\end{equation}
By triangle inequality on \eqref{eq:max_bound_b_hat} and $P(|\eta_i| \leq k\sigma)\geq 1-e^{-k^2 / 2}$:
\begin{eqnarray}\label{eq:joint_small}
\nonumber & & |\eta_i + \boldsymbol{\bar{b}}_{i.}(\boldsymbol{\beta}^*-\boldsymbol{\hat{\beta}_{\mathcal{J}^c}})| \leq |\eta_i|+\underset{\boldsymbol{\bar{b}_{i.}}\in \mathcal{C}_i}{\max}|\boldsymbol{\bar{b}}_{i.}(\boldsymbol{\beta}^*-\boldsymbol{\hat{\beta})}| \\
 &\leq& (k+1)\sigma \ \forall \boldsymbol{\bar{b}}_{i.} \in \mathcal{C}_i.
\end{eqnarray}
This holds with probability $\geq 1- \left(e^{-k^2 / 2}+\frac{1}{p^2}\right)$.

Under \textbf{A2}, $|(\boldsymbol{\tilde{b}}_{i.}-\boldsymbol{\bar{b}}_{i.})\boldsymbol{\beta}^*| \geq 2(k+1)\sigma$, $k>0$. For fixed $i$, using reverse triangle inequality similar to \eqref{eq:bf_bound}, $\forall \boldsymbol{\bar{b}}_{i.} \in \mathcal{C}_i$,
\begin{eqnarray}
\label{eq:bi_bar_bound}
\nonumber P(|\eta_i + \boldsymbol{\bar{b}}_{i.}(\boldsymbol{\beta}^*-\boldsymbol{\hat{\beta}_{\mathcal{J}^c}})+(\boldsymbol{\tilde{b}}_{i.}-\boldsymbol{\bar{b}}_{i.})\boldsymbol{\beta}^*| \geq (k+1)\sigma )\\ \geq 1-\left(e^{-k^2 / 2}+\frac{1}{p^2}\right).
\end{eqnarray}
Combining the previous two results with  \eqref{eq:val_indubf}, for fixed $i$:
\begin{eqnarray}\label{eq:val_indu_bound}
P\left(|z_i - \boldsymbol{\bar{b}}_{i.}\boldsymbol{\hat{\beta}_{\mathcal{J}^c}}| \geq (k+1)\sigma \ \forall \boldsymbol{\bar{b}}_{i.} \in \mathcal{C}_i\right) \nonumber \\ \geq 1-\left(e^{-k^2 / 2}+\frac{1}{p^2}\right).
\end{eqnarray}
Since $|z_i-\boldsymbol{\tilde{b}_{i.}}\boldsymbol{\hat{\beta}_{\mathcal{J}^c}}| \leq (k+1)\sigma$ from \eqref{eq:val_no_bound}:
\begin{eqnarray}\label{eq:indubf_truebf}
\nonumber & &P\Big(|z_i-\boldsymbol{\tilde{b}_{i.}}\boldsymbol{\hat{\beta}_{\mathcal{J}^c}}| \leq |z_i - \boldsymbol{\hat{b}}_{i.}\boldsymbol{\hat{\beta}_{\mathcal{J}^c}}| \ \forall \boldsymbol{\hat{b}}_{i.} \in \mathcal{C}_i\Big)\\ &\geq& 1-2\left(e^{-k^2 / 2}+\frac{1}{p^2}\right).
\end{eqnarray}

Recall $d(z_i,\boldsymbol{\tilde{b}_{i.}}, \boldsymbol{\hat{\beta}_{\mathcal{J}^c}})=|z_i-\boldsymbol{\tilde{b}_{i.}}\boldsymbol{\hat{\beta}_{\mathcal{J}^c}}|$ and $d(z_i,\boldsymbol{\bar{b}_{i.}}, \boldsymbol{\hat{\beta}_{\mathcal{J}^c}})=|z_i - \boldsymbol{\bar{b}}_{i.}\boldsymbol{\hat{\beta}_{\mathcal{J}^c}}|$. Combining \eqref{eq:bi_bar_bound} and \eqref{eq:indubf_truebf} via union bound:
\begin{eqnarray}\label{eq:val_error_order}
\nonumber & &P\left(d(z_i,\boldsymbol{\tilde{b}_{i.}}, \boldsymbol{\hat{\beta}_{\mathcal{J}^c}}) \leq \underset{\boldsymbol{\bar{b}_{i.}} \in \mathcal{C}_i}{\min}d(z_i,\boldsymbol{\bar{b}_{i.}}, \boldsymbol{\hat{\beta}_{\mathcal{J}^c}})\right)\\ &\geq& 1-3\left(e^{-k^2 / 2}+\frac{1}{p^2}\right).
\end{eqnarray}
This completes the proof of \textbf{Result (1)}.

\noindent \textbf{Proof of Result (2):} For $i \in \mathcal{J}$ such that $\tilde{\delta}^*_i = 0$, we have  $\boldsymbol{\tilde{b}_{i.}}= \boldsymbol{{b}_{i.}}$. To show $\boldsymbol{\hat{b}_{i.}}= \boldsymbol{{b}_{i.}}$, it suffices that $|z_i - \boldsymbol{{b}}_{i.}\boldsymbol{\hat{\beta}_{\mathcal{J}^c}}| \leq |z_i - \boldsymbol{\bar{b}}_{i.}\boldsymbol{\hat{\beta}_{\mathcal{J}^c}}|$ for all $\boldsymbol{\bar{b}_{i.}} \in \mathcal{C}_i$.

Similar to \eqref{eq:val_truebf} and \eqref{eq:val_indubf}:
\begin{eqnarray}
|z_i - \boldsymbol{{b}}_{i.}\boldsymbol{\hat{\beta}_{\mathcal{J}^c}}| &\!\!\!=\!\!\!&  |\eta_i + \boldsymbol{{b}}_{i.}(\boldsymbol{\beta}^*-\boldsymbol{\hat{\beta}_{\mathcal{J}^c}})| \label{eq:val_nobf_Jc} \\
|z_i - \boldsymbol{\bar{b}}_{i.}\boldsymbol{\hat{\beta}_{\mathcal{J}^c}}| &\!\!\!=\!\!\!& |\eta_i +(\boldsymbol{{b}}_{i.}-\boldsymbol{\bar{b}}_{i.})\boldsymbol{\beta}^* + \boldsymbol{\bar{b}}_{i.}(\boldsymbol{\beta}^*-\boldsymbol{\hat{\beta}_{\mathcal{J}^c}})|\nonumber\\&&\label{eq:val_indubf_Jc}
\end{eqnarray}
Using logic similar to \eqref{eq:indubf_truebf} and $\boldsymbol{\tilde{b}_{i.}}= \boldsymbol{{b}_{i.}}$:
\begin{eqnarray*}
\nonumber & &P\Big(|z_i-\boldsymbol{{b}_{i.}}\boldsymbol{\hat{\beta}_{\mathcal{J}^c}}| \leq |z_i - \boldsymbol{\bar{b}}_{i.}\boldsymbol{\hat{\beta}_{\mathcal{J}^c}}| \ \forall \boldsymbol{\bar{b}}_{i.} \in \mathcal{C}_i\Big)\\ &\geq& 1-2\left(e^{-k^2 / 2}+\frac{1}{p^2}\right).
\end{eqnarray*}
Thus:
\begin{eqnarray}\label{eq:ineq_Jc}
\nonumber P\left(d(z_i,\boldsymbol{{b}_{i.}}, \boldsymbol{\hat{\beta}_{\mathcal{J}^c}}) \leq \underset{\boldsymbol{\hat{b}_{i.}} \in \mathcal{C}_i}{\min}d(z_i,\boldsymbol{\hat{b}_{i.}}, \boldsymbol{\hat{\beta}})\right)\\ \geq 1-2\left(e^{-k^2 / 2}+\frac{1}{p^2}\right).
\end{eqnarray}
This completes the proof. \hfill{$\blacksquare$}

\subsection{Proof of Theorem~2:}\label{sec:proof_exact_sol}
The optimization problem given in Alg.~3 in \cite{Banerjee24} for $h=1/(2\theta(1-\theta))$ is as follows:
\begin{eqnarray}\label{eq:C4W_primal}
\underset{\boldsymbol{W}}{\textrm{minimize}}& &\quad 
\sum_{j=1}^p\boldsymbol{w_{.j}}^{\top}\boldsymbol{w_{.j}} \\ \nonumber
\textrm{subject to}& & 
\mathsf{C0} : \boldsymbol{w_{.j}}^{\top}\boldsymbol{w_{.j}}/n \leq 1+h^2\sqrt{\frac{\log p}{n}} \ \forall \ j \in [p] 
\\ & & \nonumber \mathsf{C1} : \left| \left(\boldsymbol{I_p}-\frac1n\boldsymbol{W^{\top}A}\right)\right|_\infty \leq \mu_1, \\ \nonumber & & \mathsf{C2} : \left|\frac{1}{p}\left(\boldsymbol{I_n}-\frac{1}{n}\boldsymbol{A}\boldsymbol{W}^{\top}\right) \boldsymbol{A}\right|_\infty \leq \mu_2, \\ \nonumber & & \mathsf{C3} : \left|\left(\frac{\boldsymbol{AW^{\top}}}{p}-\boldsymbol{I_n}\right)\right|_{\infty} \leq \mu_3,
\end{eqnarray}
where $\mu_1\triangleq 2h^2 \sqrt{\frac{2 \log p}{n}}, \mu_2 \triangleq 4h^3\sqrt{\frac{\log 2np}{np}}+\frac{1}{n}$, $\mu_3\triangleq 2h^2 \sqrt{\frac{2 \log n}{p}}$, $h=\frac{1}{2\theta(1-\theta)}$ and $\boldsymbol{W}$ is an $n \times p$ matrix. The motivation for all four constraints, as well as for the values of $\mu_1, \mu_2, \mu_3$, has been provided in \cite{Banerjee24} (see Alg.~3). In order to prove that the optimal solution to the above problem is $\boldsymbol{w}_{i.}=\frac{p(1-\mu_3)}{\|\boldsymbol{a}_{i.}\|_2^2}\boldsymbol{a}_{i.}$ for all $i \in [n]$, we first consider a relaxed version of this problem with only the constraint $\mathsf{C3}$. The relaxed problem is given as follows:
\begin{eqnarray}\label{eq:relaxed_opt_W}
    \underset{\boldsymbol{W}}{\min}& &\quad 
\frac{1}{n}\sum_{i=1}^n\boldsymbol{w_{i.}}^{\top}\boldsymbol{w_{i.}}
\\ \nonumber
\textrm{subject to}& &  \left|\left(\frac{\boldsymbol{AW^{\top}}}{p}-\boldsymbol{I_n}\right)\right|_{\infty} \leq \mu_3.
\end{eqnarray}
The objective function in \eqref{eq:C4W_primal} is equivalent to that of \eqref{eq:relaxed_opt_W} as $\sum_{j=1}^p\boldsymbol{w_{.j}}^{\top}\boldsymbol{w_{.j}}=trace(\boldsymbol{W^\top W})=trace(\boldsymbol{W W^\top}) = \sum_{i=1}^n\boldsymbol{w_{i.}}^{\top}\boldsymbol{w_{i.}}$. We write the relaxed problem in this modified (but equivalent) form as it allows separability with respect to the rows of $\boldsymbol{W}$. The equivalent separable optimization problem for \eqref{eq:relaxed_opt_W} is given as follows for all $i \in [n]$:
\begin{eqnarray}\label{eq:separable_relaxed_opt_W}
    \underset{\boldsymbol{w_{i.}}}{\min}& &\quad  \boldsymbol{w_{i.}}^{\top}\boldsymbol{w_{i.}}/p \\ \nonumber \textrm{subject to}& &  \left\| \boldsymbol{w_{i.}}^\top \boldsymbol{A}^{\top}/p - \boldsymbol{e_i}\right\|_{\infty} \leq \mu_3,
\end{eqnarray}
where $\boldsymbol{e_i}$ is the $i^{\textrm{th}}$ row of $\boldsymbol{I_n}$.
We define $\tau(\boldsymbol{A})=\max_{i\neq k}\frac{|\boldsymbol{a}_{i.}^\top \boldsymbol{a}_{k.}|}{\|\boldsymbol{a}_{i.}\|_2^2}$. For an $n\times p$ matrix $\boldsymbol{A}$, let us define the following quantity for all $i \in [n]$,
\begin{equation}\label{eq:Ldef}
L_i := \frac{1}{p}\|\boldsymbol{a}_{i.}\|_2^2.
\end{equation} 
Also, let us define for all $i \ne k \in [n]$,
\begin{equation}\label{eq:nudef}
|\nu_{ik}| := \frac{1}{p}|\boldsymbol{a}_{i.}^\top\boldsymbol{a}_{k.}|.
\end{equation} 
We have from Result (3) of Lemma~7 of \cite{Banerjee24} , 
\begin{equation}\label{eq:max_nu}
   P\left(\underset{i \ne k \in [n]}{\max}|\nu_{ik}| \leq \frac{1}{4\theta^2(1-\theta)^2}2\sqrt{\frac{2\log n}{ p}}\right) \geq 1-\frac{2}{n^2}-\frac{2}{n^3}.
\end{equation}
Furthermore, from \eqref{eq:lower_bound_L_row} of  Lemma~\ref{th:lower_bnd_L_row}, we have,
$$P\left(\!\left(1\!+\!2\sqrt{2}\sqrt{\frac{\log n}{p}}\right)\! \geq\! L_i \!\geq\! \left(1\!-\!2\sqrt{2}\sqrt{\frac{\log n}{p}}\right) \!\right)\!  \geq\! 1\!-\!\frac{2}{n^4}.$$
Since $\mu_3 = \frac{1}{4\theta^2(1-\theta)^2}2\sqrt{2}\sqrt{\frac{\log n}{p}} \geq 2\sqrt{2}\sqrt{\frac{\log n}{p}}$. 
Therefore, 
for $i \in [n]$,
\begin{equation}\label{eq:tail_indiv_L}
    P\left(1+\mu_3\geq L_i \geq 1-\mu_3 \right)  \geq 1-\frac{2}{n^4}.
\end{equation}
Taking union bound over all $i \in [n]$ for \eqref{eq:tail_indiv_L}, we have, 
\begin{equation}\label{eq:tail_L_full}
    P\left(1+\mu_3\geq\underset{i \in [n]}{\min} L_i \geq 1-\mu_3 \right) \geq 1-\frac{2n}{n^4} \geq 1-\frac{2}{n^3}.
\end{equation}
Therefore using \eqref{eq:max_nu} and \eqref{eq:tail_L_full}, with high probability, we have  
\begin{equation}\label{eq:theta_c}
     \tau \leq \frac{\underset{i \ne k}{\max}|\nu_{ik}|}{\underset{i \in [n]}{\min}L_i} \leq  \frac{\mu_3}{1-\mu_3},
\end{equation}
where $\mu_3$ is as defined earlier in \eqref{eq:C4W_primal}.
Note that \eqref{eq:theta_c} implies that $(1-\mu_3)\tau \leq \mu_3$.
Now, we show that $\boldsymbol{w}_{i.}=\frac{p(1-\mu_3)}{\|\boldsymbol{a}_{i.}\|_2^2}\boldsymbol{a}_{i.}$ is the optimal solution of \eqref{eq:separable_relaxed_opt_W}.

\noindent\textbf{Primal feasibility:} The choice  $\boldsymbol{w}_{i.}=\frac{p(1-\mu_3)}{\|\boldsymbol{a}_{i.}\|_2^2}\boldsymbol{a}_{i.}$ is primal feasible. This is because substituting this choice into the constraint from \eqref{eq:separable_relaxed_opt_W}, we have from \eqref{eq:theta_c} that $|(1-\mu_3)\tau| \leq \mu_3$, 
\[ \left\|\frac{(1-\mu_3)p}{p\|\boldsymbol{a}_{i.}\|_2^2}\boldsymbol{a_{i.}}^\top \boldsymbol{A}^{\top} - \boldsymbol{e}_i\right\|_{\infty} \leq \max\{ |\mu_3|, |(1-\mu_3)\tau|\} \leq \mu_3\]
with high probability. For the column $\boldsymbol{a}_{i.}$, the LHS is $\left|(1-\mu_3)p {\boldsymbol{a}}^{\top}_{i.}{\boldsymbol{a}}_{i.}/(p\|\boldsymbol{a}_{i.}\|^2_2)-1\right| = \mu_3$. For other columns, the LHS become $\left|(1-\mu_3)p {\boldsymbol{a}}^{\top}_{i.}{\boldsymbol{a}}_{k.}/(p\|{a}_{i.}\|^2_2)\right| \leq |(1-\mu_3)\tau|$.

\noindent\textbf{Primal objective function value:} The primal objective function value is given by
$$\frac{1}{p}\|\boldsymbol{w}_{i.}\|^2_2 =\dfrac{(1-\mu_3)^2}{(\|\boldsymbol{a}_{i.}\|_2^2/p)^2}\|\boldsymbol{a}_{i.}\|_2^2/p =\dfrac{(1-\mu_3)^2}{\|\boldsymbol{a_{i.}}\|_2^2/p}.$$  

\noindent\textbf{The Fenchel Dual Problem:}
We derive the dual of the problem in \eqref{eq:separable_relaxed_opt_W} using Fenchel duality. For each $i \in [n]$, we are given an optimization problem of the form
\[ \min_{\boldsymbol{w}} f(\boldsymbol{w}) + g_i\left(\frac{1}{p}\boldsymbol{w}^\top \boldsymbol{A}^{\top}\right),\]
where $f$ and $g_i$ are convex. The Fenchel dual for this problem is
\[ \max_{\boldsymbol{u}} -f^*\left(-\frac{1}{p}\boldsymbol{A}\boldsymbol{u}\right) - g^*_i(-\boldsymbol{u})\]
where $f^*$ and $g^*_i$ are the convex conjugates of $f$ and $g_i$, respectively.
\\
Therefore, we take,
\begin{eqnarray*}
    f(\boldsymbol{w}) = \frac{1}{p}\|\boldsymbol{w}\|^2 \quad\textup{and} \\
    g_i(\boldsymbol{w}) = \begin{cases} 0 & \textup{if $\|\boldsymbol{w} - \boldsymbol{e}_i\|_{\infty} \leq \mu_3$}\\ \infty & \textup{otherwise}\end{cases}.
\end{eqnarray*}

Then we have
\begin{eqnarray*}
f^*(\boldsymbol{u}) &=& \sup_{\boldsymbol{w}} \langle \boldsymbol{u},\boldsymbol{w}\rangle - f(\boldsymbol{w}) = \frac{p}{4}\|\boldsymbol{u}\|^2, \\
g^*_i(\boldsymbol{u}) &=& \sup_{\boldsymbol{w}} \langle \boldsymbol{u},\boldsymbol{w}\rangle - g_i(\boldsymbol{w}) = \sup_{\|\boldsymbol{w}-\boldsymbol{e}_i\|_{\infty}\leq \mu_3}\langle \boldsymbol{u},\boldsymbol{w}\rangle \\ &=& y_i + \mu_3\sqrt{1-n/p}\|\boldsymbol{u}\|_1.
\end{eqnarray*}
This gives a dual problem in the form 
\[ \sup_{\boldsymbol{u}} - \frac{1}{4p}\boldsymbol{u}^\top \boldsymbol{A} \boldsymbol{A}^\top\boldsymbol{u} + u_i - \mu_3\|\boldsymbol{u}\|_1.\]

\noindent\textbf{Dual feasibility:}
The point $\boldsymbol{u} = \frac{2(1-\mu_3)}{\|\boldsymbol{a}_{i.}\|_2^2/p}\boldsymbol{e}_j$ is feasible for the dual (trivially). 

\noindent\textbf{Dual objective function value:} The corresponding dual objective function value is
\begin{align*}
&-\frac{1}{4p}\boldsymbol{u}^\top \boldsymbol{A}^\top \boldsymbol{A}\boldsymbol{u} + u_i - \mu_3\|\boldsymbol{u}\|_1 \\ & = -\frac{1}{4p} \|\boldsymbol{a}_{i.}\|^2 \frac{4(1-\mu_3)^2}{(\|\boldsymbol{a_{i.}}\|_2^2/p)^2} +\frac{2(1-\mu_3)}{\|\boldsymbol{a_{i.}}\|_2^2/p} - \mu_3\frac{2(1-\mu_3)}{\|\boldsymbol{a_{i.}}\|_2^2/p}\\
&= -\frac{(1-\mu_3)^2}{\|\boldsymbol{a_{i.}}\|_2^2/p}+2\frac{(1-\mu_3)^2}{\|\boldsymbol{a_{i.}}\|_2^2/p} = \frac{(1-\mu_3)^2}{\|\boldsymbol{a_{i.}}\|_2^2/p}.\end{align*}

This establishes that an optimal solution for the primal of the relaxed problem given in \eqref{eq:relaxed_opt_W} is $\boldsymbol{w}_{i.}=\frac{p(1-\mu_3)}{\|\boldsymbol{a}_{i.}\|_2^2}\boldsymbol{a}_{i.}$ for all $i \in [n]$ and an optimal solution to the dual is $\boldsymbol{u} = \frac{2(1-\mu_3)}{\|\boldsymbol{a}_{i.}\|_2^2/p}\boldsymbol{e}_j$. Since the primal and dual objective function values are equal, we know that these are the optimal solutions.

Now, in order to complete the proof, all we need to show is that the solution $\boldsymbol{w}_{i.}=\frac{p(1-\mu_3)}{\|\boldsymbol{a}_{i.}\|_2^2}\boldsymbol{a}_{i.}$ is feasible for constraints $\mathsf{C0}$, $\mathsf{C1}$ and $\mathsf{C2}$ with high probability.
This holds as when the optimal solution of a relaxed problem belongs to the feasible space of a more constrained problem, then it is also the optimal solution of the constrained problem. In the following, we define $c \triangleq \frac{p(1-\mu_3)}{\|\boldsymbol{a}_{i.}\|_2^2}$. 
\\
\\
\noindent \textbf{Satisfying constraint $\mathsf{C0}$:}
We need to show that the solution $\boldsymbol{w_{i.}} = \frac{p(1-\mu_3)}{\|\boldsymbol{a}_{i.}\|_2^2}\boldsymbol{a_{i.}}$ satisfies the condition $\boldsymbol{w_{.j}}^{\top}\boldsymbol{w_{.j}}/n \leq 1+\frac{1}{4\theta^2(1-\theta)^2}\sqrt{\frac{\log p}{n}} \ \forall \ j \in [p] $.
Therefore substituting $\boldsymbol{w_{i.}} = \frac{p(1-\mu_3)}{\|\boldsymbol{a}_{i.}\|_2^2}\boldsymbol{a_{i.}}$, we have for all $j \in [p]$, with high probability,
\begin{eqnarray}\label{eq:C0_sat}
  \nonumber  \boldsymbol{w_{.j}}^{\top}\boldsymbol{w_{.j}}/n &=& (1-\mu_3)^2 \frac{1}{n}\sum_{i=1}^n \frac{a_{ij}^2}{\|\boldsymbol{a}_{i.}\|_2^4/p^2}\\ &\leq& \frac{(1-\mu_3)^2}{(1-\mu_3)^2}  \left(1+\frac{1}{4\theta^2(1-\theta)^2}\sqrt{\frac{\log p}{n}}\right)
\end{eqnarray}
The last equality holds as $\frac{1}{n}\sum_{i=1}^n {a_{ij}^2} \leq  \left(1+\frac{1}{4\theta^2(1-\theta)^2}\sqrt{\frac{\log p}{n}}\right)$ with high probability as given in Result (4) of Lemma.~7 of \cite{Banerjee24}, and also from \eqref{eq:tail_indiv_L}. 
Therefore, for $\theta \in (0,0.5]$, we have from \eqref{eq:C0_sat},
$ \boldsymbol{w_{.j}}^{\top}\boldsymbol{w_{.j}}/n \leq \left(1+\frac{1}{4\theta^2(1-\theta)^2}\sqrt{\frac{\log p}{n}}\right)$, where $h=\frac{1}{2\theta(1-\theta)}$. 
\\
\\
\noindent \textbf{Satisfying constraint $\mathsf{C1}$:}
We need to show that the solution $\boldsymbol{w}_{i.} \triangleq \frac{p(1-\mu_3)}{\|\boldsymbol{a}_{i.}\|_2^2}\boldsymbol{a}_{i.}$ satisfies the constraint $\mathsf{C1}$, that is we need to show that
$\left\| \left(\boldsymbol{e}_{.j}-\boldsymbol{w_{.j}}^{\top}\boldsymbol{A}/p\right)\right\|_{\infty} \leq \mu_1$ for all $j \in [p]$. 
Let us define the vector $\boldsymbol{v}=\left(\boldsymbol{e}_{.j}-\boldsymbol{w_{.j}}^{\top}\boldsymbol{A}/p\right)$, where $w_{ij}=\frac{p(1-\mu_3)}{\|\boldsymbol{a}_{i.}\|_2^2} a_{ij}$ for all $i \in [n], \, j \in [p]$. Therefore, the $j$th element of $\boldsymbol{v}$ can be written as follows for $v_j>0$,
\begin{eqnarray*}
    v_j &=& \frac{1-\mu_3}{n}\sum_{i=1}^n \frac{a_{ij}^2}{\|\boldsymbol{a}_{i.}\|_2^2/p} -1 \leq \frac{(1-\mu_3)}{1-\mu_3}\frac{1}{n}\sum_{i=1}^n a_{ij}^2 -1 \\ &\leq&  \left(1+\frac{1}{4\theta^2(1-\theta)^2}\sqrt{\frac{\log p}{n}}\right)-1 \leq \mu_1
\end{eqnarray*}
The aforementioned statement holds with high probability. 
The first inequality holds using \eqref{eq:tail_indiv_L} and the second inequality comes from Result (4) of Lemma~7 of \cite{Banerjee24}. 
Now for $v_j<0$, we have with high probability,
\begin{eqnarray*}
    -v_j &=& 1-(1-\mu_3) \frac{1}{n}\sum_{i=1}^n \frac{a_{ij}^2}{\|\boldsymbol{a}_{i.}\|_2^2/p} \leq  1-\frac{1-\mu_3}{1+\mu_3} \frac{1}{n}\sum_{i=1}^n a_{ij}^2 \\ &\leq& 1-\frac{1-\mu_3}{1+\mu_3}\left(1-\frac{1}{4\theta^2(1-\theta)^2}\sqrt{\frac{\log p}{n}}\right) \\ &=& 1-\frac{1-\mu_3}{1+\mu_3} + \frac{1-\mu_3}{1+\mu_3}\frac{1}{4\theta^2(1-\theta)^2}\sqrt{\frac{\log p}{n}} \\ &\leq& \frac{2\mu_3}{1+\mu_3} + \frac{1}{4\theta^2(1-\theta)^2}\sqrt{\frac{\log p}{n}} \\ &\leq& \frac{1}{4\theta^2(1-\theta)^2}2\sqrt{\frac{2\log p}{n}} = \mu_1.
\end{eqnarray*}
The first inequality comes from \eqref{eq:tail_indiv_L} and the second inequality comes from the fact that $-\frac{1}{n}\sum_{i=1}^n {a_{ij}^2} \leq  -\left(1-\frac{1}{4\theta^2(1-\theta)^2}\sqrt{\frac{\log p}{n}}\right)$ with high probability as given in Result (4) of Lemma~7 of \cite{Banerjee24}. The third inequality comes from the fact that $\frac{1-\mu_3}{1+\mu_3} \leq 1$.
The last inequality comes from the fact that $\frac{2\mu_3}{1+\mu_3} \leq 2\mu_3 \leq  \frac{1}{4\theta^2(1-\theta)^2}\sqrt{\frac{\log p}{n}}$ for large enough $n, p (n \ll p)$ because $\mu_3 = O\left(\sqrt{\frac{\log n}{p}}\right) \ll \sqrt{\frac{\log p}{n}}$. Therefore, we have that for $v_j\geq 0$, $v_j \leq \mu_1$ and when $v_j \leq 0$, $v_j \geq -\mu_1$, both with high probability. This implies that with high probability, $|v_j| \leq \mu_1$.

 In case of $l\ne j \in [p]$, the $l^\text{th}$ element of the vector $\boldsymbol{v}$ is as follows
\begin{eqnarray*}
     |v_l| = \left|\frac{1-\mu_3}{n}\sum_{i=1}^n \frac{a_{ij}a_{il}}{\|\boldsymbol{a}_{i.}\|_2^4/p^2} \right| \leq \left|\frac{(1-\mu_3)}{(1-\mu_3)}\frac{1}{n}\sum_{i=1}^n a_{ij}a_{il} \right|  \leq \mu_1.
\end{eqnarray*}
 This holds with high probability. The first inequality comes from \eqref{eq:tail_indiv_L}, the second inequality from from Result (1) of Lemma~7  of \cite{Banerjee24}.
 This implies that the exact solution satisfies constraint $\mathsf{C1}$. 
\\
\\
\noindent \textbf{Satisfying Constraint $\mathsf{C2}$:}
We need to show that the solution ${w}_{ij} \triangleq \frac{p(1-\mu_3)}{\|\boldsymbol{a}_{i.}\|_2^2}{a}_{ij}$ for all $i \in [n], \, j \in [p]$, satisfies constraint $\mathsf{C2}$, that is we need to show that $\left|\frac{1}{p}\left(\boldsymbol{I_n}-\frac{1}{n}\boldsymbol{A}\boldsymbol{W}^{\top}\right) \boldsymbol{A}\right|_\infty \leq \mu_2$.
Let us define the matrix $\boldsymbol{U}=\frac{1}{p}\left(\boldsymbol{I_n}-\frac{1}{n}\boldsymbol{A}\boldsymbol{W}^{\top}\right) \boldsymbol{A}= \frac1p\left(\boldsymbol{A}-\frac{1}{n}\boldsymbol{A}\boldsymbol{W}^{\top}\boldsymbol{A}\right)$. 
For all $i \in [n], j \in [p]$, we expand $u_{ij}$ the same way as done as in Result (2) of Lemma~7 of \cite{Banerjee24} to get, 
\begin{eqnarray}\label{eq:u_ij_struct}
   \nonumber u_{ij}&=& \frac{1}{np}\left(\sum_{\substack{l=1\\l\neq j}}^{p}\sum_{\substack{k=1\\k\neq i}}^{n}a_{il}w_{kl}a_{kj}\right)+ \frac{a_{ij}}{np}\left(\sum_{k=1}^{n}w_{kj}a_{kj}-1\right) \\ &+& \frac{a_{ij}}{np}\left(\sum_{\substack{l=1\\l\neq j}}^{p}w_{il}a_{il}-1\right) +  \frac{p-1}{np}a_{ij}.
\end{eqnarray}
It is enough to show that $\underset{i \ne j}{\max}|u_{ij}| \leq \mu_2$ with high probability. Let us consider the structures on the RHS of \eqref{eq:u_ij_struct} one by one, using the solution for $\boldsymbol{w}_{i,.}$. We have,
\begin{eqnarray}\label{eq:u_struct_1}
 \nonumber & & \left|\frac{1}{np}\left(\sum_{\substack{l=1\\l\neq j}}^{p}\sum_{\substack{k=1\\k\neq i}}^{n}a_{il}w_{kl}a_{kj}\right)\right| \\ \nonumber &=& \left|\frac{1}{np}\left(\sum_{\substack{l=1\\l\neq j}}^{p}\sum_{\substack{k=1\\k\neq i}}^{n}\frac{(1-\mu_3)}{\|\boldsymbol{a}_{i.}\|_2^2/p}a_{il}a_{kl}a_{kj}\right)\right| 
 \\
 &\leq& \left|\frac{1}{np}\left(\sum_{\substack{l=1\\l\neq j}}^{p}\sum_{\substack{k=1\\k\neq i}}^{n}\frac{(1-\mu_3)}{(1-\mu_3)}a_{il}a_{kl}a_{kj}\right)\right| \nonumber\\ &=& \left|\frac{1}{np}\left(\sum_{\substack{l=1\\l\neq j}}^{p}\sum_{\substack{k=1\\k\neq i}}^{n}a_{il}a_{kl}a_{kj}\right)\right|.
\end{eqnarray}

The inequality in \eqref{eq:u_struct_1} holds with high probability using \eqref{eq:tail_indiv_L}.
Next, we have, 
\begin{eqnarray}\label{eq:u_struct_2}
  \nonumber & &  \left|\frac{a_{ij}}{np}\left(\sum_{k=1}^{n}w_{kj}a_{kj}-1\right)\right| \\ \nonumber &=& \left|\frac{a_{ij}}{np}\left(\sum_{k=1}^{n}\frac{(1-\mu_3)}{\|\boldsymbol{a}_{k.}\|_2^2/p}a_{kj}a_{kj}-1\right)\right| \\ &\leq& \left|\frac{a_{ij}}{np}\left(\sum_{k=1}^{n}a_{kj}^2-1\right)\right|.
\end{eqnarray}
Similarly, the last inequality of \eqref{eq:u_struct_2} holds with high probability using \eqref{eq:tail_indiv_L} and follows the same logic as given to ensure that the constraint \textsf{C1} is satisfies.
Using similar logic, we have with high probability,
\begin{equation}\label{eq:u_struct_3}
    \left|\frac{a_{ij}}{np}\left(\sum_{\substack{l=1\\l\neq j}}^{p}w_{il}a_{il}-1\right)\right| \leq  \left|\frac{a_{ij}}{np}\left(\sum_{\substack{l=1\\l\neq j}}^{p}a_{il}^2-1\right)\right|. 
\end{equation}
Therefore, using triangle inequality on \eqref{eq:u_ij_struct} to join together \eqref{eq:u_struct_1}, \eqref{eq:u_struct_2}, and \eqref{eq:u_struct_3}, we have with high probability for all $i \in [n], j \in [p]$,
\begin{eqnarray}\label{eq:u_full}
   \nonumber |u_{ij}| &\leq& \left|\frac{1}{np}\left(\sum_{\substack{l=1\\l\neq j}}^{p}\sum_{\substack{k=1\\k\neq i}}^{n}a_{il}a_{kl}a_{kj}\right)\right|+
    \left|\frac{a_{ij}}{np}\left(\sum_{k=1}^{n}a_{kj}^2-1\right)\right|
    \\ &+&\left|\frac{a_{ij}}{np}\left(\sum_{\substack{l=1\\l\neq j}}^{p}a_{il}^2-1\right)\right|+\left|\frac{p-1}{np}a_{ij}\right|.
\end{eqnarray}
Therefore using Result (2) of Lemma~7 of \cite{Banerjee24}, we have that $\underset{i \ne j}{\max} |u_{ij}| \leq \mu_2$ holds with high probability. This proves that the exact solution satisfies constraint $\mathsf{C2}$ with high probability.
\\
This completes the proof of this theorem. \hfill{$\blacksquare$}

\begin{lemma}\label{th:lower_bnd_L_row}
   Given $\boldsymbol{A}$ is an $n \times p$ matrix with entries drawn i.i.d. from Centered Bernoulli ($\theta$). Let $L_i=\|\boldsymbol{a}_{i.}\|_2^2/p$. For all $i \in [n]$,
    \begin{eqnarray}\label{eq:L_tail_row}
     \nonumber  P\left(\left(1+2\sqrt{2}\sqrt{\frac{\log n}{p}}\right) \geq L_i \geq \left(1-2\sqrt{2}\sqrt{\frac{\log n}{p}}\right) \right) \\ \geq 1-\frac{2}{n^4}.
    \end{eqnarray}
\end{lemma}
\textbf{Proof of Lemma~\ref{th:lower_bnd_L_row}:}
This proof initially follows similar logic as the proof of Lemma~1 of \cite{Banerjee25FD} with the columns of $\boldsymbol{A}$ being replaced by rows.

We have for all $i \in [n]$, $\frac{\|\boldsymbol{a}_{i.}\|_2^2}{n}=\frac{1}{n}\sum_{j=1}^p a_{ij}^2$. For a given $i \in [n]$, the variables $a_{ij}^2$ are independent for all $j \in [p]$. Hence, using the concentration inequality of Theorem 3.1.1 and Equation (3.3) of \cite{Vershynin2018}, we have for $t>0$\footnote{We have set $c = 1/2$, $\delta := t$ and $K := 2\sqrt{C_{\max}}\kappa$ in Equation (3.3) and the equation immediately preceding it in  \cite{Vershynin2018}},
 \begin{equation}\label{eq:conc_ineq_subG_row}
     P\left(\left|\|\boldsymbol{a}_{i.}\|_2^2/p- E[\|\boldsymbol{a}_{i.}\|_2^2/p]\right| \geq t \right) \leq 2e^{-\frac{pt^2}{2C_{\max}^2\kappa^4}}.
 \end{equation}
Note that since $\boldsymbol{A}$ is Centered Bernoulli ($\theta$), $C_{\max} =\kappa= 1$.
Using $E[\|\boldsymbol{a}_{i.}\|_2^2/n] =1$, \eqref{eq:conc_ineq_subG_row} can be rewritten as follows for $t>0$:
\begin{equation}\label{eq:conc_lower_subG_row}
    P\left(1+t \geq L_i \geq 1- t\right)\geq 1-2e^{-\frac{pt^2}{2}}.
\end{equation}
Putting $t:=2\sqrt{2}\sqrt{\frac{\log n}{p}}$ in \eqref{eq:conc_lower_subG_row}, we obtain:
\begin{equation}\label{eq:lower_bound_L_row}
    P\left(\left(1+2\sqrt{2}\sqrt{\frac{\log n}{p}}\right) \geq L_i \geq \left(1-2\sqrt{2}\sqrt{\frac{\log n}{p}}\right) \right) \geq \frac{2}{n^4}.
\end{equation}
This completes the proof. $\blacksquare$
\begin{lemma}[Exponential Tail Bound for the Standard Normal Distribution] \label{eq:exponential_normal}
Let $ Z \sim \mathcal{N}(0,1) $ be a standard normal random variable. Then for all $ k > 0 $, the following inequality holds:
\[
\mathbb{P}(|Z| > k) \leq 2e^{-k^2/2}.
\]\hfill{$\blacksquare$}
\end{lemma}

\begin{proof}
The result follows by applying the Chernoff bound, which is based on Markov's inequality and the moment-generating function of $ Z $. For any $ \lambda > 0 $, we have:
\[
\mathbb{P}(Z > k) = \mathbb{P}(e^{\lambda Z} > e^{\lambda k}) \leq \frac{\mathbb{E}[e^{\lambda Z}]}{e^{\lambda k}}.
\]
Since $ Z \sim \mathcal{N}(0,1) $, its moment-generating function is:
\[
\mathbb{E}[e^{\lambda Z}] = e^{\lambda^2 / 2}.
\]
Substituting this into the inequality gives:
\[
\mathbb{P}(Z > k) \leq \frac{e^{\lambda^2 / 2}}{e^{\lambda k}} = e^{\lambda^2 / 2 - \lambda k}.
\]
To obtain the tightest bound, we minimize the exponent with respect to $ \lambda $. Differentiating and setting the derivative to zero:
\[
\frac{d}{d\lambda} \left( \frac{\lambda^2}{2} - \lambda k \right) = \lambda - k = 0 \quad \Rightarrow \quad \lambda = k.
\]
Substituting $ \lambda = k $ back into the bound yields:
\[
\mathbb{P}(Z > k) \leq e^{k^2 / 2 - k^2} = e^{-k^2 / 2}.
\]
By symmetry of Gaussian random variable, the proof is completed.\hfill{$\blacksquare$}
\end{proof}

\section{Additional Experimental Results}
We now present additional simulation results to support our proposed technique for MME correction as given in the main paper. Refer to the main paper for various details of experimental settings. We enumerate our additional results here below:

\begin{enumerate}
\item The main paper (Sec. III-A) showed results in terms of sensitivity, specificity and RRMSE for \textsf{SSM} errors. Here we show additional results for permutation errors and \textsf{ASM} errors under the same experimental settings. For permutation errors, the sensitivity and specificity plots are shown in Fig.~\ref{fig:Sens_spec_perm_0.1} and \ref{fig:Sens_spec_perm_0.5} for pooling matrix drawn from Bernoulli(0.1) and Bernoulli(0.5) respectively. RMSE plots are shown in Fig.~\ref{fig:Rmse_perm_0.1} and \ref{fig:Rmse_perm_0.5}. For \textsf{ASM} errors, the plots are shown in Fig.~\ref{fig:Sens_spec_asm_0.1}, \ref{fig:Sens_spec_asm_0.5}, \ref{fig:Rmse_asm_0.1} and \ref{fig:Rmse_asm_0.5}. In all these plots, we observe that the proposed method \textsc{Cape} outperforms various baselines. 

\item So far, we have presented results for a pooling matrix $\boldsymbol{B}$ whose entries are drawn iid from $CB(0.1)$ or $CB(0.5)$. We now present results for $CB(0.3)$ as well. For permutation errors, the results are in Fig.~\ref{fig:Sens_spec_perm_0.3} and Fig.~\ref{fig:Rmse_perm_0.3}. For \textsf{SSM}, the results are in Fig.~\ref{fig:Sens_spec_ssm_0.3} and \ref{fig:Rmse_ssm_0.3}. For \textsf{ASM}, the results are presented in Fig.~\ref{fig:Sens_spec_asm_0.3} and \ref{fig:Rmse_asm_0.3}. Here again, we observe that \textsc{Cape}, the proposed method, outperforms baselines.
\end{enumerate}

\begin{figure*}
   \centering
    \includegraphics[scale=0.2]{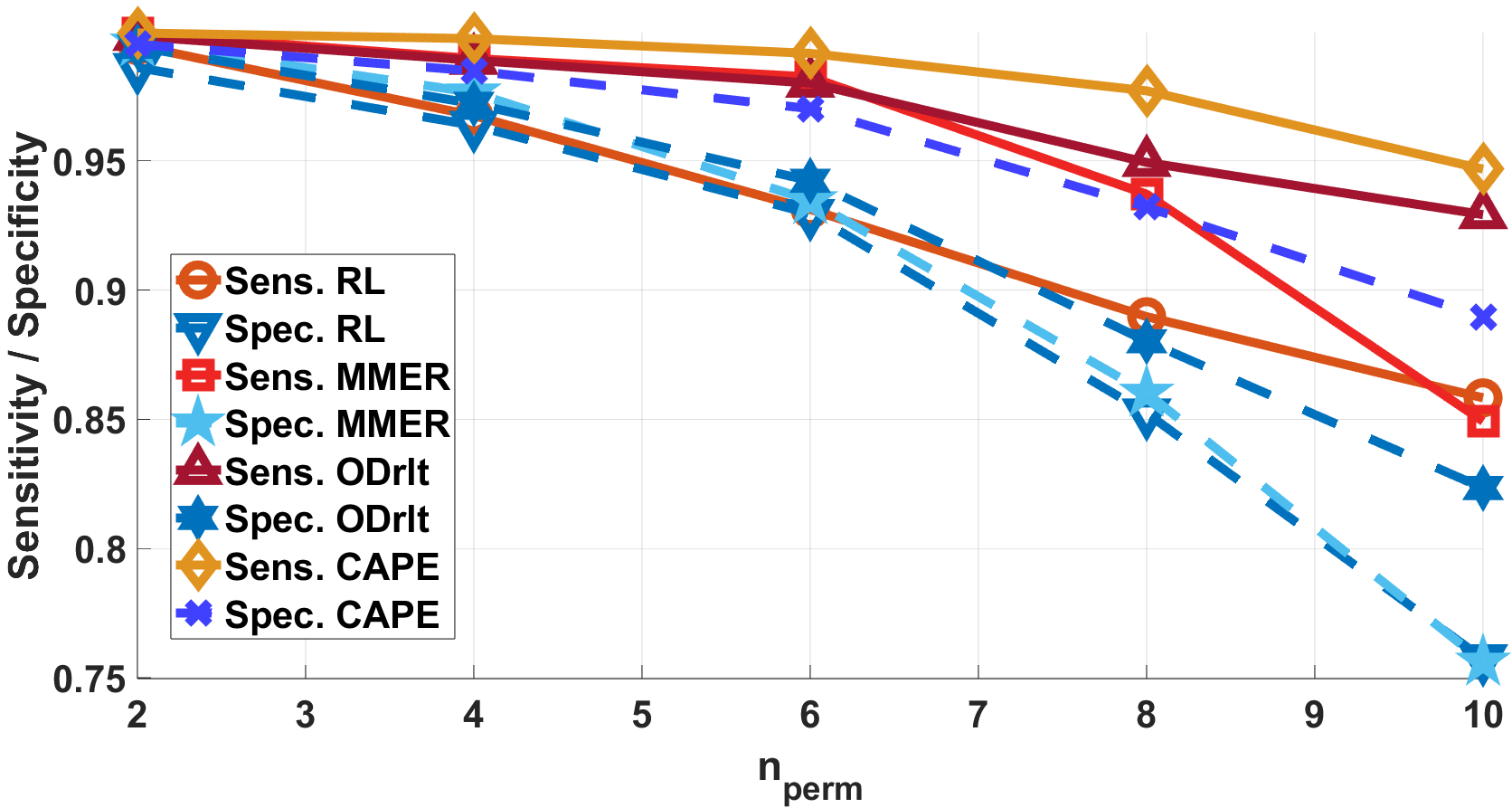}
    \includegraphics[scale=0.2]{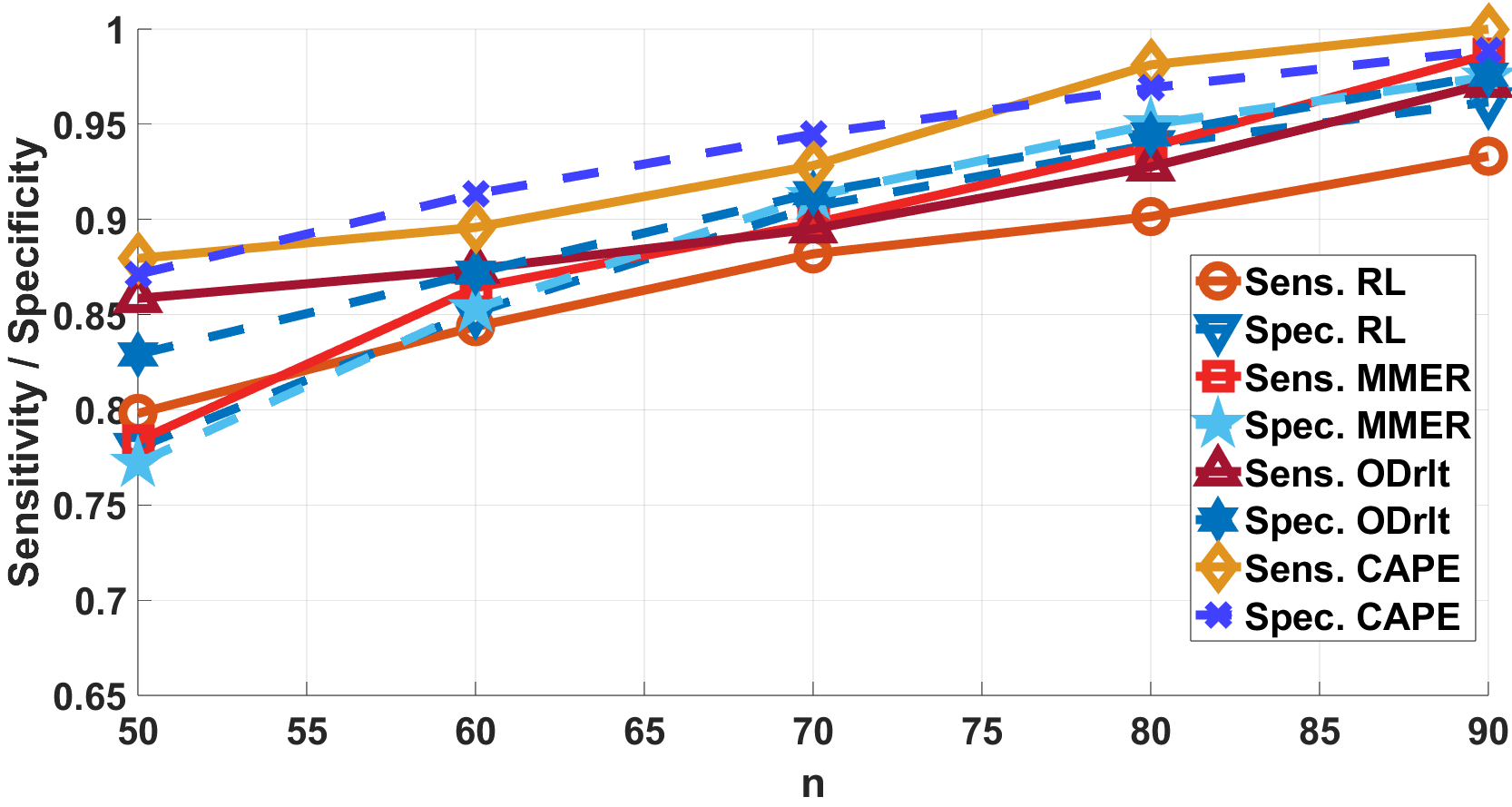}\\
    \includegraphics[scale=0.2]{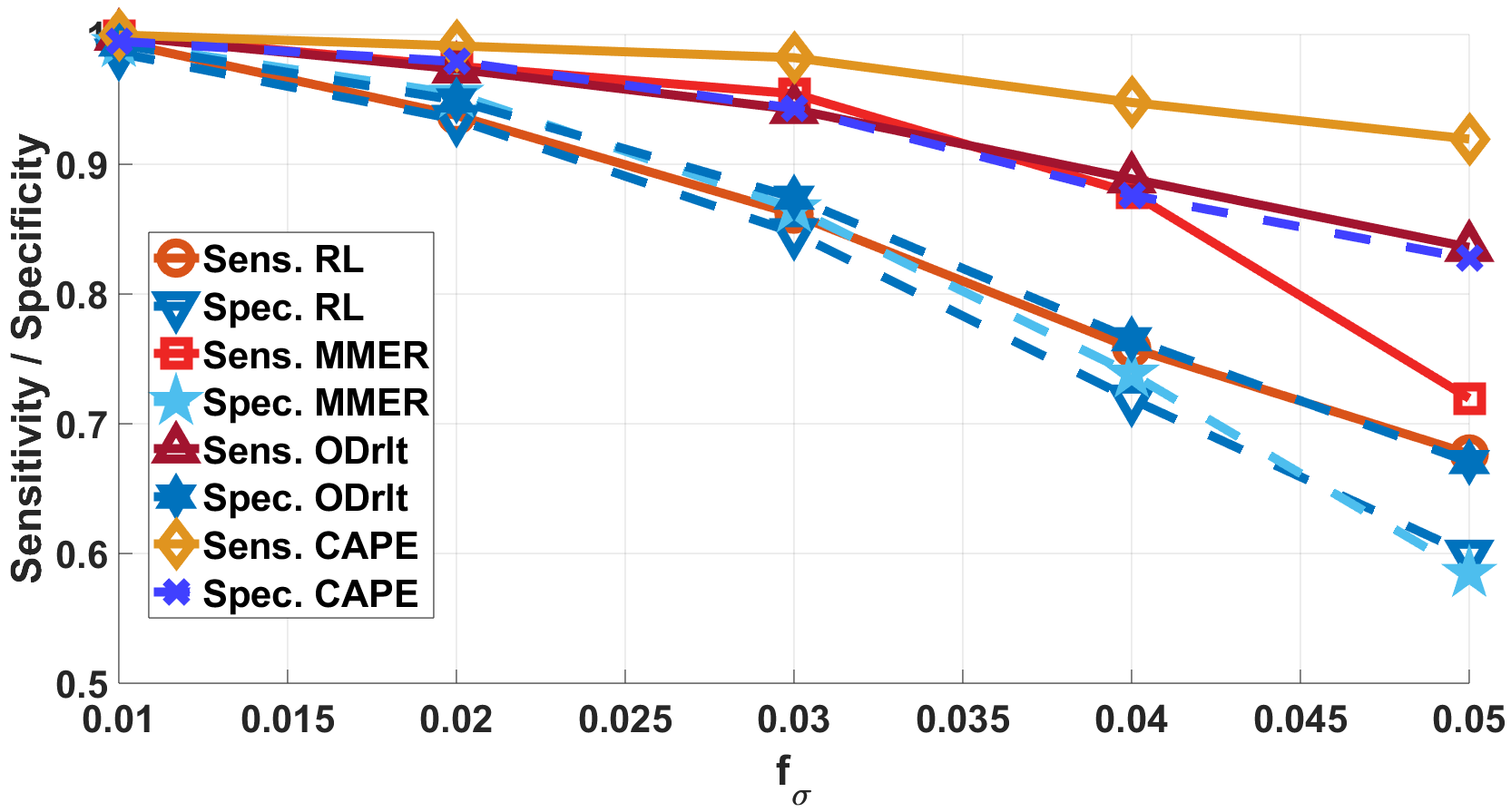}
    \includegraphics[scale=0.2]{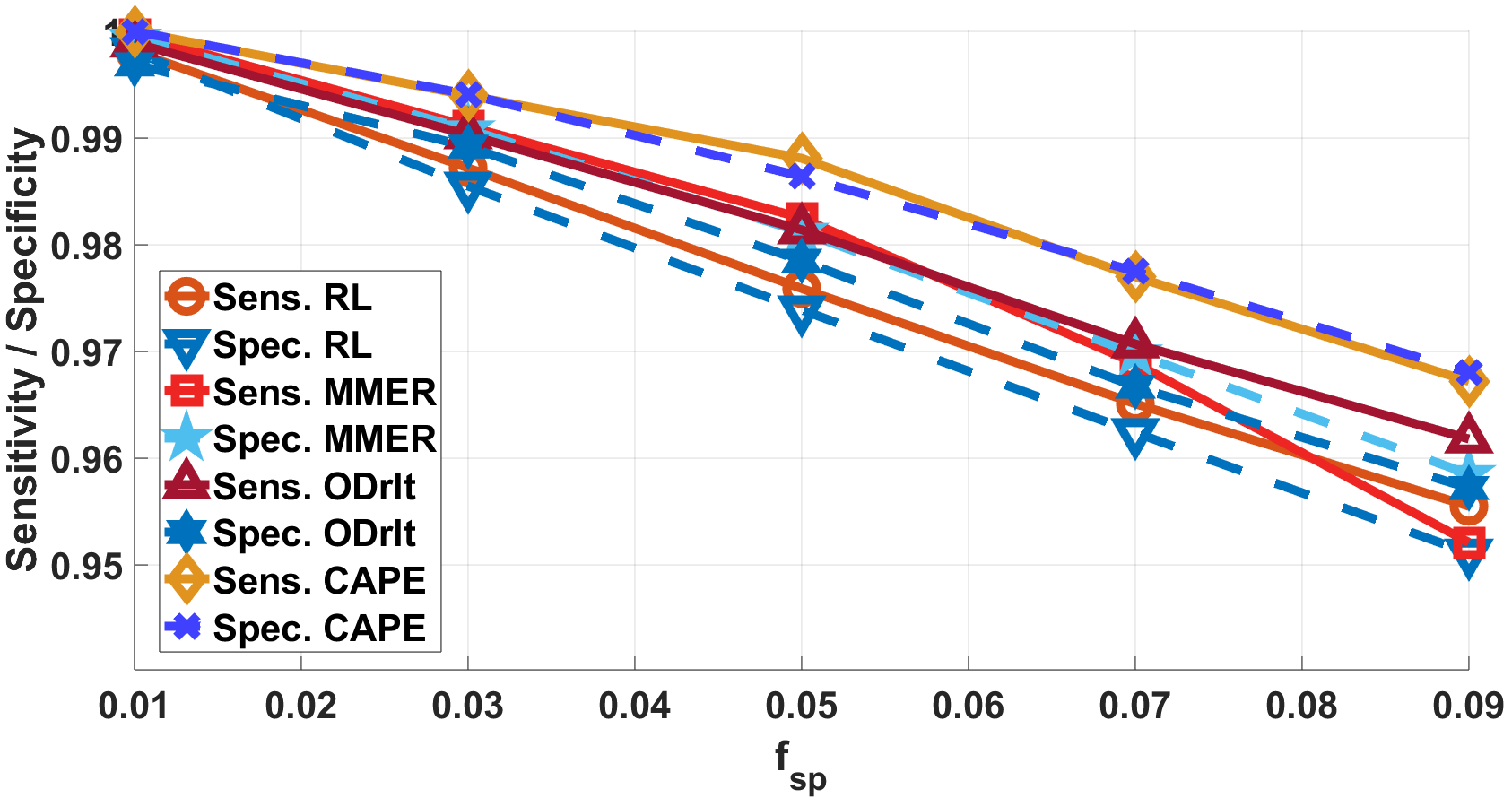}
    \caption{Sensitivity and specificity (averaged over 100 noise realizations, with $\boldsymbol{\beta}^*$, $\boldsymbol{A}$ with rows drawn i.i.d. from $CB(0.1)$, and $\boldsymbol{\delta}^*$ fixed) for detecting non-zero entries of $\boldsymbol{\beta}^*$ under permutation errors using \textsc{Cape}, \textsc{Rl}, \textsc{Mmer}, and \textsc{Odrlt}. Results shown for experiments \textsf{EA}, \textsf{EB}, \textsf{EC}, and \textsf{ED} defined in the main paper.}
    \label{fig:Sens_spec_perm_0.1}
\end{figure*}

\begin{figure*}
   \centering
    \includegraphics[scale=0.2]{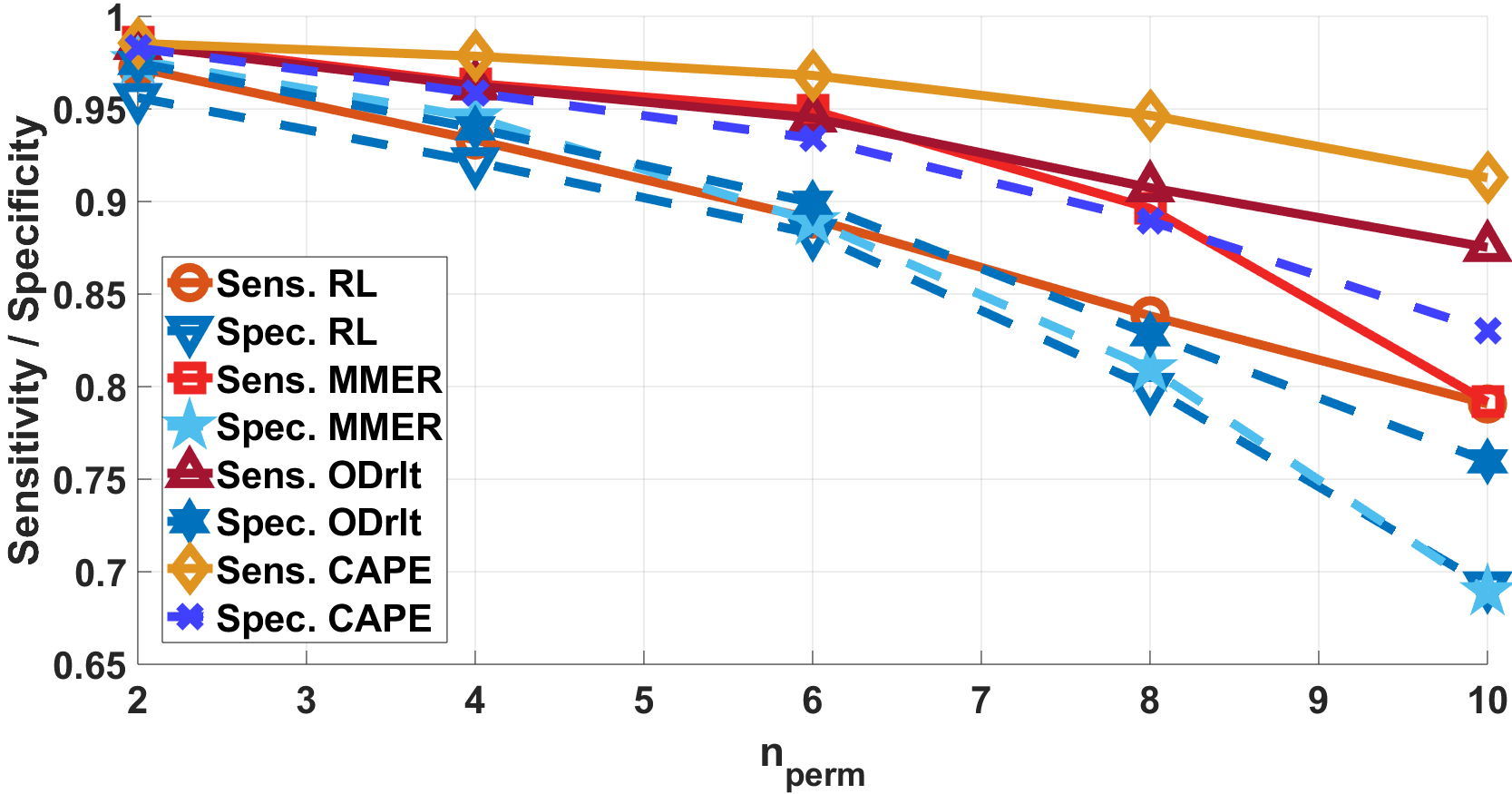}
    \includegraphics[scale=0.2]{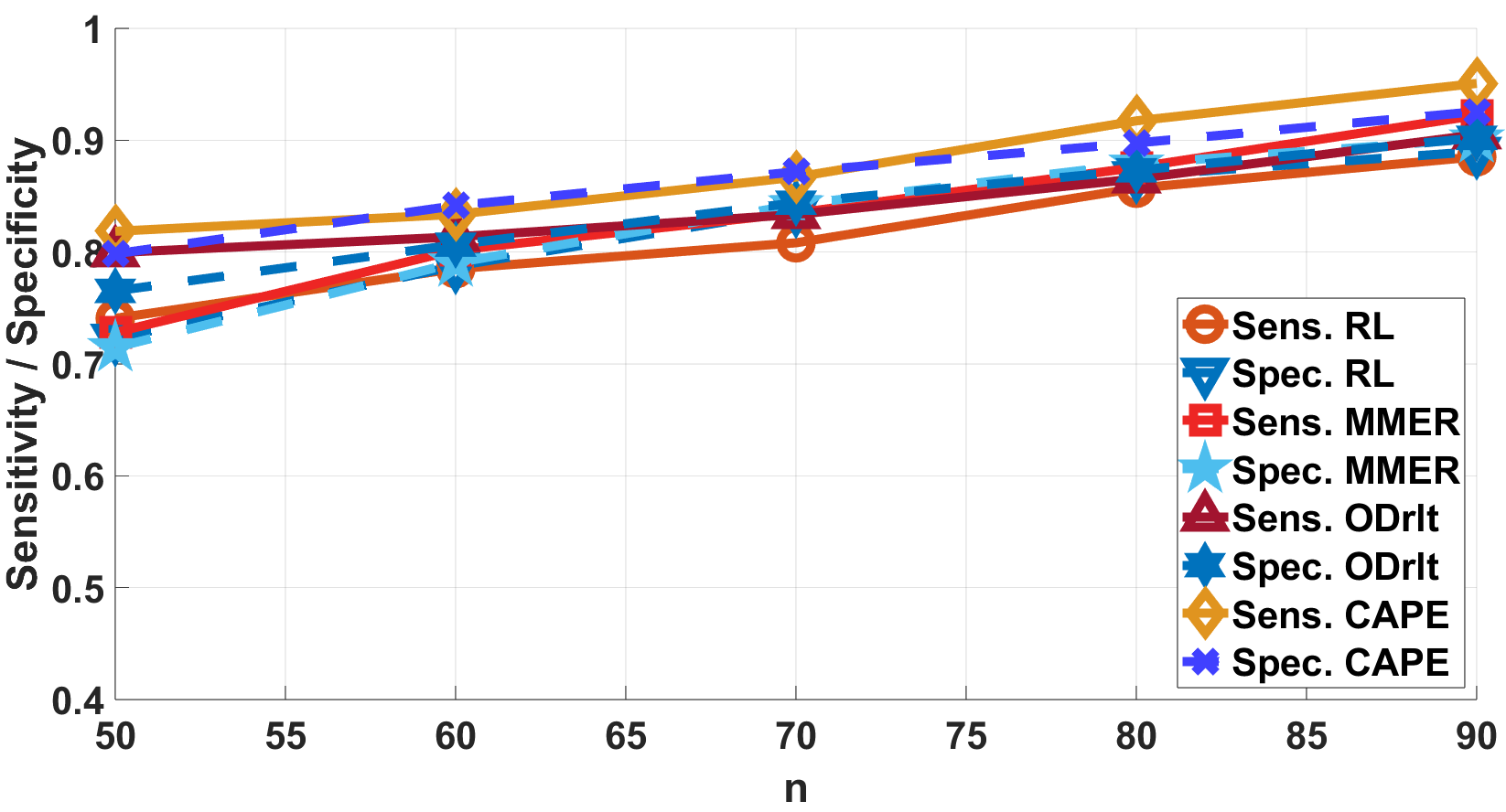}\\
    \includegraphics[scale=0.2]{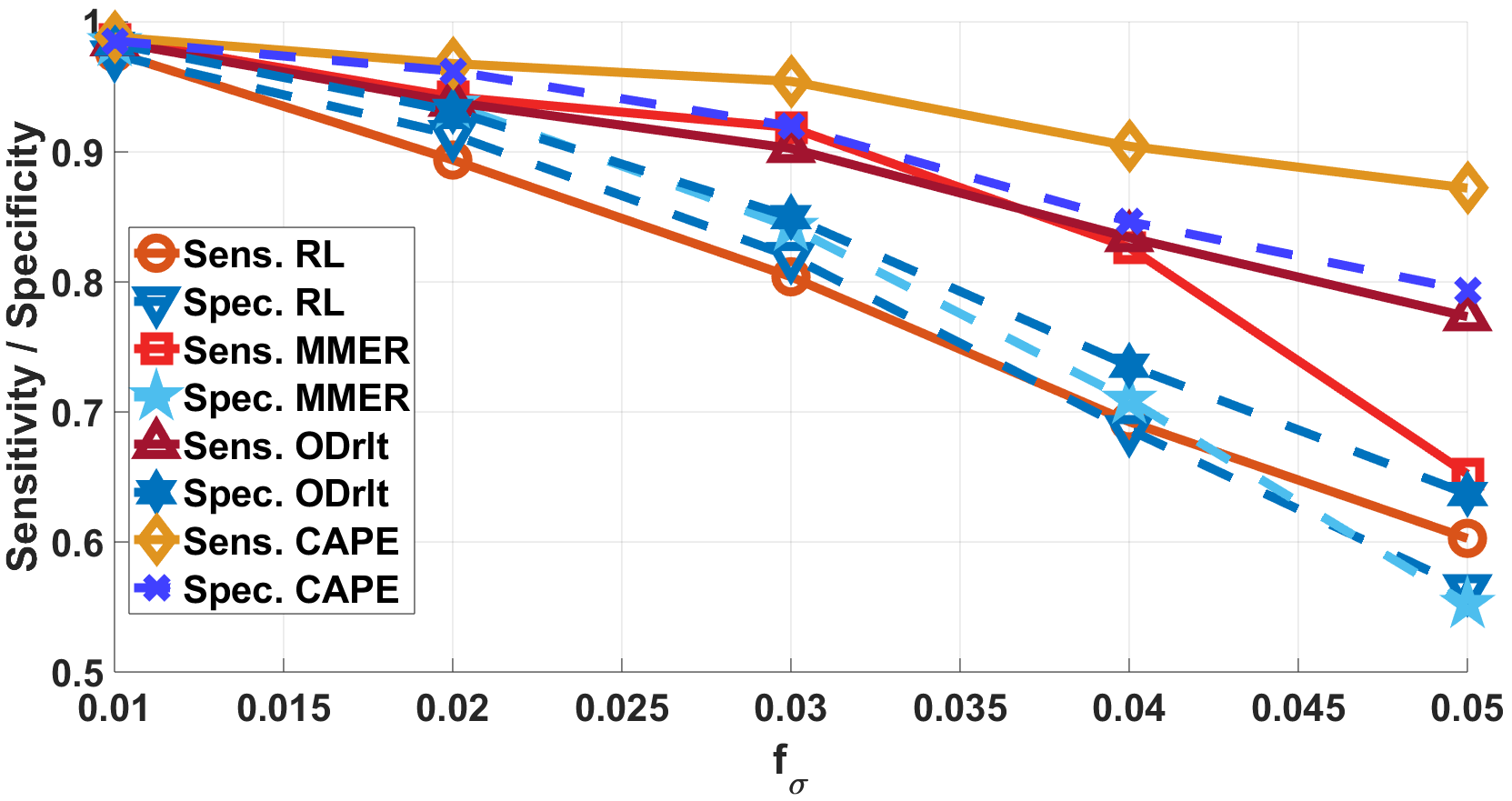}
    \includegraphics[scale=0.2]{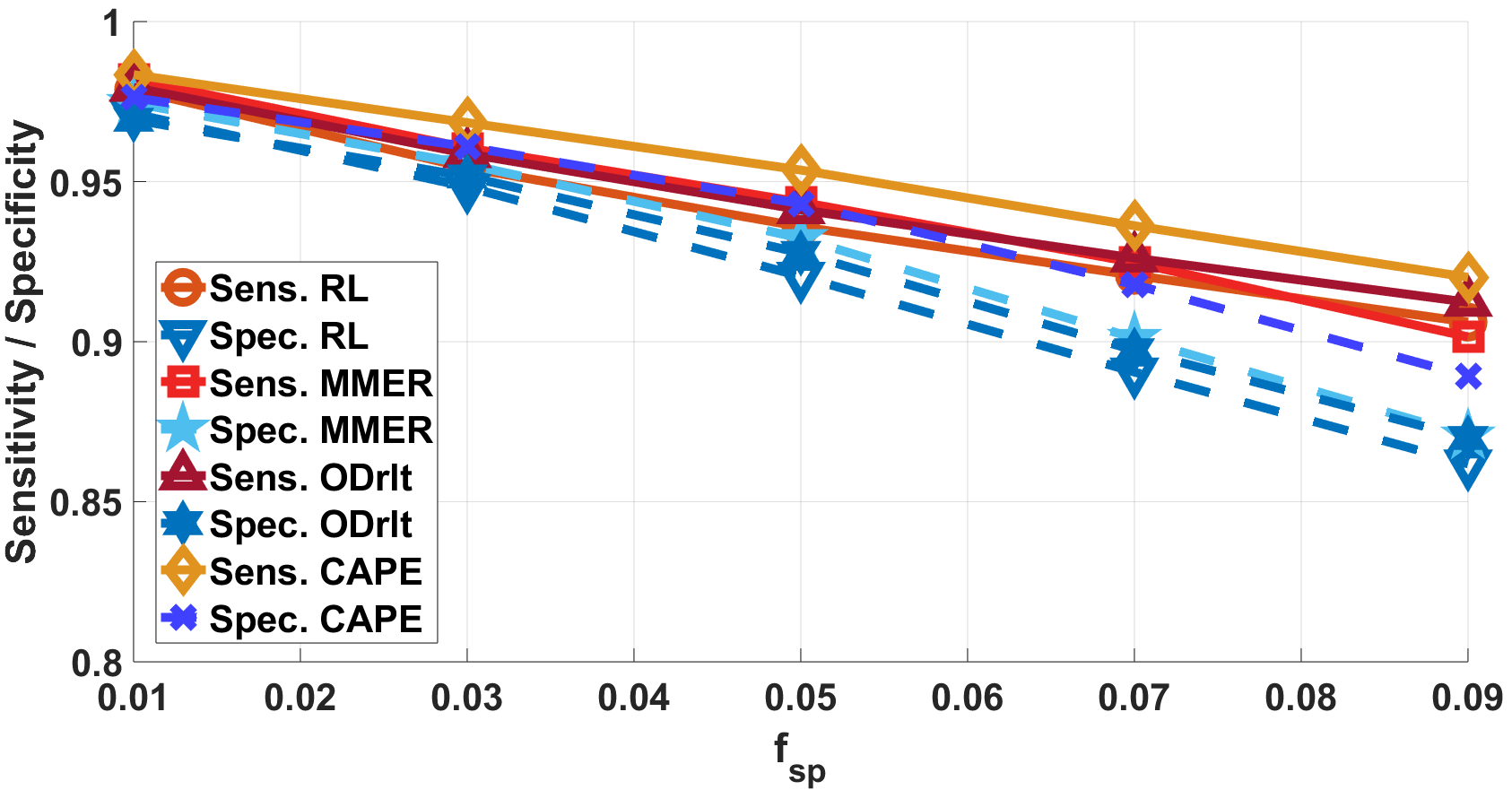}
    \caption{Sensitivity and specificity results under permutation errors for design matrices with rows i.i.d.\ from $CB(0.5)$. Results for experiments \textsf{EA}, \textsf{EB}, \textsf{EC}, and \textsf{ED} defined in the main paper.}
    \label{fig:Sens_spec_perm_0.5}
\end{figure*}

\begin{figure*}
   \centering
    \includegraphics[scale=0.2]{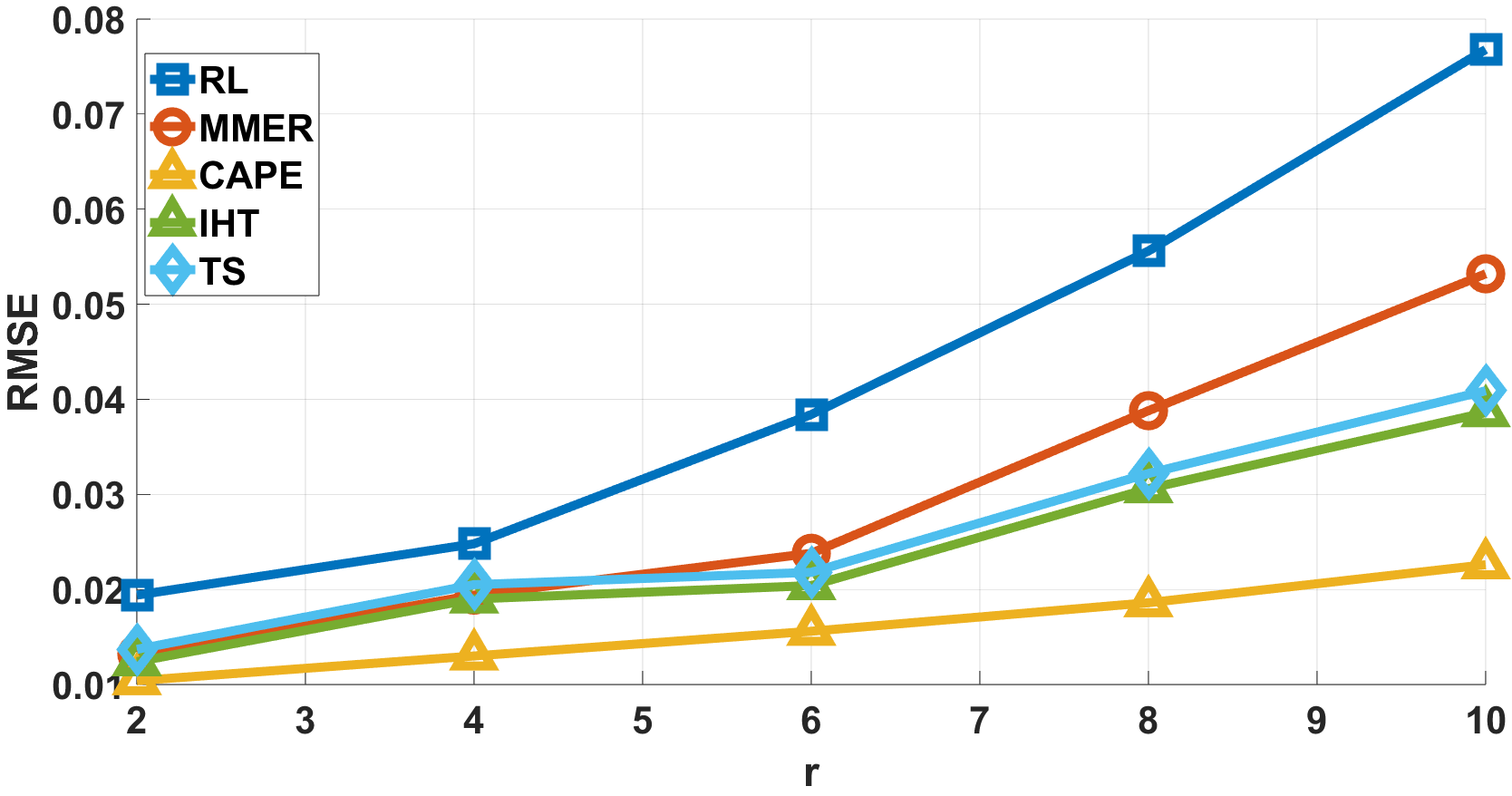}
    \includegraphics[scale=0.2]{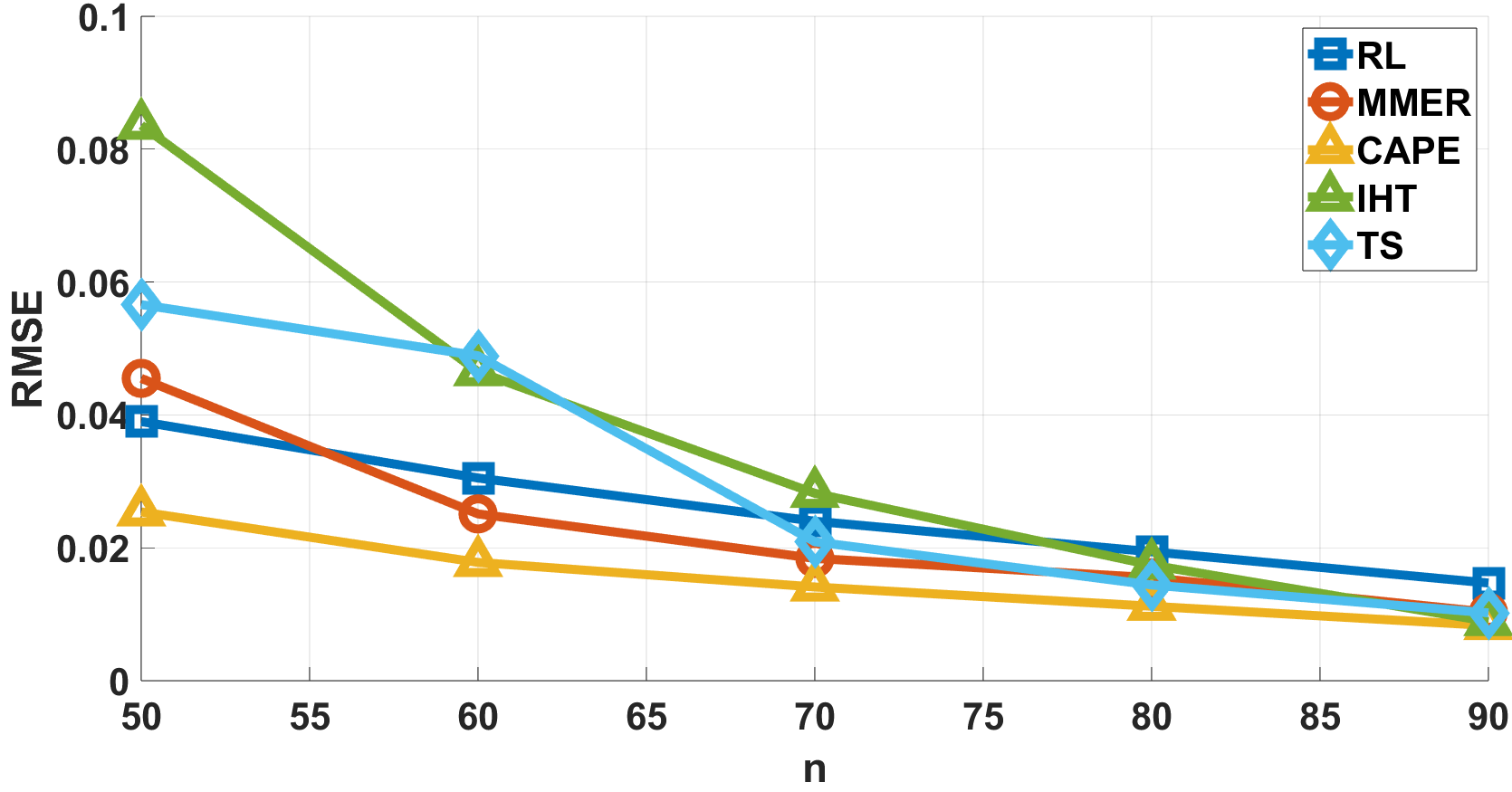}\\
    \includegraphics[scale=0.2]{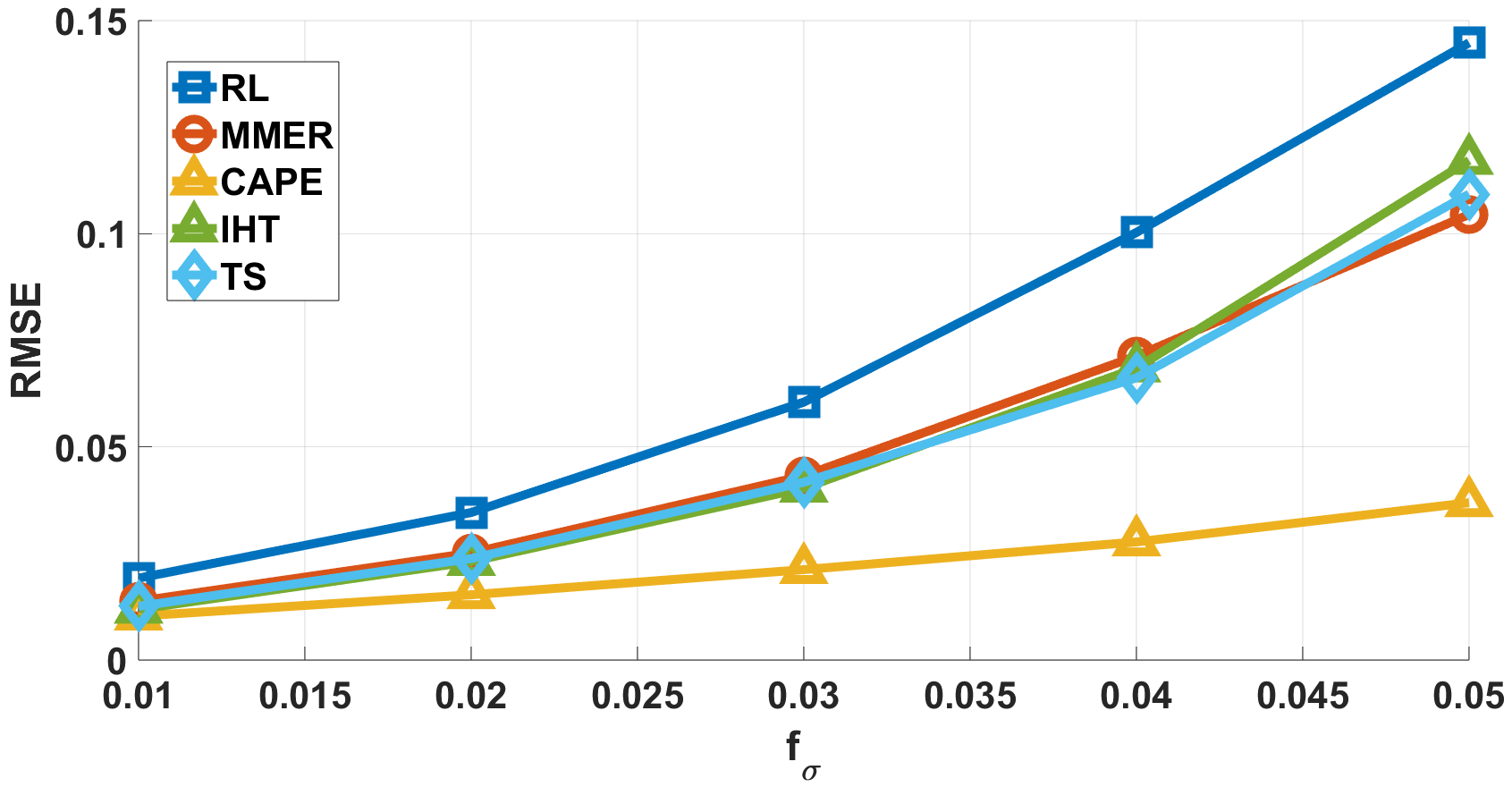}
    \includegraphics[scale=0.2]{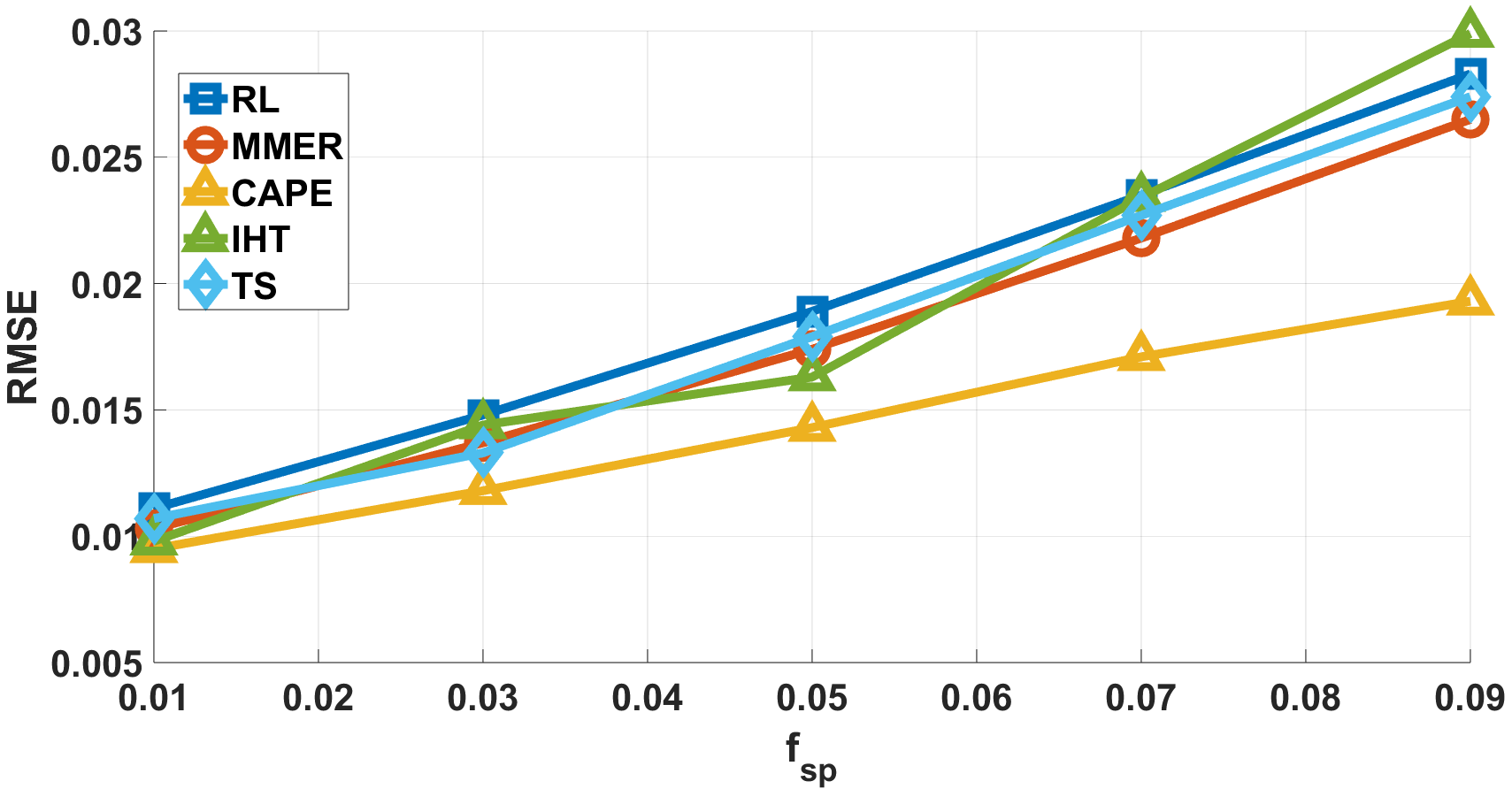}
    \caption{RRMSE (averaged over 100 noise realizations) for \textsc{Cape}, \textsc{Rl}, and \textsc{Mmer} under permutation errors with $\boldsymbol{A}$ from $CB(0.1)$. Results for experiments \textsf{EA}, \textsf{EB}, \textsf{EC}, and \textsf{ED}.}
    \label{fig:Rmse_perm_0.1}
\end{figure*}

\begin{figure*}
   \centering
    \includegraphics[scale=0.2]{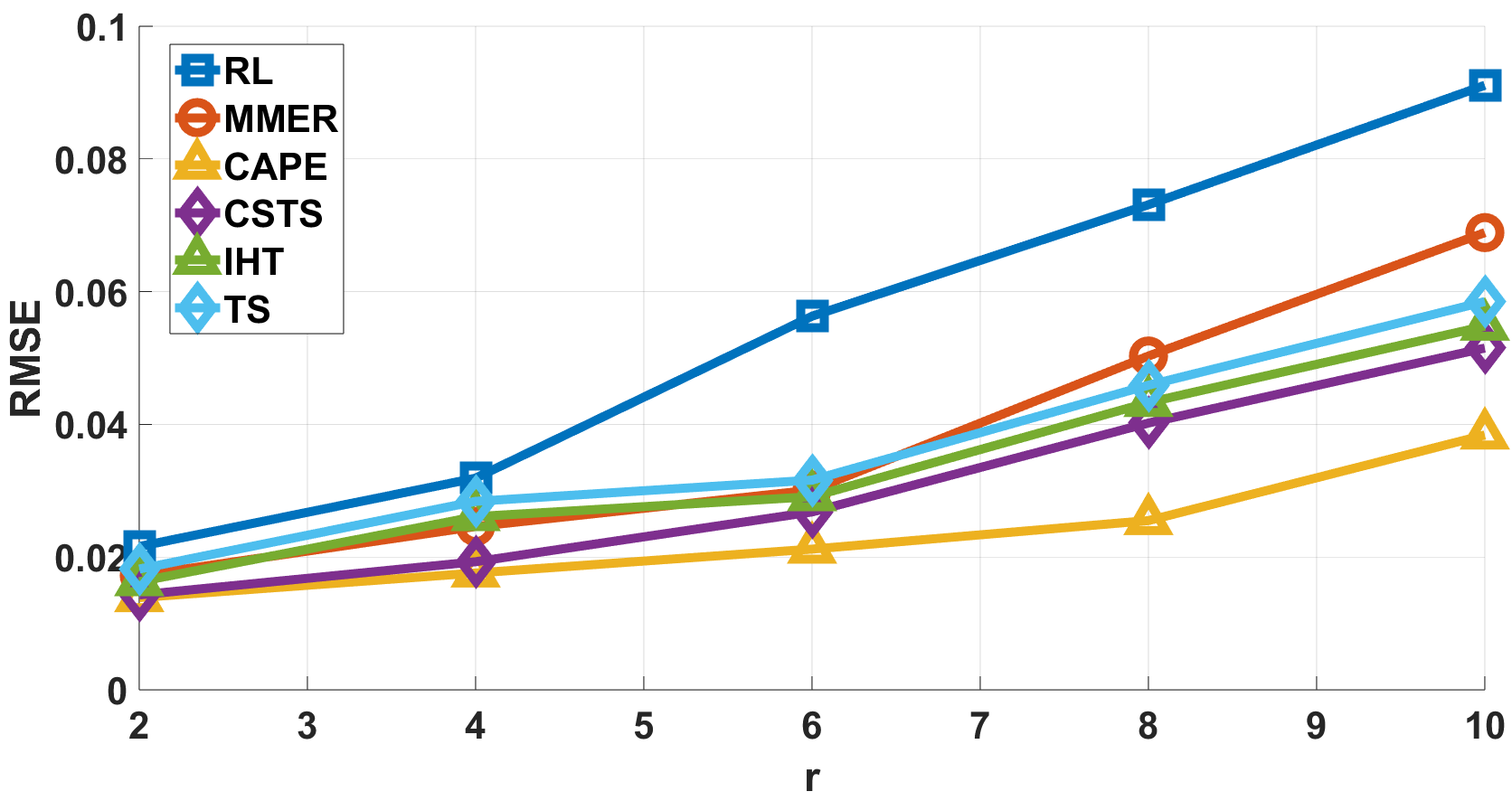}
    \includegraphics[scale=0.2]{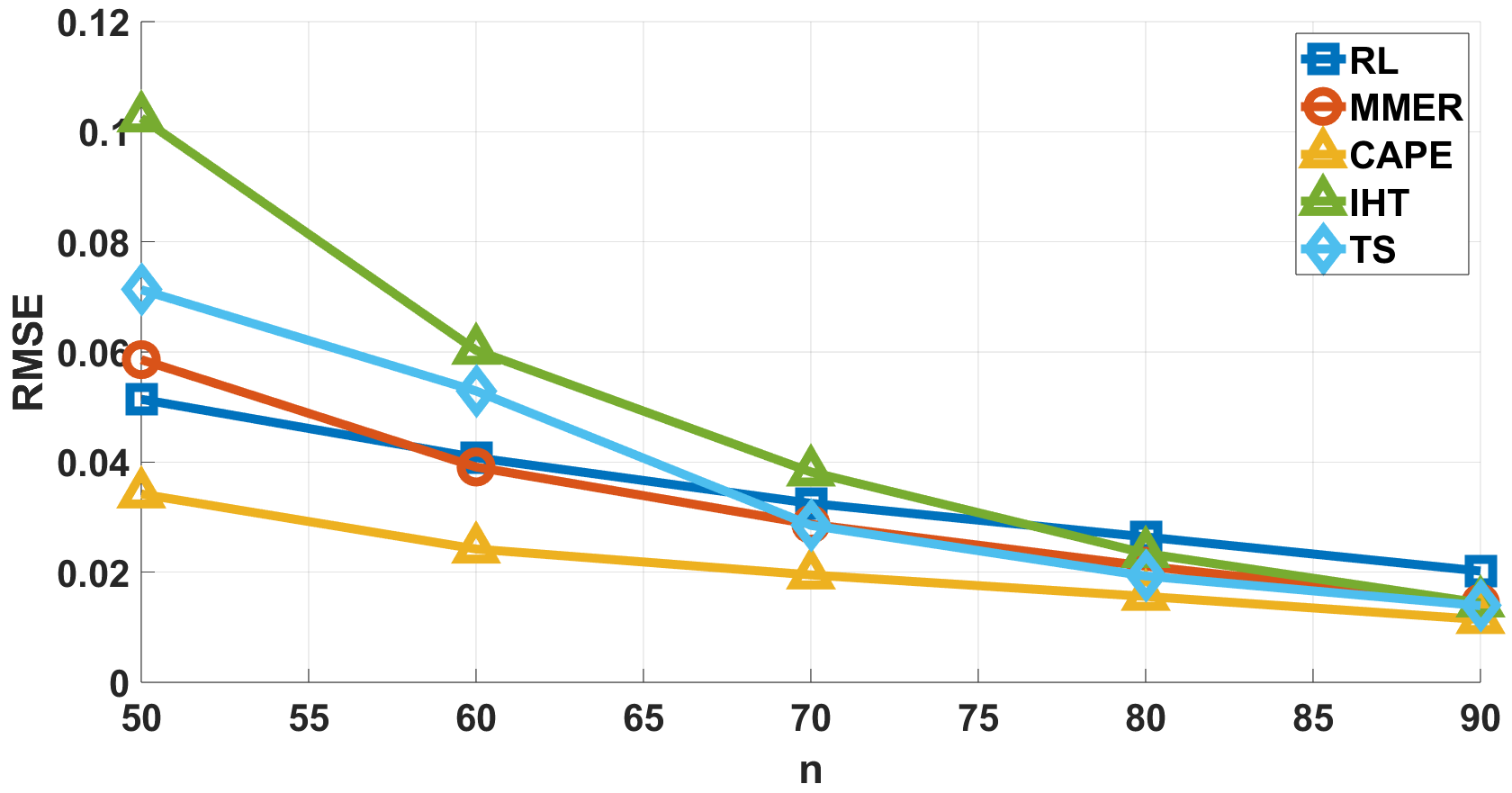}\\
    \includegraphics[scale=0.2]{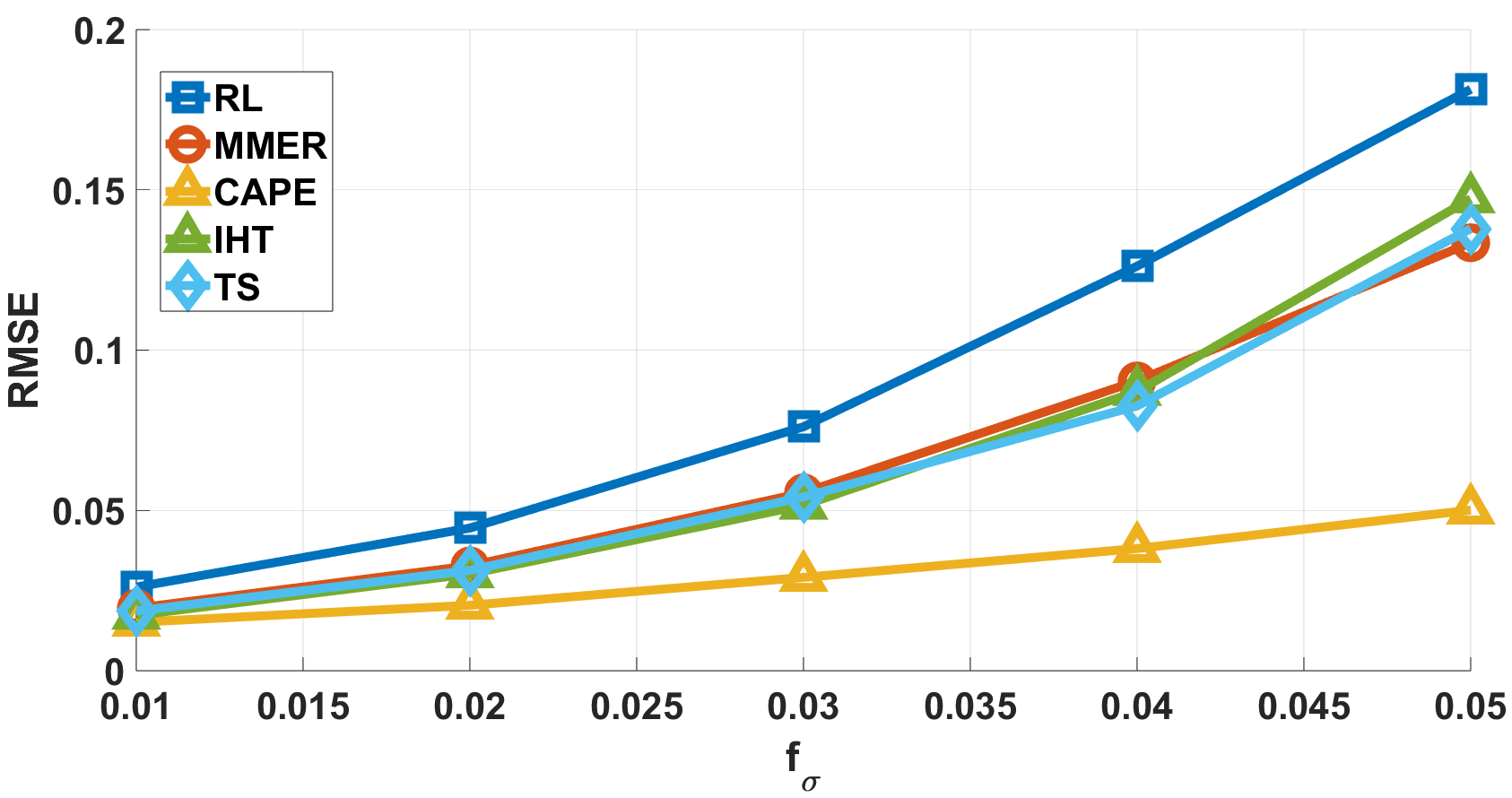}
    \includegraphics[scale=0.2]{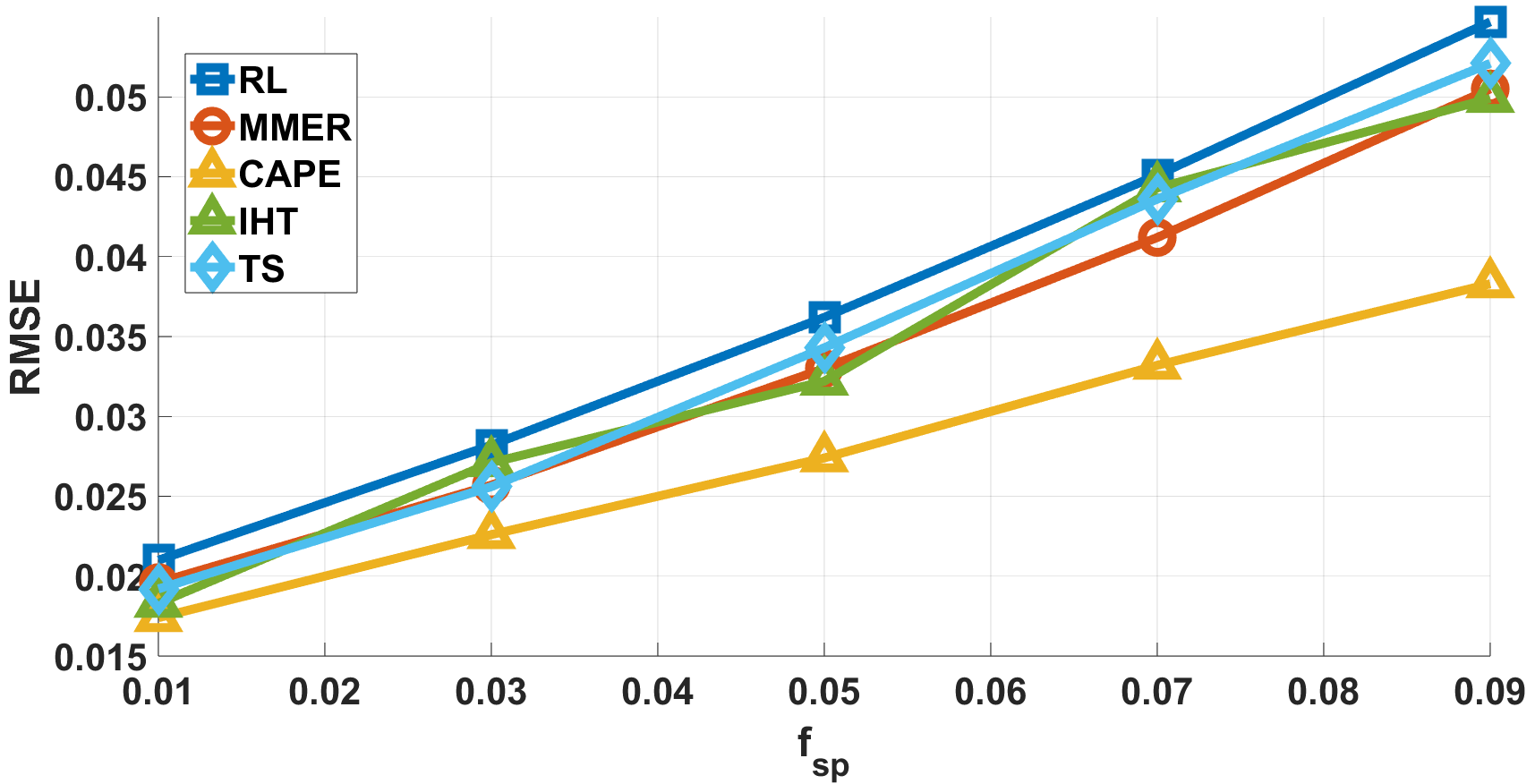}
    \caption{RRMSE comparisons for design matrices drawn from $CB(0.5)$. Results for experiments \textsf{EA}, \textsf{EB}, \textsf{EC}, and \textsf{ED} defined in the main paper.}
    \label{fig:Rmse_perm_0.5}
\end{figure*}

\begin{figure*}
   \centering
    \includegraphics[scale=0.2]{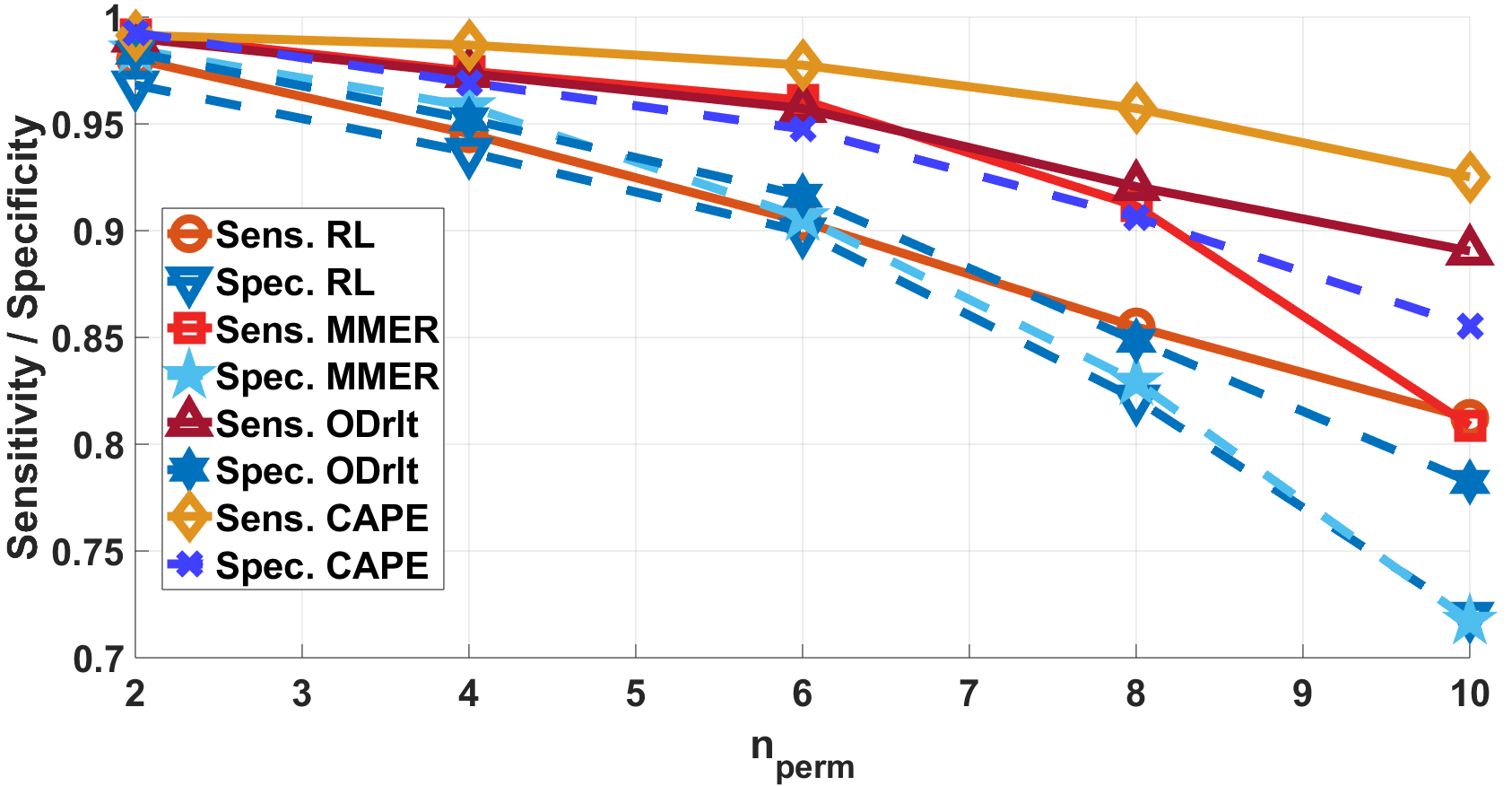}
    \includegraphics[scale=0.2]{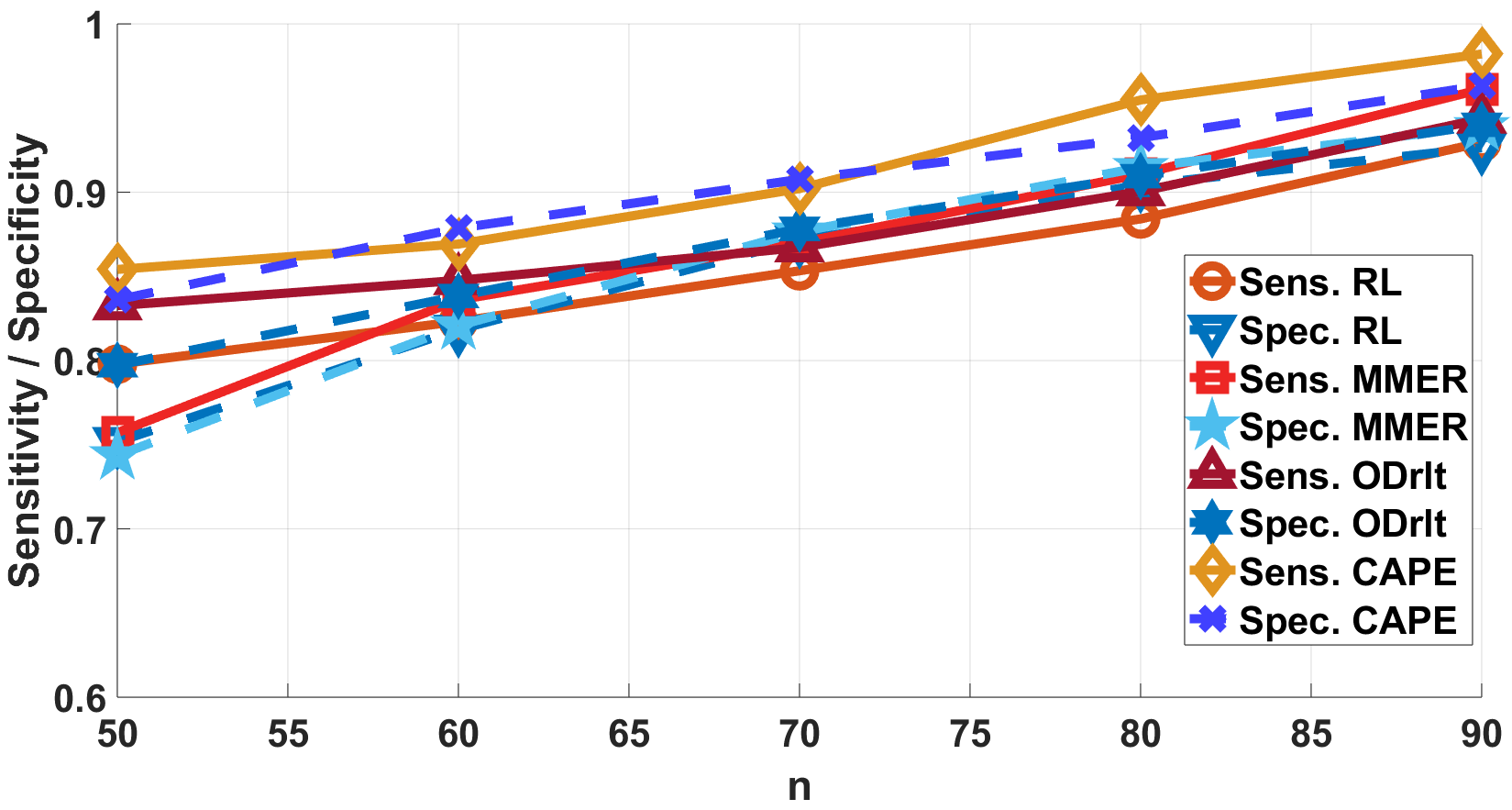}\\
    \includegraphics[scale=0.2]{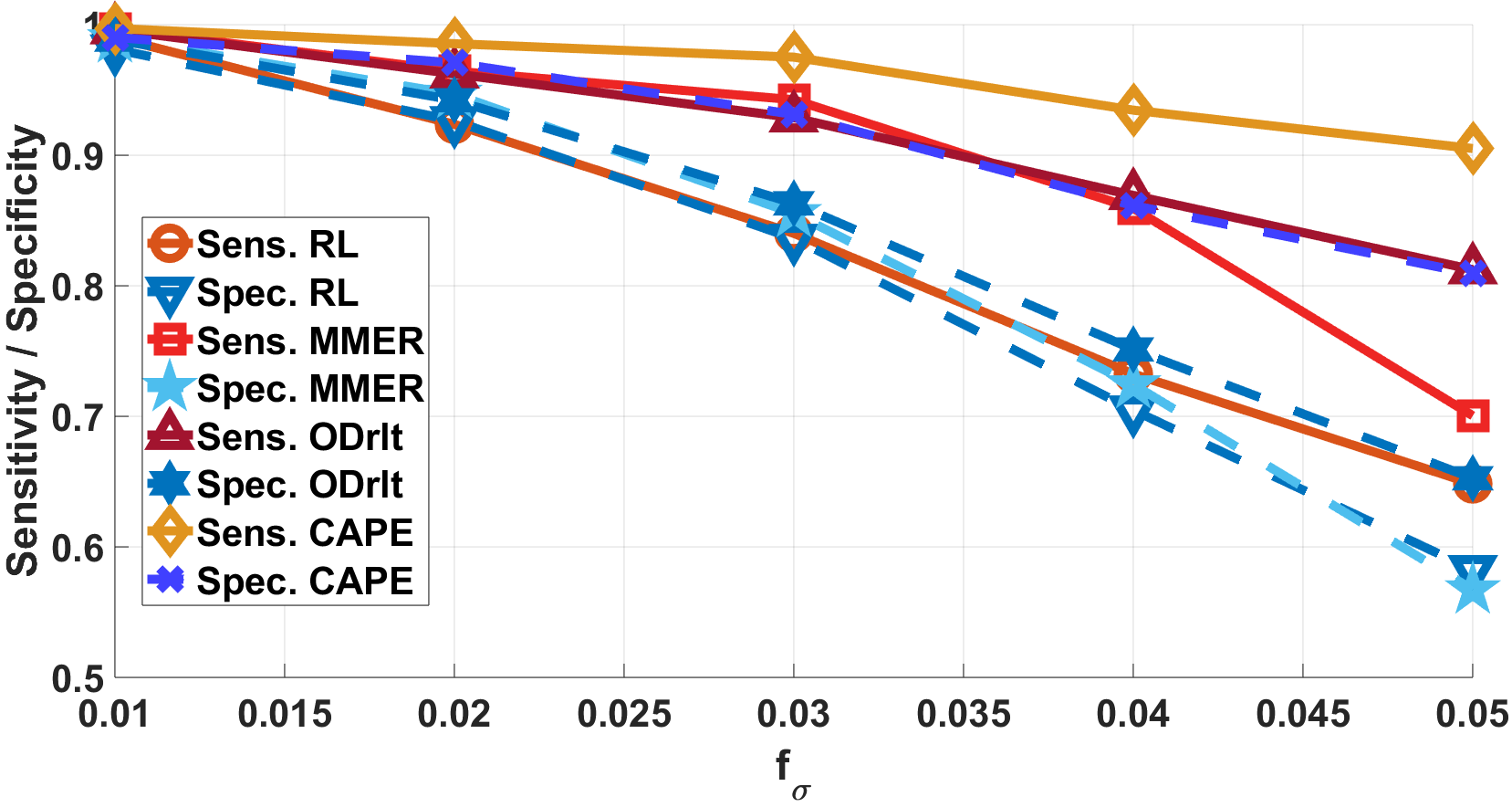}
    \includegraphics[scale=0.2]{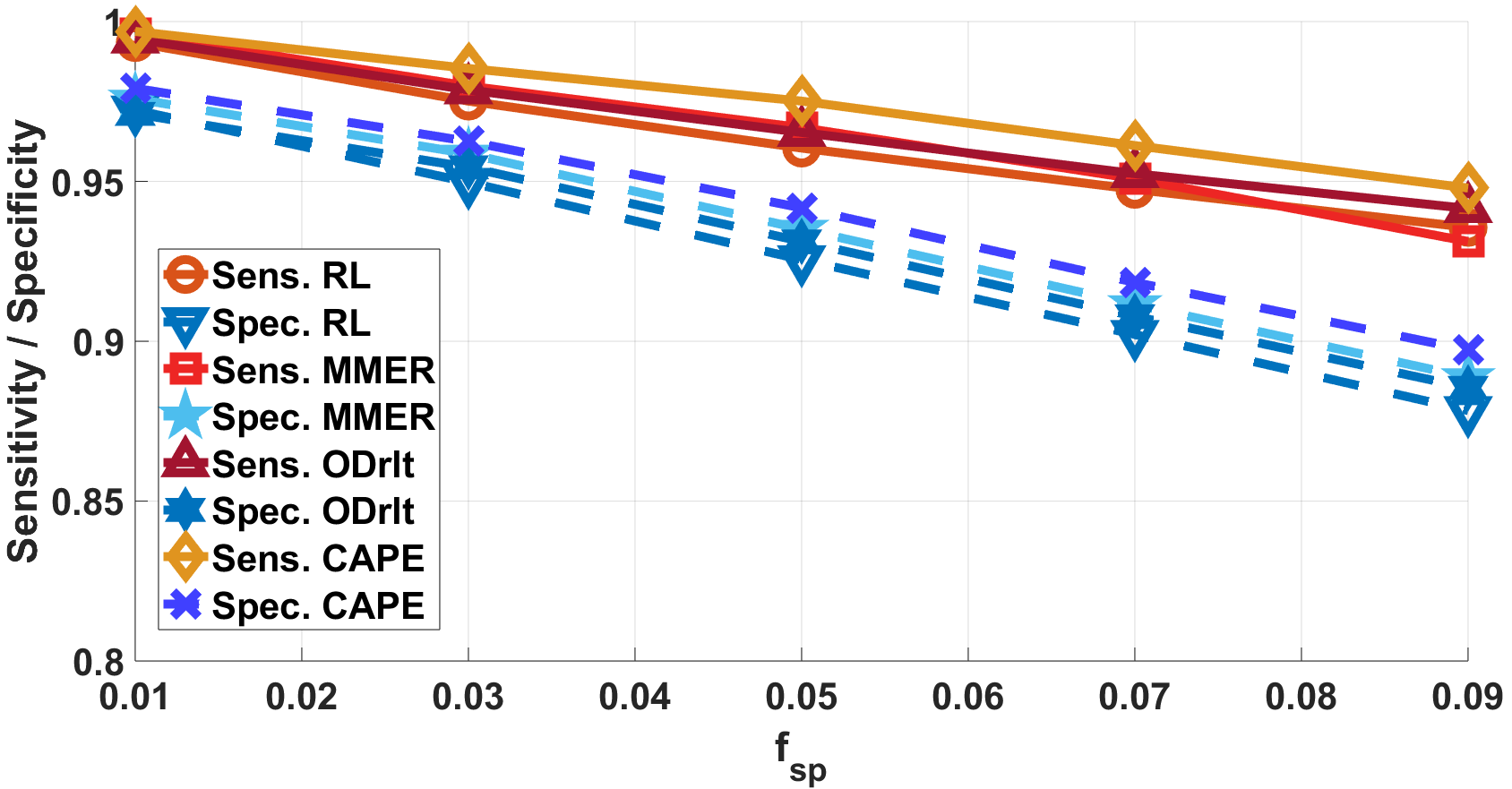}
    \caption{Sensitivity and specificity (averaged over 100 noise realizations) under permutation errors with $\boldsymbol{A}$ from Centered Bernoulli($0.3$). Results for experiments \textsf{EA}, \textsf{EB}, \textsf{EC}, and \textsf{ED} defined in the main paper.}
    \label{fig:Sens_spec_perm_0.3}
\end{figure*}

\begin{figure*}
   \centering
    \includegraphics[scale=0.2]{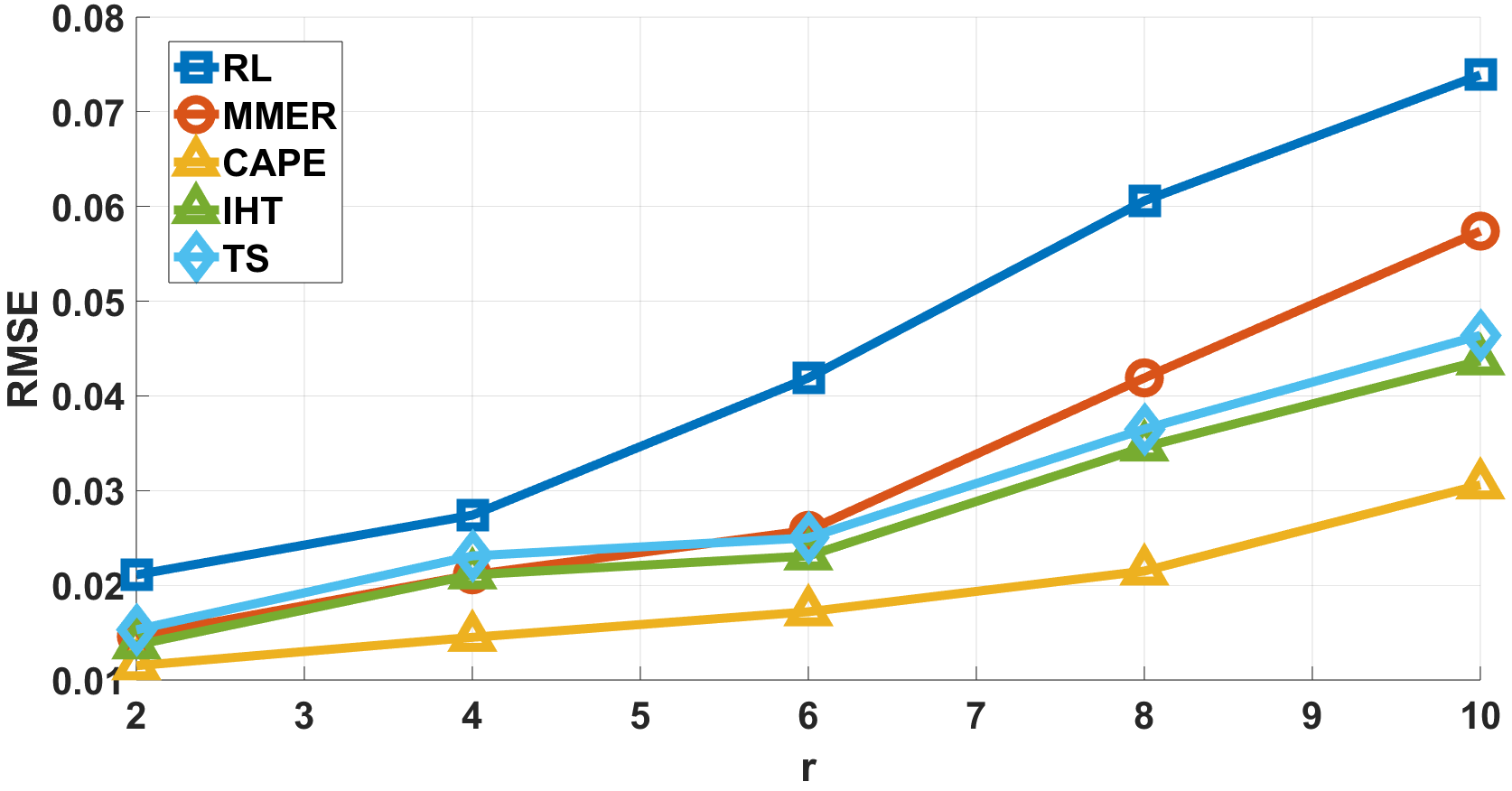}
    \includegraphics[scale=0.2]{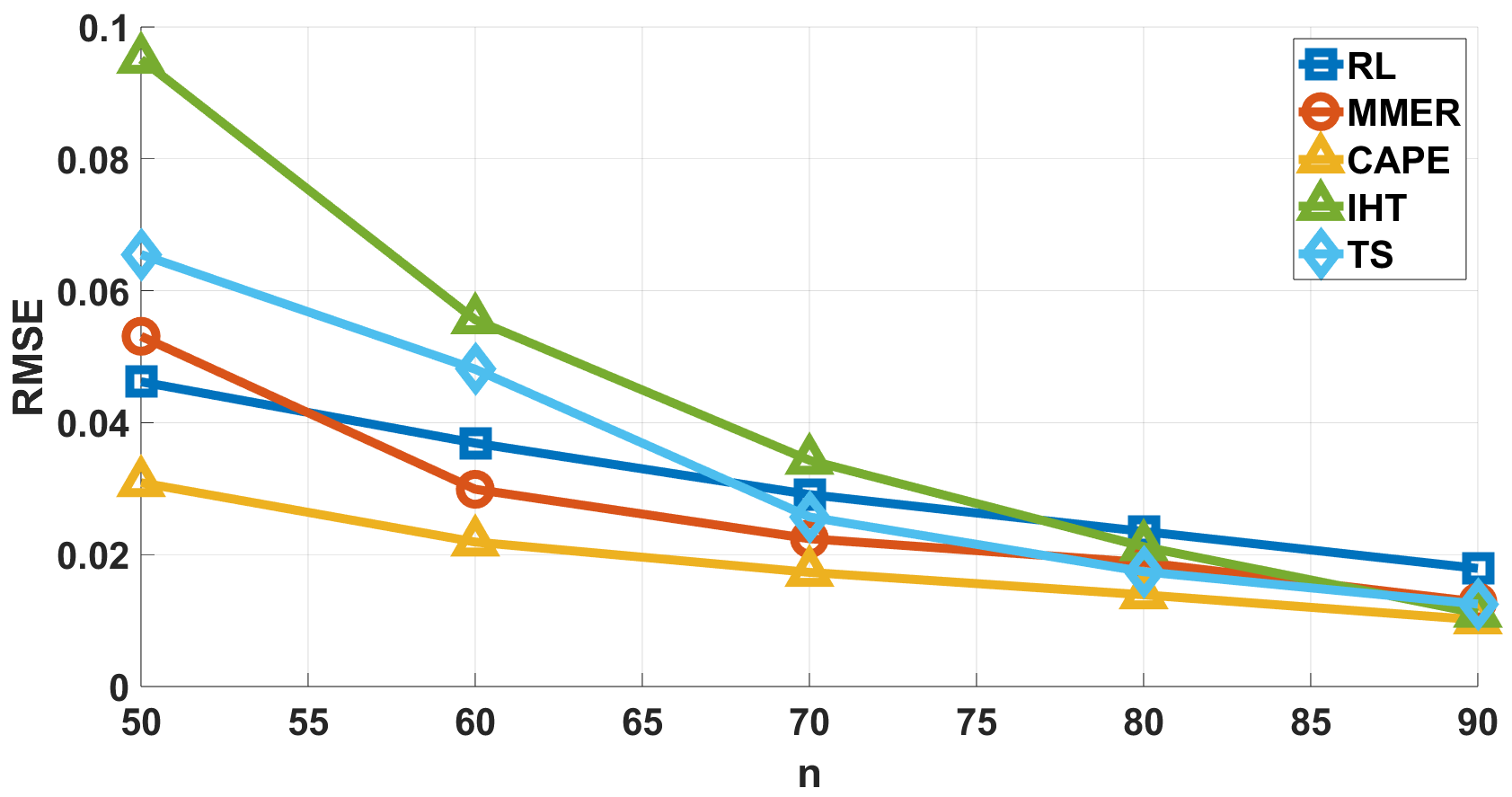}\\
    \includegraphics[scale=0.2]{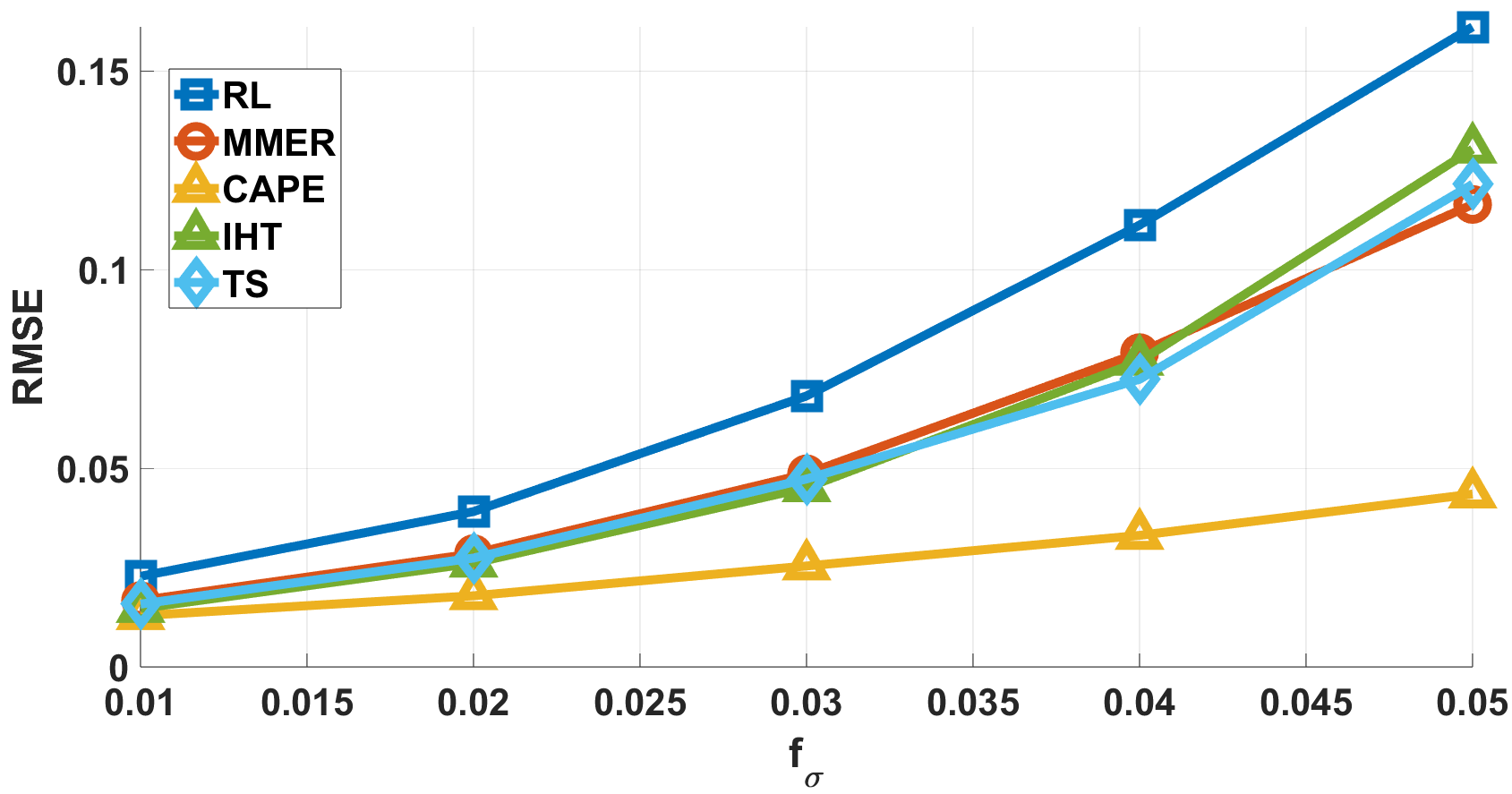}
    \includegraphics[scale=0.2]{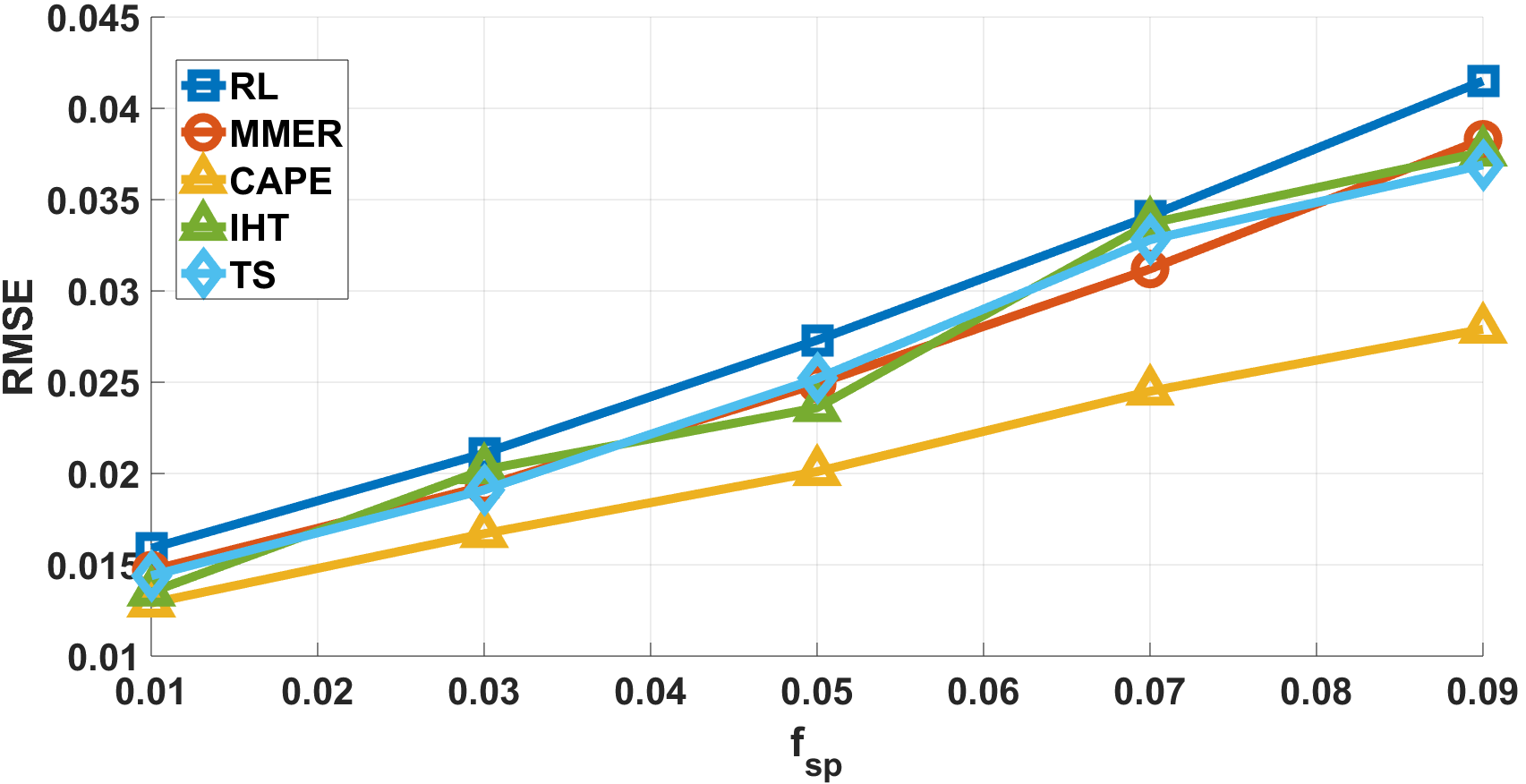}
    \caption{RRMSE results (averaged over 100 noise realizations) for \textsc{Cape}, \textsc{Rl}, and \textsc{Mmer} under permutation errors with $\boldsymbol{A}$ from Centered Bernoulli($0.3$). Results for experiments \textsf{EA}, \textsf{EB}, \textsf{EC}, and \textsf{ED} defined in the main paper.}
    \label{fig:Rmse_perm_0.3}
\end{figure*}

\begin{figure*}[t]
   \centering
    \includegraphics[scale=0.2]{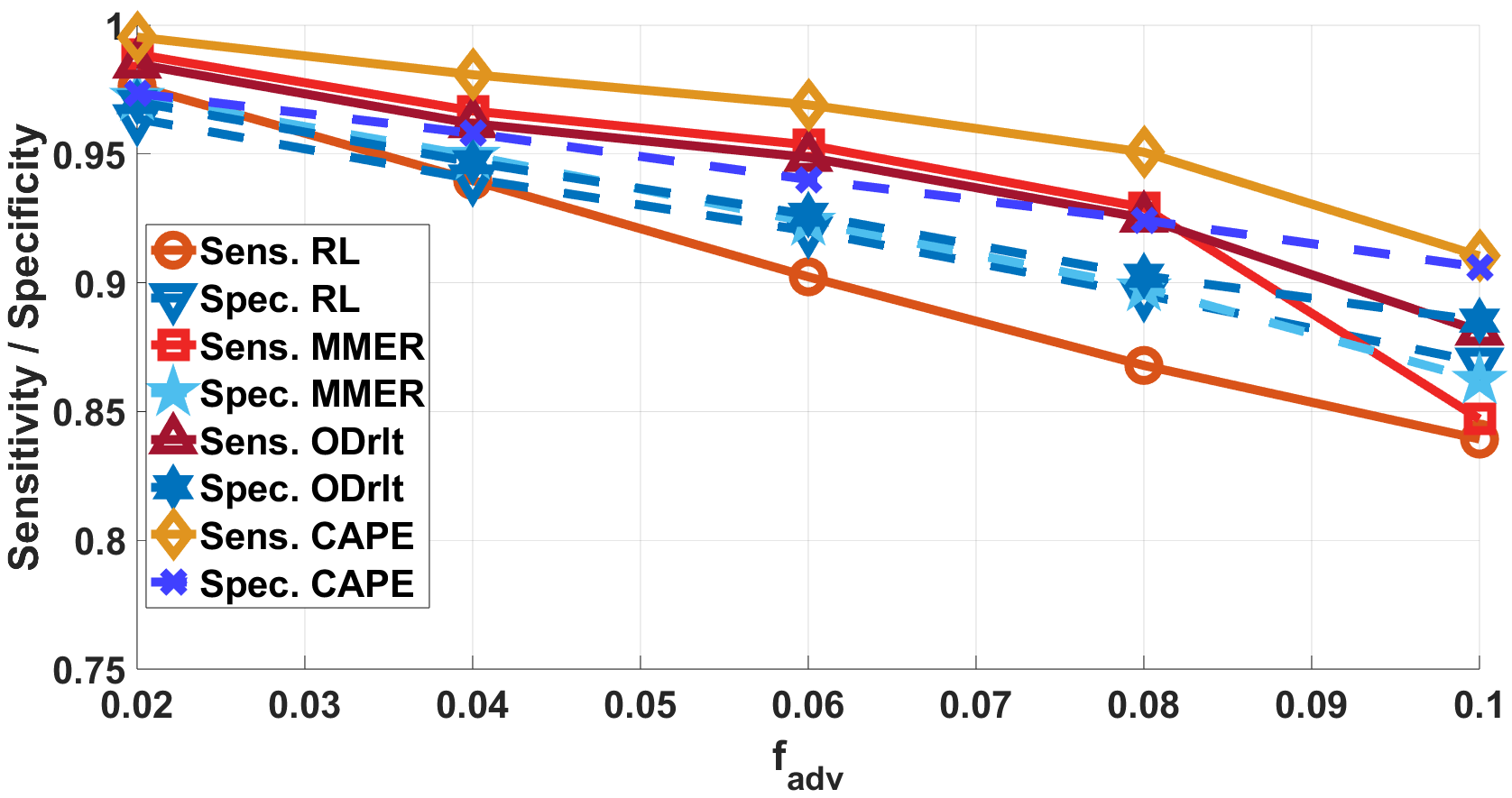}
    \includegraphics[scale=0.2]{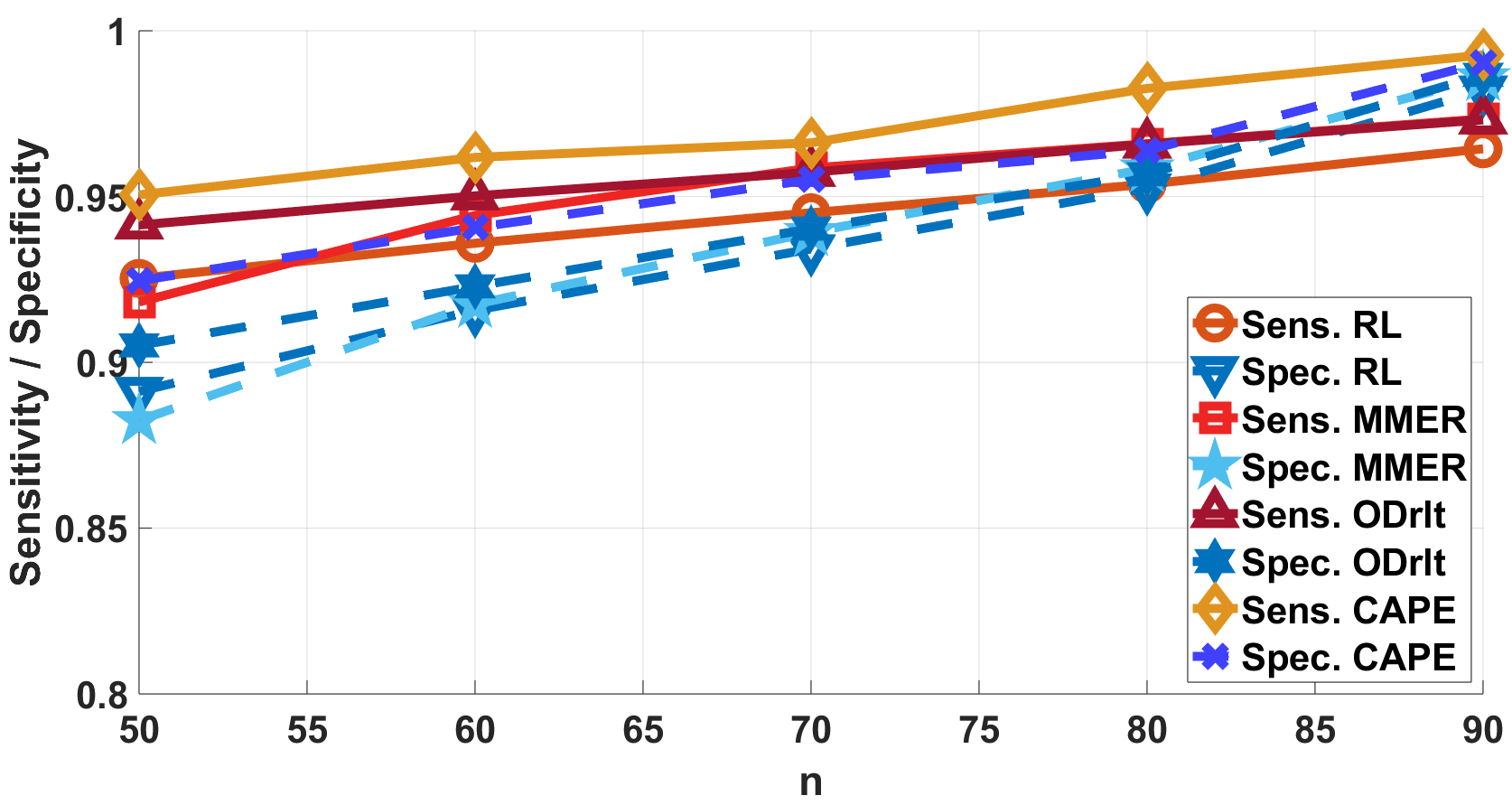}\\
    \includegraphics[scale=0.2]{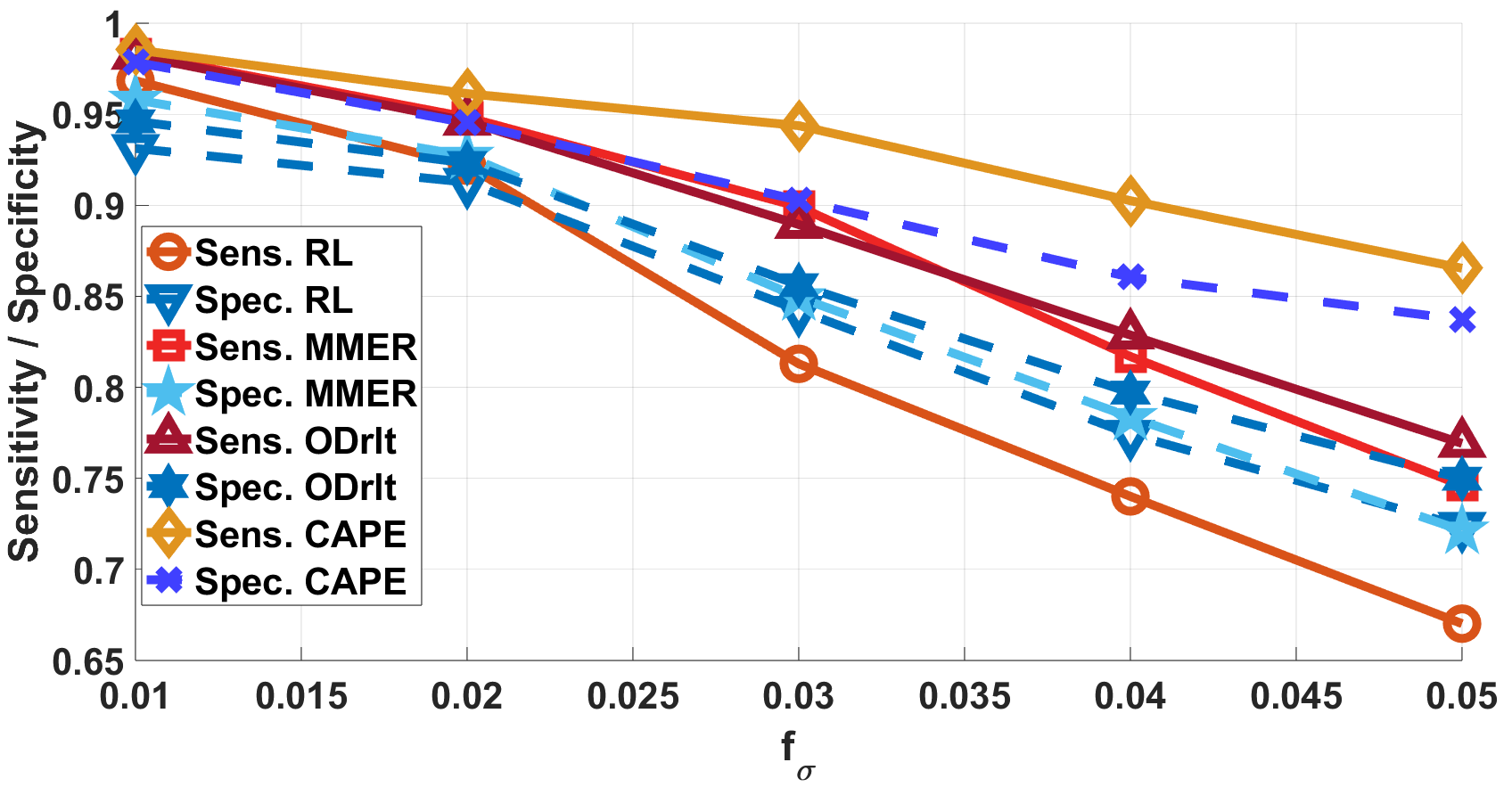}
    \includegraphics[scale=0.2]{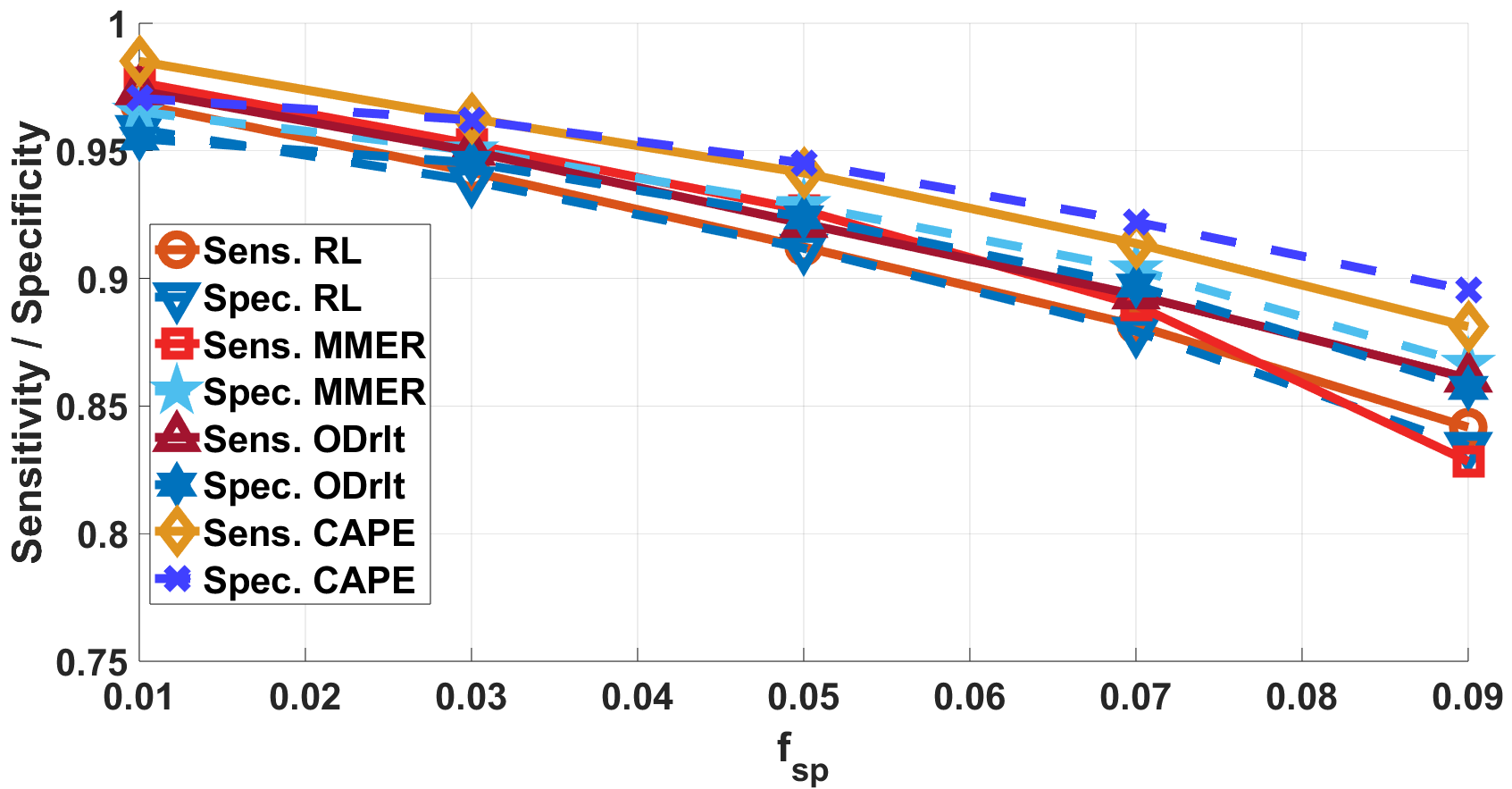}
    \caption{Average sensitivity and specificity (100 noise realizations) under \textsf{SSM} MMEs with $\theta=0.3$. Panels correspond to (\textsf{EA}), (\textsf{EB}), (\textsf{EC}), (\textsf{ED}) defined in the main paper. Methods: \textsc{Cape}, RL, \textsc{Mmer}, \textsc{Odrlt}.}
    \label{fig:Sens_spec_ssm_0.3}
\end{figure*}

\begin{figure*}[t]
   \centering
    \includegraphics[scale=0.19]{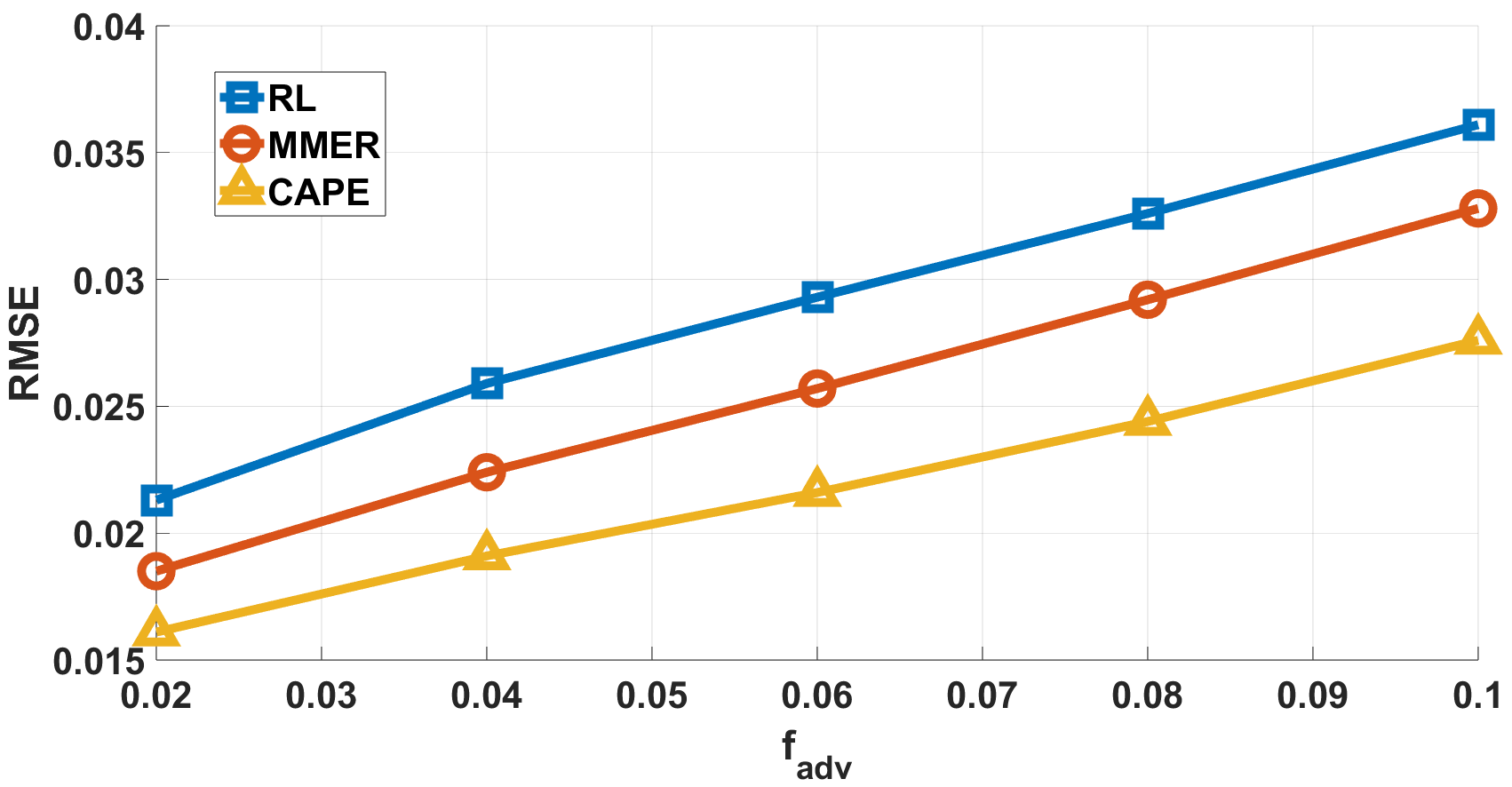}
    \includegraphics[scale=0.19]{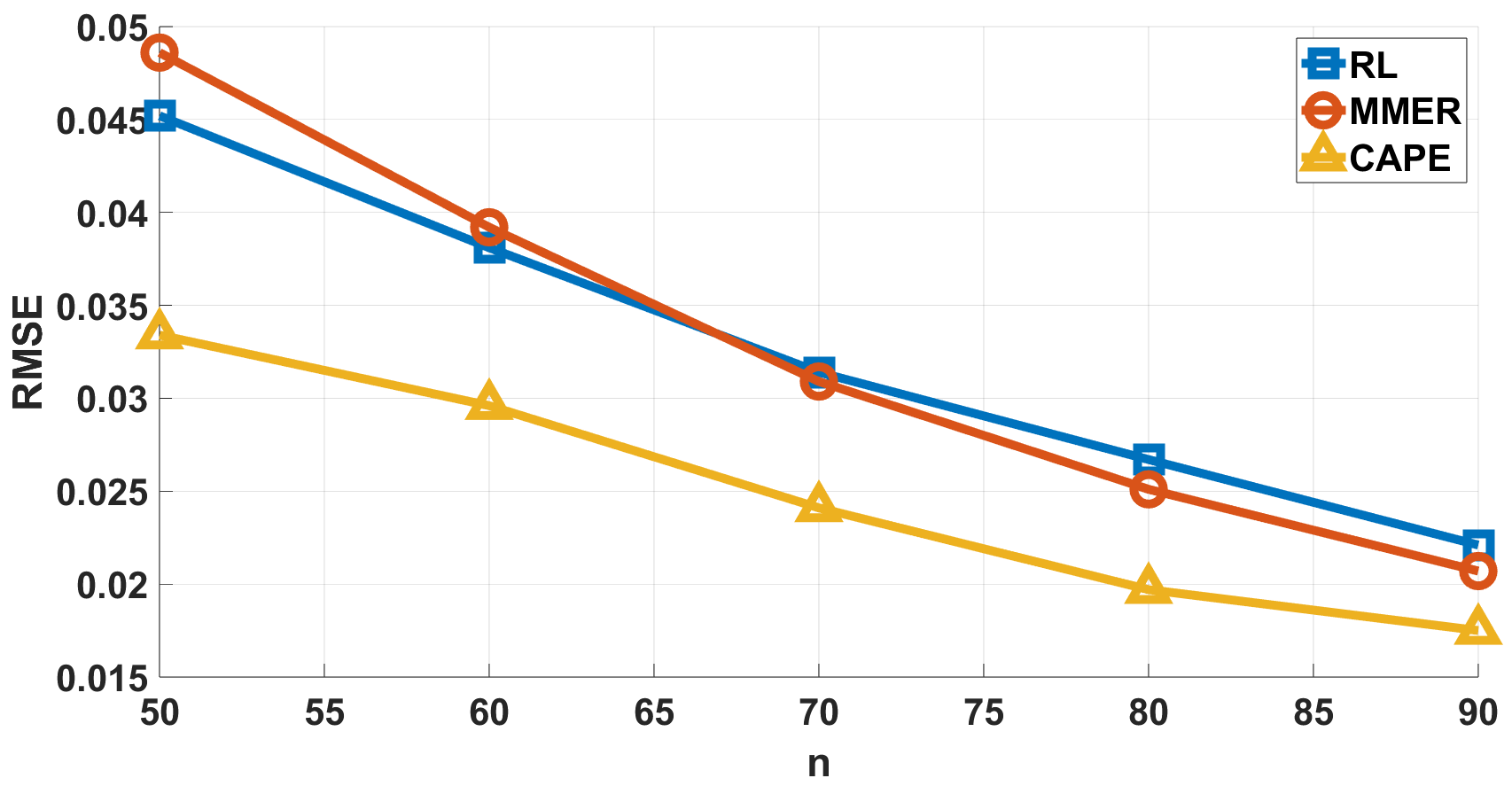}\\
    \includegraphics[scale=0.2]{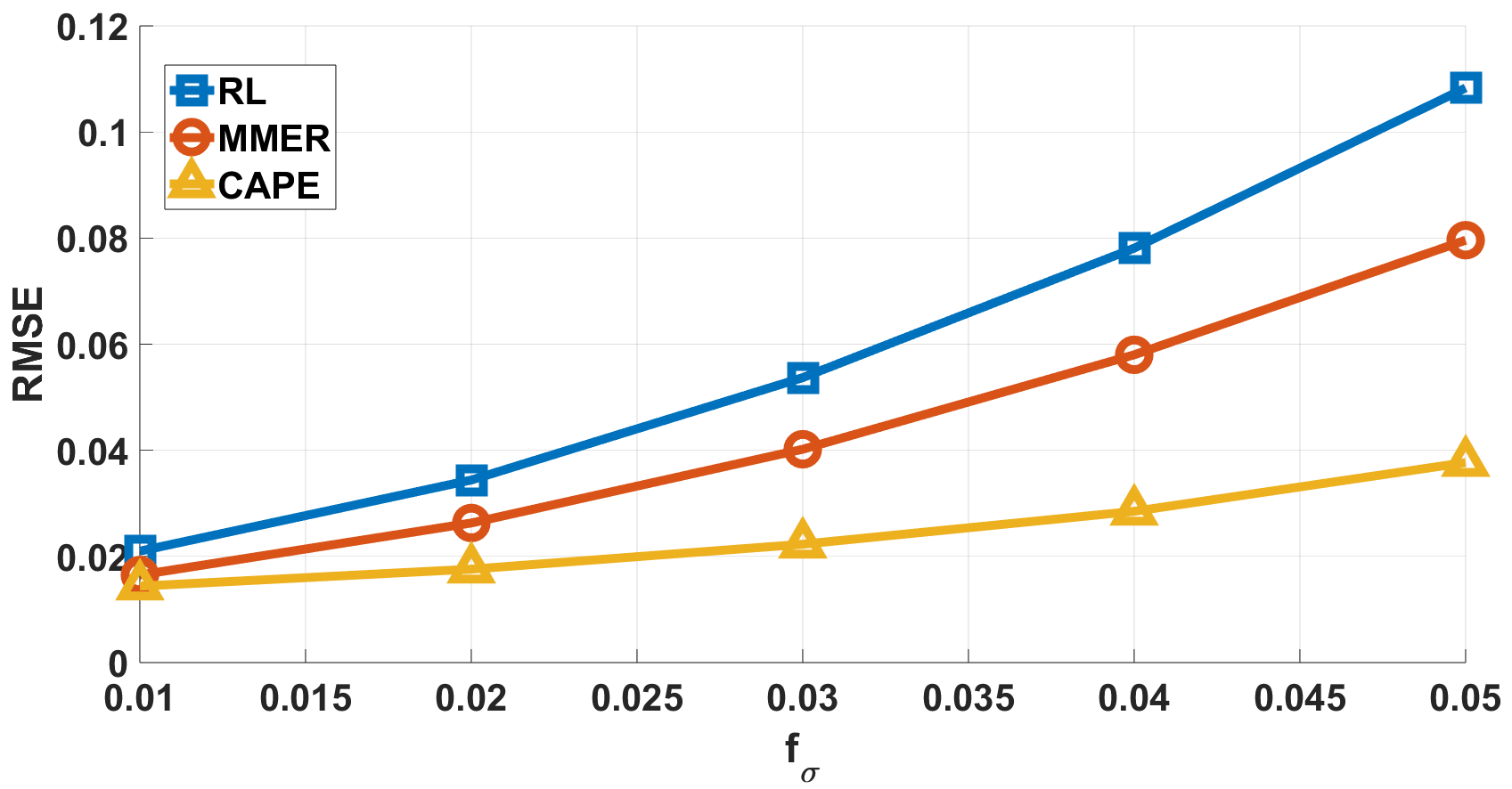}
    \includegraphics[scale=0.2]{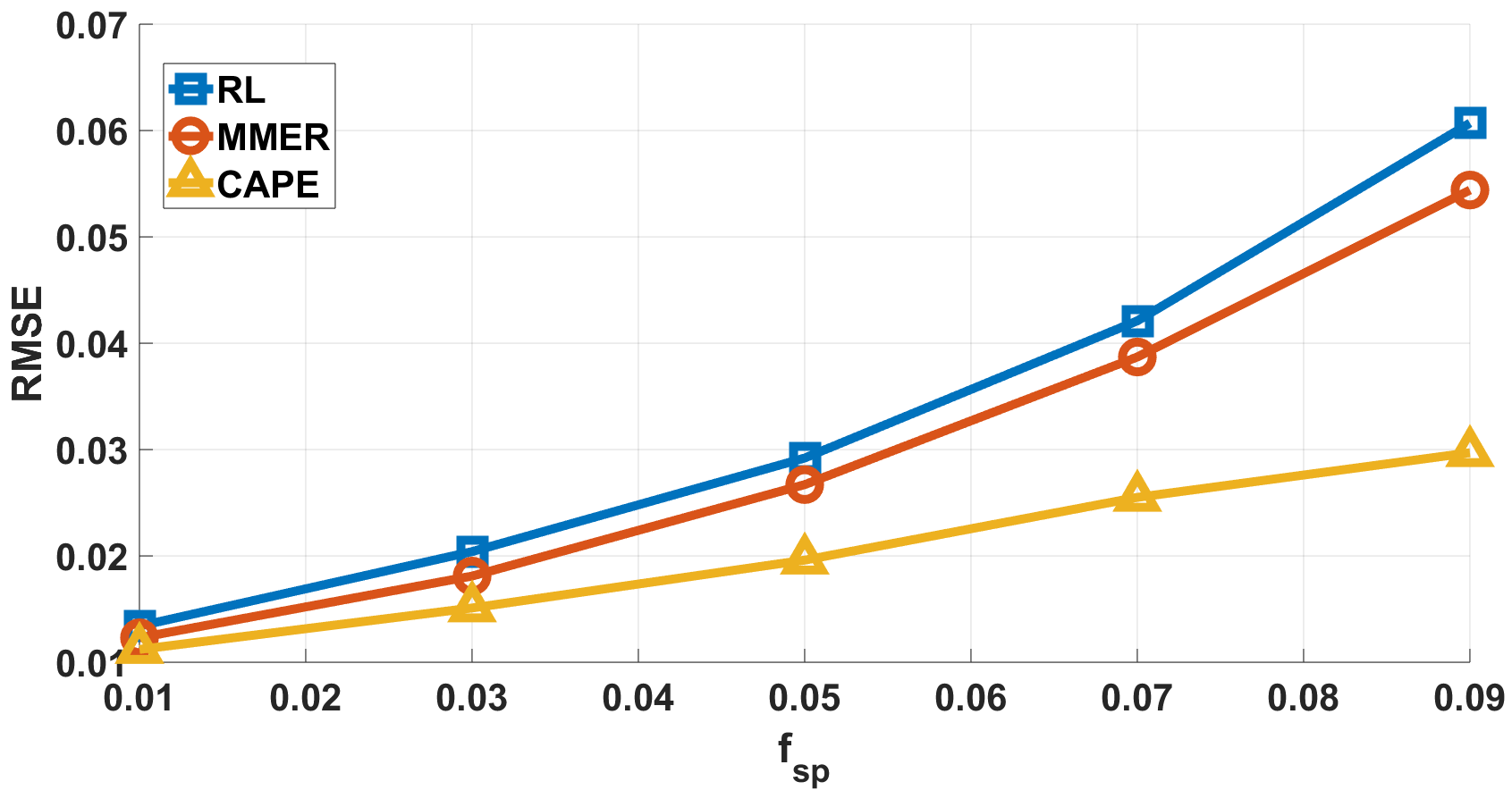}
    \caption{Average RRMSE (100 noise realizations) under \textsf{SSM} mismatches for $\theta=0.3$. Panels correspond to (\textsf{EA}), (\textsf{EB}), (\textsf{EC}), (\textsf{ED}) defined in the main paper. Methods: \textsc{Cape}, RL, \textsc{Mmer}.}
    \label{fig:Rmse_ssm_0.3}
\end{figure*}

\begin{figure*}
   \centering
    \includegraphics[scale=0.2]{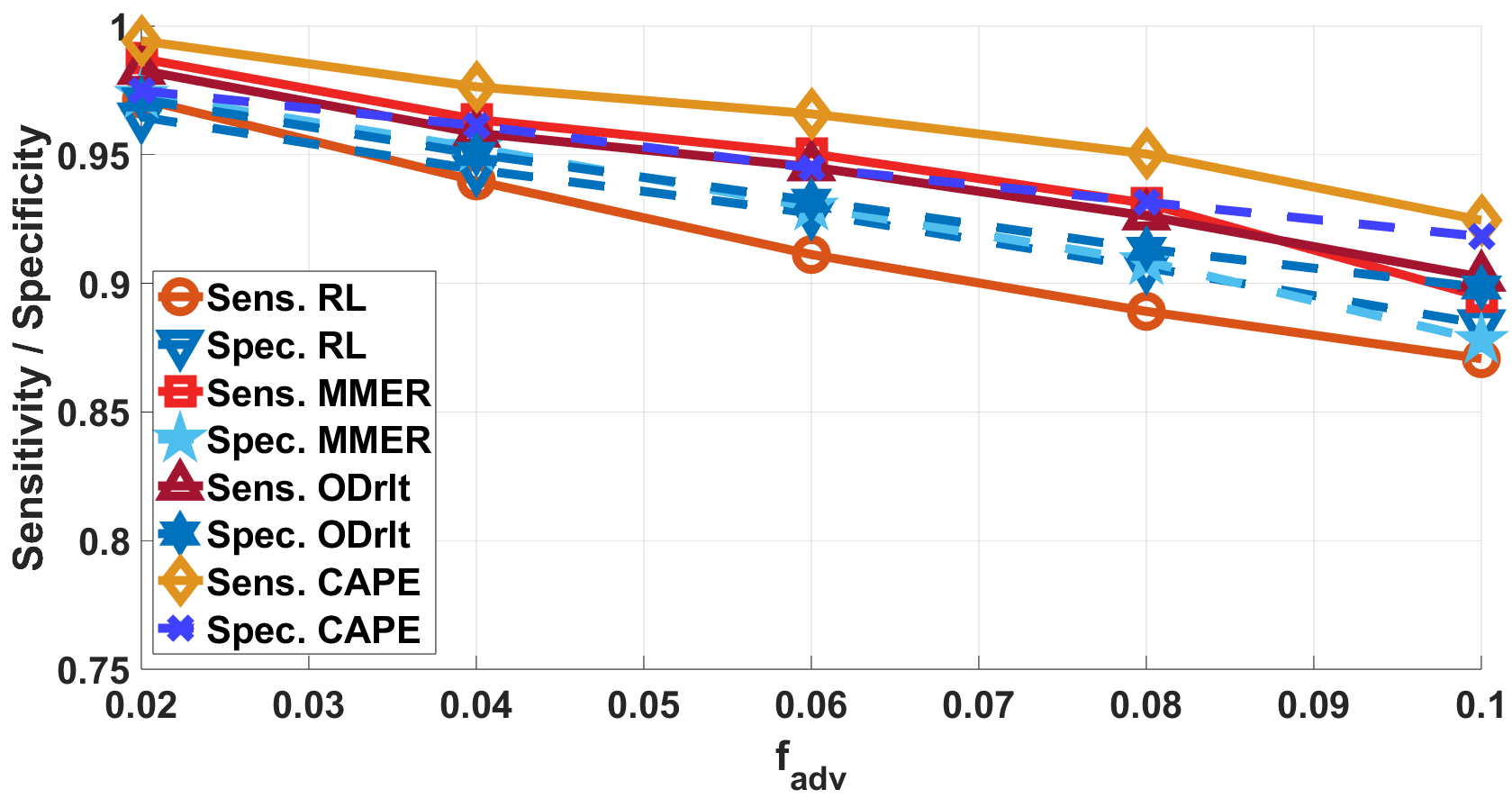}
    \includegraphics[scale=0.2]{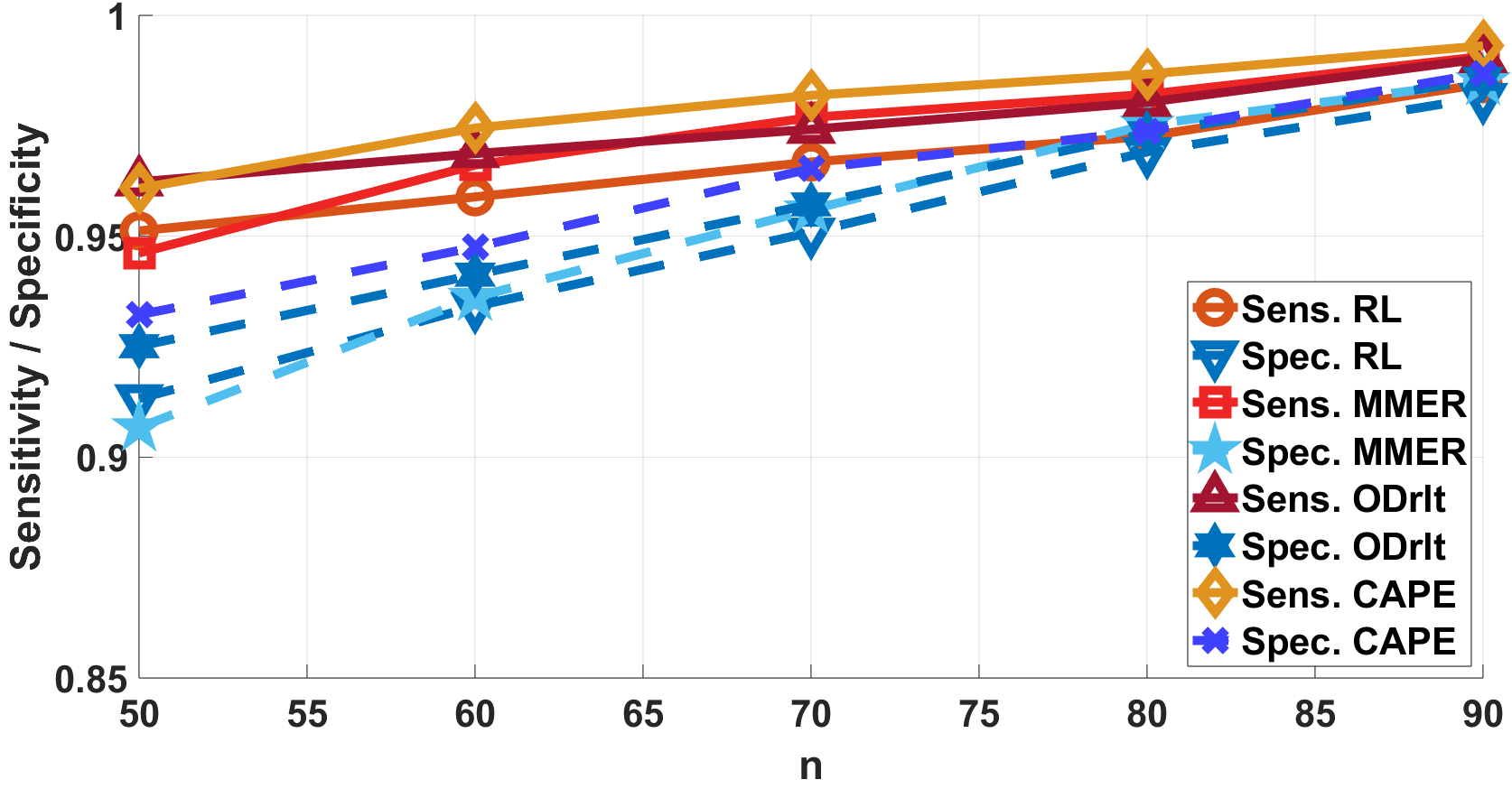}\\
    \includegraphics[scale=0.2]{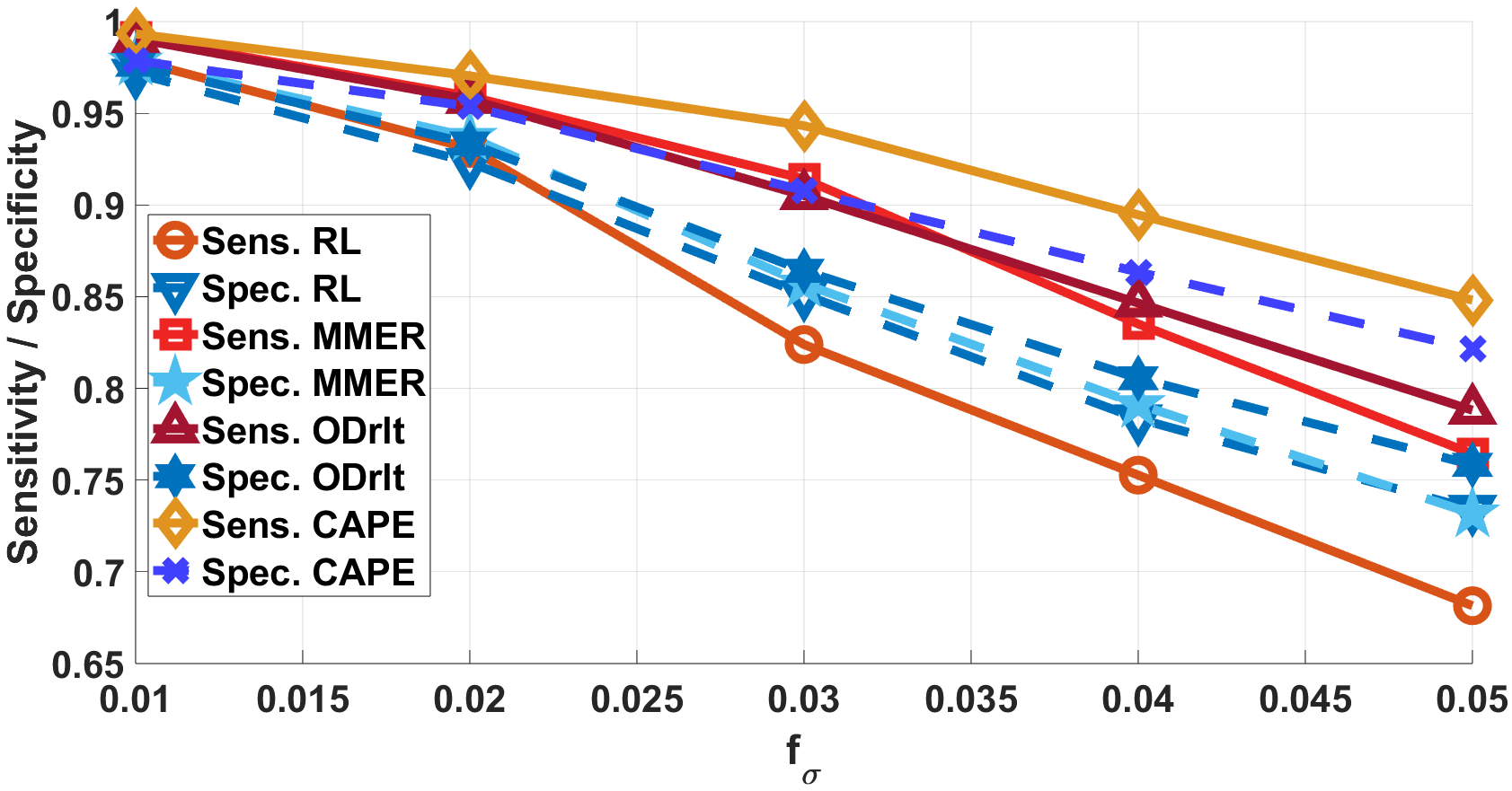}
    \includegraphics[scale=0.2]{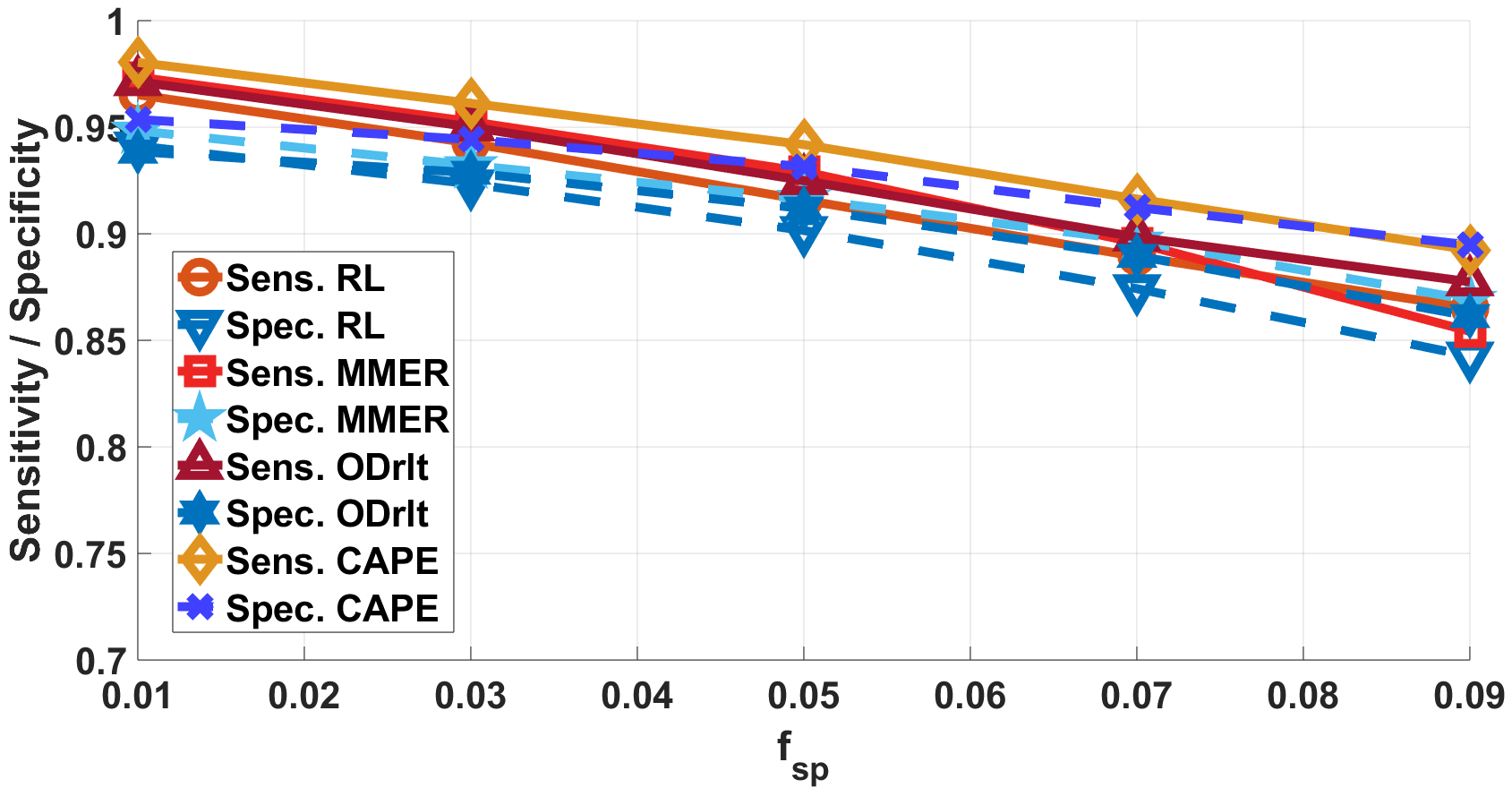}
    \caption{Average Sensitivity and Specificity (over 100 independent noise runs keeping $\boldsymbol{\beta}^*$, $\boldsymbol{A}$, and $\boldsymbol{\delta}^*$ fixed) for detecting defective samples (i.e., non-zero values of $\boldsymbol{\beta}^*$) using
    \textsc{Cape},    Robust \textsc{Lasso} (\textsc{Rl}), \textsc{Mmer}, and \textsc{Odrlt} in the presence of \textsf{ASM} MMes.
    The design matrix $\boldsymbol{A}$ has rows i.i.d. from Centered Bernoulli($0.1$). 
    Left to right, top to bottom: results for experiments (\textsf{EA}), (\textsf{EB}), (\textsf{EC}), (\textsf{ED}) defined in the main paper.}  
    \label{fig:Sens_spec_asm_0.1}    
\end{figure*}

\begin{figure*}
   \centering
    \includegraphics[scale=0.2]{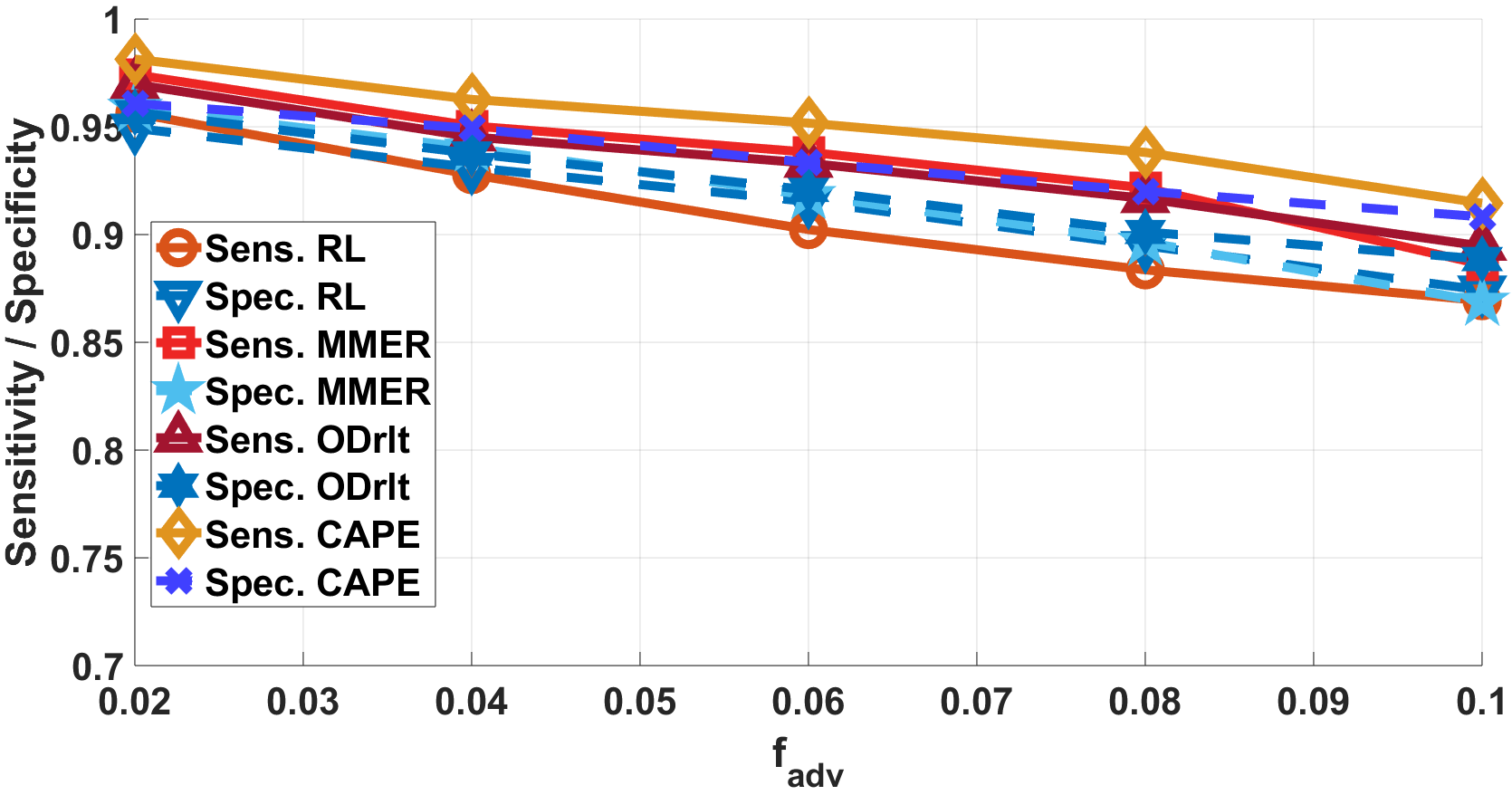}
    \includegraphics[scale=0.2]{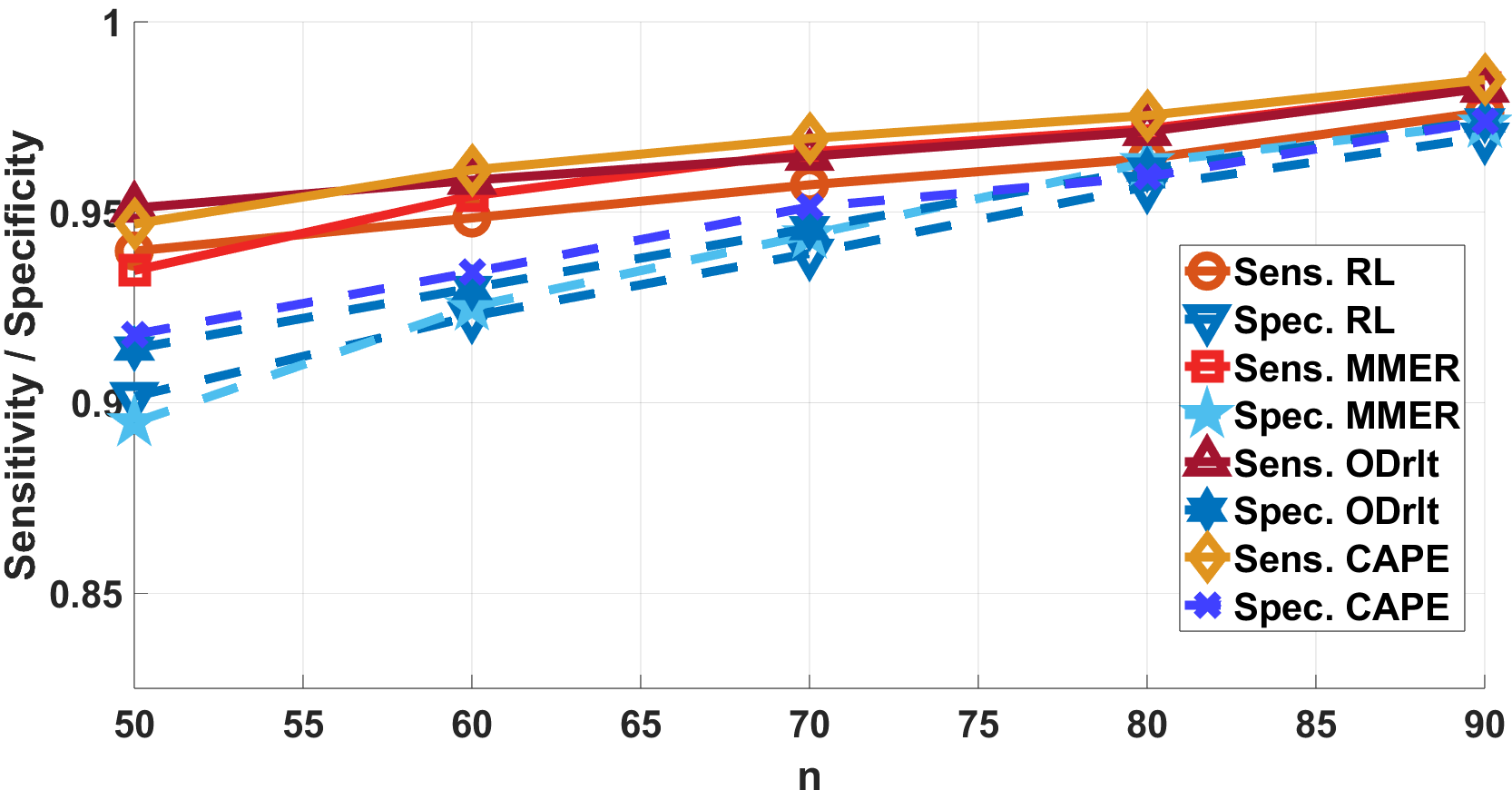}\\
    \includegraphics[scale=0.2]{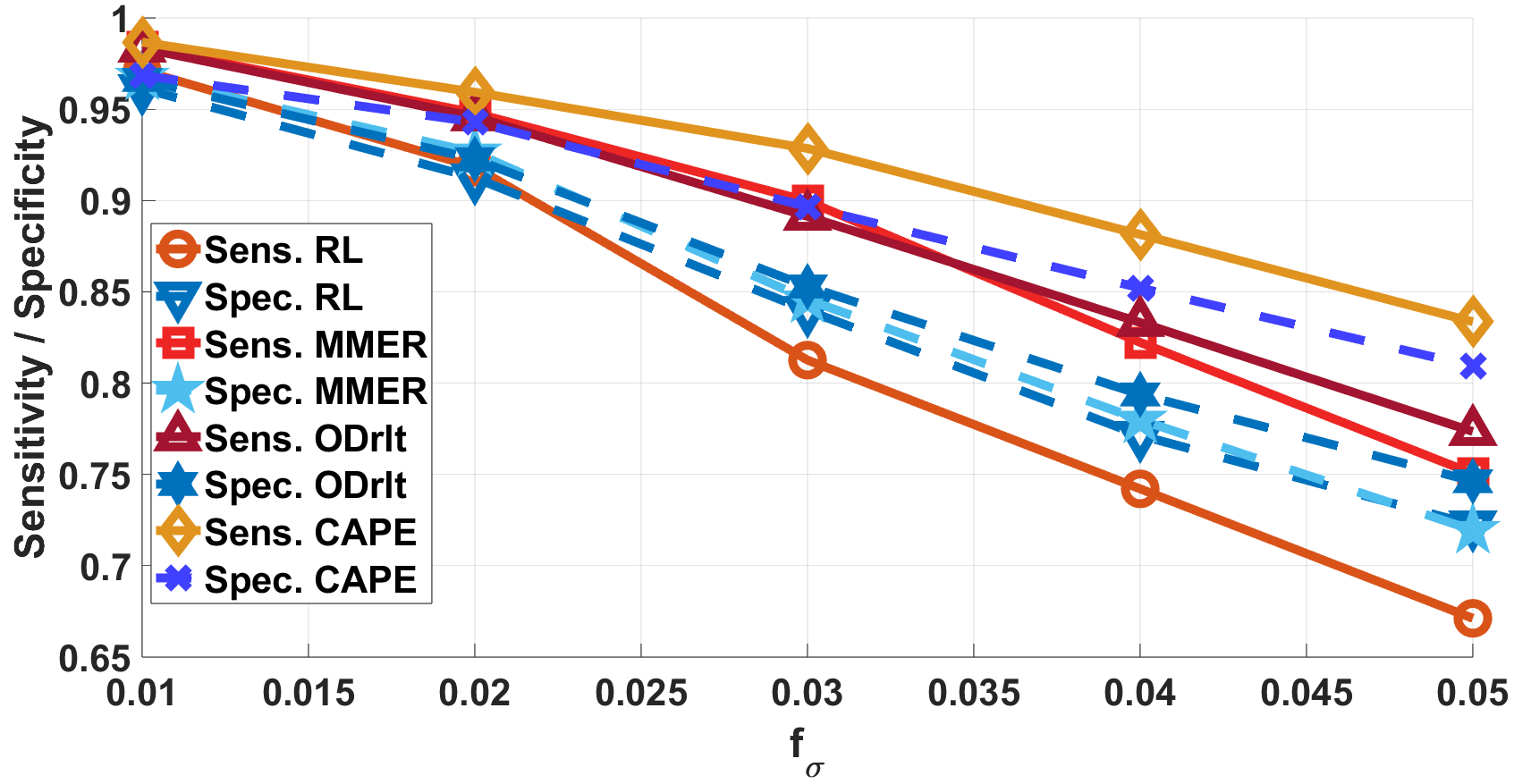}
    \includegraphics[scale=0.2]{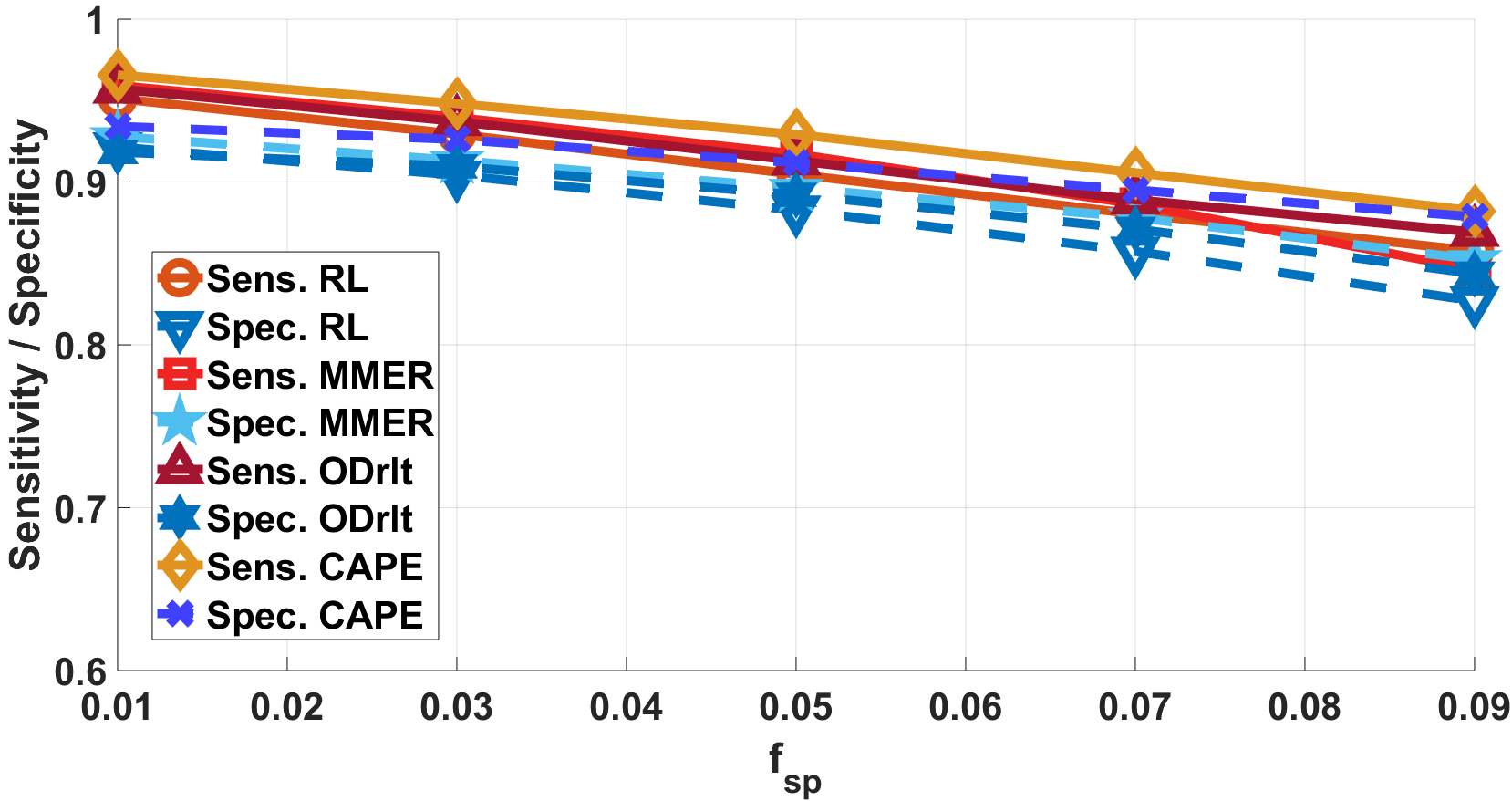}
    \caption{Average Sensitivity and Specificity (over 100 independent noise runs keeping $\boldsymbol{\beta}^*$, $\boldsymbol{A}$, and $\boldsymbol{\delta}^*$ fixed) for detecting defective samples using
    \textsc{Cape},    Robust \textsc{Lasso} (\textsc{Rl}), \textsc{Mmer}, and \textsc{Odrlt} in the presence of \textsf{ASM} MMes.
    The design matrix $\boldsymbol{A}$ has rows i.i.d. from Centered Bernoulli($0.3$). 
    Left to right, top to bottom: results for experiments (\textsf{EA}), (\textsf{EB}), (\textsf{EC}), (\textsf{ED}) defined in the main paper.}  
    \label{fig:Sens_spec_asm_0.3}    
\end{figure*}

\begin{figure*}
   \centering
    \includegraphics[scale=0.2]{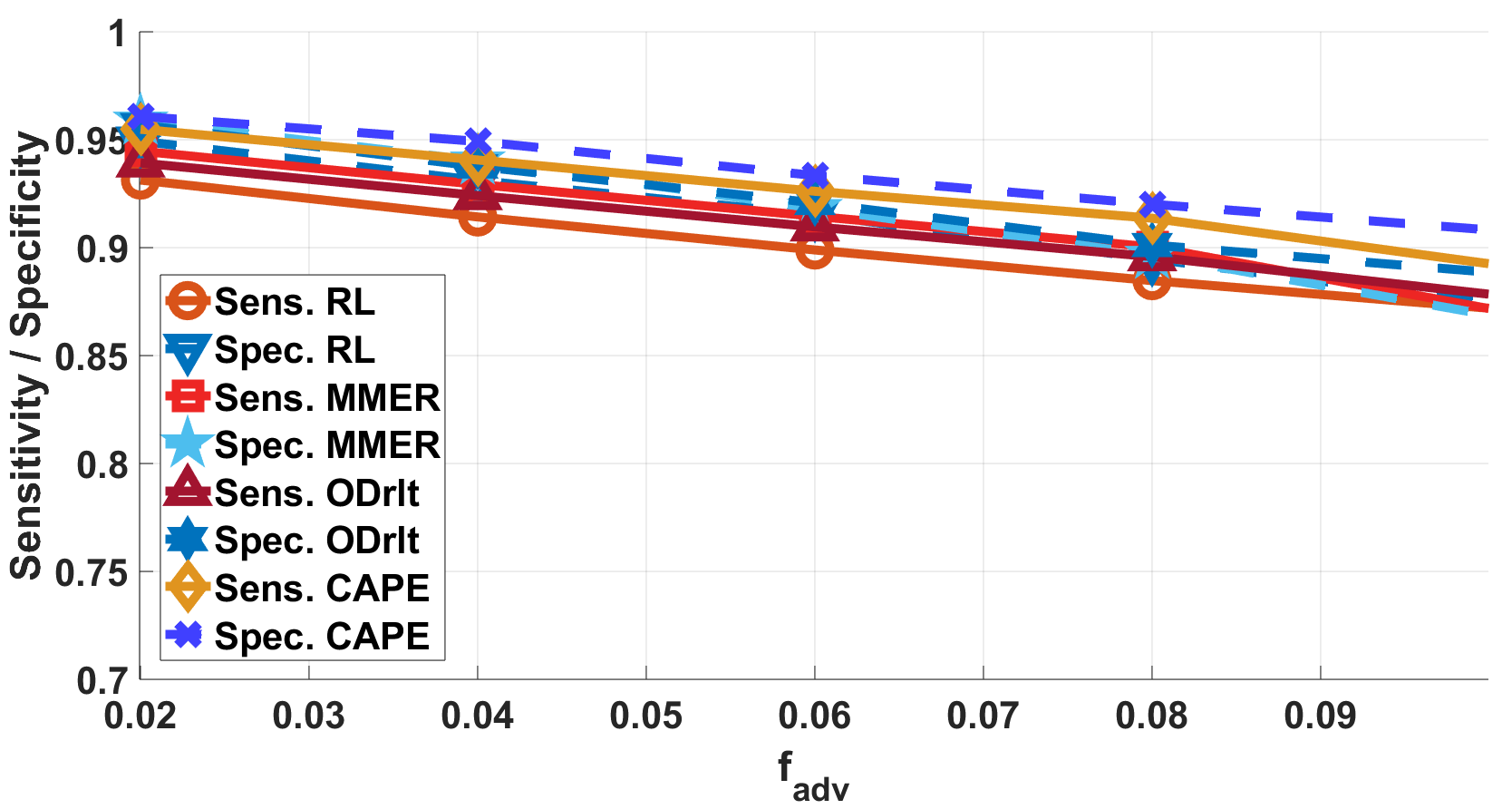}
    \includegraphics[scale=0.2]{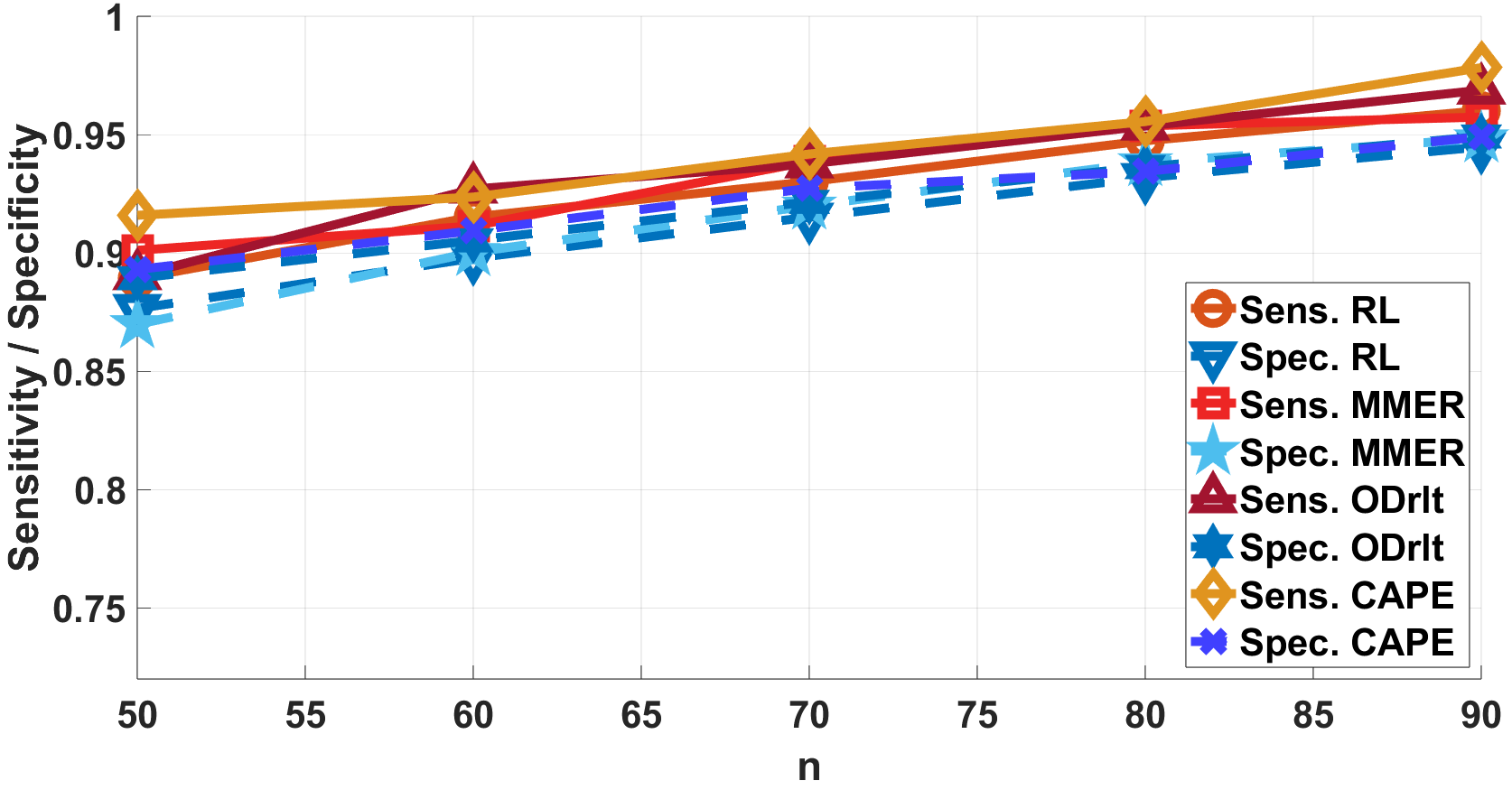}\\
    \includegraphics[scale=0.2]{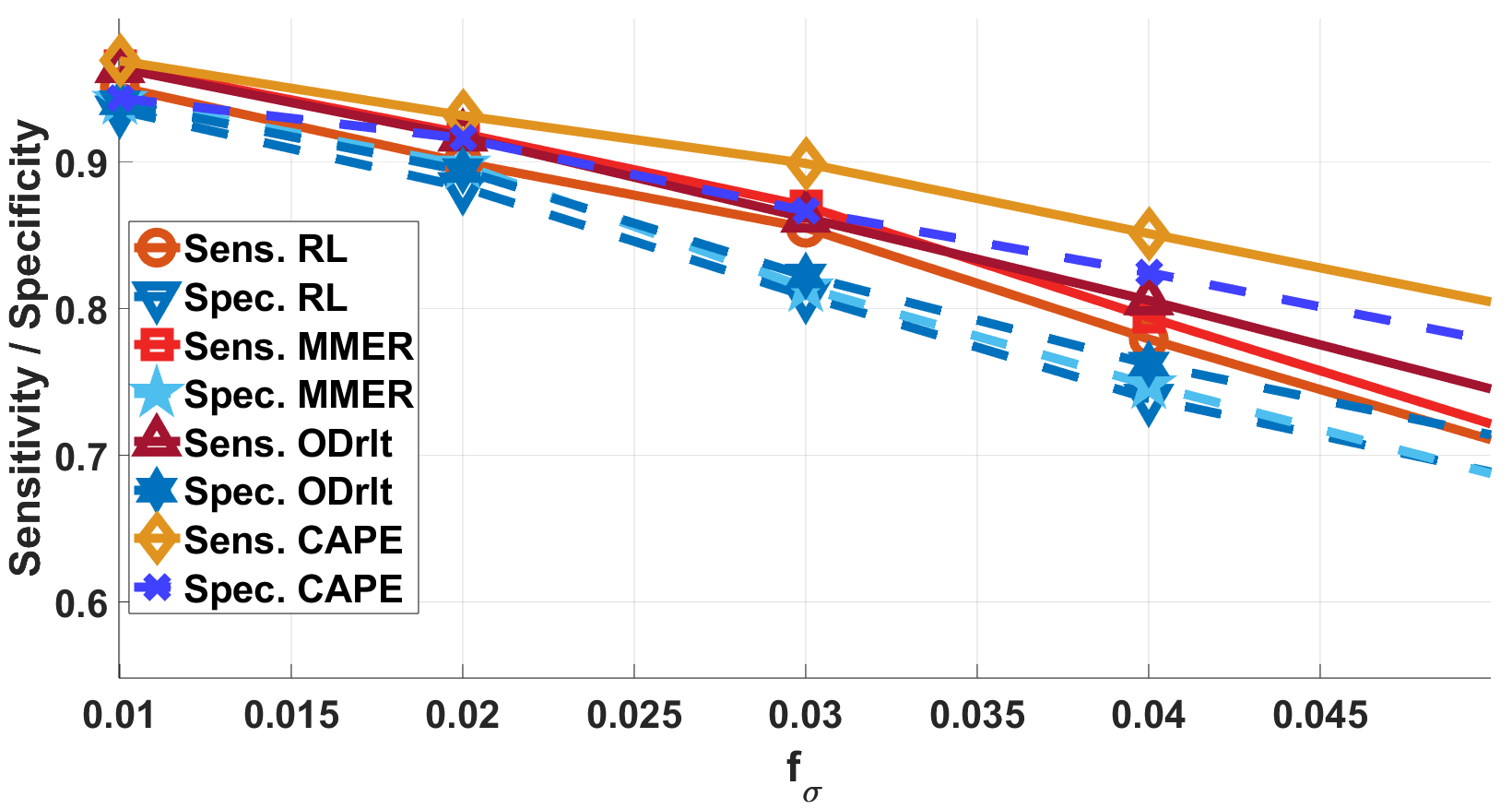}
    \includegraphics[scale=0.2]{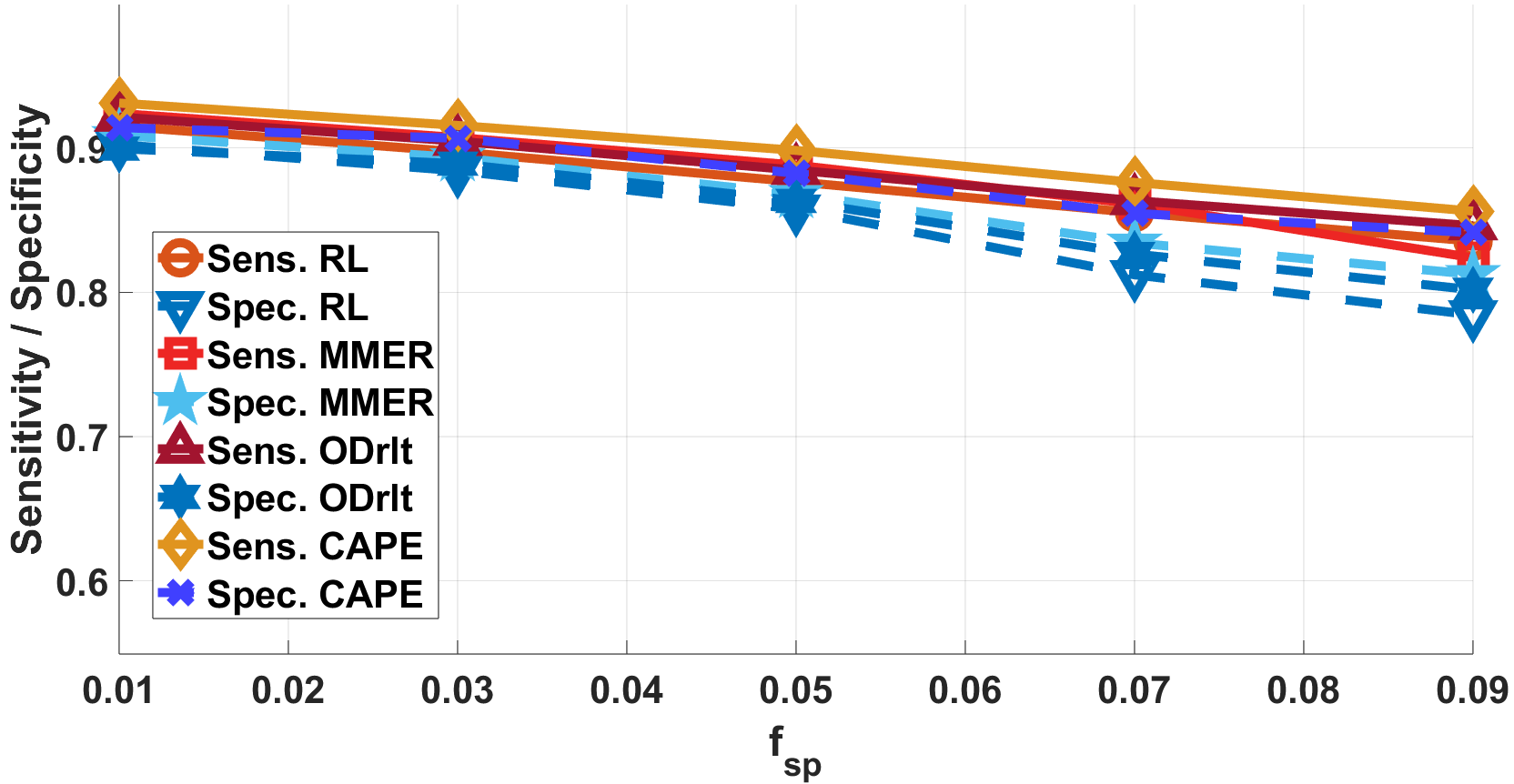}
    \caption{Average Sensitivity and Specificity (over 100 independent noise runs keeping $\boldsymbol{\beta}^*$, $\boldsymbol{A}$, and $\boldsymbol{\delta}^*$ fixed) for detecting defective samples using
    \textsc{Cape},    Robust \textsc{Lasso} (\textsc{Rl}), \textsc{Mmer}, and \textsc{Odrlt} in the presence of \textsf{ASM} MMEs.
    The design matrix $\boldsymbol{A}$ has rows i.i.d. from $CB(0.5)$. Left to right, top to bottom: results for experiments (\textsf{EA}), (\textsf{EB}), (\textsf{EC}), (\textsf{ED}) defined in the main paper.}  
    \label{fig:Sens_spec_asm_0.5}    
\end{figure*}

\begin{figure*}
   \centering
    \includegraphics[scale=0.19]{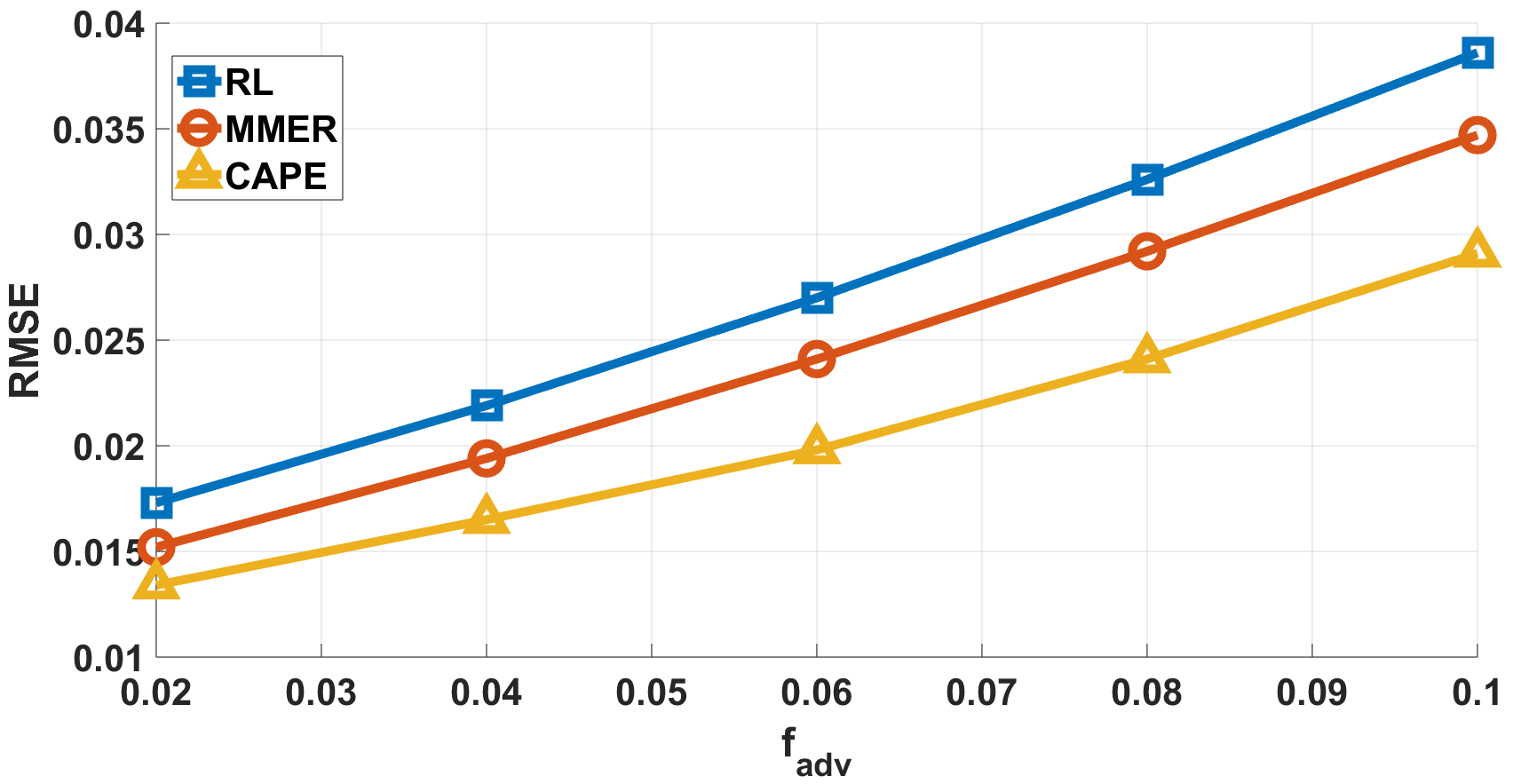}
    \includegraphics[scale=0.19]{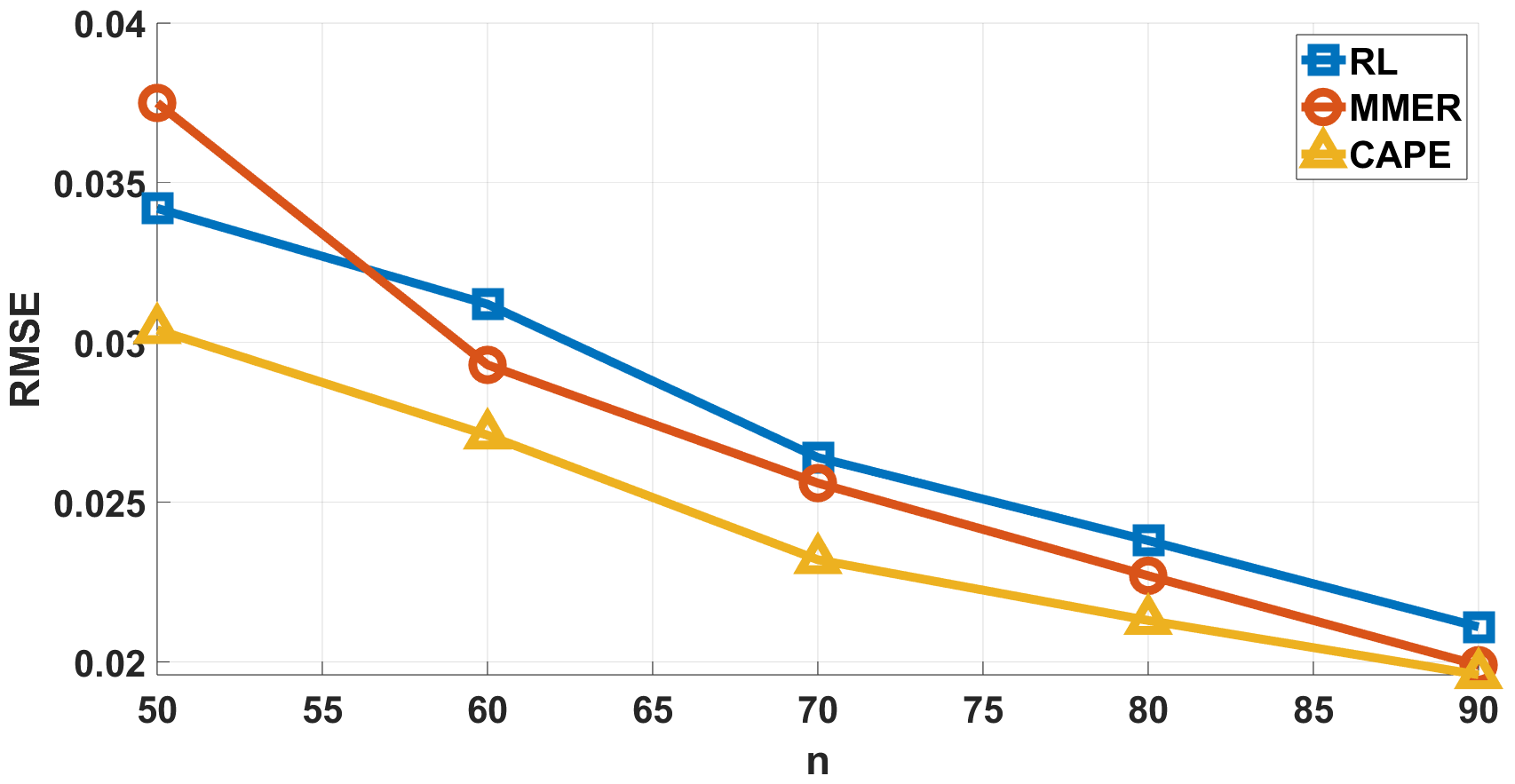}\\
    \includegraphics[scale=0.2]{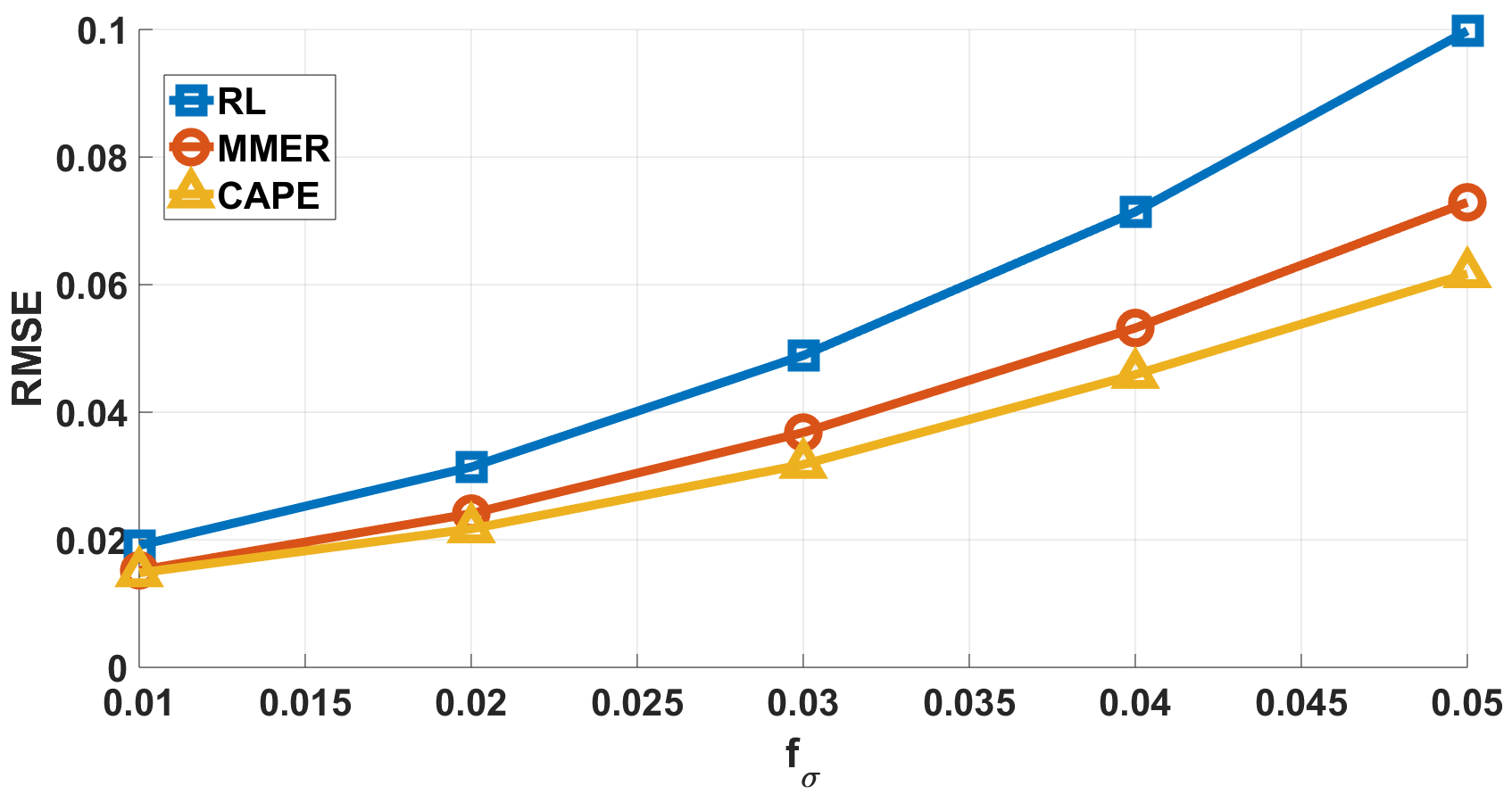}
    \includegraphics[scale=0.2]{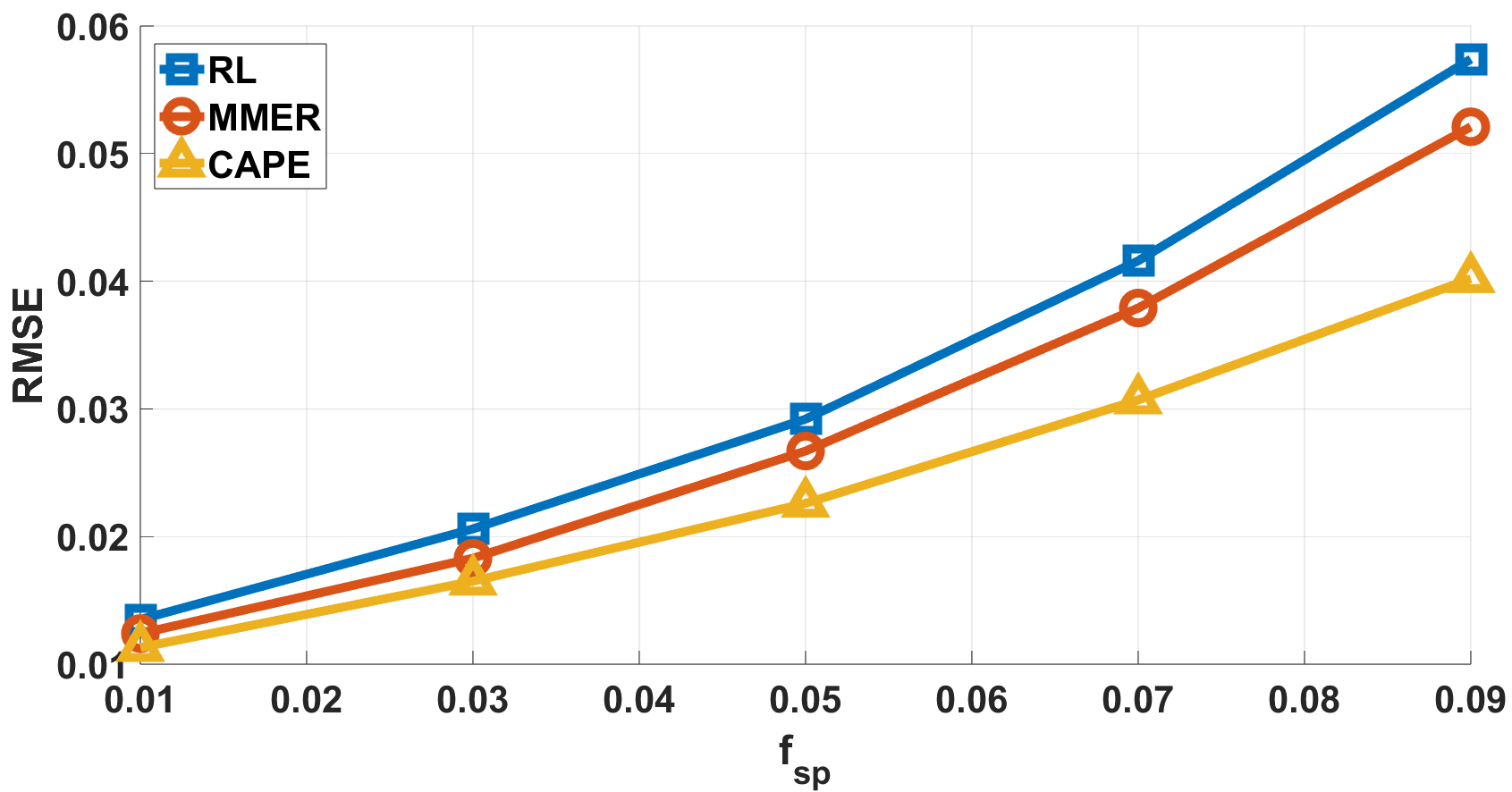}
    \caption{Average RRMSE (over 100 independent noise runs keeping $\boldsymbol{\beta}^*$, $\boldsymbol{A}$, and $\boldsymbol{\delta}^*$ fixed) for detecting defective samples using
    \textsc{Cape},    Robust \textsc{Lasso} (\textsc{Rl}), and \textsc{Mmer} in the presence of \textsf{ASM} MMes.
    The design matrix $\boldsymbol{A}$ has rows i.i.d. from Centered Bernoulli($0.1$). 
    Left to right, top to bottom: results for experiments (\textsf{EA}), (\textsf{EB}), (\textsf{EC}), (\textsf{ED}) defined in the main paper.}  
    \label{fig:Rmse_asm_0.1}    
\end{figure*}

\begin{figure*}
   \centering
    \includegraphics[scale=0.2]{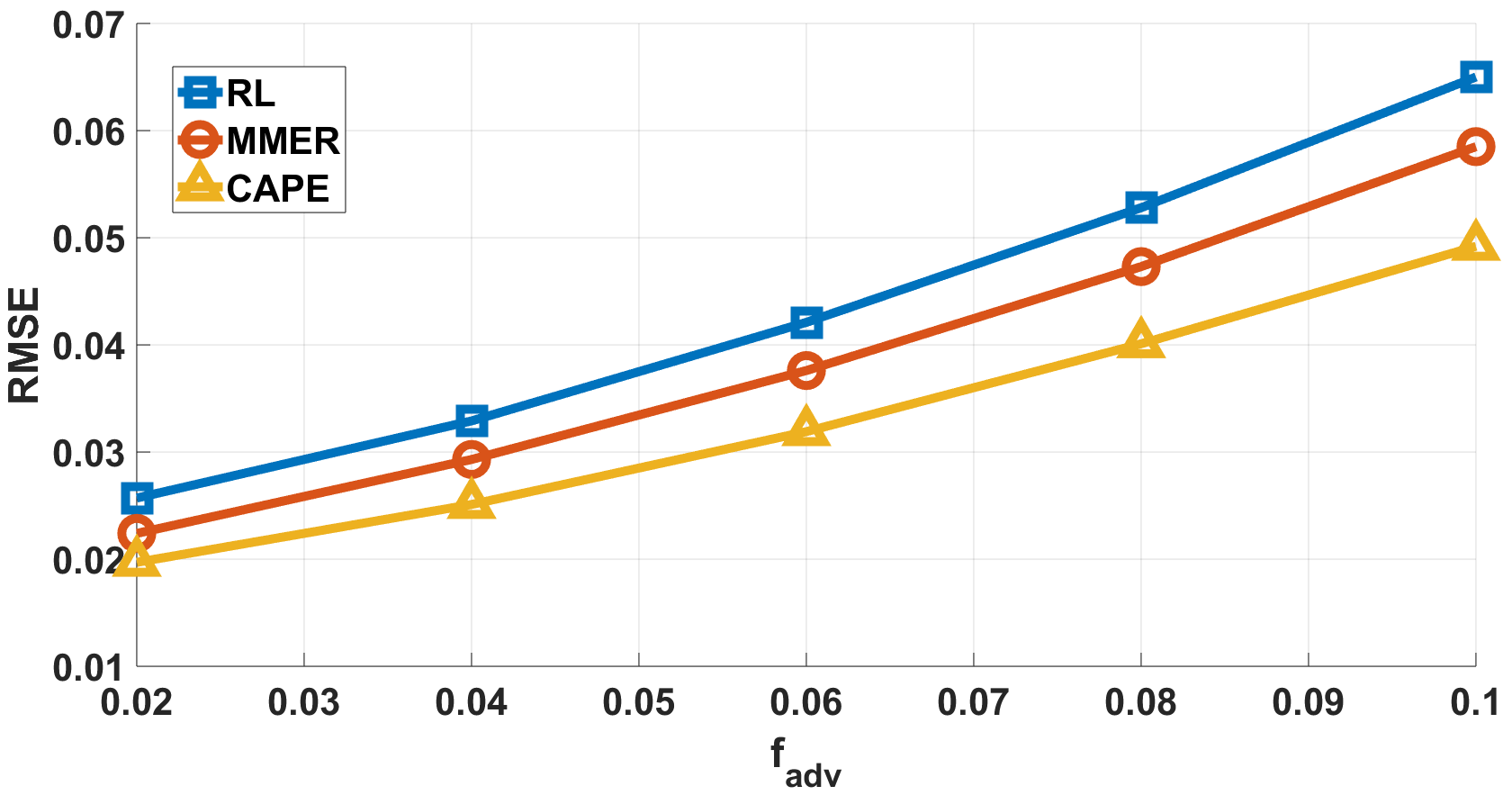}
    \includegraphics[scale=0.2]{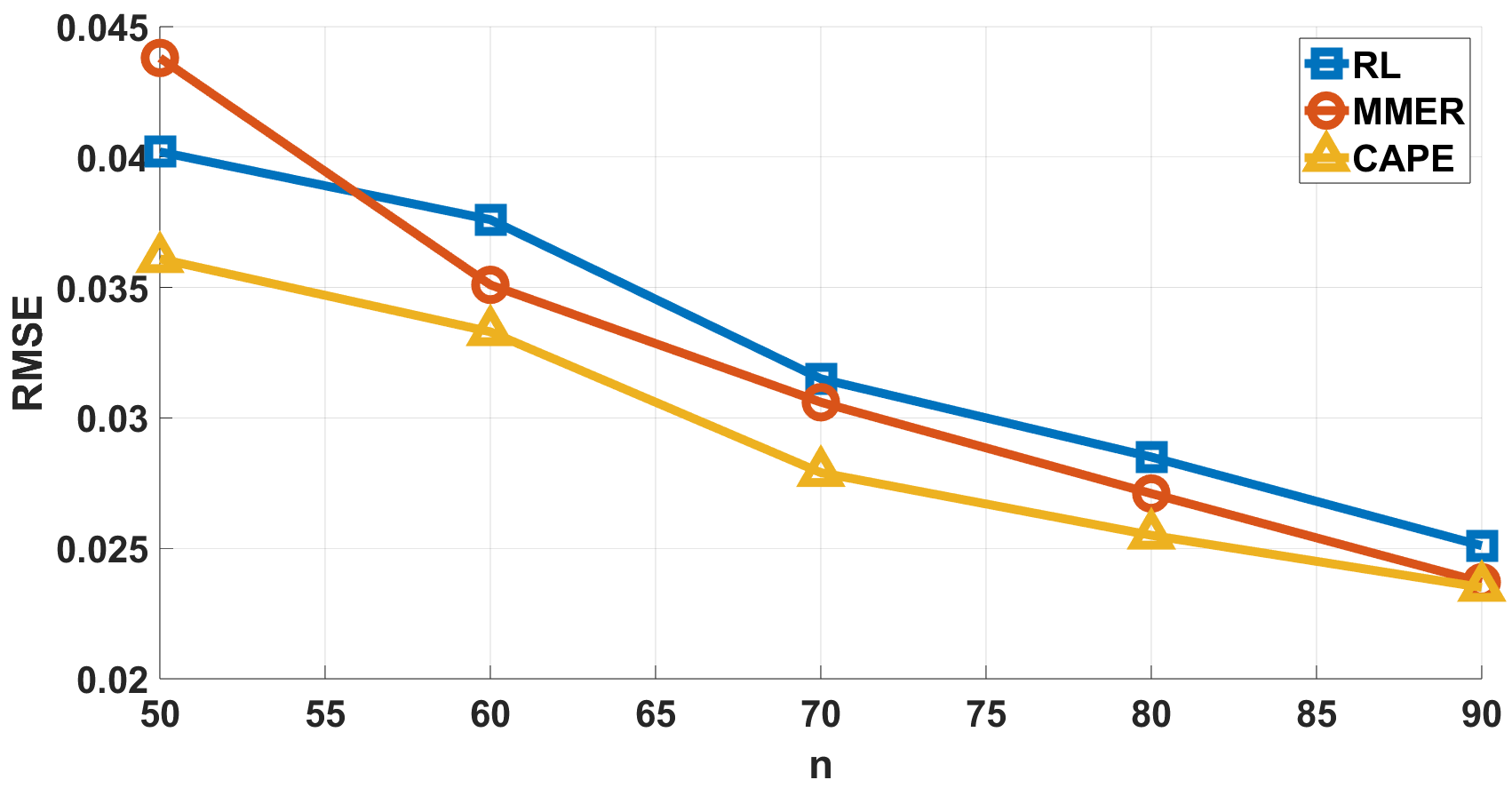}\\
    \includegraphics[scale=0.2]{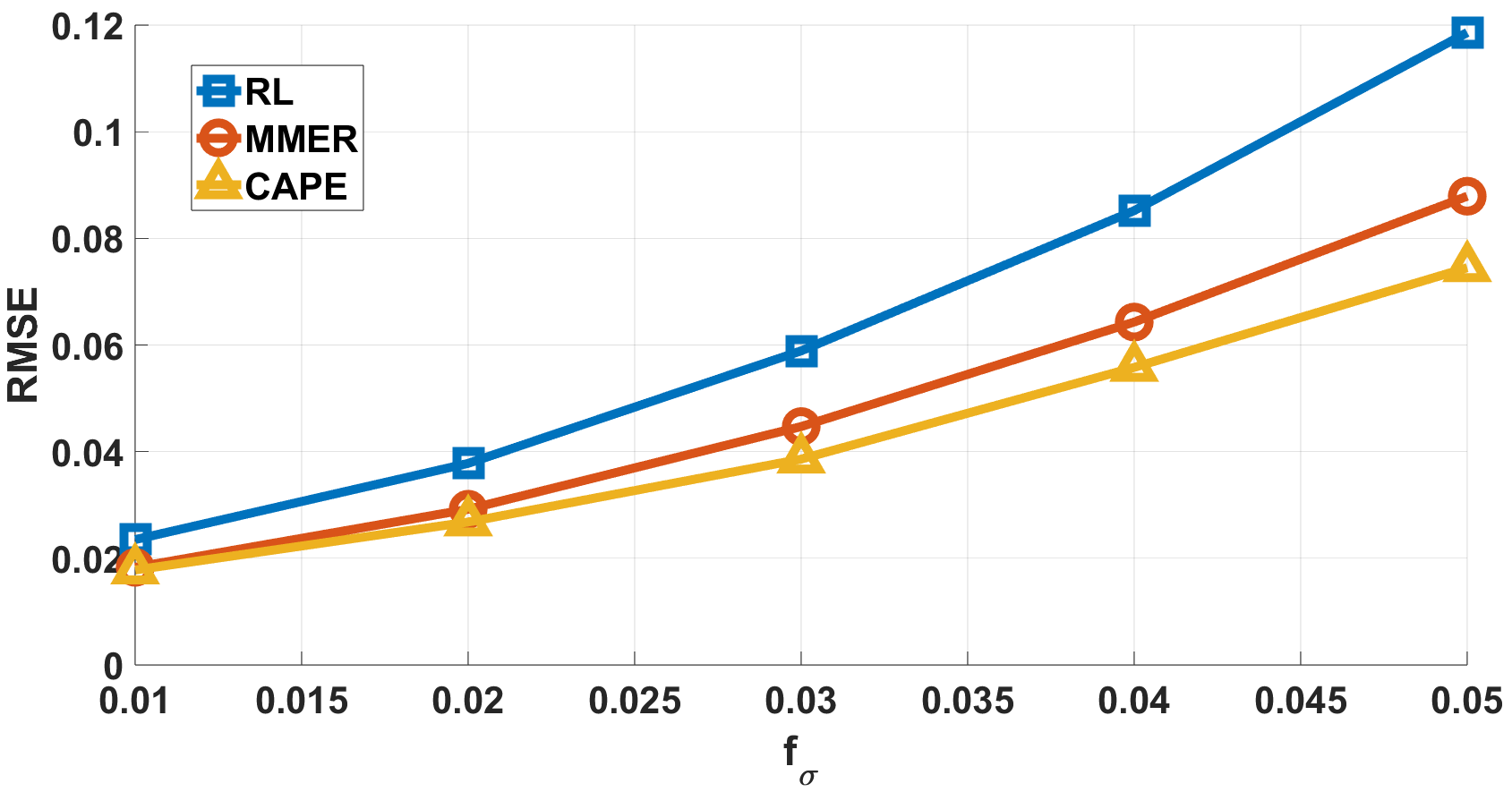}
    \includegraphics[scale=0.2]{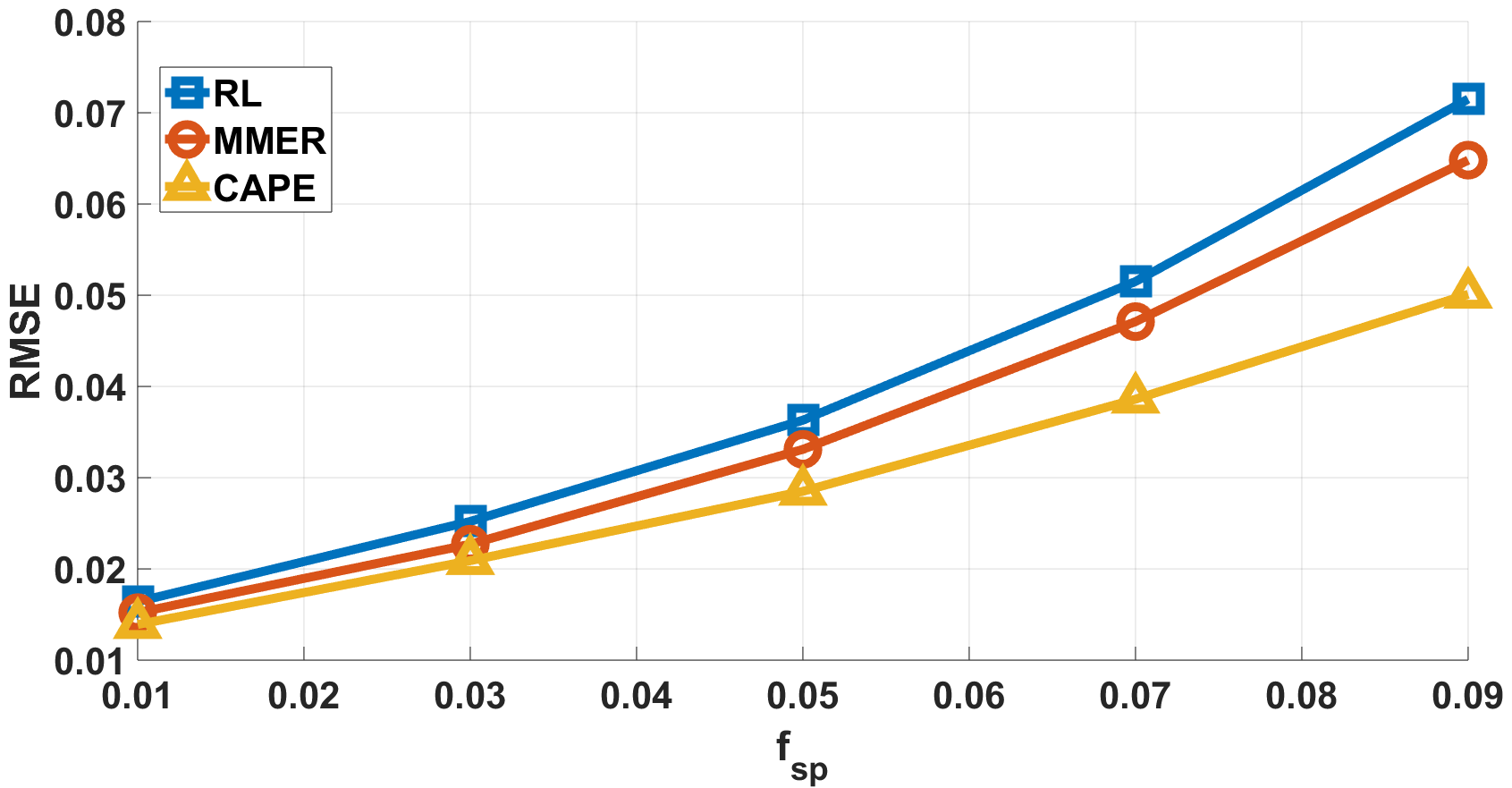}
    \caption{Average RRMSE for $\boldsymbol{\beta}^*$ using
    \textsc{Cape},    Robust \textsc{Lasso} (\textsc{Rl}), and \textsc{Mmer} under \textsf{ASM} MMes.
    The design matrix $\boldsymbol{A}$ has rows i.i.d. from Centered Bernoulli($0.3$). 
    Left to right, top to bottom: results for experiments (\textsf{EA}), (\textsf{EB}), (\textsf{EC}), (\textsf{ED}) defined in the main paper.}  
    \label{fig:Rmse_asm_0.3}    
\end{figure*}

\begin{figure*}
   \centering
    \includegraphics[scale=0.19]{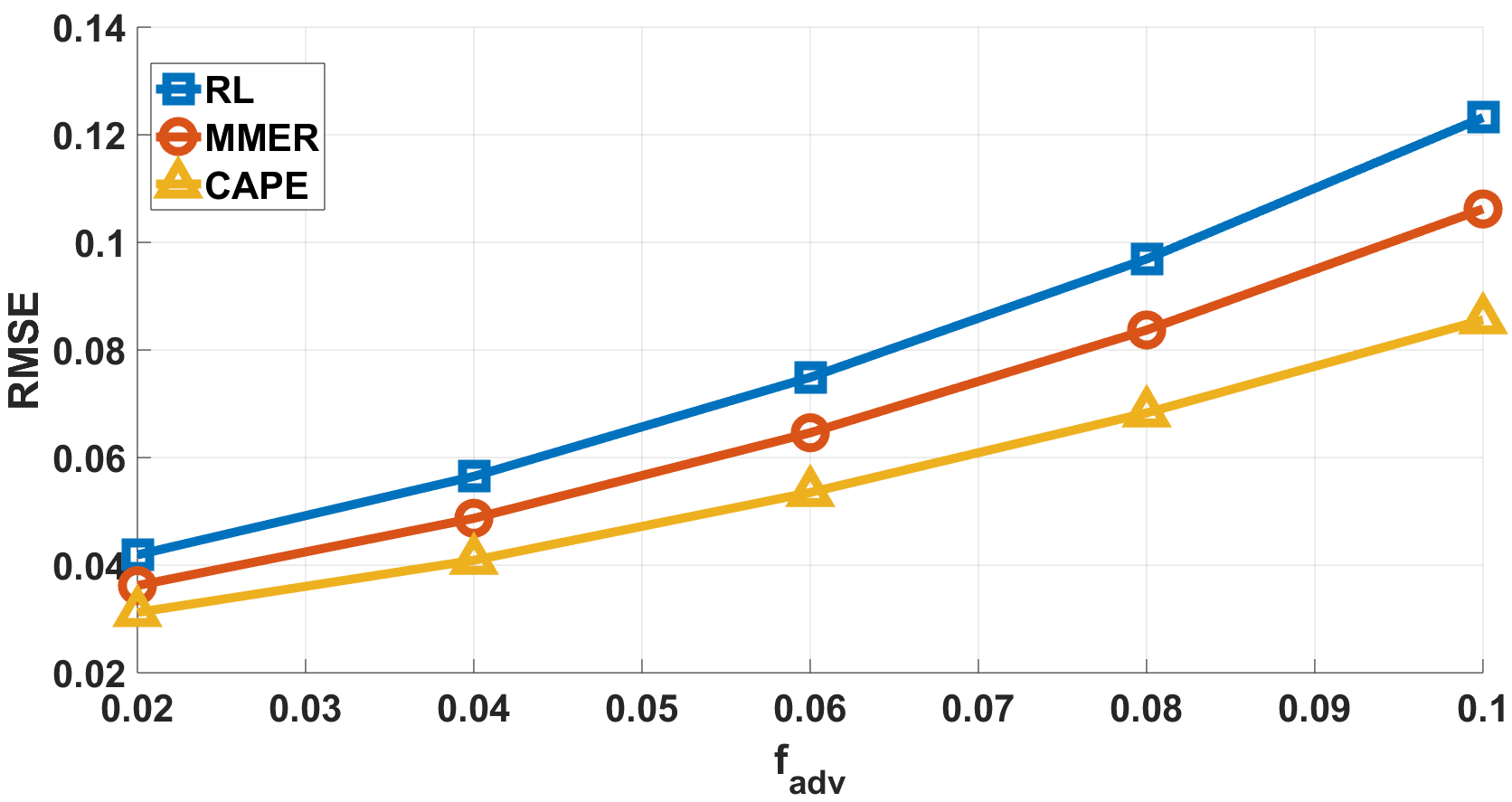}
    \includegraphics[scale=0.19]{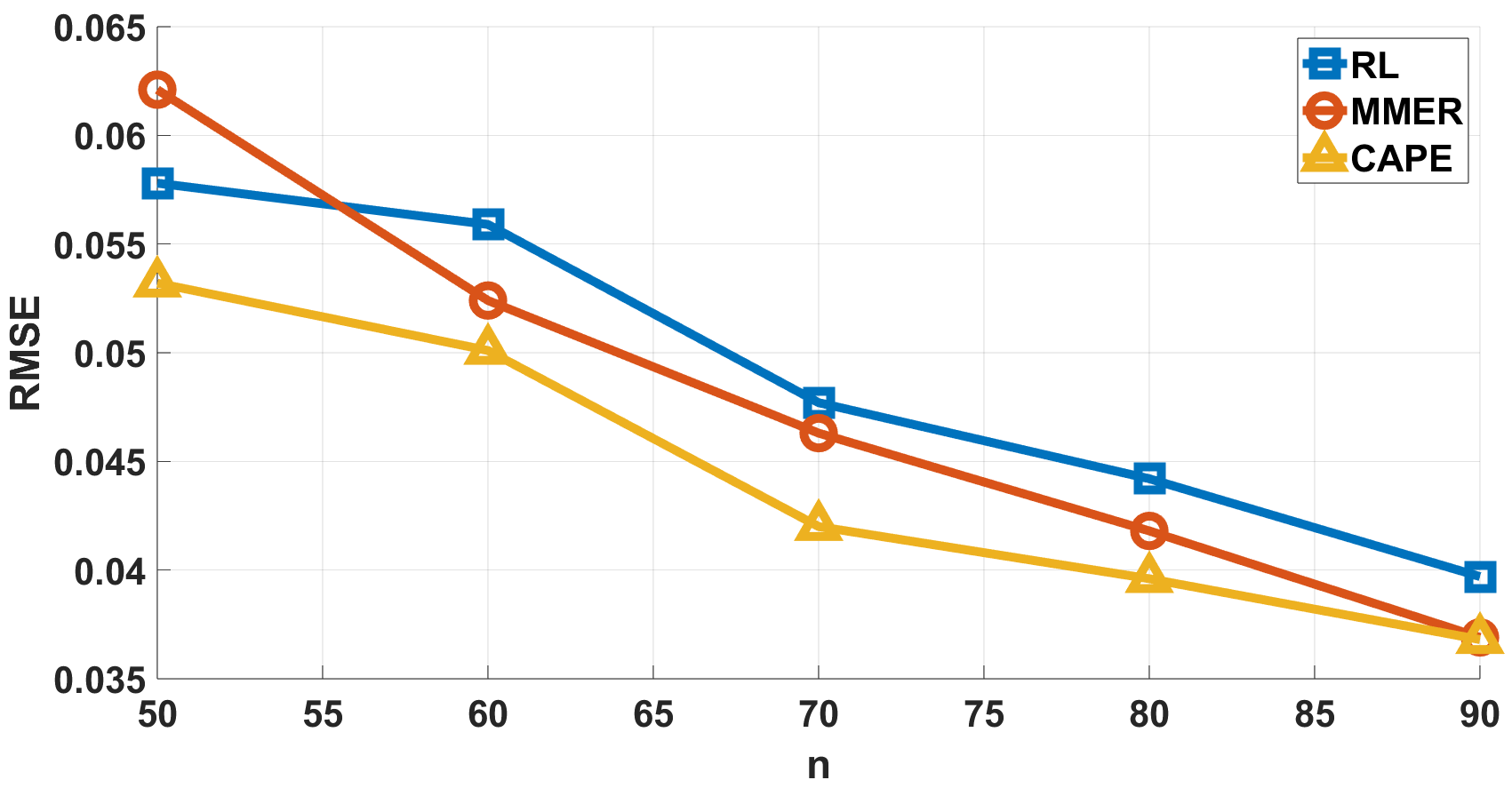}\\
    \includegraphics[scale=0.2]{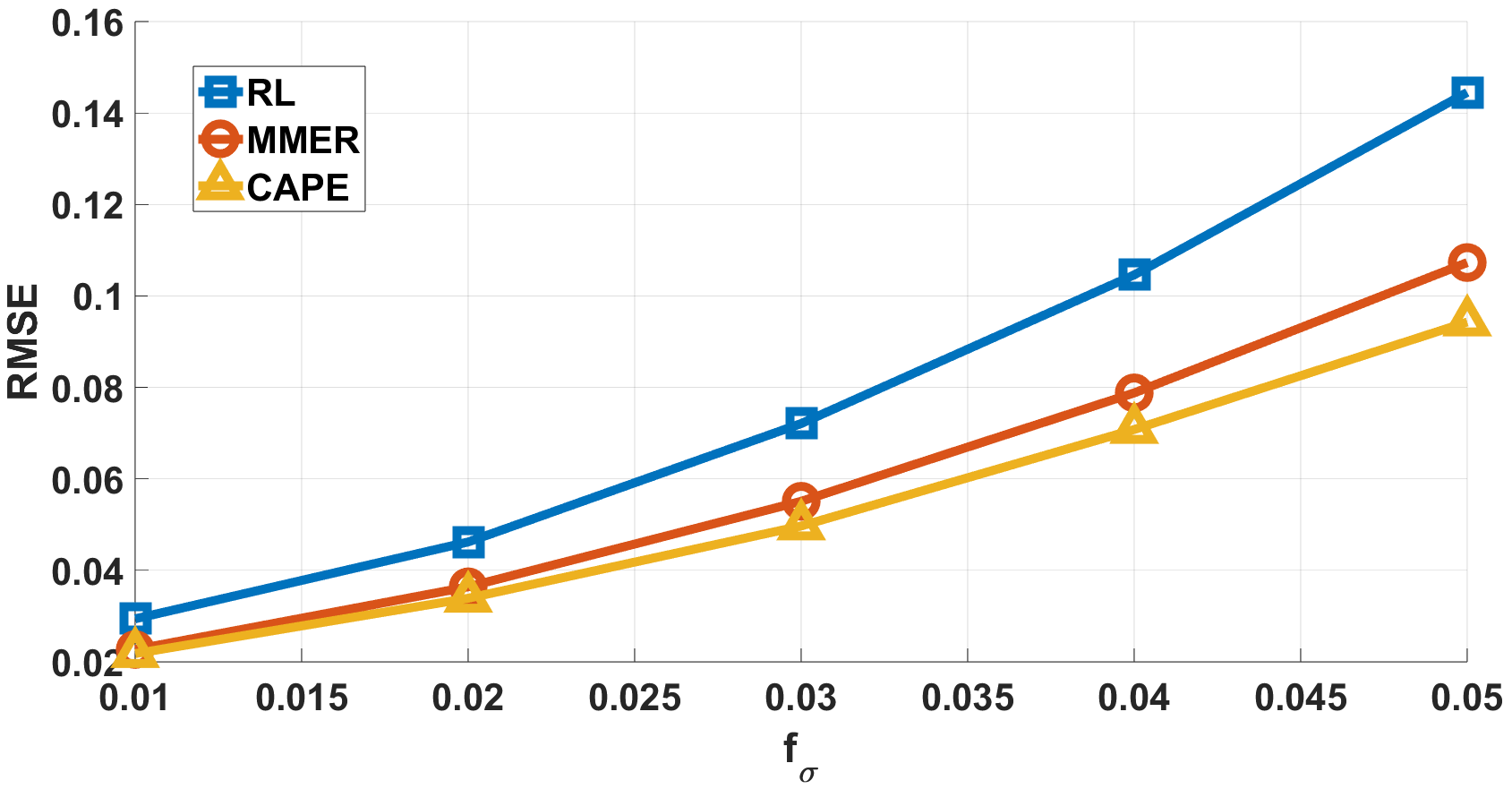}
    \includegraphics[scale=0.2]{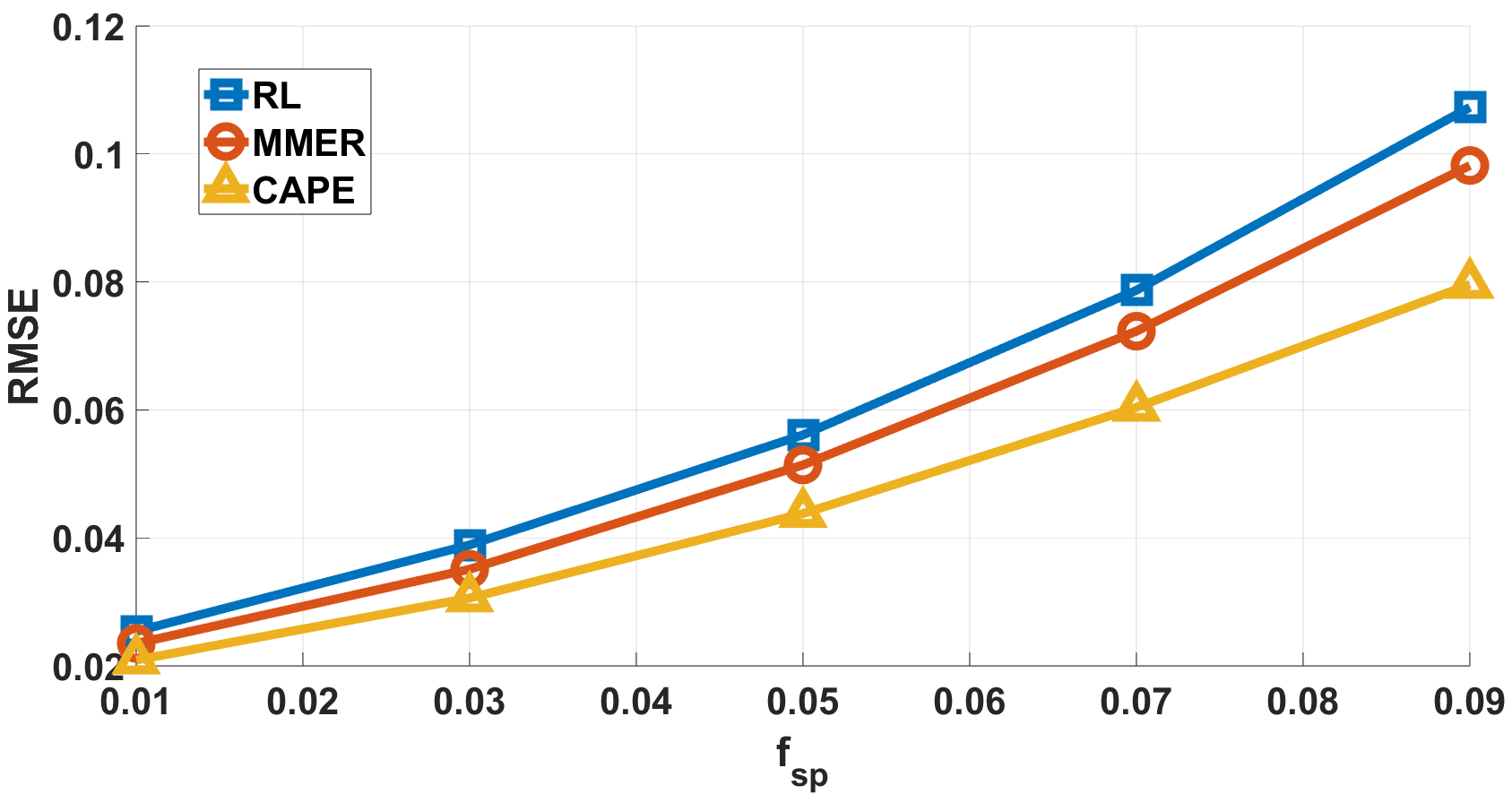}
    \caption{Average RRMSE for $\boldsymbol{\beta}^*$ using
    \textsc{Cape},    Robust \textsc{Lasso} (\textsc{Rl}), and \textsc{Mmer} under \textsf{ASM} MMes.
    The design matrix $\boldsymbol{A}$ has rows i.i.d. from Centered Bernoulli($0.5$). 
    Left to right, top to bottom: results for experiments (\textsf{EA}), (\textsf{EB}), (\textsf{EC}), (\textsf{ED}) defined in the main paper.}  
    \label{fig:Rmse_asm_0.5}    
\end{figure*}

\subsection{Correction of Probabilistic Errors in Group Testing \cite{Cheragchi2011}}
The work of \cite{Cheragchi2011} studies probabilistic and structured errors in the pooling matrix $\boldsymbol{B}$, where a 1-entry may flip to 0 with probability $\vartheta' = 1-\vartheta$, while 0-entries remain unchanged. They propose a combinatorial Distance Decoder (\textsc{Dist-D}) to recover the binary signal vector from binary measurements under such asymmetric errors. Viewing these flips as MMEs, we apply our correction algorithm \textsc{Cape} to correct them and compare its performance against \textsc{Dist-D}.

For a proper comparison between \textsc{Cape} and \textsc{Dist-D} \cite{Cheragchi2011}, we first binarize the real-valued signal $\boldsymbol{\beta}^*$ by setting $x^*_j = 1$ if $\beta^*_j > 0$ and $x^*_j = 0$ otherwise. The method in \cite{Cheragchi2011} requires binary measurements, so we generate $\boldsymbol{u} := \boldsymbol{B}\boldsymbol{x}^*$ using AND/OR operations, with $\boldsymbol{B}$ drawn from a Bernoulli$(0.5)$ model. For \textsc{Cape}, we construct $\boldsymbol{z} = \boldsymbol{B}\boldsymbol{\beta}^* + \boldsymbol{\eta}$, where we add Gaussian noise $\boldsymbol{\eta} \sim \mathcal{N}(0,\sigma^2\boldsymbol{I})$ with $\sigma$ determined by $f_\sigma=0.01$. Thus, \textsc{Dist-D} estimates $\boldsymbol{x}^*$ from $(\boldsymbol{u},\boldsymbol{B})$, while \textsc{Cape} estimates $\boldsymbol{\beta}^*$ from $(\boldsymbol{z},\boldsymbol{B})$.

To evaluate \textsc{Cape} against \textsc{Dist-D}, we compared their Sensitivity and Specificity under two experimental settings:
(\textit{i}) varying the flip-probability parameter $\vartheta$ over the range $[0.75:0.05:0.95]$, and
(\textit{ii}) varying the number of measurements $n$ from $250$ to $450$ in increments of $50$.
Throughout these experiments, we fixed $p=500$ and $s=10$. For \textsc{Dist-D}, the error parameter $e$ was selected according to the formula in Sec.~VI of \cite{Cheragchi2011}. It is important to note that the model in \cite{Cheragchi2011} incorporates no measurement noise in $\boldsymbol{u}$ beyond the bit-flip errors induced by the corruption of $\boldsymbol{B}$.

For the centering step required in \textsc{Cape}, we set $y_i = z_i - \tilde{z}$ for all $i \in [n]$, where $\tilde{z} := \boldsymbol{1}^{\top}\boldsymbol{\beta}^*$ is treated as a noiseless reference measurement. To approximate this quantity, we compute $\boldsymbol{1}^{\top}\boldsymbol{\beta}^*$ fifty times and average the results, which sufficiently suppresses additive noise to regard $\tilde{z}$ as noise-free. Rather than reducing from $2n$ to $n$ measurements, we adopt this centering strategy to demonstrate that mild correlated noise does not impair the correction procedure, and to avoid discarding half of the measurements solely for centering.

\begin{table}[ht]
\centering
\begin{tabular}{|c|c|c|c|c|}
\hline
\rowcolor{Gray}
$\vartheta$ & Sens.\textsc{Cape} & Sens.\textsc{Dist-D} & Spec.\textsc{Cape} & Spec.\textsc{Dist-D}  \\
\hline
0.75 & 0.924 & 0.901 & 0.975 & 0.972 \\
0.80 & 0.967 & 0.941 & 0.986 & 0.984 \\
0.85 & 0.991 & 0.984 & 0.995 & 0.995 \\
0.90 & 0.998 & 0.992 & 0.998 & 0.997 \\
0.95 & 1     & 0.998 & 1 & 1     \\
\hline
\end{tabular}
\caption{Sensitivity (Sens.) and Specificity (Spec.) values for \textsc{Cape} and \textsc{Dist-D} across different $\vartheta$ values where $\vartheta' := 1-\vartheta$ is the probability of a bitflip in $\boldsymbol{B}$. The other parameters are fixed at $p=500,n=400,s=10$.}
\label{tab:sens_spec_theta}
\end{table}

\begin{table}[ht]
\centering
\begin{tabular}{|c|c|c|c|c|}
\hline
\rowcolor{Gray}
$n$ & Sens.\textsc{Cape} & Sens.\textsc{Dist-D} & Spec.\textsc{Cape} & Spec.\textsc{Dist-D}  \\
\hline
250 & 0.854 & 0.791 & 0.902 & 0.891 \\
300 & 0.912 & 0.875 & 0.964 & 0.955 \\
350 & 0.975 & 0.936 & 0.992 & 0.986 \\
400 & 0.998 & 0.989 & 1 & 1     \\
450 & 1     & 1     & 1     & 1     \\
\hline
\end{tabular}
\caption{Sensitivity and Specificity values for \textsc{Cape} and \textsc{Dist-D} across different $n$ values. The other parameters are fixed at $p=500,s=10,\vartheta=0.9$.}
\label{tab:sens_spec_n}
\end{table}

In Tables~\ref{tab:sens_spec_theta} and \ref{tab:sens_spec_n}, we see that the correction algorithm \textsc{Cape} outperforms \textsc{Dist-D} for varying $\vartheta$ and $n$, despite the presence of some additive noise in $\boldsymbol{z}$ for \textsc{Cape}. 

\bibliographystyle{plain}
\bibliography{refs}